\documentclass{lmcs} 
\pdfoutput=1
\usepackage[utf8]{inputenc}

\usepackage{lastpage}
\lmcsdoi{20}{3}{24}
\lmcsheading{}{\pageref{LastPage}}{}{}%
{Jul.~13,~2021}{Sep.~18,~2024}{}


\usepackage{stmaryrd} 
\usepackage{amssymb}
\usepackage{hyperref}
\usepackage{Guyboxes}
\usepackage{mathpartir}
\usepackage[all]{xy}
\usepackage{color}
\usepackage{cmll}
\usepackage{thm-restate}
\usepackage{mathrsfs} 
\usepackage{mathabx}  
\usepackage{tikz}     

\newtheorem*{question}{Question}

\newcommand{\IPA}{\mathsf{IA}_{\sslash}}
\newcommand{\IA}{\mathsf{IA}}
\newcommand{\PCF}{\mathsf{PCF}}
\newcommand{\PCFpar}{\mathsf{PCF}_{\sslash}}

\newcommand{\tx}{\mathbb{X}}
\newcommand{\ty}{\mathbb{Y}}
\newcommand{\tunit}{\mathbb{U}}
\newcommand{\tbool}{\mathbb{B}}
\newcommand{\tnat}{\mathbb{N}}
\newcommand{\var}{\mathbf{ref}}
\newcommand{\sem}{\mathbf{sem}}

\newcommand{\tskip}{\mathbf{skip}}
\newcommand{\ttrue}{\mathbf{t\!t}}
\newcommand{\tfalse}{\mathbf{f\!f}}
\newcommand{\ite}[3]{\mathbf{if}\,#1\,#2\,#3}
\newcommand{\tsucc}{\mathbf{succ}}
\newcommand{\pred}{\mathbf{pred}}
\newcommand{\iszero}{\mathbf{iszero}}
\newcommand{\tlet}[3]{\mathbf{let}~ #1 =#2\,\mathbf{in}~#3}
\newcommand{\Y}{\mathcal{Y}}
\newcommand{\newref}{\mathbf{newref}}
\newcommand{\tin}{\mathbf{in}}
\newcommand{\newsem}{\mathbf{newsem}}
\newcommand{\grab}{\mathbf{grab}}
\newcommand{\release}{\mathbf{release}}
\newcommand{\mkvar}{\mathbf{mkvar}}
\newcommand{\mksem}{\mathbf{mksem}}
\newcommand{\plet}[5]{\mathbf{let}~\left(\begin{array}{rcl} #1 &=&#2\\#3 &=&#4 \end{array}\right)\,\mathbf{in}~#5}
\newcommand{\tpred}{\mathbf{pred}}

\newcommand{\pto}{\rightharpoonup}
\newcommand{\tuple}[1]{\langle #1 \rangle}
\newcommand{\dom}{\mathsf{dom}}
\newcommand{\eval}{\Downarrow}
\renewcommand{\obs}{\sim}
\renewcommand{\L}{\mathcal{L}}
\newcommand{\intr}[1]{\llbracket #1 \rrbracket}
\newcommand{\tensor}{\otimes}
\newcommand{\lin}{\multimap}
\newcommand{\ees}{\text{\boldmath$1$}}
\newcommand{\Alt}{\text{$\downuparrows$-$\mathsf{Plays}$}}
\newcommand{\prefix}{\sqsubseteq}
\newcommand{\bij}{\simeq}
\newcommand{\sym}{\cong}
\renewcommand{\dom}{\mathsf{dom}}
\newcommand{\cod}{\mathsf{cod}}
\newcommand{\emptyar}{\text{\boldmath$1$}}
\newcommand{\iso}{\cong}
\newcommand{\simstrat}{\approx}
\newcommand{\ot}{\leftarrow}
\newcommand{\restrict}{\upharpoonright}
\newcommand{\inter}{\circledast}
\renewcommand{\cc}{\mathrm{\,c\!\!\!\!c\,}}
\newcommand{\evm}{\mathsf{ev}}
\newcommand{\der}{\mathsf{der}}
\newcommand{\dig}{\mathsf{dig}}
\newcommand{\mon}{\mathsf{mon}}
\newcommand{\C}{\mathcal{C}}
\renewcommand{\div}{\Uparrow}
\newcommand{\PAlt}{\text{$\mathscr{P}$-$\Alt$}}
\newcommand{\aintr}[1]{[#1]}
\newcommand{\lbl}{\mathrm{lbl}}
\newcommand\mybox[2][]{\tikz[overlay]\node[fill=blue!20,inner sep=2pt,
anchor=text, rectangle, rounded corners=1mm,#1] {#2};\phantom{#2}}
\newcommand{\ensquare}[1]{\mybox[fill=blue!20]{\ensuremath{#1}}}
\newcommand{\Qu}{\mathcal{Q}}
\newcommand{\An}{\mathcal{A}}
\newcommand{\mwrite}[1]{\mathbf{w}#1}
\newcommand{\ok}{\checkmark}
\renewcommand{\read}{\mathbf{r}}
\newcommand{\mgrab}{\mathbf{g}}
\newcommand{\mrelease}{\mathbf{r\!l}}
\newcommand{\assign}{\mathsf{assign}}
\newcommand{\deref}{\mathsf{deref}}
\newcommand{\sgrab}{\mathsf{grab}}
\newcommand{\srelease}{\mathsf{release}}
\newcommand{\cell}{\mathsf{cell}}
\newcommand{\ssem}{\mathsf{lock}}
\newcommand{\Inn}{\mathrm{Inn}}
\newcommand{\comp}{\mathrm{comp}}
\newcommand{\NAlt}{\text{$\circlearrowleft$-$\mathsf{Plays}$}}
\newcommand{\ntilde}[1]{\mathscr{S}_-(#1)}
\newcommand{\ptilde}[1]{\mathscr{S}_+(#1)}
\newcommand{\bsigma}{{\boldsymbol{\sigma}}}
\newcommand{\pr}{\partial}
\newcommand{\btau}{{\boldsymbol{\tau}}}
\newcommand{\NAltStrat}{\text{$\circlearrowleft$-$\mathsf{Unf}$}}
\newcommand{\bmu}{{\boldsymbol{\mu}}}
\newcommand{\bnu}{{\boldsymbol{\nu}}}
\newcommand{\labr}{\mathsf{r}}
\newcommand{\labl}{\mathsf{l}}
\newcommand{\confp}[1]{\mathscr{C}^+(#1)}
\newcommand{\tildep}[1]{\mathscr{S}^+(#1)}
\newcommand{\bcc}{\mathbf{\,c\!\!\!\!c\,}}
\newcommand{\balpha}{{\boldsymbol{\alpha}}}
\newcommand{\brho}{{\boldsymbol{\rho}}}
\newcommand{\blambda}{{\boldsymbol{\lambda}}}
\newcommand{\bpi}{{\boldsymbol{\pi}}}
\newcommand{\init}{\mathsf{init}}
\newcommand{\bder}{\mathbf{der}}
\newcommand{\bdig}{\mathbf{dig}}
\newcommand{\bmon}{\mathbf{mon}}
\newcommand{\cleq}{\trianglelefteq}
\newcommand{\D}{\mathcal{D}}
\newcommand{\force}{\mathsf{force}}
\newcommand{\preval}{\mathsf{eval}}
\newcommand{\bcell}{\mathbf{cell}}
\newcommand{\bsem}{\mathbf{lock}}
\newcommand{\id}{\mathrm{id}}
\newcommand{\rintr}[1]{\llparenthesis #1 \rrparenthesis}
\newcommand{\fv}{\mathsf{fv}}
\newcommand{\new}{\mathbf{new}}
\newcommand{\assrt}{\mathbf{assume}}
\newcommand{\tnot}{\mathbf{not}}
\newcommand{\PNAlt}{\text{$\mathscr{P}$-$\NAlt$}}
\newcommand{\Con}{\mathrm{Con}}
\newcommand{\gcc}{\mathsf{gcc}}
\newcommand{\gce}{\mathsf{gce}}
\newcommand{\just}{\mathsf{j}}
\newcommand{\oedges}{\text{$O$-$\mathsf{edges}$}}
\newcommand{\pedges}{\text{$P$-$\mathsf{edges}$}}
\newcommand{\x}{\mathsf{x}}
\newcommand{\y}{\mathsf{y}}
\newcommand{\coll}{{\smallint\!}}
\newcommand{\collo}{{\smallint\hspace{-7.4pt}\circ}}
\newcommand{\depth}{\mathsf{depth}}
\newcommand{\ind}{\mathrm{ind}}
\newcommand{\rf}[1]{\mathsf{mf}(#1)}
\newcommand{\longcov}[1]{{\stackrel{#1}{\mathrel-\joinrel\relbar\joinrel\subset\,}}}
\newcommand{\fo}{\mathsf{f\!o}}
\newcommand{\fst}{{\mathsf{{sh}}}}
\newcommand{\Q}{\mathscr{Q}}
\newcommand{\flow}{\mathsf{flow}}
\newcommand{\q}{q}
\newcommand{\inj}{\mathsf{inj}}
\newcommand{\recomp}{\mathsf{recomp}}
\newcommand{\bdelta}{{\boldsymbol{\delta}}}
\renewcommand{\v}{\mathsf{v}}
\newcommand{\confc}[1]{\mathscr{C}^{\mathsf{c}}(#1)}
\newcommand{\w}{\mathsf{w}}
\newcommand{\labm}{\mathsf{m}}
\newcommand{\scott}{\sqsubseteq}
\newcommand{\confv}[1]{\mathscr{C}^v(#1)}
\newcommand{\tildev}[1]{\mathscr{S}^v(#1)}
\newcommand{\proj}[1]{#1\!\!\downarrow}
\newcommand{\slice}{\mathsf{slice}}
\newcommand{\apred}{\mathsf{pred}}

\newcommand{\ev}[1]{|#1|}
\newcommand{\conflict}{\mathrel{\#}}
\newcommand{\pol}{\mathrm{pol}}
\newcommand{\imc}{\rightarrowtriangle}
\newcommand{\conf}[1]{\mathscr{C}(#1)}
\newcommand{\cov}{{{\,\mathrel-\joinrel\subset\,}}}
\renewcommand{\tilde}[1]{\mathscr{S}(#1)}

\renewcommand{\qu}{\mathbf{q}}
\newcommand{\run}{\qu}
\newcommand{\done}{\checkmark}

\definecolor{grey}{rgb}{.7,.7,.7}
\newcommand{\grey}[1]{{\color{grey}#1}}
\definecolor{Red}{cmyk}{0,1,2,0}
\newcommand{\red}[1]{{\color{Red}#1}}
\definecolor{Green}{rgb}{0,1,0}
\newcommand{\green}[1]{{\color{Green}#1}}
\definecolor{Cyan}{cmyk}{1,0,0,0}
\newcommand{\cyan}[1]{{\color{Cyan}#1}}
\definecolor{purple}{rgb}{.7, 0,1}
\newcommand{\purple}[1]{{\color{purple}#1}}

\newcommand{\NegAlt}{\text{$\downuparrows$-$\mathsf{Strat}$}}
\newcommand{\NegAltwb}{\NegAlt^{\mathsf{wb}}}
\newcommand{\NegAltinn}{\NegAlt^{\mathsf{inn}}}
\newcommand{\NegAltinnwb}{\NegAlt^{\mathsf{wb,inn}}}
\newcommand{\NegAltvis}{\NegAlt^{\mathsf{vis}}}
\newcommand{\NegAltviswb}{\NegAlt^{\mathsf{wb,vis}}}
\newcommand{\Games}{\mathsf{Vis}}
\newcommand{\CG}{\text{$\imc$-$\mathsf{Strat}$}}
\newcommand{\NNegAltwb}{\NNegAlt^{\mathsf{wb}}}
\newcommand{\NNegAlt}{\text{$\circlearrowleft$-$\mathsf{Strat}$}}
\newcommand{\CGwb}{\CG^{\mathsf{wb}}}
\newcommand{\GM}{\mathsf{GM}}
\newcommand{\PNStrat}{\text{$\mathscr{P}$-$\NAltStrat$}}
\newcommand{\CGwbinn}{\CG^{\mathsf{wb},\mathsf{inn}}}
\newcommand{\AltStrat}{\text{$\downuparrows$-$\mathsf{Unf}$}}
\newcommand{\CGseq}{\CG^{\mathsf{seq}}}
\newcommand{\CGwbseq}{\CG^{\mathsf{wb},\mathsf{seq}}}
\newcommand{\CGwbseqinn}{\CG^{\mathsf{wb},\mathsf{seq},\mathsf{inn}}}
\newcommand{\CGseqinn}{\CG^{\mathsf{seq},\mathsf{inn}}}
\newcommand{\Rel}{\mathsf{Rel}}
\newcommand{\CGwbvis}{\CG^{\mathsf{wb}, \mathsf{vis}}}

\newcommand{\pview}[1]{\ulcorner#1\urcorner}
\newcommand{\pviews}[1]{\ulcorner\!\ulcorner #1 \urcorner\!\urcorner}

\newdir{|>}{!/4.5pt/@{|}*:(1,-.2)@^{>}*:(1,+.2)@_{>}}
\newdir{pb}{:(1,-1)@^{|-}}
\def\pb#1{\save[]+<16 pt,0 pt>:a(#1)\ar@{pb{}}[]\restore}

\theoremstyle{plain}

\begin{document}

\title[Disentangling Parallelism and Interference in Game Semantics]{Disentangling
Parallelism and Interference\texorpdfstring{\\}{} in Game Semantics}

\author[S. Castellan]{Simon
Castellan\lmcsorcid{0000-0001-5886-5793}}[a]
\address{Inria, Univ Rennes, CNRS, IRISA}
\email{Simon.Castellan@inria.fr} 

\author[P. Clairambault]{Pierre
Clairambault\lmcsorcid{0000-0002-3285-6028}}[b]
\address{CNRS, Aix Marseille Univ, LIS, Marseille, France}
\email{Pierre.Clairambault@cnrs.fr}

\begin{abstract}
\noindent Game semantics is a denotational semantics presenting compositionally
the computational behaviour of various kinds of effectful programs. One
of its celebrated achievement is to have obtained full abstraction
results for programming languages with a variety of computational
effects, in a \emph{single} framework. This is known as the
\emph{semantic cube} or \emph{Abramsky's cube}, which for sequential
deterministic programs establishes a correspondence between certain
conditions on strategies (``innocence'', ``well-bracketing'',
``visibility'') and the absence of matching computational effects.

Outside of the sequential deterministic realm, there are still a wealth
of game semantics-based full abstraction results; but they no
longer fit in a unified canvas. In particular, Ghica and Murawski's
fully abstract model for shared state concurrency ($\IPA$) does
\emph{not}
have a matching notion of pure parallel program -- we say that
parallelism and interference (\emph{i.e.} state plus semaphores) are
\emph{entangled}.
In this paper we construct a causal version of Ghica and Murawski's
model, also
fully abstract for $\IPA$. We provide compositional conditions
\emph{parallel innocence} and
\emph{sequentiality}, respectively banning interference and parallelism,
and leading to four full abstraction results. To our knowledge, this is
the first extension of Abramsky's \emph{semantic cube} programme beyond
the sequential deterministic world.
\end{abstract}

\maketitle

\section*{Introduction}

How to \emph{prove} that a program $P$ is correct, or
equivalent to $P'$? This simple question,
prerequisite for formally validating software, lies at
the heart of decades of work in semantics.
Its study prompted a wealth of developments, each with its
methodology and scope. \emph{Operational semantics} axiomatizes
execution
directly on syntax, while \emph{denotational semantics} gives meaning
to programs by embedding them in a syntax-independent
mathematical space.

\emph{Operational semantics} is powerful and extensible,
perfectly fit for
formalization in a proof assistant -- it is, for instance, behind the
celebrated
CompCert project \cite{DBLP:journals/cacm/Leroy09}. On the other
hand, its deployment often follows from
ad-hoc choices, and it is not robust to variations in the language. It
is tied
to syntax and struggles with compositionality\footnote{Operational
semantics \emph{can} be made compositional, but behind lie denotational
structures: for instance, the operational semantics behind the recent
\emph{Compositional CompCert} \cite{DBLP:conf/popl/StewartBCA15} ``bears
much in common'' (quoting the paper) with Ghica and Tzevelekos'
operational reconstruction of game semantics
\cite{DBLP:journals/entcs/GhicaT12}.}. \emph{Denotational semantics} is
syntax-independent, and often more principled. It is a great tool to
reason about program equivalence (two programs being equivalent if they
denote the same object), to prove general properties of languages
(\emph{e.g.} termination), and it comes with compositional reasoning
principles. The wider mathematical space in which programs are embedded
sometimes suggests new useful constructs (it is the birth story of
Linear Logic \cite{DBLP:journals/tcs/Girard87}). In exchange, it is more
mathematically demanding and often quite brittle: distinct fragments of
the same language may require radically different representations.
Traditional denotational semantics (\emph{e.g.} Scott domains) model
programs as functions, through their input/output behaviour. Effects
(\emph{e.g.} state, non-determinism, \emph{etc}) can be captured via
monads which do not readily combine. Though combining effects has been a
driving question in denotational semantics these past decades, it is
hardly a streamlined process. For instance, though there is significant
recent research activity around domain settings supporting probabilities
and higher-order
\cite{DBLP:conf/lics/StatonYWHK16,DBLP:journals/pacmpl/VakarKS19}, it is
unclear how they combine with non-determinism
\cite{DBLP:journals/mscs/Goubault-Larrecq17}, let alone concurrency; nor
how all these models relate together. 

\emph{Game semantics} \cite{ho,ajm}, though also denotational, takes a
different approach: instead of a function it represents a
program as a \emph{strategy}, a collection of
(representations of) its interactions against execution
environments. Once executions are first-class citizens
(called \emph{plays}) one can characterise those achievable with
specific effects. This led to a wealth of fully abstract models,
rewarded in 2017 by the Alonzo Church Award (from the ACM SIGLOG, the
EATCS, the EACSL, and the Kurt G\"odel Society). To cite the
announcement:

\begin{quote}
``Game semantics has changed
the landscape of programming language semantics by giving a unified
view of the denotational universes of many different languages. This is
a remarkable achievement that was not previously thought to be within
reach.''
\end{quote}

But are games models truly ``unified''? For \emph{deterministic
sequential} programs, absolutely: various degrees of control
and state are indeed captured as additional conditions on one single
canvas \cite{abramsky1999game} -- this is the \emph{semantic
cube} or \emph{Abramsky cube}. But beyond the sequential
deterministic world, the picture is not so clear. The classic fully
abstract models for \emph{finite non-determinism}
\cite{DBLP:conf/lics/HarmerM99}, for \emph{probabilistic choice}
\cite{DBLP:conf/lics/DanosH00} or for \emph{parallelism} \cite{gm} all
rely on the presence of state. Until recently, there were no fully
abstract model for any of these \emph{without state} -- or in
the language of game semantics, there were no notions of
non-deterministic, probabilistic or parallel \emph{innocence}. Following
the phrasing of the title, our understanding of these
effects was \emph{entangled} with state.

However, this picture is currently shifting.
Recently, two notions of non-deterministic innocence were proposed
independently \cite{lics14,DBLP:conf/lics/TsukadaO15} -- the two
settings also handling probabilistic innocence
\cite{DBLP:journals/corr/TsukadaO14,lics18}. Technically, these settings
differ significantly. But conceptually, both enrich strategies with
explicit \emph{branching} information. Though the novelty may seem
minor, this is in fact a major schism
with respect to traditional game semantics, in that this branching
information is typically not observable. So instead of a strategy being
merely a formal description of how a program is observed by a certain
type of contexts, the model starts to carry more intensional,
\emph{causal} information, typically inaccessible to the environment but
which nonetheless finds its use in capturing compositionally the
computational behaviour expressible by certain programming features.
This suggests that to disentangle parallelism and state, we must
adequately represent the \emph{branching structure} of parallel
computation, the (non-observable) causal patterns of pure
parallel programs.

Enter \emph{concurrent games}. Concurrent games are a family of game
semantics models questioning in various ways the premise that the basic
building block should be totally, chronologically ordered
\emph{plays}. Pioneered by Melli\`es and others
\cite{DBLP:conf/lics/AbramskyM99,DBLP:conf/concur/Mellies04,DBLP:conf/concur/MelliesM07,DBLP:conf/tlca/FaggianP09},
they have lately been under intense development, prompted by
new definitions due to Rideau and Winskel \cite{lics11}. The name comes from their relationship with the
so-called \emph{true concurrency} approach to concurrency theory,
following which one represents causal dependence and independence of
events explicitly rather than resorting to
interleavings. Besides making concurrent games a natural target to
model concurrent languages and process calculi
\cite{DBLP:conf/concur/CastellanC16,DBLP:journals/pacmpl/CastellanY19},
it provides us with the required causal description of programs.

\paragraph{\textbf{Contributions.}} We \emph{disentangle} parallelism
and state -- or rather parallelism and \emph{interference}, which we
intend to also encompass semaphores. More precisely, we provide a fully
abstract model of Idealized Parallel Algol ($\IPA$), the paradigmatic
language used in the game semantics literature to study shared memory
concurrency on top of a higher-order language. Our model is a causal
version of that of Ghica and Murawski \cite{gm}, which
additionally supports compositional conditions of \emph{parallel
innocence} and \emph{sequentiality} respectively eliminating
interference and parallelism. Accordingly the paper presents \emph{four}
full abstraction results, following all combinations of parallelism and
interference on top of the pure language $\PCF$. Thus this is a
\emph{semantic square} \cite{abramsky1999game}, the first such result
pushing Abramsky's programme beyond the sequential deterministic world.

Of the four full abstraction results glued together, three
are classics: Hyland and Ong's full abstraction
for $\PCF$ \cite{ho}, Abramsky and McCusker's full abstraction for
Idealized Algol ($\IA$) \cite{am}, and Ghica and Murawski's full
abstraction for $\IPA$ \cite{gm}. The fourth result is a
variation of the full abstraction for $\PCF$ with respect to
parallel evaluation initially presented in conference format
in \cite{lics15} -- in particular, the notion
of parallel innocence comes from there\footnote{The paper \cite{lics15}
had two main contributions: a new games model called
\emph{thin concurrent games}, and \emph{parallel innocence}. The
detailed construction of the former appears in \cite{cg2}, but not
parallel innocence. The present paper provides, among other things,
detailed proofs for the second contribution of \cite{lics15}.} and was
developed as part as the first author's PhD thesis
\cite{DBLP:phd/hal/Castellan17}. 

These four results \cite{ho,am,gm,lics15} vary significantly in their
technical underpinnings. For the purposes of this paper, this left us
with the task, more challenging than anticipated, of providing the
\emph{glue}. Accordingly, a significant part of the paper revisits the
results of \cite{ho} and \cite{ajm} in a language closer to concurrent
games, mixing ideas from HO \cite{ho}, AJM \cite{ajm} and
\emph{asynchronous} \cite{DBLP:conf/lics/Mellies05} games.  In
doing so we hope that this paper, gathering in a single framework
several important developments of the field, could also serve as a
modern entry point to game semantics.  Accordingly we wrote it with the
newcomer in mind, not assuming prior knowledge on game
semantics. The development is self-contained, with however a number of
details postponed to the appendix. We also take the
time to show how our model relates to other game semantics frameworks,
hopefully helping the reader get a panoramic perspective on the field. More
generally, we try to keep the text as pedagogical as
possible. This of course, has a cost in that the paper is intimidatingly
lengthy; and we hope the readers will excuse us for that.

\paragraph{\textbf{Outline.}} In Section~\ref{sec:syntax}, we start
by describing $\IPA$ and its fragments.
In Section~\ref{sec:gs}, we introduce our version of alternating games,
its interpretation of $\PCF$, and
link with more traditional game semantics.
In Section~\ref{sec:seq_effects}, we show how (the absence of) control
and state may be captured via conditions of strategies -- we present
\emph{Abramsky's cube} and some of its consequences. In Section~\ref{sec:cg}, we present our causal fully abstract model for $\IPA$,
based on \emph{thin concurrent games}. In Section~\ref{sec:par_inn} we
develop one of the key contributions of this paper, \emph{parallel
innocence}: we leverage the causal description of programs offered by
thin concurrent games to characterize the causal shapes definable with
pure parallel higher-order programs. In Section~\ref{sec:fa_ia_pcf}, we
study the
\emph{sequential} fragment of our causal games model, and by
linking it with the sequential model of Sections~\ref{sec:gs} and
\ref{sec:seq_effects} we show full abstraction results for $\IA$ and
$\PCF$. Finally, in Section~\ref{sec:definability} we prove our last
full abstraction result, for $\PCFpar$.

\section{\texorpdfstring{$\IPA$}{IA//} and its fragments}\label{sec:syntax}
\emph{Idealized Parallel Algol} ($\IPA$) is a higher-order, simply-typed,
call-by-name concurrent language with shared memory and semaphores. We
also introduce fragments: 
\[
\begin{array}{cl}
\PCFpar & \text{is the fragment without interference,}\\
\IA & \text{is the fragment without parallelism, and}\\
\PCF & \text{has neither interference nor parallelism.}
\end{array}
\]

\subsection{Types} The \textbf{types} of $\IPA$ are the following, 
highlighting types relative to interference.
\[
\begin{array}{rcll}
A, B &::=& \tunit \mid \tbool \mid \tnat \mid A \to B & \PCF\\
&&\mid \var \mid \sem &\text{+interference}
\end{array}
\]

Above, $\tunit$ is a \emph{unit} type with only one value, and
$\tbool$ and $\tnat$ are types for \emph{booleans} and
\emph{natural numbers}.
In the presence of interference,
$\var$ is a type for \emph{references} storing natural numbers, while
$\sem$ is for \emph{semaphores}. We refer to
$\tunit, \tbool$ and $\tnat$ as \emph{ground types}, and
use $\tx, \ty$ to range over those.
Let us now give the term constructions and typing rules.

\subsection{Terms and Typing} We
define the terms of the language directly via typing rules.

\begin{figure}
\boxit[$\PCF$]{
\begin{mathpar}
\inferrule
	{ }
	{ \Gamma \vdash \tskip : \tunit }
\and
\inferrule
	{ }
	{ \Gamma \vdash \ttrue : \tbool }
\and
\inferrule
	{ }
	{ \Gamma \vdash \tfalse : \tbool }
\and
\inferrule
	{ }
	{ \Gamma \vdash n : \tnat }
\and
\inferrule
	{ }
	{ \Gamma, x : A \vdash x : A }
\and
\inferrule
	{ \Gamma, x : A \vdash M : B }
	{ \Gamma \vdash \lambda x^A.\,M : A\to B }
\and
\inferrule
	{ \Gamma \vdash M : A \to B \\
	  \Gamma \vdash N : A }
	{ \Gamma \vdash M\,N : B }
\and
\inferrule
	{ \Gamma \vdash M : \tunit \\
	  \Gamma \vdash N : \tx }
	{ \Gamma \vdash M;\,N : \tx }
\and
\inferrule
	{ \Gamma \vdash M : \tbool \\
	  \Gamma \vdash N_1 : \tx \\
	  \Gamma \vdash N_2 : \tx }
	{ \Gamma \vdash \ite{M}{N_1}{N_2} : \tx }
\and
\inferrule
	{ \Gamma \vdash M : \tnat }
	{ \Gamma \vdash \tsucc\,M : \tnat }
\and
\inferrule
	{ \Gamma \vdash M : \tnat }
	{ \Gamma \vdash \tpred\,M : \tnat }
\and
\inferrule
	{ \Gamma \vdash M : \tnat }
	{ \Gamma \vdash \iszero\,M : \tbool }
\and
\inferrule
	{ \Gamma, x : \tx \vdash M : \ty \\ \Gamma \vdash N : \tx }
	{ \Gamma \vdash \tlet{x}{N}{M} : \ty }
\and
\inferrule
	{ \Gamma \vdash M : A \to A }
	{ \Gamma \vdash \Y\,M : A }
\end{mathpar}
}
\boxit[$\mathsf{+interference}$]{
\begin{mathpar}
\inferrule
{ \Gamma, x : \var \vdash M : \tx }
{ \Gamma \vdash \newref\,x\!\!:=\!n\,\tin\,M : \tx }
\and
\inferrule
{ \Gamma \vdash M : \var \\
  \Gamma \vdash N : \tnat }
{ \Gamma \vdash M\!\!:=\!N : \tunit }
\and
\inferrule
{ \Gamma \vdash M : \var }
{ \Gamma \vdash !M : \tnat }
\and
\inferrule
{ \Gamma, x : \sem \vdash M : \tx }
{ \Gamma \vdash \newsem\,x\!\!:=\!n\,\tin\,M : \tx }
\and
\inferrule
{ \Gamma \vdash M : \sem }
{ \Gamma \vdash \grab\,M : \tunit }
\and
\inferrule
{ \Gamma \vdash N : \sem }
{ \Gamma \vdash \release\,N : \tunit }
\and
\inferrule
{ \Gamma \vdash M : \tnat \to \tunit \\
  \Gamma \vdash N : \tnat }
{ \Gamma \vdash \mkvar\,M\,N : \var }
\and
\inferrule
{ \Gamma \vdash M : \tunit \\
  \Gamma \vdash N : \tunit }
{ \Gamma \vdash \mksem\,M\,N : \sem }
\end{mathpar}
}
\boxit[$\mathsf{+parallelism}$]{
\begin{mathpar}
\inferrule
{ \Gamma, x_1 : \tx, x_2 : \tx \vdash M : \ty \\
  \Gamma \vdash N_1 : \tx \\
  \Gamma \vdash N_2 : \tx }
{ \Gamma \vdash \plet{x_1}{N_1}{x_2}{N_2}{M} : \ty }
\end{mathpar}
}
\caption{Typing rules for $\IPA$}
\label{fig:typing}
\end{figure}

\textbf{Contexts} are lists $x_1 : A_1, \dots,
x_n : A_n$. \textbf{Typing judgments} have the form $\Gamma \vdash M :
A$ with $\Gamma$ a context and $A$ a type. In addition to 
Figure \ref{fig:typing}, we consider present an explicit
exchange rule allowing us to permute the order of variable declarations
in contexts.  The eliminator rules for basic datatypes are restricted to
eliminate only to ground types -- general eliminators are defined as
syntactic sugar: \emph{e.g.} a conditional to $\var$ may be obtained as
\[
\inferrule
{ \Gamma \vdash M : \tbool \\ 
  \Gamma \vdash N_1 : \var \\
  \Gamma \vdash N_2 : \var }
{ \Gamma \vdash \mkvar\,(\lambda
x.\,\ite{M}{(N_1\!\!:=\!x)}{(N_2\!\!:=\!x)})\,(\ite{M}{!N_1}{!N_2})
: \var } 
\]

The \textbf{bad variable} and \textbf{bad semaphore} constructs
$\mkvar$ and $\mksem$ are a common occurrence in the game semantical
literature.  While a ``\emph{good}'' reference is tied to a memory
location, many game models also comprise so-called ``\emph{bad
variables}'' inhabiting $\var$ but not behaving as actual variables.
Full abstraction results \cite{am,gm}
often require a corresponding syntactic construct $\mkvar$ allowing one
to \emph{form} bad variables by appending arbitrary read and write
methods\footnote{Though McCusker proved that equational full
abstraction holds for $\IA$ without $\mkvar$
\cite{DBLP:conf/mfps/McCusker03}.}. The same holds for semaphores.

\subsection{Further syntactic sugar.}\label{subsec:sugar}
First of all, for any type $A$ there is a \textbf{divergence}
$\vdash \bot_A : A$, any looping program.
Given $\Gamma \vdash M, N : \tunit$,
an equality test $\Gamma \vdash M =_\tunit N : \tbool$
may be defined as $M;\,N;\,\ttrue$. Likewise, for $\Gamma \vdash M, N :
\tbool$ we define $\Gamma \vdash M =_\tbool N : \tbool$ as
$\ite{M}{N}{(\ite{N}{\tfalse}{\ttrue})}$, and $\Gamma \vdash M=_\tnat N
: \tbool$ similarly, with the obvious recursive program.
\newpage

We refer to constants of ground type as \textbf{values}; we use $v$ to
range over those, and $n, b$ or $c$ to range over
values of respective types $\tnat, \tbool$ or $\tunit$. We introduce
a $n$-ary case construct branching on all
values of ground types. By abuse of notation, we write $V
\subseteq_f \tx$ for any finite subset of the values of ground type
$\tx$. Writing $V = \{v_1, \dots, v_n\}$, we set
\[
\begin{array}{l}
\mathbf{case}\,M\,\mathbf{of}\\
\hspace{20pt}v_1 \mapsto N_1\\
\hspace{20pt}v_2 \mapsto N_2\\
\hspace{20pt}\dots\\
\hspace{20pt}v_n \mapsto N_n
\end{array}
\qquad
\stackrel{\text{def}}{=}
\qquad
\begin{array}{l}
\mathbf{let}\,x\,=\,M\,\mathbf{in}\\
\hspace{35pt}\mathbf{if}\,x\,=_\tx\,v_1\,\mathbf{then}\,N_1\\
\hspace{16pt}\mathbf{else}\,\mathbf{if}\,x\,=_\tx\,v_2\,\mathbf{then}\,N_2\\
\hspace{40pt}\dots\\
\hspace{16pt}\mathbf{else}\,\mathbf{if}\,x\,=_\tx\,v_n\,\mathbf{then}\,N_n\\
\hspace{16pt}\mathbf{else}\,\bot
\end{array}
\]
of type $\ty$ in context $\Gamma$ if $\Gamma \vdash M : \tx$ and $\Gamma
\vdash N_i : \ty$ for all $1\leq i \leq n$.

The $\mathbf{let}$ construct is crucial in this paper: as we shall see
later on, strategies may evaluate a variable \emph{once}, and provide a
different continuation for each possible value. This behaviour cannot be
replicated strictly without $\mathbf{let}$, see Section~\ref{subsubsec:def_inn} for a more detailed discussion.

\subsection{Operational semantics.} We recall the small-step
operational semantics \cite{gm}. We fix a
countable set $\L$ of \textbf{memory locations}. A \textbf{store} is a
partial map $s : \L \pto \mathbb{N}$ with finite domain where
$\mathbb{N}$ stands, overloading notations, for natural
numbers.  \textbf{Configurations} of the operational semantics are
tuples $\tuple{M, s}$ where $s$ is a store with $\dom(s) = \{\ell_1,
\dots, \ell_n\}$ and $\Sigma \vdash M : A$ with $\Sigma = \ell_1 : \var,
\dots, \ell_i : \var, \ell_{i+1} : \sem, \dots, \ell_n : \sem$.

Reduction rules have the form $\tuple{M, s} \leadsto \tuple{M', s'}$
where $\dom(s) = \dom(s')$; we write $\leadsto^*$ for the reflexive
transitive closure.
If $\vdash M : \tx$, we write $M\eval$ if $\tuple{M, \emptyset}
\leadsto^* \tuple{v, \emptyset}$ for some value $v$.
We give in Figure \ref{fig:operational} the reduction rules -- there and
from now on in the paper we use the notation $\uplus$ to denote the
usual set-theoretic union, when it is known disjoint. For rules which
do not interact with the state, we omit the state component -- it is
simply left unchanged by stateless basic reductions, and propagated
upwards by stateless context rules.

\begin{figure}
\begin{minipage}{0.34\linewidth}
\boxit[Basic red. for $\PCF$]{
\vspace{9pt}
\[
\scalebox{.96}{$
\begin{array}{rcl}
(\lambda x^A.\, M)\,N &\leadsto& M[N/x] \\
\tskip;\,N &\leadsto& N \\
\ite{b}{N_{\ttrue}}{N_{\tfalse}} &\leadsto& N_b \\
\tsucc\,n &\leadsto& n+1 \\
\pred\,0 & \leadsto& 0\\
\pred\,(n+1) & \leadsto& n\\
\iszero\,0 &\leadsto& \ttrue \\
\iszero\,(n+1) &\leadsto& \tfalse \\
\Y\,M &\leadsto& M\,(\Y\,M)\\
\!\!\tlet{x}{v}{M} &\leadsto& M[v/x]
\end{array}
$}
\]
\vspace{8.5pt}
}
\end{minipage}
\hfill
\begin{minipage}{0.645\linewidth}
\boxit[Basic reductions for $\mathsf{interference}$]{
\vspace{5pt}
\[
\scalebox{.96}{$
\begin{array}{rcl}
\newref\,x\,\tin\,v &\leadsto& v \\
\newsem\,x\,\tin\,v &\leadsto& v \\
(\mkvar\,M\,N)\!\!:=\!n &\leadsto& M\,n \\
!(\mkvar\,M\,N) &\leadsto& N \\
\grab(\mksem\,M\,N) &\leadsto& M \\
\release(\mksem\,M\,N) &\leadsto& N
\end{array}$}
\]
}
\boxit[Interfering reductions]{
\[
\scalebox{.96}{$
\begin{array}{rcl}
\tuple{!\ell, s \uplus \{\ell \mapsto n\}} &\leadsto& \tuple{n, s
\uplus \{\ell \mapsto n\}} \\
\tuple{\ell\!\!:=\!n, s \uplus \{\ell \mapsto \_\}} &\leadsto&
\tuple{\tskip, s \uplus \{\ell \mapsto n\}} \\
\tuple{\grab(\ell), s \uplus \{\ell \mapsto 0\}} &\leadsto&
\tuple{\tskip, s \uplus \{\ell \mapsto 1\}} \\
\!\!\tuple{\release(\ell), s \uplus \{\ell \mapsto n\}} &\leadsto&
\tuple{\tskip, s \uplus \{\ell \mapsto 0\}} ~~~~\text{($n>0$)}
\end{array}$}
\]
}
\end{minipage}
\boxit[Basic reduction for $\mathsf{parallelism}$]{
\vspace{5pt}
\[
\scalebox{.96}{$
\begin{array}{rcl}
\plet{x_1}{v_1}{x_2}{v_2}{M} &\leadsto& M[v_1/x_1,v_2/x_2]
\end{array}$}
\]
}
\boxit[Stateless context rules]{
\begin{mathpar}
\scalebox{.96}{$
\inferrule
{ M \leadsto M' }
{ M\,N \leadsto M'\,N }
$}
\and
\scalebox{.96}{$
\inferrule
{ M \leadsto M' }
{ \ite{M}{N_1}{N_2} \leadsto \ite{M'}{N_1}{N_2} }
$}
\and
\scalebox{.96}{$
\inferrule
{ M \leadsto M' }
{ \tsucc\,M \leadsto \tsucc\,M' }
$}
\and
\scalebox{.96}{$
\inferrule
{ M \leadsto M' }
{ !M \leadsto\,!M' }
$}
\and
\scalebox{.96}{$
\inferrule
{ M \leadsto M' }
{ \iszero\,M \leadsto \iszero\,M' }
$}
\and
\scalebox{.96}{$
\inferrule
{ N \leadsto N' }
{ M\!\!:=\!N \leadsto M\!\!:=\!N' }
$}
\and
\scalebox{.96}{$
\inferrule
{ M \leadsto M' }
{ \grab(M) \leadsto \grab(M') }
$}
\and
\scalebox{.96}{$
\inferrule
{ M \leadsto M' }
{ \release(M) \leadsto \release(M') }
$}
\and
\scalebox{.96}{$
\inferrule
{ M \leadsto M' }
{ M\!\!:=\!v \leadsto M'\!\!:=\!v }
$}
\and
\scalebox{.96}{$
\inferrule
{ N \leadsto N' }
{ \tlet{x}{N}{M} \leadsto \tlet{x}{N'}{M} }
$}
\and
\scalebox{.96}{$
\inferrule
{ N_1 \leadsto N'_1 }
{ \plet{x_1}{N_1}{x_2}{N_2}{M} \leadsto \plet{x_1}{N'_1}{x_2}{N_2}{M} }
$}
\and
\scalebox{.96}{$
\inferrule
{ N_2 \leadsto N'_2 }
{ \plet{x_1}{N_1}{x_2}{N_2}{M} \leadsto \plet{x_1}{N_1}{x_2}{N'_2}{M} }
$}
\end{mathpar}
}
\boxit[Stateful context rules]{
\vspace{5pt}
\begin{mathpar}
\scalebox{.96}{$
\inferrule*[right=$\text{($\ell\in \L~\mathrm{fresh}$)}$]
{ \tuple{M[\ell/x], s\uplus \{\ell \mapsto n\}} \leadsto
\tuple{M'[\ell/x], s'\uplus \{\ell \mapsto n'\}} }
{ \tuple{\newref\,x\!\!:=\!n\,\tin\,M,s} \leadsto
\tuple{\newref\,x\!\!:=\!n'\,\tin\,M',s'} }
$}
\and
\scalebox{.96}{$
\inferrule*[right=$\text{($\ell\in \L~\mathrm{fresh}$)}$]
{ \tuple{M[\ell/x], s\uplus \{\ell \mapsto n\}} \leadsto
\tuple{M'[\ell/x], s'\uplus \{\ell \mapsto n'\}} }
{ \tuple{\newsem\,x\!\!:=\!n\,\tin\,M,s} \leadsto
\tuple{\newsem\,x\!\!:=\!n'\,\tin\,M',s'} }
$}
\end{mathpar}
}
\caption{Operational semantics of $\IPA$}
\label{fig:operational}
\end{figure}

\subsection{Fragment languages}
Besides $\PCF$, we consider three main languages of interest:
\[
\begin{array}{rcl}
\PCFpar &=& \PCF + \mathsf{parallelism}\\
\IA &=& \PCF + \mathsf{interference}\\
\IPA &=& \PCF + \mathsf{interference} + \mathsf{parallelism}
\end{array}
\]

$\IA$ is a variant of \emph{Idealized Algol with active expressions}
\cite{am}, differing only in that it has semaphores. This is not a
significant difference, as semaphores are definable from state in a
sequential language.
Likewise, $\IPA$ is close to the language of \cite{gm}: it differs only
in that the parallelism operation is more general.  For $\Gamma \vdash M
: \tunit$ and $\Gamma \vdash N : \tunit$ we may define their parallel
composition $\Gamma \vdash M \parallel N : \tunit$ (as in \cite{gm}) by
\[
M \parallel N ~~ = ~~ \plet{x}{M}{y}{N}{\tskip}\,.
\]

Conversely, for \emph{e.g.} $\Gamma \vdash N_1 : \tnat$, $\Gamma
\vdash N_2 : \tnat$ and $\Gamma, x_1 : \tnat, x_2 : \tnat \vdash M : A$,
the present parallel let construction is definable via state and
parallel composition of commands:
\[
\begin{array}{rcl}
\plet{x_1}{N_1}{x_2}{N_2}{M} &=& \newref\,v_1\!\!:=\!0\,\tin\\
                     & & \newref\,v_2\!\!:=\!0\,\tin\,
(v_1\!\!:=\!N_1 \parallel v_2\!\!:=\!N_2);\,M[!v_1/x_1, !v_2/x_2]
\end{array}
\]

\subsection{Observational Equivalence and Full
Abstraction}\label{sec:def_obs}

Here, $\L$ may refer to any of the fragments above. A \textbf{$\L$-context} for the judgment $\Gamma
\vdash A$ is a term $C[]$ of $\L$ with a hole, s.t. for any $\Gamma
\vdash M : A$ in $\L$ we have $\vdash C[M] : \tunit$ obtained by
replacing the hole with $M$. Two terms $\Gamma \vdash M, N : A$ of $\L$
are \textbf{$\L$-observationally equivalent} iff
\[
M \obs_\L N \qquad \Leftrightarrow \qquad \text{for all $C[]$ a
$\L$-context for $\Gamma \vdash A$, $\quad(C[M]\eval \quad
\Leftrightarrow \quad C[N]\eval)$}
\]

We omit $\L$ when it is clear from the context.
Observational equivalence is usually regarded as the canonical
equivalence on
programs: $\L$-observationally equivalent programs are intercheangeable
as long as the evaluation context is in $\L$.
Accordingly, denotational semantics often aims to capture observational
equivalence.
An interpretation of programs $\intr{-}$ into some mathematical universe
is called
\textbf{fully abstract} whenever
\[
M \obs N \qquad \Leftrightarrow \qquad \intr{M} = \intr{N}
\]
for all $\Gamma \vdash M, N : A$.
\emph{Full abstraction} is a gold standard in denotational
semantics, as it captures the best possible match between a language and
its semantics, ensuring that the denotational semantics is
\emph{complete} for proving equivalence between programs.

\section{Game Semantics for $\PCF$}
\label{sec:gs}

\subsection{Games and Strategies} We present first a game semantics of
$\PCF$. Though it is sequential, our presentation is non-standard, mixing
features of AJM \cite{ajm}, HO \cite{ho} and asynchronous games
\cite{DBLP:conf/lics/Mellies05} -- this is to facilitate the interplay
between all the games models involved. We skip a number
of details, found in Appendix \ref{app:alt}.

Game semantics presents higher-order computation as an exchange of
tokens between two players, called ``Player'' and ``Opponent''.
\emph{Player} stands for the program under evaluation -- events/moves
attributed to Player are observable computational events resulting from
its execution: calls to variables, program phrases converging to a
value.  \emph{Opponent} stands for the execution environment. Their
interaction follows rules depending on the type of the program under
scrutiny. In setting up a game semantics the first step is to extract
from the type a structure, called a \emph{game} or an \emph{arena},
which presents all the observable computational events available when
interacting on this type, along with their respective causal
dependencies.

\subsubsection{Affine arenas} We first introduce our representation of
types as games in the \emph{affine} case, \emph{i.e.} if any
computational event can appear at most once -- this is merely 
to first help the reader build up intuition before handling replication.

Consider $(\tunit \to \tunit) \to \tbool$, where affineness implies that
each argument may be called at most once. Once a call-by-name execution
on that type is initiated, the available observable events are the
following: \emph{(1)} the term may directly converge to $\ttrue$ or
$\tfalse$, without evaluating its argument; \emph{(2)} it may call its
argument (\emph{i.e.} it evaluates to $\lambda f^{\tunit \to
\tunit}.\,M$ with $M$ having $f$ in head position). In the case
\emph{(2)} the control goes back to the environment, which plays for
$f$: it may prompt $f$ to return the unique value $\tskip$, or to itself
call its argument.  Finally, if $f$ calls its argument, the
corresponding sub-term may reduce to a value.

\begin{figure}
\begin{minipage}{.3\linewidth}
\[
\xymatrix@R=8pt@C=5pt{
&&&\qu^-
\ar@{.}[dll]
\ar@{.}[d]
\ar@{.}[drr]\\
&\run^+
\ar@{.}[dl]
\ar@{.}[d]&&\ttrue^+\ar@{~}[rr]&&\tfalse^+\\
\run^-  \ar@{.}[d]&\done^-\\
\done^+\\
}
\]
\caption{An affine arena}
\label{ex:arena1}
\end{minipage}
\hfill
\begin{minipage}{.3\linewidth}
\[
\xymatrix@R=11.5pt@C=10pt{
&&\qu^-&\qu^-\\
\qu^+   \ar@{.}[urr]
\ar@{~}[r]&
\qu^+   \ar@{.}[urr]&
\done^+ \ar@{.}[u]&
\done^+ \ar@{.}[u]\\
\done^- \ar@{.}[u]&
\done^- \ar@{.}[u]\\~
}
\]
\caption{$\tunit \lin (\tunit \tensor \tunit)$}
\label{ex:ar_lin_tens}
\end{minipage}
\hfill
\begin{minipage}{.33\linewidth}
\[
\xymatrix@R=-3.5pt{
(\tunit
\ar@{}[r]|\to&
\tunit)   \ar@{}[r]|\to&
\tbool\\
&&{{\qu^{-}}}\\
&{\run^{+}}\\
{\run^{-}}&\\
{\done^+}&&\\
&{\done^-}\\
&&{\ttrue^+}}
\]
\caption{An alternating play}
\label{ex:play1}
\end{minipage}
\end{figure}

Overall, these events along with their causal dependencies give rise to
the diagram in Figure \ref{ex:arena1}. It is read from top to bottom,
with the dashed lines representing the dependency relation. Nodes are
called \emph{moves} or \emph{events}, and are labeled with a polarity,
$-$ for events due to the environment, and $+$ for events due to the
program.  Finally, the wiggly line between $\ttrue^+$ and $\tfalse^+$
indicates \emph{conflict}: it represents the fact that only one of these
two values may be observed in one execution, whereas all the other pairs
of events could conceivably appear together. The reader may convince
themselves that indeed, the diagram does represent the observable events
in a call-by-name evaluation of $(\tunit \to \tunit) \to \tbool$ as
outlined in the previous paragraph.  We insist that those are the
computational events that are \emph{observable} in the interface with
the environment: the program may perform internal computation; a program
in an extension of $\PCF$ with state could possibly store and read
from a local variable, \emph{etc}. But those are not
\emph{observable} by a context, thus are not represented in the arena.

To formalize the arena as a mathematical structure, we use \emph{event
structures}\footnote{More precisely, those are \emph{prime event
structures with binary conflict}.}:

\begin{defi}\label{def:es}
An \textbf{event structure (\emph{es})} is a triple $E = (\ev{E},
\leq_E, \conflict_E)$, where $\ev{E}$ is a (countable) set of
\textbf{events}, $\leq_E$ is a partial order called \textbf{causal
dependency} and $\conflict_E$ is an irreflexive symmetric binary
relation on $\ev{E}$ called \textbf{conflict}, satisfying:
\[
\begin{array}{rl}
\text{\emph{finite causes}:}& \forall e\in \ev{E},~[e]_E = \{e'\in
\ev{E} \mid
e'\leq_E e\}~\text{is finite}\\
\text{\emph{conflict inheritance}:}&
\forall e_1 \conflict_E e_2,~\forall e_2 \leq_E e'_2,~e_1 \conflict_E
e'_2\,.
\end{array}
\]

An \textbf{event structure with polarities (\emph{esp})} is an event
structure $A$ together with a function $\pol_A : \ev{A} \to \{-,+\}$
assigning to each event a \emph{polarity}.
\end{defi}

Figure \ref{ex:arena1} displays an esp. The wiggly line indicates
conflict, but we will not put wiggly lines between all conflicting pairs
of events, as long as missing conflicts may be deduced by \emph{conflict
inheritance}. A conflict that cannot be deduced by inheriting an earlier
conflict is called a \textbf{minimal conflict}.  As with Figure
\ref{ex:arena1}, we will represent types as esps.  In fact, the event
structures arising via the interpretation of types have a very
restricted form.  In the definition below, we use the notation $e \imc_E
e'$ in an event structure $E$ to mean \textbf{immediate causality},
\emph{i.e.} $e <_E e'$ with no other event strictly in between.

\begin{defi}\label{def:arena}
An \textbf{arena} is an esp $(A, \leq_A, \conflict_A, \pol_A)$
satisfying:
\[
\begin{array}{rl}
\text{\emph{alternating:}} & \text{if $a_1 \imc_A a_2$, $\pol_A(a_1)
\neq \pol_A(a_2)$,}\\ 
\text{\emph{forestial:}} & \text{if $a_1 \leq_A a$ and $a_2 \leq_A a$,
then $a_1 \leq_A a_2$ or $a_2 \leq_A a_1$,}\\ 
\text{\emph{race-free:}} & \text{if $a_1, a_2 \in \ev{A}$ are in minimal
conflict, then $\pol_A(a_1) = \pol_A(a_2)$.} 
\end{array}
\]

Besides, a \textbf{$-$-arena} additionally satisfies the condition:
\[
\begin{array}{rl}
\text{\emph{negative:}} & \text{if $a \in \min(A)$, then $\pol_A(a) =
-$,}
\end{array}
\]
where $\min(A)$ stands for the set of minimal events of $A$.
\end{defi}

Types will only yield $-$-arenas, but throughout the paper we will
use the general case. Finally, though we motivated Definition
\ref{def:es} with arenas, event structures will have other uses.
Notably, from Section~\ref{sec:cg} onwards, strategies will also be
event structures.

\subsubsection{Basic Constructions}\label{subsubsec:basic_constr} We
give a few basic constructions on event structures and arenas which will
allow us to construct in a systematic way, from any type of $\PCF$, a
$-$-arena.

\begin{figure}
\begin{minipage}{0.15\linewidth}
\[
\xymatrix@R=10pt@C=10pt{
\qu^-
\ar@{.}[d]\\
\done^+
}
\]
\vspace{2pt}
\caption{$\tunit$}
\label{fig:ar_unit}
\end{minipage}
\hfill
\begin{minipage}{.24\linewidth}
\[
\xymatrix@R=15pt@C=10pt{
&\qu^-
\ar@{.}[dl]
\ar@{.}[dr]\\
\ttrue^+\ar@{~}[rr]&&
\tfalse^+
}
\]
\caption{$\tbool$}
\label{fig:ar_bool}
\end{minipage}
\hfill
\begin{minipage}{.24\linewidth}
\[
\xymatrix@R=15pt@C=10pt{
&&\qu^-
\ar@{.}[dll]
\ar@{.}[dl]
\ar@{.}[d]
\ar@{.}[dr]\\
0^+     \ar@{~}[r]&
1^+     \ar@{~}[r]&
2^+     \ar@{~}[r]&
\dots
}
\]
\caption{$\tnat$}
\label{fig:ar_nat}
\end{minipage}
\hfill
\begin{minipage}{.24\linewidth}
\[
\xymatrix@R=7pt@C=5pt{
\qu_{\grey{0}}^-        \ar@{.}[d]&
\qu_{\grey{1}}^-        \ar@{.}[d]&
\qu_{\grey{2}}^-        \ar@{.}[d]&
\dots\\
\done^+_{\grey{0}}&\done^+_{\grey{1}}&\done^+_{\grey{2}}&\dots
}
\]
\caption{$\oc \tunit$}
\label{fig:oc_unit}
\end{minipage}
\end{figure}

We give $-$-arenas for the ground types of $\PCF$, in
Figures \ref{fig:ar_unit}, \ref{fig:ar_bool} and \ref{fig:ar_nat}, using
the same notations $\tunit, \tbool$ and $\tnat$ for the arenas as for
the types. For $\tnat$, even though the picture only shows conflict
between neighbours, all positive events are meant to be in pairwise
conflict.

We write $\ees$ for the \textbf{empty \emph{es}}, with no event. If $A$
is an \emph{esp}, we write $A^\perp$ for its \textbf{dual}, the
\emph{esp} with same events, causality and conflict, but the opposite
polarities, \emph{i.e.} $\pol_{A^\perp}(a) = -\pol_A(a)$ for all $a \in
\ev{A}$.  The \emph{simple parallel composition} is defined as follows.

\begin{defi}
If $E_1, E_2$ are two es, their \textbf{simple parallel
composition} $E_1 \parallel E_2$ has 
\[
\begin{array}{rcrcl}
\text{\emph{events}:} &~~~& \ev{E_1\parallel E_2} &=& \{1\}\times
\ev{E_1} ~~ \uplus ~~ \{2\}\times \ev{E_2}\\
\text{\emph{causality}:} && (i, e) \leq_{E_1 \parallel E_2} (j, e')
&\Leftrightarrow&
i=j~\&~e \leq_{E_i}
e'\\
\text{\emph{conflict}:} && (i, e) \conflict_{E_1 \parallel E_2} (j, e')
&\Leftrightarrow& i=j ~\&~ e \conflict_{E_i} e'\,.
\end{array}
\]

Moreover, if $E_1$ and $E_2$ have polarities (\emph{i.e.} are esp), then
$E_1 \parallel E_2$ also has polarities, defined as $\pol_{E_1 \parallel
E_2}(1, e) = \pol_{E_1}(e)$ and $\pol_{E_1 \parallel E_2}(2, e) =
\pol_{E_2}(e)$.
\end{defi}

By extension, we often write $X \parallel Y$ for the tagged disjoint
union $(\{1\} \times X) \uplus (\{2\} \times Y)$ of two sets $X$ and
$Y$. In the simple parallel composition of arenas $A$ and $B$, the two
are side by side with no interaction.  The arena $A \parallel B$
adequately represents a tensor type $A \tensor B$ where the
two resources $A$ and $B$ may be accessed in any order -- although
$\PCF$ does not have such a type, this construction will play an
important role in the sequel. We also introduce the \textbf{product}
$A_1 \with A_2$ of $A_1$ and $A_2$ $-$-arenas, defined as for
$A_1\parallel A_2$ with conflict
\[
(i, e) \conflict_{A_1 \with A_2} (j, e') ~~ \Leftrightarrow ~~ i\neq j
~\vee~ (i = j \wedge e \conflict_{A_i} e')\,,
\]
\emph{i.e.} $A_1$ and $A_2$ are in conflict. The constructions
$\parallel$ and $\with$ have obvious $n$-ary generalizations.  We also
introduce another construction on arenas, the affine arrow $A\lin B$:

\begin{defi}\label{def:lin_wo}
Let $A, B$ be arenas with $B$ \textbf{pointed}, \emph{i.e.} with
\emph{exactly one} minimal $b_0 \in \ev{B}$. 

The \textbf{affine arrow} $A\lin B$ has the components of $A^\perp \parallel B$ except for
causality, set as:
\[  
\begin{array}{rcl}
\leq_{A\lin B}  &=& \leq_{A^\perp \parallel B} \cup\,\{((2, b_0), (1, a))
\mid a \in
\ev{A}\}\,.
\end{array}
\]
\end{defi}

This completes an interpretation of $\PCF$ types as pointed
$-$-arenas capturing the causal dependency between computational events
in an affine evaluation. For instance, on $A \to B$
computation starts in $B$, but as soon as the initial move of $B$ has
been played computation in $A$ may start, with polarity reversed. At
this point, the reader may verify that indeed, the arena $(\tunit \lin
\tunit) \lin \tbool$ obtained by applying these constructions is indeed
the one in Figure \ref{ex:arena1}.

\subsubsection{General arrow}\label{subsubsec:gen_arr}
Definition \ref{def:lin_wo} suffices for the types of $\PCF$ (which
yield pointed arenas). But we aim to show that strategies have the
structure of a \emph{Seely category}, a traditional categorical model
for Intuitionistic Linear Logic -- and that structure includes tensors,
which do not preserve pointedness.  To generalize $A \lin B$ for $B$
non-pointed, it is natural to set one copy of $A$ for each initial move
of $B$. More concretely, $A \lin B$ has events and polarities
\[
\ev{A \lin B} = (\parallel_{b \in \min(B)} A)^\perp \parallel B\,,
\]
where $\min(B)$ is the set of \emph{minimal events} of $B$. The order
has $(2, b) \leq (2, b')$ iff $b \leq_B b'$, $(1, (b, a)) \leq (1, (b',
a'))$ iff $b=b'$ and $a \leq_A a'$, $(2, b) \leq (1, (b', a))$ iff
$b=b'$, and $(1, (b, a)) \leq (2, b')$ never, exactly matching the arrow
arena of HO games \cite{ho}. But having two copies of $A$ is in tension
with affineness, so we use conflict to tame
this copying.  The construction is illustrated in Figure
\ref{ex:ar_lin_tens}, displaying the arena $\tunit \lin (\tunit \tensor
\tunit)$. There are two copies of the $\tunit$ on the left, but still,
linearity is guaranteed by the addition of conflict.

To define conflict, writing $\chi_{A, B} : \ev{A\lin B}
\to \ev{A^\perp \parallel B}$ for the obvious map, we use:

\begin{restatable}{lem}{conflictlin}\label{lem:conflict_lin}
Consider $A$ and $B$ two $-$-arenas.

Then, there is a unique $\conflict_{A\lin B}$ making $A\lin B$ a
$-$-arena such that for all down-closed finite $x \subseteq \ev{A\lin
B}$, $x \in \conf{A\lin B}$ iff $\chi_{A, B}\,x \in \conf{A^\perp
\parallel B}$ with $\chi_{A, B}$ injective on $x$. 
\end{restatable}
\begin{proof}
See Appendix \ref{app:arrow}.
\end{proof}

\subsubsection{Playing on Arenas}
Now we formulate a notion of execution, relying on the fact that event
structures support a natural notion of \emph{state} or \emph{position},
called \emph{configuration}.

\begin{defi}
A (finite) \textbf{configuration} of es $E$ is a finite set $x \subseteq
\ev{E}$ which is 
\[  
\begin{array}{ll}
\text{\emph{down-closed}:}& \forall e\in x,~\forall e' \in
\ev{E},~e'\leq_E e\implies
e' \in x\\
\text{\emph{consistent}:}& \forall e,e'\in x,~\neg(e \conflict_E e')\,.
\end{array}
\]

We write $\conf{E}$ for the set of finite configurations on $E$.
\end{defi}

For $x, y \in \conf{E}$, we write $x \cov y$ if there is $e \in \ev{E}$
such that $e \not \in x$ and $y = x \cup \{e\}$; $\cov$ is the
\textbf{covering relation}. If $x \cov x \cup \{e\}$, we say that $x$
\textbf{enables} $e$ or \textbf{extends by} $e$, written $x \vdash_E e$.
Configurations of an arena represent \emph{valid
execution states}. We may now leverage this to define \emph{plays},
which provide a mathematical notion of \emph{execution}.

\begin{defi}\label{def:alt_play}
An \textbf{alternating play} on arena $A$ is a sequence $s = s_1
\dots s_n$ which is:
\[  
\begin{array}{ll}
\text{\emph{valid}:}& \forall 1\leq i \leq n,~\{s_1, \dots, s_i\} \in
\conf{A}\,,\\
\text{\emph{non-repetitive}:}& \forall 1\leq i,j \leq n,~s_i =s_j
\implies i=j\,,\\
\text{\emph{alternating}:}& \forall 1\leq i \leq n-1,~\pol_A(s_i)\neq
\pol_A(s_{i+1})\,,\\
\text{\emph{negative}:} & \text{if $n\geq 1$, then $\pol_A(s_1) = -$.}
\end{array}
\]

We write $\Alt(A)$ for the set of alternating plays on $A$.
\end{defi}

The notation $\Alt(A)$ means to suggest that an alternating play
has two possible states: O if $s$ has even length and the
last move (if any) is by Player, and P otherwise: each new move
transitions between them.  We denote the empty play with $\varepsilon$,
and the prefix ordering with $\prefix$. In the sequel we sometimes
apply $\Alt(-)$ to esps other than arenas.

Plays record individual executions, by giving a
chronological account of events observed throughout computation. For
instance, Figure \ref{ex:play1} displays a play on the arena $(\tunit
\lin \tunit) \lin \tbool$ of Figure \ref{ex:arena1}. It is also read
from top to bottom. Each move corresponds to a node in Figure
\ref{ex:arena1} -- as each move in the arena corresponds to a given type
component, the identity of each move in Figure \ref{ex:play1} is
signified by its position under the matching type component.

\subsubsection{Strategies} Given a term of type $A$ we may, given the
adequate technical machinery, ask whether a given play describes a valid
execution for that term. The play of Figure \ref{ex:play1}, for
instance, describes a valid execution for $\lambda f^{\tunit
\to \tunit}.\,f\,\tskip;\,\ttrue : (\tunit \to \tunit) \to \tbool$:
after Opponent starts computation, reduction immediately
gets stuck with a variable $f$ in head position. This is an observable
event, corresponding to Player calling its argument with $\run^+$. Then,
Opponent proceeds to call his argument with $\run^-$, triggering the
evaluation of the subterm $\tskip$. This (trivially) converges to a
value, which is observable and corresponds to $\done^+$. The control
goes back to $f$ (Opponent), which evaluates to $\tskip$ as well via
observable $\done^-$. This triggers the evaluation of $\ttrue$, leading
to the observable $\ttrue^+$ that terminates computation.

Figure \ref{ex:play1} represents one possible execution of
$\lambda f^{\tunit \to \tunit}.\,f\,\tskip;\,\ttrue : (\tunit \to
\tunit) \to \tbool$. In general a term is represented by a
\emph{strategy}, which aggregates all possible executions.

\begin{defi}\label{def:alt_strat}
A \textbf{alternating strategy} $\sigma : A$ on $-$-arena $A$ is
$\sigma \subseteq \Alt(A)$ which is:
\[  
\begin{array}{ll}
\text{\emph{non-empty}:} & \varepsilon \in \sigma\\
\text{\emph{prefix-closed}:} & \forall s\sqsubseteq s' \in \sigma,~s \in
\sigma\\
\text{\emph{deterministic}:} & \forall s a_1^+, s a_2^+ \in \sigma,~a_1
= a_2\\
\text{\emph{receptive}:} & \forall s \in \sigma,~sa^- \in \Alt(A)
\implies sa \in
\sigma
\end{array}
\]

An \textbf{alternating prestrategy} $\sigma : A$ satisfies
all these conditions except for \emph{receptive}.
\end{defi}

In this definition we have started using a convention followed
throughout this paper: when introducing an event, we sometimes
annotate it with a superscript to indicate its polarity. For instance,
``$\forall s a_1^+ \in \sigma, \dots$ is shorthand for $\forall s a_1
\in \sigma$ s.t. $\pol(a_1) = +, \dots$''.

We will see later on how to compute the strategy for a
term. It is a strength of game semantics that this may be done either
compositionally by induction on the syntax following the methodology of
denotational semantics, or operationally via an abstract machine
\cite{DBLP:journals/entcs/GhicaT12}.

\subsection{Replication and symmetry} \label{subsec:replication}
In this paper we introduce early on the machinery for
replication.  It requires a small jump in abstraction, but
fixes the arenas once and for all.

\subsubsection{Arenas with symmetry}\label{subsubsec:ar_sym} 
Figure \ref{ex:arena1} displays the arena corresponding to affine
executions\footnote{Affineness is enforced by \emph{non-repetitive}
in Definition \ref{def:alt_play}. Rather than expand arenas,
it is tempting to simply lift it. For this to be sound, it
becomes then necessary to include additional structure in plays: the
\emph{justification pointers}. This is the choice made in \emph{HO
games} \cite{ho}. This will be detailed in Section~\ref{subsec:pointers}.} on type $(\tunit \to \tunit) \to \tbool$.  To go
beyond affineness, we \emph{expand} the arena to allow multiple calls to
arguments -- for $(\tunit \to \tunit) \to \tbool$, we obtain an infinite
arena as drawn in Figure \ref{fig:replication}.

\begin{figure}
\begin{minipage}{.47\linewidth}
\begin{center}
\includegraphics[scale=.6]{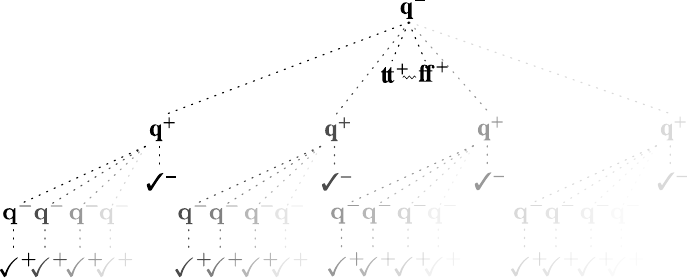}
\end{center}
\caption{$(\tunit \to \tunit) \to \tbool$ with replication}
\label{fig:replication}
\end{minipage}
\hfill
\begin{minipage}{.48\linewidth}
\[
\scalebox{.71}{$
\xymatrix@R=10pt@C=10pt{
&&&\qu^-
\ar@{.}[dll]
\ar@{.}[d]
\ar@{.}[drr]\\
&\qu^+_{\grey{0}}       \ar@{.}[dl]
\ar@{.}[d]
\ar@{.}[dr]
&&\ttrue^+&&\qu^+_{\grey{1}}
\ar@{.}[dl]
\ar@{.}[d]
\ar@{.}[dr]\\
\qu^-_{\grey{{0,0}}}
\ar@{.}[d]
&\qu^-_{\grey{0,1}}
\ar@{.}[d]
&\done^-_{\grey{0}}&&\qu^-_{\grey{1,0}}
\ar@{.}[d]
&\qu^-_{\grey{1,1}}
\ar@{.}[d]&\done^-_{\grey{1}}\\
\done^+_{\grey{0,0}}&\done^+_{\grey{0,1}}&&&
\done^+_{\grey{1,0}}&\done^+_{\grey{1,1}}
}$}
\]
\caption{Configuration of $\oc (\oc \tunit \lin \tunit) \lin \tbool$}
\label{fig:exp_conf}
\end{minipage}
\end{figure}
In the picture, it seems like \emph{e.g.} all moves $\qu^+$ are
interchangeable. This is true in spirit but every move must
be a distinct event of the arena.
Concretely, the expanded arena is computed following the methodology of
linear logic: the type $(\tunit \to \tunit) \to \tbool$ is represented
by $\oc (\oc \tunit \lin \tunit) \lin \tbool$ rather than
$(\tunit \lin \tunit) \lin \tbool$. Here, $\oc$ is an \emph{exponential
modality} in the sense of linear logic. The full definition of $\oc A$
will appear in Definition \ref{def:bang}, but its events are $\ev{\oc A}
= \mathbb{N} \times \ev{A}$, pairs $(n,a)$ where $n$ is called a
\emph{copy index}. So in reality, a precise picture of the arena for
$(\tunit \to \tunit) \to \tbool$ with replication would be a version of
Figure \ref{fig:replication} where some events are tagged by 
\emph{copy index} -- see Figure \ref{fig:exp_conf} for an
example of a configuration of $\oc (\oc \tunit \lin \tunit) \lin \tbool$
with explicit indices pictured as grey subscripts.

Expanding the arena so opens up the way to replication
without compromising the \emph{non-repetitive} condition: a strategy may
replay the ``same'' move but with different copy indices.
But then, it is necessary to identify strategies behaving in the same way
save for the choice of copy indices. To that end, following
the approach initiated in \cite{lics14} we enrich arenas with a notion
of \emph{symmetry}, capturing reindexings between configurations. 

\begin{defi}\label{def:isofam}
An \textbf{isomorphism family} on event structure $E$ is a set
$\tilde{E}$ of
bijections between configurations of $E$, satisfying the additional
conditions:
\[
\begin{array}{rl}
\text{\emph{groupoid:}}& \tilde E\text{~contains identity bijections;
is closed
under composition and inverse.}\\
\text{\emph{restriction:}}&
\text{for all~$\theta : x \bij y \in \tilde{E}$ and $x \supseteq x' \in
\conf{E}$,}\\
&\text{there is a (necessarily) unique $\theta \supseteq \theta' \in
\tilde{E}$
such that $\theta' : x' \bij y'$.}\\
\text{\emph{extension:}}&\text{for all $\theta : x \bij y \in
\tilde{E}$, $x
\subseteq x' \in \conf{E}$,}\\
&\text{there is a (not necessarily unique) $\theta \subseteq \theta' \in
\tilde{E}$ such that $\theta' : x' \bij y'$.}
\end{array}
\]

Then $(E, \tilde E)$ is an \textbf{event structure with
symmetry (\emph{ess})}. If $A$ has polarities preserved by $\tilde{A}$,
$A$ is an \textbf{event structure with symmetry and polarities
(\emph{essp})}.
\end{defi}

If $A$ is an ess, we refer to the elements of $\tilde{A}$ as
\textbf{symmetries}. We write $\theta : x \sym_{A} y$ to mean that
$\theta : x \simeq y$ is a bijection such that $\theta \in \tilde{A}$,
and write $x = \dom(\theta)$ and $y = \cod(\theta)$.  It is an easy
exercise to prove that symmetries are automatically order-isomorphisms
\cite{symmetry}, where configurations inherit a partially ordered
structure from the causal dependency of $A$.  We regard isomorphism
families as \emph{proof-relevant} equivalence relations: they convey the
information of which configurations are interchangeable, witnessed by an
explicit bijection.

From now on, \textbf{arenas} have an isomorphism
family. It comprises only identity symmetries on basic arenas $\tunit,
\tbool, \tnat$ and $\emptyar$. The
previous constructions on arenas extend transparently:
$A^\perp$ has the same symmetries as $A$. The symmetries on $A\parallel
B$ are those of the form
\[
\begin{array}{rcrcl}
\theta_A \parallel \theta_B &:& x_A \parallel x_B &\sym& y_A \parallel
y_B\\
&& (1, a)&\mapsto &(1, \theta_A(a))\\
&& (2, b)&\mapsto &(2, \theta_B(b))
\end{array}
\]
for $\theta_A : x_A \sym_A y_A$ and $\theta_B : x_B \sym_B y_B$. Those
on $A \with B$ are the symmetries on $A\parallel B$ that are bijections
between configurations of $A\with B$, \emph{i.e.} one of $\theta_A$ and
$\theta_B$ must be empty. Note that these constructions
$\parallel$ and $\with$ apply to arbitrary event structures with
symmetry.

For $A \lin B$, if $x, y \in
\conf{A\lin B}$ and $\theta : x \simeq y$ is any bijection, defining
first $\chi_{A, B}\,\theta$ as
\[
\chi_{A,B}\,x \stackrel{\chi_{A,B}^{-1}}{\simeq} x
\stackrel{\theta}{\simeq} y \stackrel{\chi_{A,B}}{\simeq}
\chi_{A,B}\,y\,,
\]
we set $\theta : x \sym_{A\lin B} y$ when $\theta$ is an
order-isomorphism satisfying $\chi_{A,B}\,\theta : \chi_{A,B}\,x
\sym_{A^\perp \parallel B} \chi_{A,B}\,y$.

The main arena construction introducing new symmetric events is the
\textbf{exponential}:

\begin{defi}\label{def:bang}
Let $A$ be a $-$-arena. The $-$-arena $\oc A$ has components
\[
\begin{array}{rrcl}
\text{\emph{events}:} & \ev{\oc A} &=& \tnat \times \ev{A}\\
\text{\emph{causality}:} & (i, a) \leq_{\oc A} (j, a') &\Leftrightarrow&
i=j~\&~a\leq_A a'\\
\text{\emph{conflict}:} & (i, a) \conflict_{\oc A} (j, a')
&\Leftrightarrow&
i=j~\&~a \conflict_A a'\\
\text{\emph{polarities}:} & \pol_{\oc A}(i, a) &=& \pol_A(a)
\end{array}
\]
along with isomorphism family comprising as symmetries those bijections
of the form
\[
\begin{array}{rcrcl}
\theta &:& \parallel_{n\in \mathbb{N}} x_n &\iso& \parallel_{n\in
\mathbb{N}} y_n\\
&& (n, a) &\mapsto& (\pi(n), \theta_n\,a)
\end{array}
\]
for some permutation $\pi \in \varsigma(\mathbb{N})$ and some family
$(\theta_n)_{n\in \mathbb{N}}$ with $\theta_n : x_n \sym_{A} y_{\pi(n)}$
for all $n \in \mathbb{N}$.
\end{defi}

This definition applies in general to any ess.  Figure \ref{fig:oc_unit}
shows the plain \emph{esp} of $\oc \tunit$ with copy
indices indicated as grey subscripts -- its symmetries are all
order-isomorphisms between configurations.  While $\oc(-)$ does not
match a type construction of $\PCF$, we shall follow
Girard \cite{DBLP:journals/tcs/Girard87} and define the arrow type of
arenas \emph{with} replication as $A \to B  = ~ \oc A \lin B$.

\subsubsection{Symmetry on plays and strategies} Symmetry will allow us
to identify strategies, but it should also
affect how strategies play. In the presence of explicit copy indices, a
fundamental property is \emph{uniformity}. Intuitively, a strategy is
uniform if its behaviour does not depend (up to symmetry) on the
specific copy indices used by its environment. 

The first step towards uniformity is to transport symmetry to plays.

\begin{defi}\label{def:alt_play_equiv}
Let $A$ be an arena and $s, t \in \Alt(A)$. We say that $s$ and $t$ are
\textbf{symmetric}, written $s \sym_A t$, if $s$ and $t$ have the same
length, and we have
\[
\theta_{s,t}^j = \{(s_i,t_i)\mid 1\leq i \leq j\} : \{s_1, \dots, s_j\}
\sym_A \{t_1,
\dots, t_j\}
\]
a symmetry in $\tilde{A}$ for all $1\leq j \leq n$; 
writing $s = s_1 \dots s_n$ and $t = t_1 \dots t_n$.
\end{defi}

Those readers familiar with AJM games may find comfort in the following
fact.

\begin{fact}\label{fact:alt_ajm}
For an arena $A$, the tuple $\tuple{\ev{A}, \pol_A, \Alt(A), \sym_A}$ is
an AJM game \cite{ajm}.
\end{fact}

This ignores the Question/Answer labeling in
AJM games, which we shall handle later on. The proof
is a straightforward exercise. For the experts, we mention that this
association of arenas to AJM games does not respect the arena
constructions because constructions on AJM games enforce local
alternation, while $\Alt(-)$ does not. As in HO games \cite{ho}, in our
presentation local alternation will only follow from the P-visibility
condition.

From the connection with AJM games it seems natural to import the AJM
uniformity:

\begin{defi}\label{def:simstrat}
For $A$ an arena and $\sigma, \tau : A$ alternating prestrategies,
we write
$\sigma \simstrat \tau$ iff:
\[
\begin{array}{rl}
\text{$\to$-\emph{simulation:}} & \forall sa^+ \in \sigma,~t\in \tau,~s
\sym_A t
\implies \exists b^+,~tb^+ \in \tau~\&~sa^+ \sym_A tb^+\\
\text{$\ot$-\emph{simulation:}} & \forall s \in \sigma,~tb^+\in \tau,~s
\sym_A t
\implies \exists a^+,~sa^+ \in \sigma~\&~sa^+ \sym_A tb^+\\
\text{$\to$-\emph{receptive:}} & \forall sa^- \in \sigma,~t \in
\tau,~sa^-\,\sym_A\,tb^- \implies tb^- \in \tau\\ 
\text{$\ot$-\emph{receptive:}} & \forall s \in \sigma,~tb^- \in
\tau,~sa^-\,\sym_A\,tb^- \implies sa^- \in \sigma
\end{array}
\]

This defines a per $\simstrat$ on prestrategies\footnote{For 
\emph{strategies}, $\to, \ot$-\emph{receptive} are
subsumed by \emph{receptive}. But these are necessary for uniformity to
apply to \emph{prestrategies} which might not be receptive -- this
generalization
will be used in the technical development.} on $A$.
A prestrategy $\sigma : A$ is \textbf{uniform} iff $\sigma \simstrat
\sigma$.
\end{defi}

Uniformity is crucial. For the interpretation 
to respect $\beta$-equivalence we must identify strategies that play the
``same moves'', but with different copy indices. For instance, we must
consider equal the two strategies $\tau_0, \tau_1 : \tunit \to \tunit$
with unique maximal play:
\[
\xymatrix@R=-5pt@C=15pt{
\tau_0&:& \oc \tunit    \ar@{}[r]|\lin& \tunit\\
&&&\qu^-\\
&&\qu_{\grey{0}}^+\\
&&\done^-_{\grey{0}}\\
&&&\done^+
}
\qquad
\qquad
\qquad
\qquad
\xymatrix@R=-5pt@C=15pt{
\tau_1 &:& \oc \tunit      \ar@{}[r]|\lin& \tunit\\
&&&\qu^-\\
&&\qu_{\grey{1}}^+\\
&&\done^-_{\grey{1}}\\
&&&\done^+
}
\]

But this quotient is risky. Let us apply both $\tau_0$ and $\tau_1$ to
$\sigma : \oc \tunit$ with only maximal play $\qu_{\grey{0}}^-
\done^+_{\grey{0}}$.  Though we have yet to define composition, 
the application of $\tau_0$ to $\sigma$ must
converge, while that of $\tau_1$ to $\sigma$ must diverge. So
$\tau_0$ and $\tau_1$ cannot be safely identified as they are
distinguishable.  In fact here, the
culprit is $\sigma$: it is not \emph{uniform}. Since $\qu^-_{\grey{0}}
\done_{\grey{0}}^+ \in \sigma$, uniformity of $\sigma$ would imply that
$\qu^-_{\grey{1}} \done_{\grey{1}}^+ \in \sigma$ as well, breaking the
counter-example.

From now on, all \textbf{alternating (pre)strategies} are assumed
uniform.

\subsection{Interpretation of $\PCF$}
The interpretation follows the
methodology of denotational semantics, resting on the fact that arenas
and strategies form a category with adequate structure.  In the main
text we only outline this fairly routine construction -- though this
should be enough to read the paper -- but the construction
is detailed in Appendix \ref{app:alt}.

\subsubsection{Category}
The category $\NegAlt$ has objects $-$-arenas, and morphisms 
the \emph{alternating strategies} (\emph{strategies} for
short) on $A \lin B$. The \textbf{composition} of $\sigma : A \lin B$ and $\tau : B \lin C$
\[
\tau \odot \sigma : A \lin C
\]
follows the usual game semantics process of \emph{parallel
interaction} followed by \emph{hiding}.

First, the \textbf{pre-interactions} are
sequences $u \in \ev{(A\lin B) \lin C}^*$ satisfying
\emph{valid} of Definition \ref{def:alt_play}. A pre-interaction $u$
has three \textbf{restrictions}, with the following types:
\[
u \restrict A, B \in \ev{A\lin B}^*\,,
\qquad
u \restrict B, C \in \ev{B\lin C}^*\,,
\qquad
u \restrict A, C \in \ev{A\lin C}^*\,,
\]
defined in the obvious way -- see Appendix \ref{app:alt_comp}. Given
prestrategies $\sigma : A \lin B$ and $\tau : B \lin C$, an
\textbf{interaction} $u \in \tau \inter \sigma$ is a pre-interaction $u
\in \ev{(A \lin B) \lin C}^*$ satisfying:
\[
u \restrict A, B \in \sigma\,,
\qquad
u \restrict B, C \in \tau\,,
\qquad
u \restrict A, C \in \Alt(A\lin C)\,.
\]

The \textbf{composition} of $\sigma$ and $\tau$ comprises
all $s \in \Alt(A\lin C)$ with a \textbf{witness}:
\[
\tau \odot \sigma = \{u \restrict A, C \mid u \in \tau \inter
\sigma\}\,;
\]
it follows that $\tau \odot \sigma : A \lin C$ is a prestrategy; and a
strategy if $\sigma$ and $\tau$ are.

Composition is associative on prestrategies, but admits identities only
for strategies: the \textbf{copycat} strategies. If $A$ is a
$-$-arena and $s \in \ev{A \lin A}$, there are left and right
\emph{restrictions}
\[
s \restrict \labl \in \ev{A}^*\,,
\qquad
\qquad
s \restrict \labr \in \ev{A}^*\,,
\]
defined in the obvious way (see Appendix \ref{app:alt_id}). For $s \in
\Alt(A\lin A)$, $s$ is a \emph{copycat play} iff \emph{(1)} for all
even-length prefix $s' \prefix s$ we have $s' \restrict \labl = s'
\restrict \labr$, and \emph{(2)} for all $(1, (a_1, a_2)) \in
\ev{s}$, if $a_2 \in \min(A)$, then $a_1 = a_2$ -- a move initial on the
left must be justified by the same move on the right. Writing
$\cc_A$ for the set of all copycat plays, we have $\cc_A : A \lin A$ a
strategy as required. For any strategy $\sigma : A \lin B$ we have
$\cc_B \odot \sigma \odot \cc_A = \sigma$, making $\NegAlt$ a category.

\begin{rem}
Our model shares with AJM games \cite{ajm} the management
of the equivalence $\simstrat$ on strategies.  All our
constructions on strategies must preserve $\simstrat$. For
most of them it is clear, but composition requires some care (see
Appendix \ref{app:basic_alt_sym}). Operations on strategies therefore
lift transparently to $\simstrat$-equivalence classes, and one can then
consider $\NegAlt$ to have as morphisms $\simstrat$-equivalence classes
of strategies (as is done in \cite{ajm}). This is fine, but it does
contrast with how we (also following the practice in AJM games) often
refer to specific concrete strategies as being ``the interpretation of''
specific terms. So we refrain from quotienting, and consider $\NegAlt$
as having \emph{concrete} strategies as morphisms, and homsets
$\NegAlt(A, B)$ additionally equipped with an equivalence relation
$\simstrat$ which all operations preserve.
This way the interpretation of terms yields \emph{concrete}
representatives, but categorical laws only hold up to $\simstrat$. In
the sequel we refer only to the plain algebraic structures (as in
``symmetric monoidal closed category'', ``cartesian closed category'',
\emph{etc}), with it being understood that laws for these algebraic
structures only hold up to $\simstrat$ and that for any construction we
consider, there is a proof obligation that it preserves $\simstrat$.
\end{rem}

\subsubsection{Further structure}
If $A$ and $B$ are $-$-arenas, their \emph{tensor} is simply
$A\tensor B = A\parallel B$ their parallel composition. For $\sigma_1
: A_1 \lin B_1$ and $\sigma_2 : A_2 \lin B_2$, the strategy
\[
\sigma_1 \tensor \sigma_2 : A_1 \tensor A_2 \lin B_1 \tensor B_2\,,
\]
defined via adequate restrictions (see Appendix
\ref{app:alt_tensor}), plays as $\sigma_1$ on $A_1, B_1$ and
$\sigma_2$ on $A_2, B_2$ -- this gives a symmetric monoidal
structure, with structural isomorphisms copycat strategies.
Moreover, $\NegAlt$ is \emph{cartesian}. Its \emph{terminal object} is
the empty $-$-arena $\ees$; the \textbf{product} of $A$ and $B$ is
the $A \with B$. This forms a \emph{cartesian product}: there
are projections
\[
\pi_A : A \with B \lin A
\qquad
\qquad
\pi_B : A \with B \lin B
\]
acting as copycat, and for $\sigma : C \lin A$ and $\tau : C \lin B$,
their \textbf{pairing} $\tuple{\sigma, \tau} : C \lin A\with B$ is
defined simply as the as the set-theoretic union of $\sigma$ and $\tau$
(modulo the obvious relabeling).

Finally, for any $-$-arenas $\Gamma, A$ and
$B$, there is an iso
$\Gamma \tensor A \lin B \iso \Gamma \lin (A \lin B)$,
\emph{i.e.} a bijection on events preserving and respecting all
structure. This yields a bijection
\[
\Lambda_{\Gamma, A, B} 
:
\NegAlt(\Gamma \tensor A, B) 
\simeq
\NegAlt(\Gamma, A \lin B)
\]
between the corresponding sets of strategies. Exploiting this,
we define the \textbf{evaluation}
\[
\evm_{A, B} = \Lambda_{A\lin B, A, B}^{-1}(\cc_{A\lin B}) : (A\lin B)
\tensor A
\lin B\,,
\]
and the universal property for monoidal closure is then a direct
verification. We conclude:

\begin{prop}
The category $\NegAlt$ is cartesian and symmetric monoidal closed.
\end{prop}

\subsubsection{Exponential}\label{subsubsec:seely}
On $\NegAlt$, $\oc$ gives an \emph{exponential} in the sense of Linear
Logic \cite{DBLP:journals/tcs/Girard87}: a functor, with
natural transformations $\der_A : \oc A \lin A$ and $\dig_A : \oc A \lin
\oc \oc A$ making $(\oc, \der, \dig)$ a comonad.  Moreover, there are
natural isomorphisms $\mon_{A, B}^2 : \oc A \tensor \oc B \lin \oc (A
\with B)$ and $\mon_{A, B}^0 : \ees \lin \oc \ees$, satisfying the
coherence laws of a \emph{Seely category} \cite{panorama}. So
the Kleisli category $\NegAlt_\oc$ is cartesian closed, and hence a
model of the simply-typed $\lambda$-calculus.  The construction is
routine, and follows the lines of AJM games \cite{ajm} -- see Appendix
\ref{app:alt_sym}.

In the sequel, given a Seely category $\C$ and a
morphism $f \in \C(\oc A, B)$, we shall write $f^\dagger \in \C(\oc A,
\oc B)$ for its \textbf{promotion}, defined as $\oc f \circ\, \dig_A$ --
in particular, recall that Kleisli composition of $f \in \C_\oc(A,B)$
and $g \in \C_\oc(B, C)$ may be defined as $g \circ_\oc f = g \circ
f^\dagger \in \C_\oc(A, C)$.

\subsubsection{Recursion}\label{subsubsec:recursion}
Strategies on arena $A$ may be partially ordered by inclusion; this
forms a \emph{pointed dcpo}. All operations on strategies are continuous
with respect to $\subseteq$.

Writing $\NegAlt_\oc(A, B)$ for the dcpo of strategies on $A \to B
=\,\oc A \lin B$, the operation
\[
\begin{array}{rcrcl}
F &:& \NegAlt_\oc(\emptyar, (A\to A)\to A) &\to& \NegAlt_\oc(\emptyar, (A\to A)\to A)\\
&& \sigma &\mapsto& \lambda f^{A \to A}.\,f\,(\sigma\,f)
\end{array}
\]
written in $\lambda$-calculus syntax following the cartesian closed
structure of $\NegAlt_\oc$, is continuous. Its least fixed point $\Y_A
\in \NegAlt_\oc(\emptyar, (A\to A)\to A)$ is transported to
$\NegAlt_\oc(\Gamma, (A\to A)\to A)$ by composition with the terminal
projection.  For each $\sigma \in \NegAlt_\oc(\Gamma, A \to A)$,
\[
\Y_A\,\sigma \simstrat \sigma\,(\Y_A\,\sigma)
\]
so a fixed point operator up to $\simstrat$, as needed to interpret
recursion.

It is a curiosity already in AJM games \cite{ajm} that the
recursive equation for the fixpoint combinator must be solved in the
domain of \emph{concrete} strategies, rather than
$\simstrat$-equivalence classes.  To the best of our knowledge it is not
known if the partial order induced by inclusion on
$\simstrat$-equivalence classes of strategies has the adequate
completeness properties to solve this, \emph{i.e.} if the
quotient of $\NegAlt$ and $\NegAlt_\oc$ by $\simstrat$ are dcpo-enriched
categories.

\subsubsection{Interpretation}\label{subsubsec:intr_pcf}
\emph{Types} of $\PCF$ are interpreted as
$-$-arenas: we set $\intr{\tunit} = \tunit$, $\intr{\tbool} = \tbool$,
$\intr{\tnat} = \tnat$ and $\intr{A\to B} = \oc \intr{A} \lin \intr{B}$
yielding for any type $A$ an arena $\intr{A}$. A \emph{context} $\Gamma
= x_1 :
A_1, \dots, x_n : A_n$ is interpreted as $\intr{\Gamma} =
\bigwith_{1\leq i
\leq n} \intr{A_i}$. A term $\Gamma \vdash M : A$ yields
\[
\intr{M} \in \NegAlt_\oc(\intr{\Gamma}, \intr{A})\,.
\]

We skip the details of the interpretation of the $\lambda$-calculus
combinators, which follows the standard interpretation of
the simply-typed $\lambda$-calculus in a cartesian closed category
\cite{lambekscott}.

We specify strategies for $\PCF$ combinators. For constants,
$\intr{\tskip} : \tunit, \intr{\ttrue} : \tbool$, $\intr{\tfalse} :
\tbool$ and $\intr{n} : \tnat$ are the corresponding obvious strategies
replying immediately the corresponding value. For the others the
interpretation is in Figure \ref{fig:intr_basic_pcf}, annotating
strategy operations with $\oc$ to emphasize that they are in the Kleisli
category $\NegAlt_\oc$. 
\begin{figure}
\[
\scalebox{.9}{$
\xymatrix@R=-5pt@C=0pt{
\mathsf{seq} &\!\!\!\!\!\!:\!\!\!\!\!\!& \oc (\tunit
&\with&\tunit)&\lin& \tunit\\
&&&&&&\qu^-\\
&&\qu^+_{\grey{0}}\\
&&\done^-_{\grey{0}}\\
&&&&\qu^+_{\grey{1}}\\
&&&&\done^-_{\grey{1}}\\
&&&&&&\done^+
}$}
\qquad
\scalebox{.9}{$
\xymatrix@R=-5pt@C=-1pt{
\mathsf{if} &\!\!\!\!\!\!:\!\!\!\!\!\!& \oc (\tbool & \with& \tx & \with
&\tx)
&\lin& \tx\\
&&&&&&&&\qu^-\\
&&\qu^+_{\grey{0}}\\
&&\ttrue^-_{\grey{0}}\\
&&&&\qu^+_{\grey{1}}\\
&&&&v^-_{\grey{1}}\\
&&&&&&&&v^+
}$}
\qquad
\scalebox{.9}{$
\xymatrix@R=-5pt@C=-1pt{
\mathsf{if} &\!\!\!\!\!\!:\!\!\!\!\!\!& \oc (\tbool &\with & \tx & \with
&\tx)
&\lin& \tx\\
&&&&&&&&\qu^-\\
&&\qu^+_{\grey{0}}\\
&&\tfalse^-_{\grey{0}}\\
&&&&&&\qu^+_{\grey{1}}\\
&&&&&&v^-_{\grey{1}}\\
&&&&&&&&v^+
}$}
\]
\medskip
\[
\scalebox{.9}{$
\xymatrix@R=-5pt@C=0pt{
\mathsf{succ} &\!\!\!\!\!\!:\!\!\!\!\!\!&
\oc \tnat &\lin& \tnat\\
&&&&\qu^-\\
&&\qu^+_{\grey{0}}\\
&&n^-_{\grey{0}}\\
&&&&(n+1)^+
}$}
\qquad\qquad
\scalebox{.9}{$
\xymatrix@R=-5pt@C=0pt{
\mathsf{iszero} &\!\!\!\!\!\!:\!\!\!\!\!\!&
\oc \tnat &\lin & \tbool\\
&&&&\qu^-\\
&&\qu^+_{\grey{0}}\\
&&0_{\grey{0}}^-\\
&&&&\ttrue^+
}$}
\qquad\qquad
\scalebox{.9}{$
\xymatrix@R=-5pt@C=5pt{
\mathsf{iszero} &\!\!\!\!\!\!:\!\!\!\!\!\!&
\oc \tnat &\lin &\tbool\\
&&&&\qu^-\\
&&\qu^+_{\grey{0}}\\
&&(n+1)_{\grey{0}}^-\\
&&&&\tfalse^+
}$}
\]
\caption{Basic strategies for $\PCF$}
\label{fig:intr_pcf}
\end{figure}
\begin{figure}
\begin{minipage}{.5\linewidth}
\begin{eqnarray*}
\intr{M;\,N} &=& \mathsf{seq}\odot_\oc \tuple{\intr{M}, \intr{N}}\\
\intr{\ite{M}{N_1}{N_2}} &=& \mathsf{if}\odot_\oc \tuple{\intr{M},
\intr{N_1},
\intr{N_2}}\\
\intr{\tsucc\,M} &=& \mathsf{succ} \odot_\oc \intr{M}\\
\intr{\tpred\,M} &=& \mathsf{pred} \odot_\oc \intr{M}\\
\intr{\iszero\,M} &=& \mathsf{iszero} \odot_\oc \intr{M}\\
\intr{\tlet{x}{N}{M}} &=& \mathsf{let}_{\tx,\ty} \odot_\oc
\tuple{\intr{N},
\Lambda^\oc(\intr{M})}
\end{eqnarray*}
\caption{Interpretation of basic combinators}
\label{fig:intr_basic_pcf}
\end{minipage}
\hfill
\begin{minipage}{.46\linewidth}
\vspace{0pt}
\[
\scalebox{.9}{$
\xymatrix@R=-1pt@C=-3pt{~\\
\mathsf{pred} &\!\!\!\!\!\!:\!\!\!\!\!\!&
\oc \tnat &\lin& \tnat\\
&&&&\qu^-\\
&&\qu^+_{\grey{0}}\\
&&0^-_{\grey{0}}\\
&&&&0^+\\~\\~
}$}
\qquad
\scalebox{.9}{$
\xymatrix@R=-1pt@C=-3pt{~\\
\mathsf{pred} &\!\!\!\!\!\!:\!\!\!\!\!\!&
\oc \tnat &\lin& \tnat\\
&&&&\qu^-\\
&&\qu^+_{\grey{0}}\\
&&(n+1)^-_{\grey{0}}\\
&&&&n^+\\~\\~
}$}
\]
\caption{Strategy for $\pred$}
\label{fig:strat_pred}
\end{minipage}
\end{figure}
The strategies used are
in Figures \ref{fig:intr_pcf}, \ref{fig:strat_pred} and
\ref{fig:let}. Save for $\mathbf{let}$, the diagram displays
exhaustively their maximal plays, defining them completely. For
$\mathbf{let}$, the strategy implements a \emph{memoization} mechanism:
it evaluates on $\tx$ obtaining a value $v$, which is then fed to the
function argument each call, without re-evaluating it.
The play shown for $\mathbf{let}$ is not
maximal as Opponent could play some $\qu_{\grey{1, i_{n+1}}}^-$ after.
We will see in Section~\ref{subsec:vis_inn} that it \emph{is} fully
informative: there is only one \emph{innocent} strategy that includes
these plays. Finally, the interpretation of recursion is set with
$\intr{\Y_A\,M} = \Y_A \odot_\oc \intr{M}$, completing the definition.

This interpretation satisfies the main property expected of a
denotational semantics:

\begin{prop}[Adequacy]\label{prop:adequacy}
For any $\vdash M : \tunit$, $M \eval$ if and only if $\intr{M}\eval$.
\end{prop}

Note there are only two strategies on $\tunit$: the minimal
$\{\varepsilon, \qu^-\}$ matching any
\emph{diverging} program, and the \emph{converging}
$\{\varepsilon, \qu^-, \qu^-\done^+\}$. For $\sigma : \tunit$,
we write $\sigma \eval$ if $\sigma$ converges and $\sigma \div$ if
$\sigma$ diverges. We omit the proof which is standard using
logical relations, see \emph{e.g.} \cite{ho}.

This immediately entails soundness for observational equivalence:

\begin{cor}
Let $\Gamma \vdash M, N : A$ be any terms of $\PCF$. If $\intr{M}
\simstrat \intr{N}$,
then $M \obs N$.
\end{cor}
\begin{proof}
Assume $\intr{M} \simstrat \intr{N}$, and consider a context $C[-]$ such
that $C[M]\eval$. By Proposition \ref{prop:adequacy},
$\intr{C[M]}\eval$.  But $\intr{C[M]} = \intr{C[-]}\odot_\oc \intr{M}
\simstrat \intr{C[-]}\odot_\oc \intr{N} = \intr{C[N]}$, so $C[N] \eval$
by Proposition \ref{prop:adequacy}. The other direction also holds,
hence $M \obs N$.
\end{proof}

Computational adequacy is the standard to express that a model
accurately describes computation in the language. In fact in game
semantics the connection with operational semantics is much stronger, as
highlighted earlier. We will elaborate on that in Section~\ref{sec:seq_effects}.

\subsection{HO games}\label{subsec:pointers}
Before exploring this computational content, we highlight
the connection with HO games \cite{ho}, based on representing
plays up to symmetry as \emph{plays with pointers}.

\subsubsection{Plays with pointers.}
First, a convention. For $A$ a $-$-arena and $s \in
\Alt(A)$, then $\ev{s} \in \conf{A}$ has two order structures: it is
totally ordered chronologically as prescribed by $s$, and has a partial
order imported from $\leq_A$. When representing plays, we often
annotate them with the immediate causal dependency generating
$\leq_A$. For instance, Figure \ref{fig:plays_pointers} shows it for 
\begin{figure}
\begin{minipage}{.32\linewidth}
\[
\xymatrix@C=0pt@R=-5.5pt{
\oc(\tx &\with& (\tx &\to& \ty)) &\lin& \ty\\
&&&&&&\qu^-\\
\qu^+_{\grey{0}}\\
v^-_{\grey{0}}\\
&&&&\qu^+_{\grey{1}}\\
&&\qu^-_{\grey{1,i_1}}\\
&&v^+_{\grey{1, i_1}}\\~\\~\\
&&\dots\\~\\~\\~\\
&&\dots\\~\\~\\~\\
&&\qu^-_{\grey{1, i_n}}\\
&&v^+_{\grey{1, i_n}}\\
&&&&w_{\grey{1}}^-\\
&&&&&&w^+
}
\]
\caption{Typical play of $\mathsf{let}$}
\label{fig:let}
\end{minipage}
\hfill
\begin{minipage}{.6\linewidth}
\[
\scalebox{.7}{$
\raisebox{100pt}{$
\xymatrix@R=-4pt{
(\tunit
        \ar@{}[r]|\to&
\tunit)   \ar@{}[r]|\to&
\tbool\\
&&\qu^-\\
&\qu^+_{\grey{0}}   
\\
\qu^-_{\grey{{0,0}}}
\\
\done^+_{\grey{{0,0}}}
\\
\qu^-_{\grey{0,1}}
\\
\done^+_{\grey{{0,1}}}
\\
&\done^-_{\grey{0}}
\\
&\qu^+_{\grey{1}}   
\\
\qu^-_{\grey{1, 0}}
\\
\done^+_{\grey{1, 0}}
\\
\qu^-_{\grey{1, 1}}
\\
\done^+_{\grey{1, 1}}
\\
&\done^-_{\grey{1}}
\\
&&\ttrue^+
}$}
\qquad\qquad
\leadsto
\qquad\qquad
\raisebox{100pt}{$
\xymatrix@R=-4pt{
(\tunit
        \ar@{}[r]|\to&
\tunit)   \ar@{}[r]|\to&
\tbool\\
&&\qu^-\\
&\qu^+_{\grey{0}}
        \ar@{.}@/^/[ur]
\\
\qu^-_{\grey{{0,0}}}
        \ar@{.}@/^/[ur]
\\
\done^+_{\grey{0,0}}
        \ar@{.}@/^/[u]
\\
\qu^-_{\grey{0,1}}
        \ar@{.}@/_/[uuur]
\\
\done^+_{\grey{0,1}}
        \ar@{.}@/^/[u]
\\
&\done^-_{\grey{0}}
        \ar@{.}@/^/[uuuuu]
\\
&\qu^+_{\grey{1}}
        \ar@{.}@/_/[uuuuuuur]
\\
\qu^-_{\grey{1, 0}}
        \ar@{.}@/^/[ur]
\\
\done^+_{\grey{{1, 0}}}
        \ar@{.}@/^/[u]
\\
\qu^-_{\grey{1, 1}}
        \ar@{.}@/_/[uuur]
\\
\done^+_{\grey{1, 1}}
        \ar@{.}@/^/[u]
\\
&\done^-_{\grey{1}}
        \ar@{.}@/^/[uuuuu]
\\
&&\ttrue^+
        \ar@{.}@/_/[uuuuuuuuuuuuu]
}$}$}
\]
\caption{Pointer annotations on plays}
\label{fig:plays_pointers}
\end{minipage}
\end{figure}
$s\in \Alt(\intr{(\tunit \to \tunit) \to \tbool})$ with $\ev{s}$ 
displayed in Figure \ref{fig:exp_conf}. The dashed lines represent
immediate causal dependency in $\leq_A$, omitted when it coincides with
juxtaposition. We call these dashed lines \emph{pointers}, going
upwards from one event to its predecessor in $A$.  As arenas are
forestial, any move has at most one pointer and only minimal
events have none.

It is worth, just this once, being extremely pedantic about the
representation used in Figure \ref{fig:plays_pointers} and others.
Recall that $\intr{(\tunit \to \tunit) \to \tbool} = \oc (\oc \tunit
\lin \tunit) \lin \tbool$. Accordingly, 
\[
\ev{\intr{(\tunit \to \tunit) \to \tbool}} = \tnat \times (\tnat \times
\ev{\tunit} + \ev{\tunit}) + \ev{\tbool}
\]
with $+$ the tagged disjoint union $A + B = \{1\} \times A \cup \{2\}
\times B$, previously also written $\parallel$.  So an event of
$\intr{(\tunit \to \tunit) \to
\tbool}$ carries a \emph{move} from $\tunit$ or $\tbool$, \emph{tags}
originating from the disjoint unions and indicating one type component,
and natural numbers, the \emph{copy indices}. In Figure
\ref{fig:plays_pointers} the information of the \emph{moves} is conveyed by the
label, \emph{i.e.} $\qu^-, \done^+$, \emph{etc}. The \emph{tag} is
conveyed by the position of the move under the corresponding type
component.  Finally, the \emph{copy indices} are given as a sequence in
grey, with the leftmost integer corresponding to the outermost $\oc$.
For instance, the move $\qu^-_{\grey{0,1}}$ really stands for $(1,
(\grey{0}, (1, (\grey{1}, \qu^-))))$.

It is often convenient to display pointers, but they are
not part of the structure of plays. If they are imported into
plays, then copy indices become essentially disposable (up to $\sym$).
To make this formal, we start by defining a notion of \emph{plays with
pointers} on a $-$-arena.

\begin{defi}\label{def:play_pointers}
An \textbf{alternating play with pointers} on $A$ is $s_1
\dots s_n \in \ev{A}^*$ which is:
\[
\begin{array}{ll}
\text{\emph{alternating}:}& \forall 1\leq i \leq n-1,~\pol_A(s_i)\neq
\pol_A(s_{i+1})\,,
\end{array}
\]
together with, for all $1\leq j \leq n$ s.t. $s_j$ is non-minimal
in $A$, the data of a \textbf{pointer} to some earlier $s_i$ s.t.
$s_i \imc_A s_j$.  We write $\PAlt(A)$ for the set of plays with
pointers on $A$.
\end{defi}

The \emph{non-repetitive} condition of
Definition \ref{def:alt_play} would make pointers redundant as
each move has a unique predecessor, and the existence of pointers would
boil down to the fact that plays reach only down-closed sets of events.
It is a useful exercise to show that non-repetitive plays with
pointers are in bijection with alternating plays, on arenas without
conflict.

Reciprocally, since repetitions are now allowed, we may use them to
represent executions with replication even without the expansion process
of Section~\ref{subsec:replication}.

\subsubsection{Meager and concrete arenas}
Definition
\ref{def:play_pointers} applies to arenas in the sense of Section~\ref{subsubsec:ar_sym}, but it ignores part of their
structure: it takes no account of \emph{conflict}, and \emph{symmetry}.
Indeed, plays with pointers originate from HO games, where
arenas are much simpler:

\begin{defi}\label{def:meagerarena}
A \textbf{meager arena} is a partial order with polarities $(A, \leq_A,
\pol_A)$ s.t.:
\[
\begin{array}{rl}
\text{\emph{alternating:}} & \text{if $a_1 \imc_A a_2$, $\pol_A(a_1)
\neq
\pol_A(a_2)$,}\\
\text{\emph{forestial:}} & \text{if $a_1 \leq_A a$ and $a_2 \leq_A a$,
then $a_1
\leq_A a_2$ or $a_2 \leq_A a_1$,}
\end{array}
\]
without conflict or symmetry. A \textbf{meager $-$-arena} additionally
satisfies:
\[
\begin{array}{rl}
\text{\emph{negative:}} & \text{if $a \in \min(A)$, then $\pol_A(a) =
-$.}
\end{array}
\]
\end{defi}

Clearly, Definition \ref{def:play_pointers} applies to meager arenas.
Each $\PCF$ type $A$ may be interpreted as a meager arena
$\aintr{A}$, setting $\aintr{\tunit} = \tunit$, $\aintr{\tbool} =
\tbool$, $\aintr{\tnat} = \tnat$ and $\aintr{A\to B} = \aintr{A} \lin
\aintr{B}$; \emph{i.e.} as for $\intr{-}$ but without the $\oc$ -- this
is exactly the interpretation in \cite{ho}. The arena $\intr{A}$ is
then an \emph{expansion} of $\aintr{A}$ -- the notion of
\emph{concrete arena} makes this explicit:

\begin{defi}\label{def:concrete-arena}
A \textbf{concrete arena} is $(A, A^0, \lbl)$ with $A$ an arena,
$A^0$ a \emph{meager arena},
\[
\lbl : \ev{A} \to \ev{A^0}\,,
\]
a \textbf{labelling} function, together satisfying the following additional
requirements:
\[
\begin{array}{rl}
\text{\emph{locally pointed:}} & \text{for all $x \in \conf{A}$, $x$ has
at most one minimal event of each polarity,}\\
\text{\emph{rigid:}} & \text{$\lbl$ preserves minimality, and preserves
immediate causality $\imc$,}\\
\text{\emph{transparent:}} & \text{for any $x, y \in \conf{A}$ and
bijection $\theta : x \simeq y$,}\\
&\text{then $\theta \in \tilde{A}$ iff $\theta$ is an order-iso
preserving $\lbl$.}
\end{array}
\]
\end{defi}

We shall update this in Section~\ref{subsubsec:concrete_pol_sym}, when further structure becomes
required.  \emph{Locally pointed} is phrased so as to allow
non-negative arenas of the form $A^\perp \parallel B$. In most
cases, for negative arenas, configurations $x \in \conf{A}$ will have at
most one minimal event.

Every basic arena $\tx$ may be regarded as the
concrete arena $(\tx, \tx, \lbl_\tx)$ with $\lbl_\tx$ the identity
function. Concrete arenas support the arena constructions $\with$ and
$\to$ with $(A\with B)^0 = A^0 \tensor B^0$, and $(A\to B)^0 =
A^0 \lin B^0$. By induction, for every type $A$
this gives us $(\intr{A}, \aintr{A}, \lbl_A)$, a pointed concrete
$-$-arena with $\lbl_A$ simply forgetting all copy indices.

\begin{rem}
\emph{Transparent} makes explicit the nature of symmetries
on arenas arising from types: as they leave all
components unchanged except \emph{copy indices}, they are exactly all
reindexings. This does not always hold outside the types
considered here. In particular, concrete arenas do not support
$\tensor$: of course condition \emph{locally pointed} fails, but more
fundamentally, valid symmetries in $\oc
(A \tensor B)$ must send $(\grey{i}, (1, a))$ and $(\grey{i}, (2, b))$
to the same
copy index $\grey{j}$, a non-local constraint, not reflected by
condition
\emph{transparent}. This is why we do not consider all arenas to be
\emph{concrete}: they fail to cover the full Seely category structure.
\end{rem}

In the sequel, we only assume arenas to be concrete when it is
explicitly mentioned.

\subsubsection{Pointers and symmetry}
Plays with pointers represent plays up to symmetry:

\begin{prop}\label{prop:pointer_plays}
Consider $A$ a concrete arena. Then, there is a function
\[
\mathscr{P} : \Alt(A)/\!\!\sym 
\quad
\to 
\quad
\PAlt(A^0)\,,
\]
injective and preserving length and prefix.
\end{prop}
\begin{proof}
For $s \in \Alt(A)$, we first construct $s^\to \in \PAlt(A)$ by
importing $\imc_{A}$.  Then, $\mathscr{P}(s)$ is obtained by applying
$\lbl_A$ pointwise. That pointers on $\mathscr{P}(s)$ are
well-formed (\emph{i.e.} that if $s_j$ points to $s_i$, then $s_i
\imc_{A^0} s_j$) follows from $\lbl$ preserving minimality and the
immediate causal order. That $\mathscr{P}$ is invariant under $\sym$
boils down to \emph{transparent}. By construction,
$\mathscr{P}$ preserves length and prefix. For injectivity,
take $s, s' \in \Alt(A)$ such that $\mathscr{P}(s) = \mathscr{P}(s')$.
Since $\mathscr{P}$ is length-preserving, $s$ and $s'$ have the same
length $n$.  Consider
\[
\theta = \{(s_i, s'_i) \mid 1\leq i \leq n\} : \ev{s} \simeq \ev{s'}
\]
the induced bijection. Since $\mathscr{P}(s) = \mathscr{P}(s')$, in
particular $s$ and $s'$ have the same pointers, so $\theta$ is an
order-isomorphism, and moreover since $\mathscr{P}(s) = \mathscr{P}(s')$
again we also have $\lbl_A(s_i)$ = $\lbl_A(s'_i)$ for all $1\leq i \leq
n$. Hence, $\theta$ is a symmetry, so by \emph{transparent}, $s \sym s'$
as required.
\end{proof}

However, $\mathscr{P}$ is not surjective. Writing $A =
\tunit \to \tbool$, the play $s \in \PAlt(\aintr{A})$ set as
\[
\xymatrix@R=-4pt{
\tunit  \ar@{}[r]|\to&
\tbool\\
&\qu^-\\
\qu^+   \ar@{.}@/^/[ur]\\
\done^- \\
&\ttrue^+
        \ar@{.}@/_/[uuu]\\
\done^- \ar@{.}@/^1pc/[uuu]\\
&\tfalse^+
        \ar@{.}@/_1pc/[uuuuu]
}
\]
is not the image of any play in $\Alt(\intr{A})$, for two reasons:
\emph{(1)} not every move is duplicated in $\intr{A}$, \emph{e.g.} there
there is only one copy of $\done^-$ for every copy of $\qu^+$ -- this
linearity discipline is enforced by \emph{non-repetitive},
which is absent in $\PAlt(\aintr{A})$; and \emph{(2)} likewise,
$\aintr{A}$ and $\PAlt(\aintr{A})$ do not account for \emph{conflict}
between $\ttrue^+$ and $\tfalse^+$ in $\intr{A}$.

\subsubsection{HO strategies} This extends to \emph{strategies}.  For
concrete arena $A$ and $\sigma : A$, then $\mathscr{P}(\sigma) =
\{\mathscr{P}(s) \mid s \in \sigma\}$ is essentially a strategy on $A^0$ in the
Hyland-Ong sense, \emph{i.e.} a prefix-closed, deterministic set of
plays with pointers. We have:

\begin{prop}\label{prop:strat_ho}
Consider $A$ a concrete arena, and prestrategies $\sigma, \tau :
\intr{A}$.

Then, $\sigma \simstrat \tau$ iff $\mathscr{P}(\sigma) =
\mathscr{P}(\tau)$.
\end{prop}
\begin{proof}
\emph{If.} Consider $\sigma, \tau : A$ s.t. $\mathscr{P}(\sigma) =
\mathscr{P}(\tau)$. For $\sigma \simstrat \tau$ we first
check $\to$-simulation. Consider $sa^+ \in \sigma, t \in \tau$ s.t.
$s \sym_A t$. But $\mathscr{P}(sa) \in \mathscr{P}(\tau)$, so there is
$t' b' \in \tau$ s.t. $\mathscr{P}(t' b') = \mathscr{P}(sa)$. Hence
by Proposition \ref{prop:pointer_plays}, $sa \sym_A t' b'$. So $t, t'
\in \tau$ and $t \sym_A t'$, with $t'b' \in \tau$. By uniformity of
$\tau$, $t b\in \tau$ for some $b$ with $t'b' \sym_A tb$, so $tb \sym_A
sa$ as well.  The condition $\ot$-\emph{simulation} is
symmetric. For $\to$-\emph{receptive}, assume $sa^- \in \sigma, t \in
\tau$ and $sa^- \sym_A tb^-$. Since $\mathscr{P}(sa^-) \in
\mathscr{\tau}$, there is $t' b' \in \tau$ s.t. $\mathscr{P}(sa) =
\mathscr{P}(t'b')$, \emph{i.e.} $sa \sym_A tb$. But then $t'b' \in \tau$
and $t' b' \sym_A t b$, so by uniformity of $\tau$ we have $tb \in
\tau$. Finally, $\ot$-\emph{receptive} is symmetric.

\emph{Only if.} Consider $\sigma, \sigma' : A$ s.t. $\sigma
\simstrat \sigma'$, and take $\mathscr{P}(s) \in \mathscr{P}(\sigma)$
for some $s \in \sigma$. By induction on $s$, we build some $s' \in
\sigma'$ s.t. $s \sym_A s'$: for positive extensions this follows
from $\sigma \simstrat \sigma'$; for negative extensions 
from the \emph{extension} condition on isomorphism families and
the $\to$-\emph{receptive} condition on uniformity. But then by
Proposition \ref{prop:pointer_plays} we have $\mathscr{P}(s) =
\mathscr{P}(s')$, so $\mathscr{P}(\sigma) \subseteq
\mathscr{P}(\sigma')$. The argument is symmetric, so
$\mathscr{P}(\sigma) = \mathscr{P}(\sigma')$ as desired.
\end{proof}

Plays with pointers permit a presentation of strategies up to
$\simstrat$, avoiding copy indices. They provide the foundation for HO
games \cite{ho}, where the interpretation of types is essentially
$\aintr{-}$ (without conflict), and plays carry pointers. We
include the classical example showing that though one may
choose copy indices or pointers, one cannot avoid both.

\begin{exa} The \emph{Kierstead terms} $\vdash K_x, K_y :
((\tbool \to
\tbool) \to \tbool) \to \tbool$ are defined as
\[
K_x = \lambda F^{(\tbool \to \tbool) \to \tbool}.\,F\,(\lambda
x.\,F\,(\lambda y.\,x))\,,
\qquad
\qquad
K_y = \lambda F^{(\tbool \to \tbool) \to \tbool}.\,F\,(\lambda
x.\,F\,(\lambda y.\,y))\,.
\]

Their respective interpretations in $\NegAlt_\oc$ have distinctive
plays:
\[
\xymatrix@R=-8pt@C=5pt{
K_x&:&((\tbool  \ar@{}[r]|\to& 
\tbool)         \ar@{}[r]|\to&
\tbool)         \ar@{}[r]|\to&
 \tbool\\
&&&&&\qu^-\\
&&&&\qu^+_{\grey{0}}    \ar@{.}@/^/[ur]\\
&&&\qu^-_{\grey{0,i}}   \ar@{.}@/^/[ur]\\
&&&&\qu^+_{\grey{i+1}}  \ar@{.}@/^/[uuur]\\
&&&\qu^-_{\grey{i+1,j}} \ar@{.}@/^/[ur]\\
&&\qu^+_{\grey{0,i,j}}  \ar@{.}@/^/[uuur]
}
\qquad
\qquad
\xymatrix@R=-8pt@C=5pt{
K_y&:&((\tbool  \ar@{}[r]|\to& 
\tbool)         \ar@{}[r]|\to&
\tbool)         \ar@{}[r]|\to&
 \tbool\\
&&&&&\qu^-\\
&&&&\qu^+_{\grey{0}}    \ar@{.}@/^/[ur]\\
&&&\qu^-_{\grey{0,i}}   \ar@{.}@/^/[ur]\\
&&&&\qu^+_{\grey{i+1}}  \ar@{.}@/^/[uuur]\\
&&&\qu^-_{\grey{i+1,j}} \ar@{.}@/^/[ur]\\
&&\qu^+_{\grey{i+1,j,0}} \ar@{.}@/^/[ur]
}
\]

Here pointers are redundant, and computed from the identity
of moves. In particular, in both plays the $\qu^+$ ``points to'' the
unique $\qu^-$ with compatible copy indices. Mapping these 
through $\mathscr{P}$, we get two plays with pointers that only differ
through their pointers. In HO games, the Kierstead terms are only
distinguished by pointers\footnote{It is necessary to go up to
third-order types to find such examples. Pointers are redundant up to
second-order types, which is the starting point of \emph{algorithmic
game semantics} \cite{DBLP:journals/tcs/GhicaM03}.}. It is 
crucial to keep them separate: it is a surprisingly challenging exercise
to find a $\PCF$ context that separates them.
\end{exa}

Plays with pointers are powerful, and indeed the game semantics literature
is strongly biaised towards HO games (as opposed to
AJM games). This, however, has two costs. Firstly, plays
with pointers are not a natural inductive structure, making their
manipulation sometimes inelegant or unwieldy (so-called ``pointer
surgery''). Proposals have been made for clean formalizations,
\emph{e.g.} through nominal sets
\cite{DBLP:journals/entcs/GabbayG12}. Another cost 
is that replication is so hard-wired into the model that
it does not enjoy a clean linear decomposition. Enforcing linearity is
slightly awkward and relies on additional structure
\cite{DBLP:books/daglib/0094282}.

In this work we stick with $\NegAlt$ rather than
adopting plays with explicit pointers. Among other things this will ease
the relationship with the forthcoming \emph{thin concurrent games},
which we do not know how to formulate with pointers in general. Besides, in $\NegAlt$, pointers \emph{can} be directly obtained from the arena, and as such may be used as in
HO games\footnote{Another work blurring the lines between HO and AJM is
\cite{DBLP:journals/entcs/AbramskyJ09} where AJM games are 
equipped with a \emph{function} able to rebuild pointers without
the need to explicitly integrate them in plays. All the data of a game
in the sense of \cite{DBLP:journals/entcs/AbramskyJ09} can be computed
from an arena in our sense, but our arenas are more primitive.}. In
fact, pointers play a central role in this paper. From now on, all
representations of plays will display pointers. In contrast, we will
often omit copy indices as most of the time they convey no useful
information; one can regard this convention as drawing $\mathscr{P}(s)$
rather than $s$.

\section{Sequential Computational Effects in Game Semantics}
\label{sec:seq_effects}

We now explore the model constructed above, introducing
the traditional ``semantic cube''.

The plays of a term are computed denotationally, by induction on
syntax. However, given a term, an experienced game semanticist will be
able to directly list its plays, without going through the intricate
definition of the interpretation.  This is because as discussed before,
plays represent the operational behaviour of the term: rather than
denotationally, they can be obtained directly from the term by
operational means
\cite{DBLP:conf/lics/DanosHR96,DBLP:conf/fossacs/Jaber15,DBLP:journals/entcs/GhicaT12,DBLP:conf/csl/LevyS14}.
\begin{figure}
\[
\xymatrix@R=-4pt{
(\tbool \ar@{}[r]|\to&
\tbool  \ar@{}[r]|\to&
\tunit) \ar@{}[r]|\to&
\tbool\\
&&&\qu^-&&
\ensquare{\lambda f^{\tbool \to \tbool \to
\tunit}.\,f\,((f\,\tfalse\,\tfalse);\,\ttrue)\,\tfalse;\,\ttrue}\\
&&\qu^+ \ar@{.}@/^/[ur]&&&
\lambda f^{\tbool \to \tbool \to
\tunit}.\,\underline{f}\,((f\,\tfalse\,\tfalse);\,\ttrue)\,\tfalse;\,\ttrue\\
\qu^-   \ar@{.}@/^/[urr]&&&&&
\lambda f^{\tbool \to \tbool \to
\tunit}.\,f\,\ensquare{((f\,\tfalse\,\tfalse);\,\ttrue)}\,\tfalse;\,\ttrue\\
&&\qu^+ \ar@{.}@/^/[uuur]&&&
\lambda f^{\tbool \to \tbool \to
\tunit}.\,f\,((\underline{f}\,\tfalse\,\tfalse);\,\ttrue)\,\tfalse;\,\ttrue\\
\qu^-   \ar@{.}@/^/[urr]&&&&&
\lambda f^{\tbool \to \tbool \to
\tunit}.\,f\,((f\,\ensquare{\tfalse}\,\tfalse);\,\ttrue)\,\tfalse;\,\ttrue\\
\tfalse^+&&&&&
\lambda f^{\tbool \to \tbool \to
\tunit}.\,f\,((f\,\underline{\tfalse}\,\tfalse);\,\ttrue)\,\tfalse;\,\ttrue\\
&&\done^-       \ar@/_/@{.}[uuu]&&&
\lambda f^{\tbool \to \tbool \to
\tunit}.\,f\,((f\,\tfalse\,\tfalse);\,\ensquare{\ttrue})\,\tfalse;\,\ttrue\\
\ttrue^+        \ar@{.}@/^1pc/[uuuuu]
&&&&&\lambda f^{\tbool \to \tbool \to
\tunit}.\,f\,((f\,\tfalse\,\tfalse);\,\underline{\ttrue})\,\tfalse;\,\ttrue\\
&\qu^-  \ar@{.}@/^/[uuuuuuur]&&&&
\lambda f^{\tbool \to \tbool \to
\tunit}.\,f\,((f\,\tfalse\,\tfalse);\,\ttrue)\,\ensquare{\tfalse};\,\ttrue\\
&\tfalse^+&&&&
\lambda f^{\tbool \to \tbool \to
\tunit}.\,f\,((f\,\tfalse\,\tfalse);\,\ttrue)\,\underline{\tfalse};\,\ttrue\\
&&\done^-       \ar@{.}@/_1.5pc/[uuuuuuuuu]&&&
\lambda f^{\tbool \to \tbool \to
\tunit}.\,f\,((f\,\tfalse\,\tfalse);\,\ttrue)\,\tfalse;\,\ensquare{\ttrue}\\
&&&\ttrue^+     \ar@{.}@/_1.5pc/[uuuuuuuuuuu]&&
\lambda f^{\tbool \to \tbool \to
\tunit}.\,f\,((f\,\tfalse\,\tfalse);\,\ttrue)\,\tfalse;\,\underline{\ttrue}
}
\]
\caption{Illustration of the operational contents of game semantics}
\label{fig:op_gam_ex}
\end{figure}
This is illustrated in Figure \ref{fig:op_gam_ex}.  Opponent moves
trigger the evaluation of a subterm, which appears boxed.  The following
Player move then corresponds to the head (\emph{i.e.} leftmost) variable
occurrence (or constant) of the subterm being evaluated. The pointers
from Player moves correspond to the stage where the variable in head
position was abstracted, or to the function call being returned by the
value in head position. More specifically, Figure \ref{fig:op_gam_ex}
represents the interaction of the term under study with the applicative
context:
\[
C[] = []\,(\lambda x^\tbool.\,\lambda
y^\tbool.\,\ite{x}{(\ite{y}{\tskip}{\tskip})}{\tskip}) 
\]

Figure \ref{fig:op_gam_ex} is strongly inspired by the \emph{Pointer
Abstract Machine (PAM)} \cite{DBLP:conf/lics/DanosHR96}.

\subsection{Well-Bracketing} \label{subsec:wb}
Now that \emph{executions as plays} are first-class citizens,
independent of programs, we may start classifying them according to the
computational capabilities that they witness.  For instance, is the
following play a possible execution of a term?
\[
\xymatrix@R=-5pt{
(\tunit \ar@{}[r]|\to&
\tunit) \ar@{}[r]|\to&
\tbool\\
&&\qu^- && \ensquare{\lambda f^{\tunit \to \tunit}.\,f\,M}\\
&\qu^+  \ar@{.}@/^/[ur]&&& \lambda f^{\tunit \to
\tunit}.\,\underline{f}\,M\\
\qu^-   \ar@{.}@/^/[ur]&&&&\lambda f^{\tunit \to
\tunit}.\,f\,\ensquare{M}\\
&&\ttrue^+ 
        \ar@{.}@/_/[uuu]
&& \lambda f^{\tunit \to \tunit}.\,\underline{\ttrue}\hspace{15pt}
}
\]

We argue informally why this cannot be an execution in $\PCF$.  The
first action of the term is to ask its argument, so it has the form
$\lambda f^{\tunit \to \tunit}.\,f\,M$; we annotate the figure with the
corresponding operational state as in Figure \ref{fig:op_gam_ex}. In the
last line, $\ttrue$ at toplevel indicates the overall computation has
terminated to $\ttrue$. This is confusing, since operationally the
Opponent move in the third line corresponded to triggering the
evaluation of the argument of $f$. How can evaluating the argument of
$f$ cause the whole computation to terminate?

Nevertheless, this play is indeed a realistic execution, for the term
\[
\lambda f^{\tunit \to \tunit}.\,\mathbf{callcc}\,(\lambda k^{\tbool \to
\tunit}.\,f\,(k\,\ttrue);\,\bot) : (\tunit \to \tunit)\to \tbool
\]
where $\mathbf{callcc}$ is the \emph{call-with-current-continuation}
primitive originating in \emph{Scheme}, and which
famously may be typed with Peirce's law \cite{DBLP:conf/popl/Griffin90}.
The precise operational semantics of $\mathbf{callcc}$ will not
be useful for this paper, but informally $\mathbf{callcc}\,M$
immediately calls $M$, feeding it a special function $k$, the
``continuation''. When the continuation is called with value $v$,
$\mathbf{callcc}$ interrupts $M$ and returns $v$ at
toplevel, breaking the call stack discipline.

Can the play above be realised without $\mathbf{callcc}$ (or some other
\emph{control operator}, as such primitives are called)? We can show
that the answer is no, by capturing plays that ``respect the call
stack discipline'', and refining the whole interpretation 
to show that this invariant is preserved. This is the goal of the
notion of \emph{well-bracketing}. First we enrich arenas:

\begin{defi}\label{def:qa_arena}
A \textbf{Question/Answer labeling} on arena $A$ is a function
\[
\lambda_A : \ev{A} \to \{\Qu, \An\}\,
\]
invariant under symmetry (if $\theta : x \sym_A y$, then for
all $a \in x$, $\lambda_A(a) = \lambda_A(\theta(a))$) and satisfying:
\[
\begin{array}{ll}
\text{\emph{question-opening:}} & \text{if $a \in \ev{A}$ is minimal,
then $\lambda_A(a) = \Qu$,}\\
\text{\emph{answer-closing:}} & \text{if $\lambda_A(a) = \An$, then $a$
is maximal for $\leq_A$,}\\ 
\text{\emph{answer-linear:}} & \text{if $\lambda_A(a_1) = \lambda_A(a_2)
= \An$ with $a \imc_A a_1, a_2$, then $a_1 = a_2$ or $a_1 \conflict_A a_2$.}
\end{array}
\]
\end{defi}

From now on, \textbf{arenas} have a Question/Answer
labeling.  \emph{Questions} intuitively correspond to variable calls,
while \emph{Answers} correspond to returns.
\begin{figure}
\hfill
\begin{minipage}{0.1\linewidth}
\[
\xymatrix@R=15pt@C=10pt{
\qu^{-,\Qu}
        \ar@{.}[d]\\
\done^{+,\An}
}
\]
\end{minipage}
\hfill
\begin{minipage}{.3\linewidth}
\[
\xymatrix@R=15pt@C=10pt{
&\qu^{-,\Qu}
        \ar@{.}[dl]
        \ar@{.}[dr]\\
\ttrue^{+,\An}\ar@{~}[rr]&&
\tfalse^{+,\An}
}
\]
\end{minipage}
\hfill
\begin{minipage}{.4\linewidth}
\[
\xymatrix@R=15pt@C=10pt{
&&\qu^{-,\Qu}
        \ar@{.}[dll]
        \ar@{.}[dl]
        \ar@{.}[d]
        \ar@{.}[dr]\\
0^{+,\An}       \ar@{~}[r]&
1^{+,\An}       \ar@{~}[r]&
2^{+,\An}       \ar@{~}[r]&
\dots
}
\]
\end{minipage}
\hfill
\caption{Question/Answer labeling on basic arenas}
\label{fig:qa_arena}
\end{figure}
Basic arenas are enriched as shown in Figure \ref{fig:qa_arena}. For
other constructions the labeling is inherited transparently, with
$\lambda_{\oc A}(\grey{i}, a) = \lambda_A(a)$, $\lambda_{A_1\tensor
A_2}(i, a) = \lambda_{A_i}(a)$, $\lambda_{A\lin B}(2, b) =
\lambda_B(b)$, and $\lambda_{A \lin B}(1, (b, a)) = \lambda_A(a)$.

If $s \in \Alt(A)$ and $s_i$ is an answer, it cannot be minimal in $A$
by \emph{question-opening}. Its antecedent in $A$ -- its
\emph{justifier} -- must appear in $s$ as some $s_j$ with $j<i$, and
is a question by \emph{answer-closing}. We say that \textbf{$s_i$
answers $s_j$}. If a question in $s$ has an answer in $s$ we say it is
\textbf{answered} in $s$. The last unanswered question of $s$, if any,
is the \textbf{pending question}.

We now capture executions respecting the call stack
discipline as \emph{well-bracketed} plays.

\begin{defi}\label{def:alt_play_wb}
Let $s \in \Alt(A)$ be an alternating play.

It is \textbf{well-bracketed} if for all prefix $ta^{\An}
\sqsubseteq s$, $a$ answers the pending question of $t$.
\end{defi}

All plays encountered in the paper until now are well-bracketed,
with the exception of the example at the beginning of Section~\ref{subsec:wb}.
We can then define well-bracketed strategies:

\begin{defi}\label{def:alt_strat_wb}
Let $\sigma : A$ be a strategy on $A$.

It is \textbf{well-bracketed} iff for all $sa^+ \in \sigma$, if
$s$ is well-bracketed then $sa$ is well-bracketed.
\end{defi}

In other words, a well-bracketed strategy is never the first to
break the call stack discipline. Asking all plays to be well-bracketed
is too strict, as illustrated by the play
\[
\xymatrix@R=-5pt@C=10pt{
(\tunit \ar@{}[r]|\lin&\tunit)
\ar@{}[r]|\lin&(\tunit\ar@{}[r]|\lin&\tunit)\\
&&&\qu^{-,\Qu}\\
&\qu^{+,\Qu}
        \ar@{.}@/^/[urr]\\
\qu^{-,\Qu}
        \ar@{.}@/^/[ur]\\
&&\qu^{+,\Qu}
        \ar@{.}@/^/[uuur]\\
&\done^{-,\An}
        \ar@{.}@/_/[uuu]\\
&&&\done^{+,\An}
        \ar@{.}@/_/[uuuuu]
}
\]
of copycat: the last move does not answer the pending question, but
because Opponent broke the normal control flow first. There is a lluf
subcategory $\NegAltwb$ of $\NegAlt$, having well-bracketed strategies
as morphisms.  The interpretation of $\PCF$ in $\NegAlt_\oc$ in fact
yields only well-bracketed strategies, \emph{i.e.} has target
$\NegAltwb_\oc$.  This shows that indeed, the execution at the beginning
of Section~\ref{subsec:wb} cannot be realised in $\PCF$.

\subsection{Visibility and Innocence}  \label{subsec:vis_inn}
Likewise, is this play a possible execution of a term?
\[
\xymatrix@R=-6pt@C=30pt{
(\tbool \ar@{}[r]|\to&\tunit)\ar@{}[r]|\to&\tunit\\
&&\qu^-                                 && \ensquare{\lambda f^{\tbool
\to \tunit}.\,f\,M}\\
&\qu^+   \ar@{.}@/^/[ur] &              && \lambda f^{\tbool \to
\tunit}.\,\underline{f}\,M\\
\qu^-    \ar@{.}@/^/[ur]&&              && \lambda f^{\tbool \to
\tunit}.\,f\,\ensquare{M}\\
\tfalse^+ &             &               && \lambda f^{\tbool \to
\tunit}.\,f\,\underline{\tfalse}\\
\qu^-    \ar@{.}@/^/[uuur] &&           && \lambda f^{\tbool \to
\tunit}.\,f\,\ensquare{M}\\
\ttrue^+ &&                             && \lambda f^{\tbool \to
\tunit}.\,f\,\underline{\ttrue}
}
\]

Again, this seems unfeasible in $\PCF$. Again, on the right hand side we
show, assuming a term realising this play, its corresponding
operational states.  At the third and fifth moves, the \emph{same}
subterm is being evaluated; yet we get two distinct answers.
In an extension of $\PCF$ with a primitive $+$ for non-deterministic
choice, this play would be realisable by $\lambda f^{\tbool \to
\tunit}.\,f\,(\ttrue + \tfalse)$. But does it make computational sense
in a \emph{deterministic} language?

Once more, the answer is yes: the play above describes a valid execution
of the term
\[
\lambda f^{\tbool \to
\tunit}.\,\newref\,r\,\tin\,f\,(\mathbf{let}\,x=!r\,\mathbf{in}\,r:=1;\,(x>0))
: (\tbool \to \tunit) \to \tunit
\]
in $\PCF$ extended with references: $\newref\,r\,\tin\,M$
allocates a reference $r$ initialized to $0$.
\begin{figure}
\begin{minipage}{.65\linewidth}
\[
\scalebox{.85}{$
\xymatrix@R=-5pt@C=15pt{
(\tbool \ar@{}[r]|\to&\tunit)\ar@{}[r]|\to&\tunit\\
&&\qu^-                                  & \ensquare{\lambda f^{\tbool
\to
\tunit}.\,f\,(\mathbf{let}\,x=!r\,\mathbf{in}\,r:=1;\,(x>0))} & r\mapsto
0\\
&\qu^+   \ar@{.}@/^/[ur] &              & \lambda f^{\tbool \to
\tunit}.\,\underline{f}\,(\mathbf{let}\,x=!r\,\mathbf{in}\,r:=1;\,(x>0))&
r\mapsto
0\\
\qu^-    \ar@{.}@/^/[ur]&&              & \lambda f^{\tbool \to
\tunit}.\,f\,\ensquare{(\mathbf{let}\,x=!r\,\mathbf{in}\,r:=1;\,(x>0))}
& r \mapsto
0\\
\tfalse^+ &             &               & \lambda f^{\tbool
\to\tunit}.\,f\,\underline{\tfalse}
\hspace{120pt}& r \mapsto 1\\
\qu^-    \ar@{.}@/^/[uuur] &&           & \lambda f^{\tbool \to
\tunit}.\,f\,\ensquare{(\mathbf{let}\,x=!r\,\mathbf{in}\,r:=1;\,(x>0))}
& r \mapsto
1\\
\ttrue^+ &&                             & \lambda f^{\tbool \to
\tunit}.\,f\,\underline{\ttrue} \hspace{120pt}& r \mapsto 1
}$}
\]
\caption{A strategy with references}
\label{fig:games_ref_op}
\end{minipage}
\hfill
\begin{minipage}{.34\linewidth}
\[
\xymatrix@R=-3.5pt{
(\tunit \ar@{}[r]|\to&
\tunit) \ar@{}[r]|\to&
\tunit\\
&&\qu^-\\
&\qu^+  \ar@{.}@/^/[ur]\\
\grey{\qu^-}    \ar@[grey]@{.}@/^/[ur]\\
&\grey{\qu^+}   \ar@{.}@[grey]@/^/[uuur]\\
\qu^-   \ar@{.}@/^/[uuur]\\
\done^+ \ar@{.}@/^1pc/[uuu]
}
\]
\caption{Non P-visible play}
\label{fig:non_pvis}
\end{minipage}
\end{figure}
We show in Figure \ref{fig:games_ref_op} an operational description as
to how this term indeed realises this play.

Again, this cannot be realised in $\PCF$. To show this, we
give a version of \emph{innocence} \cite{ho}, formalizing 
that without state, evaluating the same subterm yields
the same response. The first step is a mathematical way to state
that two plays ``correspond to the same subterm'', like the two prefixes
of the play of Figure \ref{fig:games_ref_op} terminating with a $\qu^-$
on the left.

The operation computing (a mathematical notion of) ``current subterm''
is the \emph{P-view}:

\begin{defi}\label{def:pview}
Let $s \in \Alt(A)$. Its \textbf{P-view} is the subsequence
defined by induction:
\[
\begin{array}{rcll}
\pview{\varepsilon} &=& \varepsilon\\
\pview{s a^+} &=& \pview{s} a^+\\
\pview{s a^+_1 s' a_2^-} &=& \pview{s} a_1^+ a_2^- &\text{if $a_1^+
\imc_A
a_2^-$}\\
\pview{s a^+_1 s' a_2^-} &=& a_2^-                 &\text{if $a_2$ is
negative minimal
in $A$}
\end{array}
\]
\end{defi}

We take the immediate prefix for $P$-ending plays and follow the pointer
for $O$-ending plays. For instance, the prefixes of length 3 and 5 of
the play on Figure \ref{fig:games_ref_op} have the same P-view,
capturing that they correspond to the same subterm.  This is
a powerful definition -- really, the distinguishing feature of HO games
\cite{ho} -- and it often takes newcomers a while to digest.
Interestingly, our forthcoming \emph{parallel innocence} will be phrased
quite differently.

But this is not yet conclusive: if $s \in \Alt(A)$, it might
be that $\pview{s} \not \in \Alt(A)$.  For instance, in Figure
\ref{fig:non_pvis} we gray out moves not selected in computing the
P-view of $s \in \Alt(A)$ for $A = \intr{(\tunit \to \tunit) \to
\tunit}$. The subsequence of $\pview{s}$ in black is an
alternating sequence of $\ev{A}$, but fails \emph{valid} of
Definition \ref{def:alt_play}. Indeed, the ``justifier'' of $\done^+$,
its immediate dependency in $A$, is not selected -- thus 
$\ev{\pview{s}}$ is not down-closed. Accordingly, we say:

\begin{defi}\label{def:pvis}
A play $ \in \Alt(A)$ is \textbf{P-visible} if for all prefix $\forall
t \sqsubseteq s$, $\pview{t} \in \Alt(A)$.

Likewise, a strategy $\sigma : A$ is \textbf{P-visible} iff all its
plays are P-visible.
\end{defi}

So, ``computing P-views never drops pointers'', or ``Player
always points in the P-view''.  On P-visible $s \in \Alt(A)$, the P-view
always yields a well-formed (P-visible) play.

We now define \emph{innocent
strategies} as those that behave the same in any situation where
the same subterm is being evaluated, \emph{i.e.} whose behaviour only
depends on the \emph{P-view}:

\begin{defi}\label{def:innocence}
A $P$-visible alternating strategy $\sigma : A$ is \textbf{innocent} if
it satisfies:
\[
\begin{array}{ll}
\text{\emph{innocence:}} & \text{for all $sa^+ \in \sigma$, for all $t
\in \sigma$,
if $\pview{s} = \pview{t}$ then $ta^+ \in \sigma$.}
\end{array}
\]
\end{defi}

That $ta^+ \in \Alt(A)$ is well-formed relies on $\pview{s} =
\pview{t}$, so that the causal dependencies of $a$ in $A$
appear in $t$\footnote{In traditional Hyland-Ong games based on plays
with points, one would conclude the above definition 
with something like ``\dots then $ta \in \sigma$, where $a$
has the same pointer as in $sa$'', which is rarely made very formal.
Here, because pointers are derived
the above definition is rigorous and self-contained.}.  All
structural morphisms of $\NegAlt$ are innocent. Innocent strategies
compose -- though this is infamously tricky to prove,
prompting a significant line of work
investigating the structures arising from the composition of innocent
strategies \cite{curien,hhm,DBLP:conf/fossacs/ClairambaultD15}. We do
not review here the proof of stability under composition.

The interpretation of $\PCF$ yields only innocent
strategies, \emph{i.e.} targets the cartesian closed lluf
subcategory $\NegAltinn_\oc$ of innocent strategies. 
Hence, the play at the beginning of Figure \ref{subsec:vis_inn} cannot
be realised in $\PCF$. We also get a cartesian closed lluf subcategory
$\NegAltinnwb_\oc$ with well-bracketing. Finally, the weaker
\emph{P-visibility} is also preserved under the categorical operations,
forming lluf sub-cartesian closed categories $\NegAltvis_\oc$ and
$\NegAltviswb_\oc$.

\subsection{Full Abstraction for $\PCF$} We have now
eliminated all non $\PCF$-definable behaviour. We review the
corresponding \emph{definability} and \emph{intensional full
abstraction} arguments.

\subsubsection{Definability}\label{subsubsec:def_inn}
Call a \textbf{P-view} on arena $A$ any $s \in \Alt(A)$ 
invariant under P-view, \emph{i.e.} $\pview{s} = s$ -- 
those are exactly the $s \in \Alt(A)$ such that for all $t
s_i^+ s_{i+1}^- \sqsubseteq s$, we have $s_i \imc_A s_{i+1}$, in other
words \emph{Opponent always points to the previous move}.  We motivated
P-views as a way to address specific ``\emph{subterms}'' of 
a strategy -- it might therefore not be a surprise that those are the
key to reconstruct a term from an innocent strategy. We write
\[
\pviews{\sigma} = \{\pview{s} \mid s \in \sigma\}
\]
for the set of P-views of $\sigma$. If $\sigma$ is innocent,
then it is simple that $\pviews{\sigma} \subseteq \sigma$.
Moreover, $\sigma$ can then be recovered as the set of P-visible $s \in
\Alt(A)$ such that for all $t \sqsubseteq s$, $\pview{t} \in
\pviews{\sigma}$.

For $\sigma : A$ innocent, $\pviews{\sigma}$ is not a
strategy as in general it fails receptivity. It is however easily
verified to be a \emph{prestrategy} -- and in particular \emph{uniform}.
Moreover, we have:

\begin{prop}\label{prop:causal_inn}
For $\sigma, \tau : A$ innocent strategies on $A$, we have
$\pviews{\sigma} = \pviews{\tau}$ iff $\sigma = \tau$.

Likewise, $\pviews{\sigma} \simstrat \pviews{\tau}$ if and only if
$\sigma \simstrat \tau$.
\end{prop}
\begin{proof}
We only detail the second statement. Firstly, if $\sigma \simstrat
\tau$, it is direct that $\pviews{\sigma} \simstrat \pviews{\tau}$ as
$\pviews{\sigma} \subseteq \sigma$ and $\pviews{\tau} \subseteq \tau$
and the bisimulation game of Definition \ref{def:simstrat} preserves
P-views.

If $\pviews{\sigma} \simstrat \pviews{\tau}$, take $sa^+ \in \sigma,
t \in \tau$ s.t. $s \sym_A t$. In particular $\pview{s} \sym_A
\pview{t}$ and $\pview{s} a^+ \in \pviews{\sigma}$. By $\to$-extension,
there is $b^+$ s.t. $\pview{t} b^+ \in \pviews{\tau}$, so $t b^+ \in
\tau$ by innocence. This proves $\to$-extension, $\ot$-extension
is symmetric and $\to, \ot$-receptivity follow by receptivity of
$\sigma, \tau$.
\end{proof}

So innocent strategies have two representations: a full $\sigma
: A$ satisfying Definition \ref{def:innocence}; or, following
Proposition \ref{prop:causal_inn}, the set $\pviews{\sigma}$.
Anticipating on later developments, we refer to $\pviews{\sigma}$ as
the
\emph{causal} presentation of $\sigma$. In traditional innocent game
semantics, the forest of P-views is called (notably by Curien
\cite{curien2006notes}) the \emph{meager} representation, while the
set of plays is \emph{fat}. Here this 
is misleading, because plays in $\pviews{\sigma}$ still carry explicit
copy indices. In particular $\pviews{\sigma}$ has branches matching 
all copyable Opponent moves, which is ``fat''.

To recover the meager representation, we show:

\begin{prop}\label{prop:meager_inn}
Consider $A$ a concrete arena and $\sigma, \tau : A$ innocent
strategies.

If $\mathscr{P}(\pviews{\sigma}) = \mathscr{P}(\pviews{\tau})$, then
$\sigma \simstrat \tau$.
\end{prop}
\begin{proof}
Let $\sigma, \tau : \intr{A}$ be innocent strategies on $A$ and assume
that $\mathscr{P}(\pviews{\sigma}) = \mathscr{P}(\pviews{\tau})$.

By Proposition \ref{prop:strat_ho}, $\pviews{\sigma} \simstrat
\pviews{\tau}$. Then, by Proposition \ref{prop:causal_inn}, it follows
that $\sigma \simstrat \tau$.
\end{proof}

This, at last, provides the \emph{meager} representation.

These representations have distinct advantages: composition is only
directly defined on the fat representation; but it is the
meager one that bridges innocent strategies and syntax and
allows definability. An innocent alternating strategy
$\sigma : A$ is \textbf{finite} iff $\mathscr{P}(\pviews{\sigma})$ is
finite. Its \textbf{size} is simply the cardinal of that set.
Definability simply follows the meager form:

\begin{thm}\label{th:definability}
Let $A$ be a $\PCF$ type, and $\sigma : \intr{A}$ be a finite
well-bracketed innocent
strategy.

Then, there is a $\PCF$ term $\vdash M : A$ s.t. $\intr{M} \simstrat
\sigma$.
\end{thm}
\begin{proof}
We describe the argument -- for more details, the reader is
referred to \cite{ho}.

Without loss of generality, $A$ has the form $A_1 \to \dots \to A_n \to
\tx$ where
for each $1\leq i \leq n$,
\[
A_i = A_{i, 1} \to \dots \to A_{i, p_i} \to \tx_i\,.
\]

We reason on $\pviews{\sigma}$, by induction on the size of $\sigma$.
If $\sigma$ has no reaction to the (unique) minimal $\qu^-$ in $\tx$
(\emph{i.e.} $\sigma = \{\varepsilon, \qu^-\}$), any diverging term will
do. Otherwise, by determinism there is exactly one move $a^+$ s.t.
$\qu^- a^+ \in \sigma$. If $a^+$ is an answer $v^+$ on $\tx$, then $M$
is the matching constant.  Otherwise, $a^+$ is the initial 
$\qu^+_{i_0}$ in some $A_{i_0}$\footnote{Here the subscripts
indicate the type component and not copy indices, which are left
un-specified.}. The situation is drawn as
\[
\xymatrix@C=0pt@R=-5pt{
A_1 &\to& \dots &\to& (A_{i, 1} &\to & \dots & \to & A_{i, p_i} & \to &
\tx_i)&\to&
\dots & \to & A_n & \to & \tx\\
&&&&&&&&&&&&&&&&\qu^-\\
&&&&&&&&&&\qu^+_{i_0}
        \ar@{.}@/^/[urrrrrr]\\
&&&&\grey{\qu_{i_0,1}^-}
        \ar@{.}@[grey]@/^/[urrrrrr]&&&&
\grey{\qu_{i_0, p_{i_0}}^-}
        \ar@{.}@[grey]@/^/[urr]&&
\grey{v^-}
}
\]
with, in grey, the possible extended P-views. For each extension
there is a residual substrategy.

We extract those -- first, if $\qu_{i_0}^+$ immediately
returns. For each value $v$ in $\tx_i$, we form
\[
\pviews{\sigma_v} = \{\qu^- s \mid \qu^- \qu_{i_0}^+ v^- s \in
\pview{\sigma}\}\,,
\]
a causal innocent strategy on $\intr{A}$ of size strictly lesser than
$\sigma$. By induction hypothesis there is $\vdash M_v : A$ with
$\intr{M_v} \simstrat \sigma_v$. As $\sigma$ is finite, there are
finite many $v$ s.t. $\sigma_v$ is non-diverging.

Alternatively, for all $1\leq j \leq p_{i_0}$, we consider P-views
$\qu^- \qu_{i_0}^+ \qu_{i_0, j}^- s \in \pviews{\sigma}$
where as a P-view, $s$  answers neither $\qu_{i_0}^+$, nor $\qu^-$
by well-bracketing. Such a P-view yields
\[
\qu_{i_0, j}^- s \in \Alt(\intr{A_1 \to \dots \to A_n \to A_{i_0, j}})
\]
a P-view where moves in $s$ formerly depending on $\qu^-$ in $\intr{A}$
are set to depend on $\qu_{i_0, j}^-$. Considering all such P-views
generates a causal innocent strategy of size strictly lesser than
$\sigma$, hence by induction hypothesis there is $\vdash M_{i_0, j} :
A_1 \to \dots \to A_n \to A_{i_0, j}$ s.t.
$\intr{M_{i_0, j}} \simstrat \sigma_{i_0, j}$.

Finally, with all this data we may form $\vdash M : A$ as
\[
\begin{array}{l}
\lambda x_1^{A_1}\dots x_n^{A_n}.\, 
\mathbf{case}\,x_{i_0}\,(M_{i_0,1}\,x_1\,\dots\,x_n) \dots (M_{i_0,
p_{i_0}}\,x_1\,\dots\,x_n)\,\mathbf{of}\\
\hspace{60pt} v_1 \mapsto M_{v_1}\\
\hspace{70pt} \dots\\
\hspace{60pt} v_p \mapsto M_{v_p}
\end{array}
\]
where $p$ is such that every $\sigma_{v_i}$ with $i >p$ is diverging.
We get, as needed $\intr{M} \simstrat \sigma$.
\end{proof}

The final statement is a careful verification following the definition
of the interpretation, see \cite{ho}. Here, $\mathbf{case}$ is the
syntax introduced in Section~\ref{subsec:sugar}, involving the
$\mathbf{let}$ construct. Without that, simply iterating $\mathbf{if}$
constructs would yield a strategy that re-computes
\[
x_{i_0}\,(M_{i_0,1}\,x_1\,\dots\,x_n) \dots (M_{i_0,
p_{i_0}}\,x_1\,\dots\,x_n)
\]
each time it matches it against a value. This is what is done in
\cite{ho} as their version of $\PCF$ does not include a $\mathbf{let}$
construct. This yields a term that is not $\simstrat$-equivalent to
$\sigma$, but is nonetheless $\obs$-equivalent (see Section~\ref{sec:def_obs}), which suffices for full abstraction. We prefer the
present more intensional definability result, and hence have included
the $\mathbf{let}$ construct\footnote{An alternative is to include a
primitive $\mathbf{case}$ evaluating its argument exactly once. The
terms then obtained via definability are easily characterised
syntactically -- dubbed \emph{$\PCF$ B\"ohm trees} by Curien,
and are studied in \cite{curien}. The definability process 
informs a concrete order-isomorphism between finite meager innocent
strategies and finite $\PCF$ B\"ohm trees, emphasizing that meager
innocent strategies \emph{are} syntax.}. 

\subsubsection{Intensional full
abstraction}\label{subsubsec:intensional_fa} Full abstraction of a
denotational model with respect to a language was defined in Section~\ref{sec:def_obs}. Of course, $\NegAlt_\oc$ is not
fully abstract for $\PCF$ as it stands. For instance, 
$\intr{\lambda x^{\tunit}.\,x;\,x} \not \simstrat \intr{\lambda
x^{\tunit}.\,x}$: game semantics displays explicitly individual calls
to $x$, so we see that the term on the left hand side evaluates $x$
twice whereas the other evaluates it once.
However, we do of course have $\lambda x^{\tunit}.\,x;\,x \obs \lambda
x^{\tunit}.\,x$; this can for instance be deduced from them having the
same interpretation in Scott domains \cite{plotkin}.

The celebrated ``full abstraction for $\PCF$'' results are in fact what
(following \cite{ajm}) we call \emph{intensional full
abstraction}.  Fixing an interpretation $\intr{-}$ of $\PCF$ into a
$\C$, we set
\[
f \obs g \qquad\Leftrightarrow\qquad \forall \alpha \in \C(A \to B,
\intr{\tunit}),
\quad (\alpha \circ \overline{f} =  \alpha \circ \overline{g})\,,
\]
for $f,g \in \C(A, B)$, with $\overline{f}, \overline{g} \in
\C(\emptyar,A \to B)$ obtained via cartesian closure, and $\emptyar$ the
terminal object of $\C$. This mimics the definition of observational
equivalence.  We say that $\C$ is \textbf{intensionally
fully abstract} for $\PCF$ iff the quotiented model $\C^\obs$ is fully
abstract.

\begin{thm}\label{th:fa_pcf}
The model $\NegAltinnwb_\oc$ is intensionally fully abstract for $\PCF$.
\end{thm}
\begin{proof}
Consider $\vdash M, N : A$ s.t. $M \obs N$, and assume $\intr{M}
\not \obs \intr{N}$, \emph{i.e.} there is a test $\alpha \in
\NegAltinnwb_\oc(\intr{A}, \intr{\tunit})$ s.t. $\alpha \odot_\oc
\intr{M} \neq \alpha \odot_\oc \intr{N}$ -- say \emph{w.l.o.g.} that
$\alpha \odot_\oc \intr{M} \eval$ converges while $\alpha \odot_\oc
\intr{N} \div$.  One may prove (see \cite{ho} for details)
that the corresponding interactions expose only a finite part of
$\alpha$, so \emph{w.l.o.g.} we can assume $\alpha$ finite. By Theorem
\ref{th:definability}, $\alpha$ is defined via a $\PCF$ term,
providing a context $C[-]$ s.t. $\intr{C[M]} = \alpha \odot_\oc
\intr{M}$ and $\intr{C[N]} = \alpha \odot_\oc \intr{N}$. But then, we
must have $C[M] \eval$ while $C[N]$ diverges by Proposition
\ref{prop:adequacy}; contradiction.
\end{proof}

Intensional full abstraction is full abstraction for an \emph{a
priori} non effective quotiented model: it
does not directly provide effective tools to reason about
observational equivalence. Instead, it is a
way of stating that we have faithfully captured the intensional
behaviour of programs, in the sense that the added tests in the model
are not able to distinguish more -- there is no ``abstraction
leak''. Often, it follows from adequacy and finite definability.

\emph{Full abstraction} is of course the preferred notion when the
quotiented model is sufficiently effective and the interpretation
computable (\emph{i.e. effectively presentable} \cite{pisanotes}). But
when it requires an undecidable quotient\footnote{For $\PCF$ this is
unavoidable as observational equivalence is undecidable already for
finitary $\PCF$ \cite{DBLP:journals/tcs/Loader01}.}, we believe it
preferable to use a different terminology: 
``\emph{intensional full abstraction}'' puts the emphasis on the model
pre-quotient. In game semantics, it is that model pre-quotient that had
the most impact.  In particular it then led to
\emph{effective} fully abstract models for stateful languages,
leveraging the results and insights above.

\subsection{Full Abstraction for $\IA$} \label{subsec:fa_ia}
The exposition in Section~\ref{subsec:vis_inn} suggests that also
without innocence, strategies are computationally relevant for
programs with mutable state. We now focus on the
game semantics of $\IA$, namely $\PCF$ extended with interference (see
Section~\ref{sec:syntax}).

\subsubsection{Interpretation of types} With respect to $\PCF$, $\IA$
adds the type $\var$ of \emph{integer references}, and the
type $\sem$ of \emph{semaphores}. Their usual game semantic interpretation
is \emph{behavioural}, in the sense that it represents
how one may interact on those types: one may read a reference or write a
new value in it; and likewise one may grab a semaphore, or release it.

\begin{figure}
\begin{minipage}{.35\linewidth}
\[
\scalebox{.9}{$
\xymatrix@C=10pt@R=10pt{
\mwrite{0}^{-,\Qu}      \ar@{~}[r]
                \ar@{.}[d]&
\mwrite{1}^{-,\Qu}      \ar@{~}[r]
                \ar@{.}[d]&
\dots           \ar@{~}[r]&
\mwrite{n}^{-,\Qu}      \ar@{~}[r]
                \ar@{.}[d]&
\dots\\
\ok^{+,\An}&
\ok^{+,\An}&&
\ok^{+,\An}
}$}
\]
\caption{$\var_w$}
\label{fig:varw}
\end{minipage}
\hfill
\begin{minipage}{.3\linewidth}
\[
\scalebox{.9}{$
\xymatrix@C=10pt@R=10pt{
&&\read^{-,\Qu}
        \ar@{.}[dll]
        \ar@{.}[dl]
        \ar@{.}[d]
        \ar@{.}[dr]\\
0^{+,\An}       \ar@{~}[r]&
1^{+,\An}       \ar@{~}[r]&
2^{+,\An}       \ar@{~}[r]&
\dots
}$}
\]
\caption{$\var_r$}
\label{fig:varr}
\end{minipage}
\hfill
\begin{minipage}{.18\linewidth}
\[
\scalebox{.9}{$
\xymatrix@C=10pt@R=9pt{
\mgrab^{-,\Qu}  \ar@{~}[r]
                \ar@{.}[d]&
\mrelease^{-,\Qu}\ar@{.}[d]\\
\ok^{+,\An}&\ok^{+,\An}
}$}
\]
\caption{$\sem$}
\label{fig:sem}
\end{minipage}
\end{figure}
To capture this, we define $-$-arenas: $\var_w = \bigwith_{n \in
\mathbb{N}} \tunit$, $\var_r = \tnat$ and $\sem   = \tunit \with
\tunit$, and set $\intr{\var} = \var_w \with \var_r$ and $\intr{\sem} =
\sem$. Although we reuse the arena constructions for $\tunit$ and
$\tnat$, for specific moves in these arenas we use the naming
conventions of Figures \ref{fig:varw}, \ref{fig:varr} and \ref{fig:sem}
-- in Figures \ref{fig:varw} and \ref{fig:varr} all distinct moves in
the same row are in pairwise conflict.

\subsubsection{Interacting with Memory and
Semaphores}\label{subsubsec:seq_inter_mem} The idea behind Abramsky and
McCusker's interpretation of state is that it is not
the operations of reading, writing, grabing or releasing a semaphore
that are effectful -- indeed, those are just requests via the interface
provided by the $\var$ and $\sem$ types and associated commands. The
strategy for a program with free reference or semaphore variables will
simply record their accesses leaving the memory and semaphores
uninterpreted. For instance, the strategy for $x : \var \vdash x:=0;\,!x
: \tnat$ includes:
\[
\xymatrix@R=-5pt@C=15pt{
\oc \var &\lin& \tnat\\
&&\qu^-\\
\mwrite{0}^+\\
\ok^-\\
\read^+\\
42^-
}
\]
a play where the value read is not the value just written. The actual
effectful computation will be handled in Section~\ref{subsubsec:newrs}
with the creation of new references and semaphores.

Accordingly, we set the interpretation of memory and semaphore accesses
as:
\begin{eqnarray*}
\intr{M:=N} &=& \assign \odot_\oc \tuple{\intr{N}, \intr{M}}\\
\intr{!M} &=& \deref \odot_\oc \intr{M}\\
\intr{\grab(M)} &=& \sgrab \odot_\oc \intr{M}\\
\intr{\release(M)} &=& \srelease \odot_\oc \intr{M}
\end{eqnarray*}
\begin{figure}
\begin{minipage}{.26\linewidth}
\[
\xymatrix@R=-7pt@C=0pt{
\oc (\tnat & \with & \var) & \lin & \tunit\\
&&&&\qu^-\\
\qu^+_{\grey{0}}\\
n^-_{\grey{0}}\\
&&\mwrite{n}^+_{\grey{1}}\\
&&\ok^-_{\grey{1}}\\
&&&&\done^+
}
\]
\caption{$\assign$}
\label{fig:assign}
\end{minipage}
\hfill
\begin{minipage}{.22\linewidth}
\[
\xymatrix@R=-.5pt@C=0pt{
\oc \var &\lin & \tnat\\
&& \qu^-\\
\read^+_{\grey{0}}\\
n^-_{\grey{0}}\\
&& n^+
}
\]
\caption{$\deref$}
\label{fig:deref}
\end{minipage}
\hfill
\begin{minipage}{.22\linewidth}
\[
\xymatrix@R=-.5pt@C=0pt{
\oc \sem &\lin& \tunit\\
&& \qu^-\\
\mgrab^+_{\grey{0}}\\
\ok^-_{\grey{0}}\\
&&\done^+
}
\]
\caption{$\sgrab$}
\label{fig:grab}
\end{minipage}
\hfill
\begin{minipage}{.22\linewidth}
\[
\xymatrix@R=-.5pt@C=0pt{
\oc \sem &\lin& \tunit\\
&& \qu^-\\
\mrelease^+_{\grey{0}}\\
\ok^-_{\grey{0}}\\
&&\done^+
}
\]
\caption{$\srelease$}
\label{fig:release}
\end{minipage}
\end{figure}
using the (innocent well-bracketed) strategies of Figures
\ref{fig:assign}, \ref{fig:deref}, \ref{fig:grab}, and
\ref{fig:release}. Finally, we set
\[
\intr{\mkvar\,M\,N} = \tuple{\tuple{\intr{M}\,n\mid n \in \mathbb{N}},
\intr{N}}\,,
\qquad
\qquad
\intr{\mksem\,M\,N} = \tuple{\intr{M}, \intr{N}}\,,
\]
where $\intr{M}\,n$ is $\intr{M}$ applied to the constant strategy $n$
(using the cartesian closed structure of $\NegAlt_\oc$), and using
implicitly the isomorphisms $\var \iso (\with_{n\in \mathbb{N}} \tunit)
\with \tnat$ and $\sem \iso \tunit \with \tunit$.

Finite definability of finite innocent well-bracketed
strategies still holds -- the proof (see \cite{am}) directly extends
that of Theorem \ref{th:definability}, using bad variables and
semaphores.

\begin{prop}\label{prop:def_ia_inn}
Consider $A$ a type of $\IA$, and $\sigma : \intr{A}$ finite innocent
well-bracketed.

Then, there is a $\IA$ term $\vdash M : A$ \emph{not using $\newref$ or
$\newsem$}, such that $\intr{M} \simstrat \sigma$.
\end{prop}

\subsubsection{Creation of References and Semaphores}\label{subsubsec:newrs}
Finally, we introduce the actual effectful behaviour. The idea is to use
non-innocent strategies $\cell_n {:} \oc \var$, $\ssem_n {:} \oc \sem$ implementing
interference. For instance, $\cell_n$ is a memory cell with $n$
currently stored. When read it returns $n$, and upon a write request for
$k \in \mathbb{N}$, it acknowledges it and proceeds as $\cell_k$.
Likewise $\ssem_0$ is the strategy for a free semaphore, and $\ssem_n$
for $n>0$ represents a semaphore in use.  Those may be simply described
as the language of prefixes of the infinite trees:
\[
\begin{array}{rclcl}
\cell_n^{\grey{I}} &=& 
\read_{\grey{i}}^- \cdot n_{\grey{i}}^+ \cdot \cell_n^{\grey{I\uplus
\{i\}}} \mid
\mwrite{k}_{\grey{i}}^- \cdot \ok_{\grey{i}}^+ \cdot \cell_k^{\grey{I
\uplus \{i\}}}
        && (i \not \in I)\\
\ssem_0^{\grey{I}} &=& \mgrab_{\grey{i}}^- \cdot \ok_{\grey{i}}^+ \cdot
\ssem_1^{\grey{I\uplus \{i\}}} \mid \mrelease_{\grey{i}}^-
        && (i \not \in I)\\
\ssem_n^{\grey{I}} &=& \mgrab_{\grey{i}}^- \mid \mrelease_{\grey{i}}^-
\cdot
\ok_{\grey{i}}^+ \cdot \ssem_0^{\grey{I \uplus \{i\}}}
        && (i \not \in I,~n>0)
\end{array}
\]
where symbols are moves in $\oc \var$ and $\oc \sem$ respectively,
separated via $\cdot$ for readability. Here, $I \subseteq_f \mathbb{N}$
collects the copy indices already used, ensuring \emph{non-repetitive}.
We set $\cell_n$ as (the prefix language of)
$\cell_n^{\grey{\emptyset}}$ and $\ssem_n$ as (the prefix language
of) $\ssem^{\grey{\emptyset}}_n$; it is direct that
$\cell_n {:} \oc \var$ and $\ssem_n {:} \oc \sem$.  However,
they are not innocent. Considering the two plays:
\[
\read_{\grey{0}}^- \cdot 0_{\grey{0}}^+,\quad\mwrite{1}_{\grey{1}}^-
\cdot
\ok_{\grey{1}}^+ \cdot \read_{\grey{0}}^- \cdot 1_{\grey{0}}^+
\quad\in\quad \cell_0\,,
\]
as $\pview{\read_{\grey{0}}^-} =
\pview{\mwrite{1}_{\grey{1}}^-\cdot \ok_{\grey{1}}^+ \cdot
\read_{\grey{0}}^-}$, innocence requires $\read_{\grey{0}}^- \cdot
1_{\grey{0}}^+ \in \cell_0$ as well, which is not the case. Of course,
it is precisely the role of $\cell$ and $\ssem$ to break innocence and
transfer information across distinct copies -- however, $\cell$ and
$\ssem$ remain \emph{P-visible} in the sense of Definition
\ref{def:pvis}.

We now complete the interpretation of $\IA$.
Consider $\Gamma, x: \var \vdash M : A$ with
\[
\intr{M} \in \NegAltviswb_\oc(\Gamma \with \var, A)
\]
omitting some brackets. Using the cartesian closed structure of
$\NegAltviswb_\oc$, we consider
\[
\Lambda_{\Gamma}^\oc(\intr{M}) \in \NegAltviswb(\oc \var, \oc \Gamma
\lin A)
\]
which we compose with the memory cell.
Summing up, for references and semaphores,
\[
\begin{array}{rclcl}
\intr{\newref\,x\!\!:=\!n\,\tin\,M}
        &=& {\Lambda_{\Gamma}^\oc}^{-1}(\Lambda_{\Gamma}^\oc(\intr{M})
\odot \cell_n)
                &\in& \NegAltviswb_\oc(\Gamma, A)\\
\intr{\newsem\,x\!\!:=\!n\,\tin\,M}
        &=& {\Lambda_{\Gamma}^\oc}^{-1}(\Lambda_{\Gamma}^\oc(\intr{M})
\odot \ssem_n)
                &\in& \NegAltviswb_\oc(\Gamma, A)\,.
\end{array}
\]

This concludes the interpretation of $\IA$ in $\NegAltviswb$.
Adequacy proceeds as in \cite{am}, undisturbed by the slightly
different technical setting of the present paper.

\begin{prop}[Adequacy]\label{prop:adequacy_ia}
For any $\vdash M : \tunit$ in $\IA$, $M \eval$ if and only if
$\intr{M}\eval$.
\end{prop}

\subsubsection{Full Abstraction} \label{subsubsec:fa_ia}
We now review the full abstraction result of \cite{am}. The argument
revolves around a fundamental \emph{factorisation} theorem, stated as
follows.

\begin{thm}[Factorisation] \label{th:factor}
Let $A$ be a type of $\IA$, and $\sigma :
\intr{A}$ be P-visible well-bracketed.

Then, there is an innocent well-bracketed $\Inn(\sigma) \in
\NegAltinnwb(\oc \var, \intr{A})$ such that
\[
\sigma \simstrat \Inn(\sigma) \odot \cell_0\,,
\]
and $\Inn(\sigma)$ is finite if $\sigma$ is finite.
\end{thm}
\begin{proof}
For any O-ending $sa^- \in
\sigma$, we wish $\Inn(\sigma)$ to act like $\sigma$, but as an innocent
strategy it may only depend on $\pview{sa^-}$. However,
$\Inn(\sigma)$ may also access the reference, so we will maintain the
invariant that the reference contains (an encoding of) the full history,
or more precisely of $\mathscr{P}(s)$.
Between $\pview{sa^-}$ and $\mathscr{P}(s)$, $\sigma$
knows the full play (up to symmetry).

Upon being called with $a^-$, $\Inn(\sigma)$ reads $\mathscr{P}(s)$
from the reference, then stores $\mathscr{P}(sab)$ in the reference (for
$s a^- b^+ \in \sigma$) and then plays $b$.
See \cite{am} for more details.
\end{proof}

Finiteness of $\Inn(\sigma)$ follows the definition of finite
\emph{innocent} strategies from Section~\ref{subsubsec:def_inn}: having
finitely many ($\sym$-equivalence classes of) P-ending P-views. However,
finiteness of non-innocent strategies has not yet been defined. We
define it now: a strategy in $\NegAlt(A)$ is \textbf{finite} iff the set
of ($\sym$-equivalence classes of) P-ending plays of $\sigma$ is finite.
Despite the common terminology, these two notions are distinct: an
innocent strategy may be finite as an innocent strategy while being
non-finite as a non-innocent strategy. This mismatch comes from the fact
that these two notions both coincide with the domain-theoretic notion of
\emph{compactness}, but in the distinct domains $\NegAltinnwb(A, B)$ and
$\NegAltwb(A, B)$ (ordered by inclusion) for $-$-arenas $A, B$. By
Proposition \ref{prop:pointer_plays}, these statements involving
$\sym$-equivalence classes may be instead phrased with plays with
pointers.

From Theorem \ref{th:factor} and Proposition \ref{prop:def_ia_inn} it is
direct that finite definability holds for $\IA$. We can deduce
immediately intensional full abstraction for $\IA$, proved as
Theorem \ref{th:factor}.

\begin{thm}\label{th:fa_ia}
The model $\NegAltviswb_\oc$ is intensionally fully abstract for $\IA$.
\end{thm}

This is exactly as Theorem \ref{th:fa_pcf}. However, in stark contrast
with Theorem \ref{th:fa_pcf}, for $\IA$ the fully abstract quotient
category is effectively presentable. In fact, 
for $\sigma, \tau : \intr{A}$,
\[
\sigma \obs \tau 
\quad
\Leftrightarrow
\quad
\mathscr{P}(\comp(\sigma)) = \mathscr{P}(\comp(\tau))
\]
where $\comp(\sigma)$ is the set of \emph{complete} plays of $\sigma$,
capturing the completed executions where both players act like P-visible
well-bracketed strategies: a play is \textbf{complete} if it is
well-bracketed, P-visible, \emph{O-visible} (the dual to P-visibility,
not detailed here), and such that every question has
an answer. The result follows from finite definability for $\IA$ \cite{am}.

This effective fully abstract model of $\IA$ is one of the most striking
results of game semantics. Observational equivalence in $\IA$ remains
undecidable with bounded integers, at fourth order without recursion
\cite{DBLP:conf/lics/Murawski03} and second-order with recursion
\cite{DBLP:conf/lics/Ong02} (of course, observational equivalence is
obviously undecidable in the full language as it is Turing-complete).
However, the model yielded sound and
complete algorithms for observational equivalence on restricted
fragments \cite{DBLP:journals/tcs/GhicaM03}, starting the field
of \emph{algorithmic game semantics}.

\subsection{The Semantic Cube} Abramsky's ``semantic cube'', often
called the ``Abramsky cube'', starts with the observation that
game semantics allows the interpretation of both control (\emph{i.e.}
$\mathbf{callcc}$) and state in the same model, \emph{i.e.} the same
category.

\subsubsection{Control} We have not given the interpretation of
$\mathbf{callcc}$, nor the corresponding full abstraction result
\cite{DBLP:conf/lics/Laird97}. In fact, in the present technical setting, we cannot
interpret $\mathbf{callcc}$. This is due, in part, to the added conflict
in arenas for basic datatypes with respect to standard HO games
\cite{ho} -- see Figures \ref{fig:ar_unit}, \ref{fig:ar_bool}, and
\ref{fig:ar_nat}. This conflict imposes that each question can be
answered at most once, which is incompatible with
$\mathbf{callcc}$\footnote{In addition to conflicts,
the incompatibility with $\mathbf{callcc}$ comes from the fact that our
interpretation of types only involves $\oc$ on arrows, and not on basic
datatypes.  To authorize control we should change \emph{e.g.} the arena
$\tbool$ to one with replicated answers, written (in the language of
tensorial logic \cite{DBLP:journals/apal/MelliesT10}) as $\neg \oc \neg
(1 \oplus 1)$.}.  In fact: 

\begin{prop}\label{prop:inn_wb}
For any $\PCF$ type $A$, any innocent $\sigma : \intr{A}$ is also
well-bracketed.
\end{prop}
\begin{proof}
First, any innocent $\sigma : \intr{A}$ is
well-bracketed iff its P-views are well-bracketed -- see
\cite{DBLP:conf/lics/Laird97} for a proof. Hence if $\sigma$ is not
well-bracketed, then there is a P-view
\[
s_1 \dots \qu^{-,\Qu} \dots \qu^{-,\Qu} \dots a^{+,\An}
\]
where $a$ answers the first $\qu$ shown rather than the pending
question, the second $\qu$ shown. But this second $\qu^{-,\Qu}$ must be
the initial move of a banged sub-arena in the interpretation of $A$, so
we can play it again.  And by innocence of $\sigma$, the following must
be a play of $\sigma$:
\[
s_1 \dots \qu^{-,\Qu} \dots \qu^{-,\Qu}_{\grey{i}} \dots a^{+,\An}
\qu^{-,\Qu}_{\grey{i+1}} \dots a^{+,\An}
\]
where both copies of $a$ point to the first $\qu^{-,\Qu}$, absurd by
\emph{non-repetitive}.
\end{proof}

This entails that in fact,
Theorem \ref{th:fa_pcf} holds for $\NegAltinn_\oc$. But no such
coincidence holds beyond innocent strategies: for Theorem \ref{th:fa_ia}
well-bracketing really is needed. In this paper we have adopted an
interpretation of ground types incompatible with $\mathbf{callcc}$.
There
is no technical obstacle to modelling $\mathbf{callcc}$ -- one can
simply drop conflicts in basic arenas and duplicate return values --
but we prefer our design, closer to linear logic and the relational
model (see Section~\ref{subsubsec:relational}).
Furthermore, control operators will play no role in the present paper
beyond the exposition of the scientific context.

\subsubsection{The Semantic Cube} \label{subsubsec:sem_cube}
We temporarily consider, for the sake of the discussion,
a setting with both control and state; say Murawski's model for
interference and control \cite{DBLP:conf/csl/Murawski07}, which is
essentially equivalent (modulo the representation with pointers) to ours
where basic arenas have no conflict and answers are replicated. Let us
call it by $\Games$.  There is
\[
\PCF + \text{interference} + \text{control} \qquad \to \qquad
\Games
\]
an adequate interpretation, so we can model a rich combination of
effects; but that is not all.

Indeed, there are four (intensional) full abstraction results:

\begin{thm}[Semantic Cube]
We have four intensional full abstraction results:
\[
\begin{array}{rcl}
\Games & \text{is fully abstract for} & \PCF +
\text{\emph{interference}} + \text{\emph{control}}\\
\Games + \text{\emph{innocence}} & \text{is fully
abstract for} & \PCF + \text{\emph{control}}\\
\Games + \text{\emph{well-bracketing}} & \text{is fully
abstract for} & \PCF + \text{\emph{interference}}\\
\Games + \text{\emph{innocence}} + \text{\emph{well-bracketing}} 
&\text{is fully abstract for} & \PCF
\end{array}
\]
\end{thm}

We have reviewed two cases before, namely $\PCF$ (Theorem
\ref{th:fa_pcf}) and $\PCF + \text{interference}$ (Theorem
\ref{th:fa_ia}).  The full abstraction result for $\PCF +
\text{control}$ is due to Laird \cite{DBLP:conf/lics/Laird97}, while for
$\PCF + \text{interference} + \text{control}$ appears in
Murawski\footnote{Murawski uses a different primitive for control, but
the difference is superficial within $\IA$.}
\cite{DBLP:conf/csl/Murawski07}.

This ``Semantic Cube'', drawn in Figure \ref{fig:cube},
expresses that the conditions on strategies capture the
behaviour generated by certain computational effects; or rather
the \emph{absence} of certain effects.
\begin{figure}
\begin{center}
\includegraphics[scale=.25]{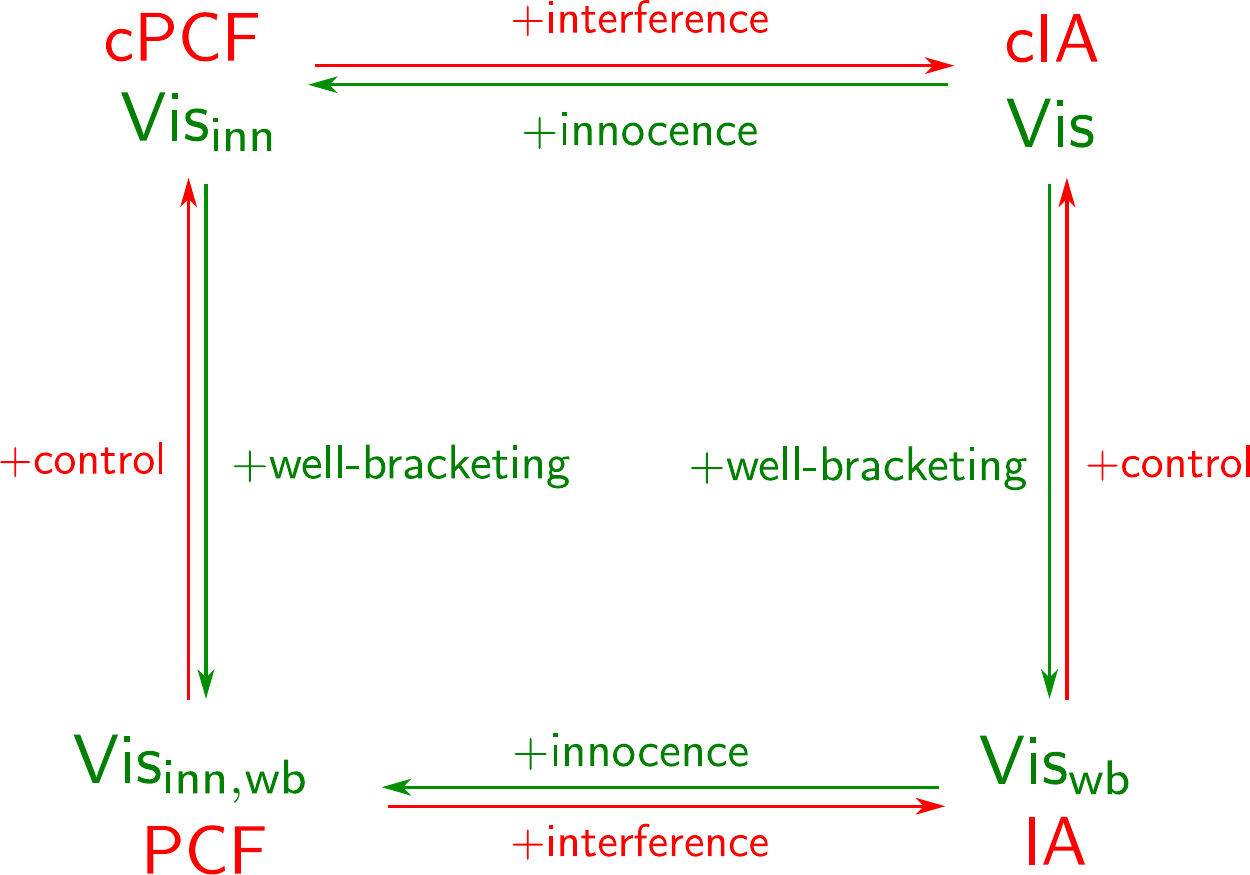}
\end{center}
\caption{The Semantic Cube}
\label{fig:cube}
\end{figure}
The achievement is noteworthy, as it is famously
difficult to \emph{combine} semantic accounts of computational effects.
But independently of purely semantic purposes, this provides us
with a microscope to study behaviourally interactions between effects in
programming languages. We demonstrate this with the following
orthogonality property\footnote{We learnt of it from a talk by Paul Levy in 2014
\cite{pblevy:chocola}.} between interference
and control which nicely illustrates the strength of game semantics:

\begin{thm}\label{th:orth_cube}
Let $\vdash M : A$ a term of $\PCF + \text{\emph{interference}} +
\text{\emph{control}}$. Assume that
\[
\begin{array}{rl}
(1)& \text{$M \obs N_1$ where $N_1$ is a term of $\PCF +
\text{\emph{interference}}$,}\\
(2)& \text{$M \obs N_2$ where $N_2$ is a term of $\PCF +
\text{\emph{control}}$;}
\end{array}
\]
then $M \obs N$ where $N$ is a term of (an infinitary extension of)
$\PCF$.
\end{thm}
\begin{proof}
Consider $\intr{M} : A$. We have seen in Section~\ref{subsubsec:fa_ia}
that for $\IA$, strategies are indistinguishable iff they have the same
\emph{complete plays}. In the presence of control this phenomenon gets
stronger: strategies are indistinguishable iff they have the
\emph{same plays} \cite{DBLP:conf/csl/Murawski07}.  Hence, $\intr{M}$ is
an innocent well-bracketed strategy (even though $M$ might internally
use state and control). It is approximated by a sequence of finite
innocent strategies which may be defined; but as the definability
process is monotone this yields an infinitary $\PCF$ term.
\end{proof}

There are reasons to expect that in a version of $\PCF$ such as ours
with a $\mathbf{let}$ construct, the innocent well-bracketed games
model is \emph{intensionally universal}, meaning that each computable
innocent well-bracketed strategy is definable\footnote{Hyland
and Ong have a \emph{extensional} universality theorem
\cite{ho}, \emph{i.e.} up to observational equivalence. In their
setting, intensional universality fails: in the absence of a
$\mathbf{let}$ construct, the strategies obtained by universality
\emph{stutter}, repeating the same move possibly many times.
\emph{Intensional} universality does not appear anywhere in
call-by-name, but it does in call-by-value
\cite{DBLP:conf/fossacs/MurawskiT13}.} -- though as far as we know,
this has not been proved.  With such a result, Theorem
\ref{th:orth_cube} would generalize to conclude the existence of simply
a term of $\PCF$, rather than an infinitary term.

\subsection{Towards Concurrency}\label{subsec:parallelism}
The reader may rightly complain that Figure
\ref{fig:cube} is not a ``semantic cube'', only a ``semantic square''.
Though we focused on control and interference, there are fully
abstract models of languages featuring general references
\cite{DBLP:conf/lics/AbramskyHM98}, exceptions
\cite{DBLP:conf/lics/Laird01}, coroutines
\cite{DBLP:conf/icalp/Laird04}, non-determinism
\cite{DBLP:conf/lics/HarmerM99}, probabilistic choice
\cite{DBLP:conf/lics/DanosH00},
concurrency \cite{DBLP:journals/entcs/Laird01,gm}, and others.
One imagines that the methodology above generalizes, and that the big
``syntactic hypercube'' of these effects is matched by a ``semantic
hypercube''.

However, there is no such ``semantic hypercube'': the works cited above
rely on a priori incompatible formal
settings. In this paper, we present steps towards such a semantic
hypercube. More precisely we aim to \emph{disentangle parallelism and
interference} in the same sense as the ``Abramsky cube'' disentangles
control and interference, \emph{i.e.} we must answer:

\begin{question}
Build a model $\CG$ with notions of \emph{parallel innocence} and
\emph{sequentiality} s.t.:
\[
\begin{array}{rcl}
\CG & \text{is fully abstract for} & \IPA\,,\\
\CG + \text{parallel innocence} & \text{is
fully abstract for} & \PCFpar\,,\\
\CG + \text{sequentiality} &\text{is
fully abstract for} & \IA\,,\\
\CG + \text{parallel innocence} + \text{sequentiality} &
\text{is fully abstract for} & \PCF\,,
\end{array}
\]
all of these being intensional full abstraction results.
\end{question}

The model should be fully abstract for $\IPA$, link
with $\NegAltinnwb_\oc$ and $\NegAltwb_\oc$ (respectively fully abstract
for $\PCF$ and $\IA$), but also support a notion of parallel
innocence yielding full abstraction for $\PCFpar$. It is natural to
start with a simple non-alternating variant of $\NegAlt$, inspired by
Ghica and Murawski's fully abstract games model for $\IPA$ \cite{gm}.

\subsubsection{Non-alternating plays and strategies} We simply relax
alternation in Definition \ref{def:alt_play}.

\begin{defi}\label{def:nalt_play}
A \textbf{non-alternating play} on $-$-arena $A$ is $s = s_1
\dots s_n$ which is:
\[
\begin{array}{ll}
\text{\emph{valid:}}& \forall 1\leq i \leq n,~\{s_1, \dots, s_i\} \in
\conf{A}\,,\\
\text{\emph{non-repetitive:}}& \forall 1\leq i,j \leq n,~s_i =s_j
\implies
i=j\,,\\
\text{\emph{negative:}} & n\geq 1 \implies \pol(s_1) = -\,.
\end{array}
\]

We write $\NAlt(A)$ for the set of non-alternating plays on $A$.
\end{defi}

The notation (inspired by \emph{template games}
\cite{DBLP:journals/pacmpl/Mellies19}), is intended to suggest that
whereas alternating plays in $\Alt(A)$ transition between two states O
and P determining which player has control, in $\NAlt(A)$ there is only
one state, in which either player may play. The intuition is simple: as
several threads might be running in parallel, their interleaving breaks
alternation.
\begin{figure}
\[
\xymatrix@R=-5pt@C=5pt{
(\tunit &\to&\tunit)&\to& \tnat\\
&&&&\qu^- 
        &\hspace{10pt}& 
           \ensquare{\lambda f^{\tunit \to
\tunit}.\,\newref\,r\!\!:=\!0\,\tin\,f\,(r:=1);\,!r}
        && r \mapsto 0 \\
&&\qu^+_{\grey{0}}
        \ar@{.}@/^/[urr] &&
        && \lambda f^{\tunit \to
\tunit}.\,\newref\,r\!\!:=\!0\,\tin\,\underline{f}\,(r:=1);\,!r
        && r \mapsto 0 \\
\qu^-_{\grey{0,0}}
        \ar@{.}@/^/[urr] &&&&
        && \lambda f^{\tunit \to
\tunit}.\,\newref\,r\!\!:=\!0\,\tin\,f\,\ensquare{(r:=1)};\,!r
        && r \mapsto 0\\
&&\done^-_{\grey{0}}
        \ar@{.}@/_/[uu] &&
        && \lambda f^{\tunit \to
\tunit}.\,\newref\,r\!\!:=\!0\,\tin\,f\,\ensquare{(r:=1)};\,\ensquare{!r}
        && r \mapsto 0\\
\done^+_{\grey{0,0}} 
        \ar@{.}@/^/[uu]&&&&
        && \lambda f^{\tunit \to
\tunit}.\,\newref\,r\!\!:=\!0\,\tin\,f\,(\hspace{5pt}\underline{\tskip}\hspace{5pt});\,\ensquare{!r}
        && r \mapsto 1 \\
&&&&1^+ 
        \ar@{.}@/_/[uuuuu]&& \lambda f^{\tunit \to
\tunit}.\,\newref\,r\!\!:=\!0\,\tin\,f\,(\hspace{5pt}\tskip\hspace{5pt});\,\underline{1}
        && r \mapsto 1\\
}
\]
\caption{Operational content of a non-alternating play}
\label{fig:comp_nalt}
\end{figure}
We show in Figure \ref{fig:comp_nalt} a non-alternating play on
$\intr{(\tunit \to \tunit) \to \tnat}$, using the same conventions as
previously. Resting on the same computational intuitions as before, we
show for each move a representation of the matching computational state
of a term realizing that play. The figure illustrates that even
$\IA$, a sequential language, allows non-alternating plays, as the
environment can evaluate subterms in parallel.  After the third and
fourth moves, two subterms are being evaluated in parallel: $r:=1$ and
$!r$, causing a \emph{data race}.

As before, we may now define \emph{strategies} as certain sets of plays.

\begin{defi}\label{def:nalt_strat}
A \textbf{non-alternating strategy} $\sigma : A$, is
$\sigma \subseteq \NAlt(A)$ which is:
\[
\begin{array}{ll}
\text{\emph{non-empty:}}        & \varepsilon \in \sigma\\
\text{\emph{prefix-closed:}}    & \forall s \sqsubseteq s' \in
\sigma,~s\in \sigma\\
\text{\emph{receptive:}}        & \forall s \in \sigma,~sa^- \in
\NAlt(A)
\implies sa \in \sigma\\
\text{\emph{courteous:}}        & \forall s a b t \in \sigma,~sb \in
\NAlt(A),~(\pol(b) = - \vee \pol(a) = +) \implies sbat \in \sigma
\end{array}
\]

A \textbf{non-alternating prestrategy} is only required to
satisfy \emph{non-empty} and \emph{prefix-closed}.
\end{defi}

Let us compare with Definition \ref{def:alt_strat}.
Besides moving to non-alternating plays, we remove determinism. Of
course, this is natural since the interleaving semantics of even a pure
parallel language represents non-deterministically the choices of the
scheduler. 
The new condition added is \emph{courtesy}, it corresponds to the
\emph{saturation} condition of \cite{gm} (the name ``courtesy'' is
imported from \cite{DBLP:conf/concur/MelliesM07}). Courtesy expresses
that the model is asynchronous. If one has $sa^+b \in \sigma$, 
there is an execution of $\sigma$ where it plays $a$, then we
observe $b$ (of any polarity). But if the surrounding computing
environment is asynchronous, nothing forces $a$ to be directly
observable by Opponent -- $a$ might get stuck in a buffer, in the
network, \emph{etc}. So then, courtesy imposes that $\sigma$ should be
\emph{stable under asynchronous delays}: if $sa^+ b \in \sigma$ and
there is no dependency from $a$ to $b$ in the arena, then $a$ can be
postponed after $b$ in $\sigma$.

\subsubsection{Well-bracketing}
We introduce \emph{well-bracketing} for non-alternating
strategies. In Ghica and Murawski's games, all plays are
\emph{well-bracketed} in the sense that they satisfy two conditions,
dubbed \emph{fork} and \emph{wait}. We adapt and introduce these conditions
now.

\begin{defi}\label{def:nalt_play_wb}
For a $-$-arena $A$, $s \in \NAlt(A)$ is \textbf{well-bracketed}
if for any $s' \prefix s$,
\[
\begin{array}{ll}
\text{\emph{fork:}} & \text{if $s' = \dots \qu^{\Qu} \dots m$
with $\qu \imc_A m$, $\qu$ must be unanswered before $m$ is played, }\\
\text{\emph{wait:}} & \text{if $s' = \dots \qu^{\Qu} \dots
a^{\An}$ with $\qu \imc_A a$, all questions justified by $\qu$ must be answered.}
\end{array}
\]
\end{defi}

This differs from the simple \emph{well-bracketing} of alternating plays
(Definition \ref{def:alt_play_wb}).  In a non-alternating setting it
does not make sense to refer to the last unanswered question as it might
originate in a different thread than the one the Player move to be
played belongs to. Instead, this condition forces plays to follow the
following protocol: a question, as long as it is not answered, may
prompt (\emph{i.e.} justify) other questions. It can only be answered
once all the questions it justified are answered, and then it will not
be able to justify anything further.  For instance, the play of Figure
\ref{fig:comp_nalt} is \emph{not} well-bracketed: the fourth move
$\done^-_{\grey{0}}$ causes a failure to \emph{wait} because it is
justified by $\qu^+_{\grey{0}}$, although the latter has justified
$\qu^-_{\grey{0,0}}$ as of yet unanswered. If we were to permute the
moves $\qu^-_{\grey{0,0}}$ and $\done^-_{\grey{0}}$ then
$\qu^-_{\grey{0,0}}$ would cause a failure to \emph{fork} as it would be
justified by a question that is already answered.

Rather than imposing this condition on all plays, we impose it on strategies.

\begin{defi}\label{def:nalt_strat_wb}
Let $\sigma : A$ be a non-alternating strategy.

It is \textbf{well-bracketed} iff for all
$sa^+ \in \sigma$, if $s$ is well-bracketed then $sa$ is well-bracketed.
\end{defi}

The play of Figure \ref{fig:comp_nalt} belongs to a well-bracketed
strategy: Opponent breaks well-bracketing first.  This is
a slight difference with Ghica and Murawski's model: we observe all
dynamic behaviour of a program of $\IPA$, even that not reachable via a
context of $\IPA$. Of course, it is always possible to restrict to
well-bracketed plays without cutting any Player behaviour (\emph{i.e.}
the strategy is cut at Opponent extensions, not Player extensions). 

\subsubsection{Limitations} The proximity of non-alternating strategies
to both $\NegAlt$ and Ghica and Murawski's model make them a tempting
foundation for this paper. However, they have two limitations,
both subtly related to their inability to record the \emph{branching
structure}.

Firstly, \emph{parallel innocence} requires adequate
causal structures, as illustrated by the following example. Is there a
parallel innocent strategy that includes the two plays with pointer
representation in Figure \ref{fig:par_branching}?
\begin{figure}
\[ 
\xymatrix@R=-2pt@C=32pt{ 
(\tbool \ar@{}[r]|\to&\tunit)\ar@{}[r]|\to&\tunit\\ 
&&\qu^-\\ 
&\qu^+   \ar@{.}@/^/[ur]\\ 
\qu^-    \ar@{.}@/^/[ur]\\ 
\tfalse^+ 
} 
\qquad\qquad 
\xymatrix@R=-2pt@C=32pt{ 
(\tbool \ar@{}[r]|\to&\tunit)\ar@{}[r]|\to&\tunit\\ 
&&\qu^-\\ 
&\qu^+   \ar@{.}@/^/[ur]\\ 
\qu^-    \ar@{.}@/^/[ur]\\ 
\ttrue^+ 
} 
\]
\caption{Two (pointer representation) of plays in a parallel innocent
strategy}
\label{fig:par_branching}
\end{figure}
Is there a program of $\PCFpar$ that may realize both?
Traditional innocence forbids that, because in a sequential
program, both plays must be visiting the same piece of syntax and
obtain the same result. In $\PCFpar$ though, the program
\[ 
\lambda f^{\tbool \to \tunit}.\,\mathbf{let}\,\left(\begin{array}{rcl} x
&=& 
f\,\ttrue \\ y &=& f\,\tfalse \end{array}\right)\,\mathbf{in}~x;\,y :
(\tbool \to 
\tunit) \to \tunit 
\]
indeed realizes these two plays, corresponding to the evaluation of
distinct threads.  A deterministic innocent strategy is determined by
(the pointer representation of) its \emph{P-views}, so we may see the
set of P-views as a witness for innocence. Analogously, what structure
may witness that a non-alternating strategy is innocent? In fact, what
is missing from the two P-views of Figure \ref{fig:par_branching} is the
\emph{branching structure}, keeping these two P-views apart and
recording how they are linked to each other. It has already been
observed \cite{lics14,DBLP:conf/lics/TsukadaO15} that
\emph{non-deterministic innocence} may be defined by replacing sets of
P-views, as witnesses for innocence, with \emph{trees} recording the
non-deterministic branching information. Here we must do the same, but
instead record the branching structure pertaining to \emph{parallelism}
as well -- which plain non-alternating strategies cannot capture
adequately \cite{DBLP:conf/concur/CastellanC16}.

Secondly, it is unclear how to endow non-alernating
strategies with appropriate notions of \emph{symmetry} and
\emph{uniformity} as in Definition \ref{def:simstrat}. Our attempts
at generalizing Definition \ref{def:simstrat} ended up suffering from
various pathologies, typically uniformity not being
preserved by hiding.
The tension with hiding comes from a play $s \in \tau \odot
\sigma$ being witnessed by distinct interactions between
$\sigma$ and $\tau$ --
this suggests again the need for an explicitly branching
structure, as then one recovers a \emph{unique witness} property.
For these two reasons, we move from plain non-alternating strategies to
\emph{thin concurrent games}. We will, however rely on:

\begin{restatable}{prop}{naltsmcc}\label{prop:nalt_smcc}
There is a symmetric monoidal closed category with products
$\NNegAltwb$, with objects $-$-arenas, and morphisms
well-bracketed non-alternating strategies on $A\lin B$.
\end{restatable}
\begin{proof}
The constructions play a minor role in this paper and are very similar
to the alternating case, so we omit them. Some details of the
construction appear in Appendix \ref{app:nalt}.
\end{proof}

These limitations call for a more intensional
setting, representing explicitly the parallel and non-deterministic
branching structures. To our knowledge, the only
games setting in the literature sufficiently expressive and mature is
\emph{thin concurrent games}, one of the possible enrichments with
symmetry of the concurrent strategies of
Rideau and Winskel \cite{lics11}.

\section{Causal Full Abstraction for \texorpdfstring{$\IPA$}{IA//}}
\label{sec:cg}

\emph{Thin concurrent games} were first introduced in \cite{lics15}, but
the detailed construction (along with significant improvements and
simplifications) is presented in \cite{cg2}.

In this section, we start with an introduction to thin concurrent games.
We omit detailed constructions (which appear in \cite{cg2}), but we
do attempt to give a self-contained introduction, in
particular providing required reasoning principles. Then, we apply this
setting to give a fully abstract model for $\IPA$, the causal sibling of
Ghica and Murawski's model \cite{gm}.
In the \emph{concurrent games} literature, strategies are often referred
to as \emph{concurrent strategies}. Here we prefer \emph{causal
strategies} to better distinguish them with non-alternating strategies,
which also represent concurrent behaviour.

\subsection{Arenas and causal strategies}
First, we must refine the symmetry on arenas.

\subsubsection{Polarized symmetry}
Arenas for causal strategies in \cite{cg2} require the following:

\begin{defi}\label{def:pol_sym}
For $A$ an arena with isomorphism family $\tilde{A}$,
a \textbf{polarized decomposition} of $\tilde{A}$ comprises
isomorphism families $\ntilde{A}, \ptilde{A}$ included in
$\tilde{A}$, s.t.:
\[
\begin{array}{rl}
\text{\emph{(1)}} & \text{If $\theta \in \ntilde{A} \cap \ptilde{A}$,
then $\theta$ is an identity bijection,}\\
\text{\emph{(2)}} & \text{If $\theta \in \ntilde{A}$ and $\theta
\subseteq^- \theta' \in \tilde{A}$, then $\theta' \in \ntilde{A}$,}\\
\text{\emph{(3)}} & \text{If $\theta \in \ptilde{A}$ and $\theta
\subseteq^+ \theta' \in \tilde{A}$, then $\theta' \in \ptilde{A}$.}
\end{array}
\]
where $\subseteq^p$ means that only (pairs of) events of polarity $p$
are added.

If $\theta : x \sym_A y$ and $\theta \in \ntilde{A}$, we write $\theta
: x \sym_A^- y$ and call $\theta$ a \textbf{negative
symmetry}. Likewise, $\theta : x \sym_A^+ y$ means that $\theta : x
\sym_A y$ with $\theta \in \ptilde{A}$, called a \textbf{positive
symmetry}.
\end{defi}

Arenas with a polarized decomposition are
\emph{thin concurrent games}\footnote{Thin concurrent games in
\cite{cg2} are more general, \emph{e.g.} they might not be alternating
or forestial.}
in the sense of \cite{cg2}. Intuitively, negative symmetries (resp.
positive) reindex Opponent (resp. Player)
moves -- though Definition \ref{def:pol_sym} does not involve ``copy
indices''. By Lemma 3.19 from \cite{cg2}, any $\theta : x \sym_A z$
factors uniquely as $\theta_+ \circ \theta_-$, with $\theta_+$
positive and $\theta_-$ negative.

\subsubsection{Constructions on arenas}\label{subsubsec:constr_arenas}
First we extend the arena constructions accordingly.

For $\tunit, \tbool$ and $\tnat$, all isomorphism families are reduced
to identity bijections between configurations. For the dual of an arena
$A$, then for all symmetry $\theta : x \sym_A y$,
\[
\theta : x \sym_{A^\perp}^+ y
\quad
\Leftrightarrow
\quad
\theta : x \sym_A^- y\,,
\qquad
\qquad
\theta : x \sym_{A^\perp}^- y
\quad
\Leftrightarrow
\quad
\theta : x \sym_A^+ y\,.
\]

For parallel composition and product, the sub-symmetries are
inherited as for the full symmetry in Section~\ref{subsubsec:ar_sym} applying to the positive and negative isomorphism
families separately. For the arrow, for $x, y \in \conf{A\lin B}$ and
$\theta : x \simeq y$ an order-isomorphism, we set $\theta \in
\ptilde{A\lin B}$ iff $\chi_{A, B}\,\theta \in \ptilde{A^\perp \parallel
B}$; and $\theta \in \ntilde{A\lin B}$ iff $\chi_{A, B}\,\theta \in
\ntilde{A^\perp \parallel B}$.

The most interesting construction is the exponential. Recall
that a symmetry in $\oc A$ is
\[
\begin{array}{rcrcl}
\theta &:& \parallel_{n\in \mathbb{N}} x_n &\iso& \parallel_{n\in
\mathbb{N}} y_n\\
&& (n, a) &\mapsto& (\pi(n), \theta_n(a))
\end{array}
\]
for some permutation $\pi \in \varsigma(\mathbb{N})$ and for some family
$(\theta_n)_{n\in \mathbb{N}}$ with $\theta_n : x_n \sym_{A} y_{\pi(n)}$
for all $n \in \mathbb{N}$. First we set $\theta \in \ntilde{\oc A}$ iff
for all $n \in \mathbb{N}$, $\theta_n : x_n \sym_A^- y_{\pi(n)}$.
Finally, we set $\theta \in \ptilde{\oc A}$ iff for all $n \in
\mathbb{N}$ such that $x_n$ is non-empty, we have $\pi(n) = n$, and
$\theta_n : x_n \sym_A^+ y_n$.

We recall from Section~\ref{subsubsec:ar_sym}
our representation of configurations of $\oc A$ with an example:
\begin{exa}\label{ex:symm}
Consider the arena $o$ with only one (negative) move $\qu^-$, and $\oc
(\oc o \lin o)$, with set of events $\mathbb{N} \times (\mathbb{N}
\times \{\qu\} \uplus \{\qu\})$. Following the notations of Section~\ref{subsubsec:ar_sym}, we draw 
\[
\xymatrix@R=5pt@C=10pt{
\qu^-_{\grey{0}}
        \ar@{.}[d] &
\qu^-_{\grey{1}}
        \ar@{.}[d]\\
\qu^+_{\grey{0,2}}&
\qu^+_{\grey{1,4}}
}
\]
for the configuration $\{(\grey{0}, \qu), (\grey{1}, \qu), (\grey{0},
(\grey{2}, \qu)), (\grey{1}, (\grey{4}, \qu))\}$, which may also be
written as $\parallel_{i\in \{0, 1\}} x_i$
for $x_0 = \{(\grey{2}, \qu)\}$ and $x_1 = \{(\grey{4}, \qu)\}$.
Writing $\gamma_{0, 1} \in \varsigma(\mathbb{N})$ for the permutation
exhanging $0$ and $1$, the symmetry described by $(\gamma_{0, 1},
(\id_{x_i})_{i\in \{0, 1\}})$ may be represented as
\[
\raisebox{10pt}{$
\xymatrix@R=5pt@C=10pt{
\qu^-_{\grey{0}}
	\ar@{.}[d] &
\qu^-_{\grey{1}}
	\ar@{.}[d]\\
\qu^+_{\grey{0,2}}&
\qu^+_{\grey{1,4}}
}$}
\qquad
\sym_{\oc (\oc o \lin o)}
\qquad
\raisebox{10pt}{$
\xymatrix@R=5pt@C=10pt{
\qu^-_{\grey{1}}
	\ar@{.}[d] &
\qu^-_{\grey{0}}
	\ar@{.}[d]\\
\qu^+_{\grey{1,2}}&
\qu^+_{\grey{0,4}}
}$}
\]
sending a move on one side to the move at the same spot on the other
side. This reindexes initial move, leaving the subsequent
moves morally unchanged: $\qu^+_{\grey{0}, \grey{2}}$ is sent to
$\qu^+_{\grey{1}, \grey{2}}$, but the first number just identifies the
causal predecessor -- the actual copy index, $\grey{2}$, is unchanged.

The symmetry above is negative. It is not positive, because
$\gamma_{\grey{0}, \grey{1}}$ is not the identity. Likewise, the two
symmetries represented below:
\[
\raisebox{10pt}{$
\xymatrix@R=5pt@C=5pt{
\qu^-_{\grey{0}}
	\ar@{.}[d] &
\qu^-_{\grey{1}}
	\ar@{.}[d]\\
\qu^+_{\grey{0,2}}&
\qu^+_{\grey{1,4}}
}$}
\quad
\sym_{\oc (\oc o \lin o)}
\quad
\raisebox{10pt}{$
\xymatrix@R=5pt@C=5pt{
\qu^-_{\grey{0}}
	\ar@{.}[d] &
\qu^-_{\grey{1}}
	\ar@{.}[d]\\
\qu^+_{\grey{0,2}}&
\qu^+_{\grey{1,5}}
}$}
\qquad\qquad
\raisebox{10pt}{$
\xymatrix@R=5pt@C=5pt{
\qu^-_{\grey{0}}
	\ar@{.}[d] &
\qu^-_{\grey{1}}
	\ar@{.}[d]\\
\qu^+_{\grey{0,2}}&
\qu^+_{\grey{1,4}}
}$}
\quad
\sym_{\oc (\oc o \lin o)}
\quad
\raisebox{10pt}{$
\xymatrix@R=5pt@C=5pt{
\qu^-_{\grey{1}}
	\ar@{.}[d] &
\qu^-_{\grey{0}}
	\ar@{.}[d]\\
\qu^+_{\grey{1,5}}&
\qu^+_{\grey{0,2}}
}$}
\]
are respectively positive, and neither positive nor negative.
\end{exa}

Why does $\oc A$ treat differently the positive and negative isomorphism
families? The permutation $\pi(-)$ corresponds to reindexing the minimal
events of $\oc A$.
Because the exponential construction is intended to
apply to negative games, $\pi(-)$ reindexes negative moves. But
symmetries in $\ptilde{\oc A}$ can only reindex positive moves, so
$\pi(-)$ must be the identity. In contrast, symmetries in $\ntilde{\oc
A}$ can perform any reindexing on the minimal events.

This intuition is
further solified by the extension of \emph{concrete arenas} to these
polarized sub-isomorphism families. This will not be required until much
later, so it only appears in Section~\ref{subsubsec:concrete_pol_sym} --
but it might still be helpful for the reader to consult it now.

\subsubsection{Causal strategies}
In contrast with traditional game models, a causal strategy is one
global object: an event structure.  It presents all execution threads
together, with explicit information on how these executions relate
via parallel and non-deterministic branching.

\begin{defi}\label{def:caus_strat}
A \textbf{causal prestrategy} $\bsigma : A$ comprises an ess
$(\ev{\bsigma}, \leq_{\bsigma}, \conflict_\bsigma, \tilde{\bsigma})$
with
\[
\pr : \ev{\bsigma} \to \ev{A}
\]
a function called the \textbf{display map}, subject to the following
conditions:
\[
\begin{array}{rl}
\text{\emph{rule-abiding:}} & \text{for $x \in \conf{\bsigma}$,
$\pr(x) \in \conf{A}$,}\\
\text{\emph{locally injective:}} & \text{for $s_1, s_2 \in x \in
\conf{\bsigma}$, if $\pr(s_1) = \pr(s_2)$ then $s_1 = s_2$,}\\
\text{\emph{symmetry-preserving:}} &
\text{for $\theta \in \tilde{\bsigma}$, $\pr(\theta) = \{(\pr(s_1),
\pr(s_2)) \mid (s_1, s_2) \in \theta\} \in \tilde{A}$,}\\ 
\text{\emph{$\sim$-receptive:}} &
\text{for $\theta : x \sym_\bsigma y$, and extensions $x \vdash_\bsigma s_1^-$,
$\pr(\theta) \vdash_{\tilde{A}} (\pr(s^-_1), a^-_2)$,}\\
&\text{there is a unique $s_2^- \in \ev{\bsigma}$ s.t. $\theta
\vdash_{\tilde{\bsigma}} (s^-_1, s^-_2)$ and $\pr(s_2^-) = a_2^-$,}\\ 
\text{\emph{thin:}} & 
\text{for $\theta : x \sym_\bsigma y$, and extension $x
\vdash_\bsigma s_1^+$,}\\
&\text{there is a \emph{unique} extension $y \vdash_\bsigma s_2^+$
such that $\theta \vdash_{\tilde{\bsigma}} (s_1^+, s_2^+)$.}
\end{array}
\]

Additionally, we say that $\bsigma$ is a \textbf{causal strategy} if it
also satisfies:
\[
\begin{array}{rl}
\text{\emph{negative:}} & \text{for $s \in \ev{\bsigma}$, if $s$ is
minimal then $s$ is negative,}\\
\text{\emph{courteous:}} & \text{for $s_1 \imc_\bsigma s_2$, if
$\pol(s_1) = +$ or $\pol(s_2) = -$ then
$\pr(s_1) \imc_A \pr(s_2)$,}\\
\text{\emph{receptive:}} & \text{for $x \in \conf{\bsigma}$, for 
$\pr(x) \vdash_A a^-$,}\\
&\text{there is a unique $x \vdash_\bsigma s^- \in \conf{\bsigma}$ such
that
$\pr(s) = a$,}
\end{array}
\]
\end{defi}

As a convention, causal strategies are ranged over by
symbols in bold font, as in \emph{e.g.} $\bsigma, \btau$.
We disambiguate some notations used in the definition. First,
$\bsigma$ implicitly comes with polarities,
imported from $A$ as $\pol_\bsigma(s) = \pol_A(\pr(s))$. We also used
the enabling relation on isomorphism families, defined
by $\theta \vdash_{\tilde{A}} (a_1, a_2)$ iff $(a_1, a_2)
\not \in \theta$ and $\theta \uplus \{(a_1, a_2)\} \in \tilde{A}$.

Conditions \emph{rule-abiding}, \emph{locally injective} and
\emph{symmetry-preserving} together amount to $\pr$ being a \textbf{map
of event structures with symmetry} \cite{symmetry}. Conditions \emph{courteous} and
\emph{receptive} play the same role as in Definition
\ref{def:nalt_strat}. The condition \emph{$\sim$-receptive} forces
strategies to treat uniformly any pairs of Opponent events symmetric
in the game.  Finally, \emph{thin} forces strategies to pick \emph{one} canonical
representative up to symmetry for positive moves. For further
explanations and discussions on those conditions, the reader is directed
to \cite{cg2}.

Causal strategies and non-alternating strategies differ fundamentally
in \emph{how} the concurrent behaviour is represented. While
non-alternating strategies present observable execution traces, causal
strategies present the causal constraints underlying the
observed behaviour.
\begin{figure}
\begin{minipage}{.35\linewidth}
\[
\xymatrix@R=5pt@C=-3pt{
&(\tunit \ar@{}[rrrrr]|\lin&&&&&\tunit)\ar@{}[rrrrr]|\lin&&&&&\tnat\\
&&&&&&&&&&&\qu^-
        \ar@{-|>}[dlllll]\\
&&&&&&\qu^+
        \ar@{-|>}[dlllll]
        \ar@{-|>}[d]
        \ar@{.}@/^/[urrrrr]\\
&\qu^-
        \ar@{-|>}[dl]
        \ar@{-|>}[dr]
        \ar@{-|>}[drrrrrrrrr]
        \ar@{.}@/^/[urrrrr]&&&&&
\done^- \ar@{-|>}[dllll]
        \ar@{-|>}[drrrr]
        \ar@{-|>}[drrrrrr]
        \ar@{.}@/^/[u]\\
\done^+ \ar@{~}[rr]
        \ar@{~}@/_1pc/[rrrrrrrrrrrr]
        \ar@{.}@/^/[ur]&&
\done^+ \ar@{~}[rrrrr]
        \ar@{.}@/_/[ul]&&&&&&&&
1^+     \ar@{~}[rr]
        \ar@{.}@/^/[uuur]&&
0^+     \ar@{.}@/_/[uuul]
}
\]
\caption{A causal strategy}
\label{fig:ex_aug1}
\end{minipage}
\hfill
\begin{minipage}{.6\linewidth}
\[
\xymatrix@R=5pt@C=20pt{
(\tunit \ar@{}[r]|\lin&\tunit)\ar@{}[r]|\lin&\tnat\\
&&\qu^-  \ar@{-|>}[dl]\\
&\qu^+   \ar@{-|>}[dl]
        \ar@{-|>}[d]
        \ar@{.}@/^/[ur]\\
\qu^-    \ar@{-|>}[d]
        \ar@{.}@/^/[ur]
        \ar@{-|>}[drr]&
\done^- \ar@{-|>}[dr]
        \ar@{.}@/^/[u]\\
\done^+ \ar@{.}@/^/[u]&&
1^+\ar@{.}@/_/[uuu]
}
\qquad\qquad
\xymatrix@R=5pt@C=20pt{
(\tunit \ar@{}[r]|\lin&\tunit)\ar@{}[r]|\lin&\tnat\\
&&\qu^-  \ar@{-|>}[dl]\\
&\qu^+   \ar@{-|>}[dl]
        \ar@{-|>}[d]
        \ar@{.}@/^/[ur]\\
\qu^-    \ar@{-|>}[d]
        \ar@{.}@/^/[ur]&
\done^- \ar@{-|>}[dl]
        \ar@{-|>}[dr]
        \ar@{.}@/^/[u]\\
\done^+ \ar@{.}@/^/[u]&&
0^+\ar@{.}@/_/[uuu]
}
\]
\caption{Two augmentations of Figure \ref{fig:ex_aug1}}
\label{fig:ex_aug2}
\end{minipage}
\end{figure}
In Figure \ref{fig:ex_aug1}, we present a causal strategy,
corresponding to a linear version of Figure \ref{fig:comp_nalt}.
In this diagram and others further on, we
draw $\bsigma : A$ by picturing the event structure
$\bsigma$ with events displayed as their image through $\pr_\bsigma$.
Whenever possible, we keep the convention to draw moves under the
corresponding type component. The causal dependency
$\leq_\bsigma$ is pictured via its immediate dependency relation
$\imc_\bsigma$. As with arenas, we show immediate conflicts
(\emph{i.e.} not inherited) as wiggly lines.
As in plays, we denote by dotted lines the
immediate causal dependency in the arena $A$: for $s \in \ev{\bsigma}$,
either $\pr_\sigma(s)$ is minimal in $A$, or there is a unique $s' \in
\ev{\bsigma}$ such that $\pr_\sigma(s') \imc_A \pr_\sigma(s)$; in which
case the diagram has a dotted line between the events representing $s$
and $s'$. Borrowing earlier terminology, we refer to $s'$ as the
\textbf{justifier} of
$s$. Finally, the symmetry $\tilde{\bsigma}$ is not shown at all (of
course, for Figure \ref{fig:ex_aug1} it would remain trivial).
Indeed it is hard to represent, but also we regard it as not being part
of the pertinent operational structure: its mere existence witnesses
uniformity.

Figure \ref{fig:ex_aug1} presents, in \emph{one} diagram, the full
behaviour of the program of Figure \ref{fig:comp_nalt} under linear
execution contexts. Drawing a strategy fully in this way is sometimes
challenging. It is often
convenient to refer to -- and draw -- consistent fragments of a causal
strategy: an \textbf{augmentation} of $\bsigma$ is any $(x,
\leq_x)$ where $x \in \conf{\bsigma}$ and $\leq_x$ is the partial order
inherited from $\leq_\bsigma$. For instance
the two maximal augmentations of Figure \ref{fig:ex_aug1} yield the
diagrams of Figure \ref{fig:ex_aug2}. In this case, the two
configurations correspond to the two resolutions of the date race
described with Figure \ref{fig:comp_nalt}: if the write wins, the read
yields $1$ and depends on $\qu^-$ (which triggered the write). If the
read wins, we read $0$ and the write acknowledgment depends on $\done^-$
(which triggered the read).
However, this representation of a strategy via its augmentations is
partial: it
forgets the explicit non-deterministic branching.

\subsubsection{Non-alternating unfolding}\label{subsubsec:gen_nalt}
Causal strategies generate non-alternating strategies:

\begin{prop}\label{prop:def_naltstrat}
Consider $A$ a $-$-arena, and $\bsigma : A$ a causal strategy.

The \textbf{non-alternating unfolding} of $\bsigma$ is (with
$\pr_\bsigma$ applied to plays move-by-move):
\[
\NAltStrat(\bsigma) = \pr_\bsigma(\NAlt(\bsigma))\,.
\]

Thus defined, $\NAltStrat(\bsigma) : A$ is a non-alternating strategy on
$A$.
\end{prop}
\begin{proof}
\emph{Non-empty} and \emph{prefix-closed} are obvious. For
\emph{receptive}, take $s \in \NAltStrat(\bsigma)$ s.t. $sa^- \in
\NAlt(A)$. By definition, there is $t \in \NAlt(\bsigma)$ s.t. $s =
\pr_\bsigma(t)$. Thus $t$, written 
\[ 
t = t_1 \dots t_n \in \NAlt(\bsigma)\,, 
\] 
is such that for all $1\leq i \leq n$, $\{t_1,
\dots, t_i\} \in \conf{\bsigma}$, and is also \emph{non-repetitive} and
\emph{negative}. Let us write $x = \{t_1, \dots, t_n\}$.  Then,
$\pr_\bsigma(x) = \{s_1, \dots, s_n\} \in \conf{A}$, and since $sa^- \in
\NAlt(A)$ we have $\pr_\bsigma(x) \vdash_A a^-$. By \emph{receptivity}
of $\bsigma$, there is (a unique) $m \in \ev{\bsigma}$ such that $x
\vdash_\bsigma m$ and $\pr_\bsigma(m) = a$. It is then direct that $t m
\in \NAlt(\bsigma)$, witnessing $sa \in \NAltStrat(\bsigma)$ as
required.
 
For \emph{courteous}, consider $s_1abs_2\in \NAltStrat(\bsigma)$ s.t.
$\pol(a) = +$ or $\pol(b) = -$, and $s_1 b \in \NAlt(A)$. By
definition, there is $t_1 m n t_2 \in \NAlt(\bsigma)$ s.t. $s_1 a b
s_2 = \pr_\bsigma(t_1 m n t_2)$. We claim that $x = \ev{t_1 n} \in
\conf{\bsigma}$. We know that $\ev{t_1 m n} \in \conf{\bsigma}$; so $x$
is consistent. If it is not down-closed, necessarily $m \imc_\bsigma n$.
So, by \emph{courtesy} of $\bsigma$, $a \imc_A b$; contradicting
$s_1 b \in \NAlt(A)$. Thus $\ev{t_1 n} \in \conf{\bsigma}$ 
from which we deduce $t_1 n m t_2 \in \NAlt(\bsigma)$, so $s_1 b a s_2 \in
\NAltStrat(\bsigma)$.
\end{proof}

For instance, the non-alternating play in Figure \ref{fig:comp_nalt} is
in the unfolding of the causal strategy in Figure \ref{fig:ex_aug1}.
This extraction of a non-alternating strategy is an instance of the
usual relationship between interleaving and ``truly concurrent'' models
for concurrency. In this paper this relationship will in particular
allow us to import well-bracketing from $\NNegAltwb$. 

\begin{defi}\label{def:caus_wb_nalt}
Consider $\bsigma : A$ a causal strategy on $-$-arena $A$.

We say that $\bsigma$ is \textbf{well-bracketed} if $\NAltStrat(\bsigma)
: A$ is well-bracketed.
\end{defi}

\subsection{A Category of Causal Strategies} We now start building the
categorical operations on causal strategies, aiming at a Seely category
$\CG$. We first focus on composition.

For $-$-arenas $A$ and $B$, a causal
strategy \textbf{from $A$ to $B$} is a causal strategy
\[
\bsigma : A^\perp \parallel B\,,
\]
in the sequel we also write $A \vdash B$ for $A^\perp \parallel B$. This
is unlike for $\NegAlt$ and $\NNegAlt$ introduced earlier, for which
morphisms from $A$ to $B$ were defined as strategies on $A\lin B$. We do
this to keep close to \cite{cg2}. When linking with $\NegAlt$ and
$\NNegAlt$ we shall deal with this  mismatch, but that will not cause us
too much trouble\footnote{Plays on $A\lin B$ carry more information than
on $A\vdash B$, namely the justifier for initial moves in $A$. 
With causal strategies, that information may be read back from
the causal structure. See Section~\ref{subsubsec:cg_mon_clos}.}.  

Composition of causal strategies is more elaborate than for play-based
strategies. We define it in several stages. Fix two causal
(pre)strategies $\bsigma : A\vdash B$ and $\btau : B \vdash C$.

\subsubsection{Synchronization}
For configurations $x^\bsigma \in \conf{\bsigma}, x^\btau \in
\conf{\btau}$, as a convention we write
\[
\pr_\bsigma(x^\bsigma) = x^\bsigma_A \parallel x^\bsigma_B \in
\conf{A\vdash B}\,,
\qquad
\qquad
\pr_\btau(x^\btau) = x^\btau_B \parallel x^\btau_C \in \conf{B \vdash
C}\,,
\]
for the corresponding projections to the game. In defining composition,
the first stage is to capture when such configurations $x^\bsigma \in
\conf{\bsigma}$ and $x^\btau \in \conf{\btau}$ may successfully
\emph{synchronise}.

\begin{defi}\label{def:caus_comp}
Consider two configurations $x^\bsigma \in \conf{\bsigma}$ and $x^\btau
\in \conf{\btau}$.
They are \textbf{causally compatible} if \emph{(1)} they are
\textbf{matching}: $x^\bsigma_B = x^\btau_B = x_B$; and \emph{(2)} if
the composite bijection
\[
\varphi_{x^\bsigma, x^\btau}
\quad
:
\quad
x^\bsigma \parallel x^\tau_C 
\quad
\stackrel{\pr_\bsigma \parallel x^\btau_C}{\simeq} 
\quad
x^\bsigma_A \parallel x_B \parallel x^\btau_C 
\quad
\stackrel{x^\bsigma_A \parallel \pr_\btau^{-1}}{\simeq}
\quad
x^\bsigma_A \parallel x^\btau\,,
\]
using local injectivity of $\pr_\bsigma$ and
$\pr_\btau$, is \textbf{secured}, in the sense that the relation
\[
(m, n) \vartriangleleft (m', n') 
\qquad
\Leftrightarrow
\qquad
m <_{\bsigma \parallel C} m'
\quad
\vee
\quad
n <_{A \parallel \btau} n'\,,
\]
defined on (the graph of) $\varphi_{x^\bsigma, x^\btau}$ by importing
causal constraints of $\bsigma$ and $\btau$, is acyclic.
\end{defi}

Two matching configurations $x^\bsigma \in \conf{\bsigma}$, $x^\btau \in \conf{\btau}$
agree on the state reached in $B$. By local injectivity of $\pr_\bsigma$
and $\pr_\btau$, this induces a bijection as above, 
thought of as the induced synchronization between events of $x^\bsigma$
and $x^\btau$ that match in $B$. But this is not enough to capture a
sensible notion of execution: some matching pairs might not be
\emph{reachable}, in the sense that $\bsigma$ and $\btau$ impose
incompatible constraints as to the execution order.
\begin{figure}
\begin{minipage}{.45\linewidth}
\[
\raisebox{50pt}{$
\xymatrix@R=5pt@C=20pt{
\tunit  \ar@{}[r]|\lin& \tunit\\
&\qu^-  \ar@{-|>}[dl]\\
\qu^+   \ar@{.}@/^/[ur]
        \ar@{-|>}[d]\\
\done^- \ar@{.}@/_/[u]
        \ar@{-|>}[dr]\\
&\done^+\ar@{.}@/_/[uuu]
}$}
\quad
\text{vs}
\quad
\raisebox{50pt}{$
\xymatrix@R=5pt@C=20pt{
(\tunit \ar@{}[r]|\lin&\tunit)\ar@{}[r]|\vdash&\tnat\\
&&\qu^-  \ar@{-|>}[dl]\\
&\qu^+   \ar@{-|>}[dl]
        \ar@{-|>}[d]\\
\qu^-    \ar@{-|>}[d]
        \ar@{.}@/^/[ur]
        \ar@{-|>}[drr]&
\done^- \ar@{-|>}[dr]
        \ar@{.}@/^/[u]\\
\done^+ \ar@{.}@/^/[u]&&
1^+\ar@{.}@/_/[uuu]
}$}
\]
\caption{Matching, secured}
\label{fig:matching1}
\end{minipage}
\hfill
\begin{minipage}{.5\linewidth}
\[
\raisebox{50pt}{$
\xymatrix@R=5pt@C=20pt{
\tunit  \ar@{}[r]|\lin& \tunit\\
&\qu^-  \ar@{-|>}[dl]\\
\qu^+   \ar@{.}@/^/[ur]
        \ar@{-|>}[d]\\
\done^- \ar@{.}@/_/[u]
        \ar@{-|>}[dr]\\
&\done^+\ar@{.}@/_/[uuu]
}$}
\quad
\text{vs}
\quad
\raisebox{50pt}{$
\xymatrix@R=5pt@C=20pt{
(\tunit \ar@{}[r]|\lin&\tunit)\ar@{}[r]|\vdash&\tnat\\
&&\qu^-  \ar@{-|>}[dl]\\
&\qu^+   \ar@{-|>}[dl]
        \ar@{-|>}[d]\\
\qu^-    \ar@{-|>}[d]
        \ar@{.}@/^/[ur]&
\done^- \ar@{-|>}[dr]
        \ar@{.}@/^/[u]
        \ar@{-|>}[dl]\\
\done^+ \ar@{.}@/^/[u]&&
0^+\ar@{.}@/_/[uuu]
}$}
\]
\caption{Matching, non-secured}
\label{fig:matching2}
\end{minipage}
\end{figure}
To illustrate this we show in Figures \ref{fig:matching1}
and \ref{fig:matching2} two attempted synchronizations between
configurations of the strategy of Figure \ref{fig:ex_aug1} and the
causal strategy for the identity $\lambda x^\tunit.\,x$. In both cases,
the configurations are matching. In Figure \ref{fig:matching1}, the
synchronization is successful and yields causally compatible pairs of
configurations. However, in Figure \ref{fig:matching2} the  induced
bijection is not \emph{secured}: the two strategies impose opposite
constraints as to the order in which the two $\done$ moves are to be
played.  Thus, this synchronization fails.  For us, this will entail
that the identity $\lambda x^\tunit.\,x$ may only synchronize
successfully with the augmentation of the program of Figure
\ref{fig:comp_nalt} appearing in Figure \ref{fig:matching1} -- so that
the only final result is $1$.

\subsubsection{Interaction}
But we must present the interaction of $\bsigma$ and $\btau$ as an event
structure. More specifically, it should be an event structure with
symmetry along with a display map:

\begin{defi}
A \textbf{pre-interaction} on $A, B, C$ is an ess $\bmu =
(\ev{\bmu},
\leq_\bmu, \conflict_\bmu, \tilde{\bmu})$ with
\[
\pr : \ev{\bmu} \to \ev{A \parallel B \parallel C}
\]
a display map subject to the following conditions:
\[
\begin{array}{rl}
\text{\emph{rule-abiding:}} & \text{for all $x \in \conf{\bmu}$,
$\pr(x) \in \conf{A\parallel B \parallel C}$,}\\
\text{\emph{locally injective:}} & \text{for all $s_1, s_2 \in x \in
\conf{\bmu}$, if $\pr(s_1) = \pr(s_2)$ then $s_1 =
s_2$,}\\
\text{\emph{symmetry-preserving:}} &
\text{for all $\theta \in \tilde{\bmu}$, $\pr(\theta) \in
\tilde{A\parallel B \parallel C}$,}
\end{array}
\]
\emph{i.e.} $\pr : \bmu \to A\parallel B \parallel C$ is a \emph{map of
event structures with symmetry}.
\end{defi}

An \textbf{isomorphism} between pre-interactions $\bmu$, $\bnu$ on $A,
B, C$ is an isomorphism $f : \bmu \iso \bnu$ in the category of event
structures with symmetry, commuting with the display maps, \emph{i.e.}
$\pr_{\bnu} \circ f = \pr_\bmu$.  The \emph{interaction} between
$\bsigma$ and $\btau$ is a pre-interaction whose configurations
correspond \emph{exactly} with pairs of causally compatible
configurations $x^\bsigma \in \conf{\bsigma}, x^\btau \in
\conf{\btau}$:

\begin{prop}\label{prop:main_interaction}
There is a pre-interaction $\btau \inter \bsigma$, the
\textbf{interaction} of $\bsigma$ and $\btau$, with
\[
(- \inter -) 
:
\{(x^\btau, x^\bsigma) \in \conf{\btau} \times \conf{\bsigma} \mid   
\text{$x^\bsigma$ and $x^\btau$ are causally compatible}\} 
\simeq
\conf{\btau \inter \bsigma}
\]
an order-iso (with causally compatible pairs ordered by
component-wise inclusion) satisfying
\[
\pr_{\btau\inter \bsigma}(x^\btau \inter x^\bsigma)
=
x^\bsigma_A \parallel x_B \parallel x^\btau_C
\]
for all $x^\bsigma \in \conf{\bsigma}$ and $x^\btau \in \conf{\btau}$
causally compatible.
\end{prop}

In particular, any $z \in \conf{\btau \inter \bsigma}$ is
written uniquely as $x^\btau \inter x^\bsigma$ for $x^\bsigma \in
\conf{\bsigma}$ and $x^\btau \in \conf{\btau}$. Thus, $\conf{\btau
\inter \bsigma}$ may be regarded as the subset of $\conf{\bsigma} \times
\conf{\btau}$ restricted to those of the matching pairs which cause no
deadlock. In fact, this property \emph{almost} suffices to characterise
the interaction uniquely -- to complete the picture, we must also
consider \emph{symmetries}. Because display maps preserve symmetry,
for $\theta^\bsigma \in \tilde{\bsigma}$ and $\theta^\btau \in
\tilde{\btau}$, 
\[
\pr_\bsigma(\theta^\bsigma) = \theta^\bsigma_A \parallel
\theta^\bsigma_B\,,
\qquad
\pr_{\btau}(\theta^\btau) = \theta^\btau_B \parallel \theta^\btau_C\,,
\]
and we can say that $\theta^\bsigma$ and $\theta^\btau$ are
\textbf{matching} if $\theta^\bsigma_B = \theta^\btau_B$, and
\textbf{causally compatible} if $\dom(\theta^\bsigma)$,
$\dom(\theta^\btau)$ are causally compatible -- or, equivalently,
$\cod(\theta^\bsigma)$ and $\cod(\theta^\btau)$ are.

Taking into account symmetry, we may strengthen Proposition
\ref{prop:main_interaction} to:

\begin{restatable}{prop}{charinter}\label{prop:main_interaction2}
There is a pre-interaction $\btau \inter \bsigma$, unique
up to iso, such that there are
\[
\begin{array}{rcrcl}
(- \inter -) 
&:&
\{(x^\btau, x^\bsigma) \in \conf{\btau} \times \conf{\bsigma} \mid
\text{$x^\bsigma, x^\btau$ causally compatible}\} 
&\simeq&
\conf{\btau \inter \bsigma}\\
(- \inter -)
&:&
\{(\theta^\btau, \theta^\bsigma) \in \tilde{\btau} \times
\tilde{\bsigma} \mid \text{$\theta^\bsigma, \theta^\btau$ 
causally compatible}\}
&\simeq&
\tilde{\btau \inter \bsigma}
\end{array}
\]
order-isomorphisms commuting with $\dom$ and $\cod$, and satisfying
\[
\pr_{\btau\inter \bsigma}(\theta^\btau \inter \theta^\bsigma)
=
\theta^\bsigma_A \parallel \theta_B \parallel \theta^\btau_C
\]
for all $\theta^\bsigma \in \tilde{\bsigma}$ and $\theta^\btau \in
\tilde{\btau}$
causally compatible.
\end{restatable}
\begin{proof}
Follows from the characterisation of the interaction as a
\emph{pullback}, whose projections
\[
\xymatrix{
\bsigma \parallel C&
\btau \inter \bsigma
        \ar[l]_{\Pi_\bsigma}
        \ar[r]^{\Pi_\btau}&
A \parallel \btau
}
\]
are maps of event structures with symmetry \cite{cg2} -- see Appendix
\ref{app:interaction} for details.
\end{proof}

There is some redundancy in this statement: first,
the action of $(-\inter -)$ on configurations coincides with 
that on identity symmetries. Reciprocally, one can
actually prove that the action of $(- \inter -)$ on symmetries, \emph{if
it exists}, is uniquely determined by that on configurations -- so the
fact that $(-\inter -)$ extends to symmetries is property rather than
structure. Nevertheless, in the sequel, we often find convenient to
perform the constructions on symmetries explicitly.
Altogether, this characterises the interaction in terms of its
\emph{states} and \emph{symmetries}.

But there is also an alternative, \emph{event-based} view: an individual event
$m \in \ev{\btau \inter \bsigma}$ may be regarded as a synchronization
between its projections $\Pi_\bsigma(m) \in \ev{\bsigma \parallel C}$
and $\Pi_\btau(m) \in \ev{A \parallel \btau}$. But we
warn against the misleading idea that $m \in \ev{\btau \inter
\bsigma}$ is determined by these projections: intuitively, there is one
event in $\btau\inter \bsigma$ for each pair $(s, t)$ of synchronizable
events, \emph{and} each distinct way to reach $s$ and $t$ conjointly in
$\bsigma$ and $\btau$.
\begin{figure}
\[
\left(
\raisebox{40pt}{$
\scalebox{.85}{$
\xymatrix@R=5pt@C=-2pt{
&(\tbool\ar@{}[rrrr]|\lin&&&&
\tunit) \ar@{}[rrrr]|\vdash&&&&\tunit\\
&&&&&&&&&\qu^-
        \ar@{-|>}[dllll]\\
&&&&&\qu^+
        \ar@{-|>}[dllll]
        \ar@{-|>}[d]\\
&\qu^-  \ar@{-|>}[dl]
        \ar@{-|>}[dr]
        \ar@{.}@/^.1pc/[urrrr]
&&&&\done^-
        \ar@{.}@/^.1pc/[u]
        \ar@{-|>}[drrrr]\\
\ttrue^+\ar@{.}@/^.1pc/[ur]
        \ar@{~}[rr]&&
\tfalse^+
        \ar@{.}@/_.1pc/[ul]&&&&&&&
\done^+ \ar@{.}@/_/[uuu]
}$}$}
\right)
\inter
\left(
\raisebox{40pt}{$
\scalebox{.85}{$
\xymatrix@R=5pt@C=-2pt{
&\tbool \ar@{}[rrrr]|\lin&&&&\tunit\\
&&&&&\qu^-
        \ar@{-|>}[dllll]\\
&\qu^+  \ar@{.}@/^.1pc/[urrrr]
        \ar@{-|>}[dl]
        \ar@{-|>}[dr]\\
\ttrue^-\ar@{.}@/^.1pc/[ur]
        \ar@{~}[rr]
        \ar@{-|>}[drrrr]&&
\tfalse^-
        \ar@{.}@/_.1pc/[ul]
        \ar@{-|>}[drrrr]\\
&&&&\done^+
        \ar@{.}@/^/[uuur]&&
\done^+ \ar@{.}@/_/[uuul]
}
$}$}
\right)
=
\left(
\raisebox{40pt}{$
\scalebox{.85}{$
\xymatrix@R=-2pt@C=-2pt{
&\grey{\tbool}
        \ar@{}[rrrr]|{\grey{\lin}}&&&&
\grey{\tunit}
        \ar@{}[rrrr]|{\grey{\vdash}}&&&&
\tunit\\
&&&&&&&&&\qu^-
        \ar@[grey]@{-|>}[dllll]\\
&&&&&\grey{\qu^\labr}
        \ar@[grey]@{-|>}[dllll]\\
&\grey{\qu^\labl}
        \ar@[grey]@{.}@/^.1pc/[urrrr]
        \ar@[grey]@{-|>}[dl]
        \ar@[grey]@{-|>}[dr]\\
\grey{\ttrue^\labr}
        \ar@[grey]@{.}@/^.1pc/[ur]
        \ar@[grey]@{~}[rr]
        \ar@[grey]@{-|>}[drrrr]&&
\grey{\tfalse^\labr}
        \ar@[grey]@{.}@/_.1pc/[ul]
        \ar@[grey]@{-|>}[drrrr]\\
&&&&\grey{\done^\labl}
        \ar@[grey]@{.}@/^/[uuur]
        \ar@[grey]@{-|>}[drrrr]&&
\grey{\done^\labl}
        \ar@[grey]@{.}@/_/[uuul]
        \ar@[grey]@{-|>}[drrrr]\\
&&&&&&&&\done^\labr
        \ar@{.}@/^/[uuuuur]&&
\done^\labr
        \ar@{.}@/_/[uuuuul]
}
$}$}
\right)
\]
\caption{Example of an interaction}
\label{fig:ex_interaction}
\end{figure}
A simple example appears in Figure
\ref{fig:ex_interaction} (ignoring for now the $\labl / \labr$
annotation and the part in grey): the final two copies of $\done$ have
the same
projections, but a different causal history. Though we shall not
unfold the concrete construction of the interaction \cite{cg2},
it might nonetheless help the reader to have an idea of what its
events are concretely defined to be. For
$x^\bsigma \in \conf{\bsigma}$ and $x^\btau \in \conf{\btau}$
causally compatible, the reflexive transitive closure of
$\vartriangleleft$ (see Definition \ref{def:caus_comp}) yields a
partial order $\leq_{x^\bsigma, x^\btau}$ on (the graph of)
$\varphi_{x^\bsigma, x^\btau}$. The \emph{events} of $\btau \inter
\bsigma$ are then precisely the causally compatible pairs $(x^\bsigma,
x^\btau)$ such that $\leq_{x^\bsigma, x^\btau}$ has a top element: the
pair $(\Pi_\bsigma(m), \Pi_\btau(m))$.

In the sequel, we shall only reason on the interaction through the proxy
of Proposition \ref{prop:main_interaction2} and forthcoming lemmas
characterizing
immediate causality in the interaction. However, to ease
the flow of the exposition, those are postponed to Section~\ref{subsubsec:caus_conf_int}.

\subsubsection{Composition} Following the traditional methodology of
game semantics, composition is defined from interaction via
\emph{hiding}. We first briefly analyse the \emph{components}
of interactions.

The projections $\Pi_\bsigma$ and $\Pi_\btau$ 
project any event of $\btau \inter \bsigma$ to a matching pair of an
event of $\bsigma \parallel C$ and an event of $A \parallel \btau$.
These projections help us classify every $p \in \ev{\btau \inter
\bsigma}$ into:
\[
\begin{array}{rl}
\text{\emph{(1)}} &
\text{$\Pi_\bsigma(p) = (1, m)$ with $m \in \ev{\bsigma}$, and
$\Pi_\btau(p) = (1, a)$ with $a \in \ev{A}$,}\\
\text{\emph{(2)}} &
\text{$\Pi_{\bsigma}(p) = (1, m)$ with $m \in \ev{\bsigma}$, and
$\Pi_\btau(p) = (2, n)$ with $n \in \ev{\btau}$,}\\
\text{\emph{(3)}} &
\text{$\Pi_\bsigma(p) = (2, c)$ with $c \in \ev{C}$, and $\Pi_\btau(p)
=
(2, n)$ with $n \in \ev{\btau}$.}
\end{array}
\]

In case \emph{(1)}, the only relevant projection is $\Pi_\bsigma(p) =
(1, m)$ as $\Pi_\btau(p) = \pr_\bsigma(m)$. We write
$p_\bsigma = m$ and $p_\btau$ is undefined, and we say that $p$
\textbf{occurs in $A$}. In case \emph{(3)}, the only relevant projection
is $\Pi_\btau(p) = (2, n)$ as $\Pi_\bsigma(p) =
\pr_\btau(n)$.  We write $p_\btau = n$ and $p_\bsigma$ is undefined, and
we say that $p$ \textbf{occurs in $C$}. Finally, in case \emph{(2)} the
two projections $\Pi_\bsigma(p) = (1, m)$ and $\Pi_\btau(p) = (2, n)$
are relevant, but we must have $\pr_\bsigma(m) = (2, b)$ and
$\pr_\btau(p) = (1, b)$ for some $b \in \ev{B}$. We write $p_\bsigma =
m$, $p_\btau = n$, we say that $p$ \textbf{occurs in $B$} and that $p$
is \textbf{synchronized}.

The definition of composition consists simply in removing all
synchronized events:

\begin{defi}
The \textbf{composition} of $\bsigma : A \vdash B$ and $\btau : B \vdash
C$ comprises components:
\[
\begin{array}{rcl}
\ev{\btau \odot \bsigma} &=& \{p \in \ev{\btau \inter \bsigma} \mid
\text{$p$ occurs in $A$ or $C$}\}\,,\\
p_1 \leq_{\btau \odot \bsigma} p_2 &\Leftrightarrow& p_1 \leq_{\btau
\inter \bsigma} p_2\,,\\
p_1 \conflict_{\btau \odot \bsigma} p_2 &\Leftrightarrow& p_1
\conflict_{\btau \inter \bsigma} p_2\,,\\
\theta : x \sym_{\btau \odot \bsigma} y &\Leftrightarrow& \exists \theta
\subseteq \theta' : x' \sym_{\btau \inter \bsigma} y'\,.
\end{array}
\]
with display map $\pr_{\btau\odot \bsigma} : \ev{\btau \odot \bsigma}
\to \ev{A \vdash C}$ obtained as restriction of $\pr_{\btau \inter
\bsigma}$.
\end{defi}

The composition of prestrategies $\bsigma : A \vdash B$ and $\btau : B
\vdash C$ gives data
\[
(\ev{\btau \odot \bsigma}, \leq_{\btau \odot \bsigma}, \conflict_{\btau
\odot \bsigma}, \tilde{\btau \odot \bsigma}, \pr_{\btau \odot \bsigma})
\]
satisfying all the axioms of Definition \ref{def:caus_strat} except,
possibly, $\sim$-receptivity\footnote{However, $\sim$-receptivity of
$\bsigma$ and $\btau$ is required for $\btau\inter \bsigma$ to form an
ess \cite{cg2}.}. When composing prestrategies, we will check
$\sim$-receptivity separately -- this only occurs
in Section~\ref{subsubsec:intr_inter}.

However, if $\bsigma$ and $\btau$ are \emph{strategies}, then so is
$\btau \odot \bsigma$ \cite{cg2}. Composition is associative up to iso
(with \textbf{isomorphisms} between causal strategies defined as between
pre-interactions above).  In Figure \ref{fig:ex_interaction}, the
composition simply keeps the events in black. This means that the
composition has two conflicting positive events, both corresponding to
$\done^+$: the model records the point of non-deterministic branching
even when it brings no observable difference. Though this does not
appear in pictures, we insist that events of $\btau \odot \bsigma$
\emph{are} certain events of $\btau \inter \bsigma$. Thus an event of
the composition always carries a unique causal explanation:
\emph{itself}.  

To parallel this event-based definition of composition, there is a
state-based characterization.  A causally compatible pair $x^\bsigma \in
\conf{\bsigma}, x^\btau \in \conf{\btau}$ is \textbf{minimal} if for all
causally compatible $y^\bsigma \in \conf{\bsigma}, y^\btau \in
\conf{\btau}$ with $y^\bsigma \subseteq x^\bsigma$, $y^\btau \subseteq
x^\btau$ with $x^\bsigma_A = y^\bsigma_A$ and $x^\btau_C = y^\btau_C$,
then $x_B = y_B$. The same definition applies to causally
compatible pairs of symmetries. 

\begin{restatable}{prop}{charcomp}\label{prop:char_conf_comp}
Consider $\bsigma : A\vdash B$, and $\btau : B \vdash C$ causal
strategies.

There is a causal strategy $\btau \odot \bsigma$, unique up to
iso, s.t. there are order-isos:
\[
\begin{array}{rcrcl}
(- \odot -) 
&\!\!:\!\!&
\{(x^\btau, x^\bsigma) \in \conf{\btau} \times \conf{\bsigma} \mid
\text{$x^\bsigma, x^\btau$ \emph{minimal} caus. comp.}\} 
&\!\!\simeq\!\!&
\conf{\btau \odot \bsigma}\\
(- \odot -)
&\!\!:\!\!&
\{(\theta^\btau, \theta^\bsigma) \in \tilde{\btau} \times
\tilde{\bsigma} \mid 
\text{$\theta^\bsigma, \theta^\btau$ \emph{minimal} caus. comp.}\}
&\!\!\simeq\!\!&
\tilde{\btau\odot \bsigma}
\end{array}
\]
commuting with $\dom$ and $\cod$; s.t., for $\theta^\bsigma \in
\tilde{\bsigma}, \theta^\btau \in \tilde{\btau}$ minimal causally compatible,
\[
\pr_{\btau\odot \bsigma}(\theta^\btau \odot \theta^\bsigma) =
\theta^\bsigma_A \parallel \theta^\btau_C\,.
\]
\end{restatable}

The \emph{minimality} requirement amounts to asking the maximal events of
$x^\bsigma$ and $x^\btau$ to occur in $A$ or $C$. As events
of $\btau\odot \bsigma$ carry their causal witness, configurations of
$\btau \odot \bsigma$ are in one-to-one correspondence with those configurations
of $\btau \inter \bsigma$ whose maximal events occur in $A$ or $C$ -- thus
Proposition \ref{prop:char_conf_comp} follows from Proposition
\ref{prop:main_interaction2} (see Appendix \ref{app:char_conf_comp}).

In fact, trailing Opponent moves do not matter as they are
forced by \emph{receptivity} and \emph{courtesy} to behave as in the
game. A configuration $x \in \conf{\bsigma}$ is \textbf{$+$-covered} iff
the top elements of $x$ (for $\leq_\bsigma$) are positive -- we write $x
\in \confp{\bsigma}$. Likewise, $\theta \in \tilde{\bsigma}$ is
\textbf{$+$-covered} if $\dom(\theta)$ (or, equivalently,
$\cod(\theta)$) is $+$-covered -- we write $\theta \in
\tildep{\bsigma}$. We have \cite{devismephd}:

\begin{restatable}{lem}{lempcovtrois}\label{lem:pcov3_main}
Consider $\bsigma, \btau : A$ two causal strategies. Assume there are
\[
\psi : \confp{\bsigma} \simeq \confp{\btau}
\qquad
\qquad
\psi : \tildep{\bsigma} \simeq \tildep{\btau}
\]
order-isomorphisms compatible with $\dom, \cod$, and display maps.

Then, $\bsigma$ and $\btau$ are isomorphic.
\end{restatable}

See Appendix \ref{app:pluscov} for the proof. Relying on this we can
finally prove:

\begin{restatable}{prop}{charcompbis}\label{prop:comp_pcov}
Consider $\bsigma : A \vdash B$ and $\btau : B \vdash C$ causal
strategies.

Then, there is a strategy $\btau \odot \bsigma : A \vdash C$, unique up
to iso, such that there are order-isos:
\[
\begin{array}{rcrcl}
(- \odot -) 
&\!\!:\!\!&
\{(x^\btau, x^\bsigma) \in \confp{\btau} \times \confp{\bsigma} \mid
\text{$x^\bsigma$ and $x^\btau$ caus. comp.}\} 
&\!\!\simeq\!\!&
\confp{\btau \odot \bsigma}\\
(- \odot -)
&\!\!:\!\!&
\{(\theta^\btau, \theta^\bsigma) \in \tildep{\btau} \times
\tildep{\bsigma} \mid \text{$\theta^\bsigma$ and $\theta^\btau$ caus.
comp.}\}
&\!\!\simeq\!\!&
\tildep{\btau\odot \bsigma}
\end{array}
\]
commuting with $\dom$ and $\cod$, and
s.t., for $\theta^\bsigma \in \tildep{\bsigma}$
and $\theta^\btau \in \tildep{\btau}$ causally compatible,
\[
\pr_{\btau\odot \bsigma}(\theta^\btau \odot \theta^\bsigma) =
\theta^\bsigma_A \parallel \theta^\btau_C\,.
\]
\end{restatable}
\begin{proof}
Relatively direct from Proposition \ref{prop:char_conf_comp} and Lemma
\ref{lem:pcov3_main}, see Appendix \ref{app:pluscov}.
\end{proof}

This is convenient as $+$-covered configurations of strategies often
have a simpler description (see \emph{e.g.} Lemma \ref{lem:cc_pcov}).
Minimality also disappears as a causally
compatible pair of $+$-covered configurations is always minimal (indeed,
a synchronized maximal event would be negative for one of the players).
This final characterization will be used often to prove equalities
between strategies. It is also of great use when linking with the
relational model (see Section~\ref{subsec:positional_collapse}), but
also for quantitative extensions (see \emph{e.g.}
\cite{lics18,DBLP:journals/pacmpl/ClairambaultV20}).

\subsubsection{Congruence} What is the right equivalence
between causal (pre)strategies?
There are a few options, several investigated in
\cite{cg2}; here we use \emph{positive isomorphism}:

\begin{defi}\label{def:pos_iso}
Consider $\bsigma, \btau : A$ two causal strategies on arena $A$.

A \textbf{positive isomorphism} $\varphi : \bsigma \simstrat \btau$ is
an isomorphism of ess satisfying
\[
\pr_\btau \circ \varphi \sim^+ \pr_\bsigma\,,
\]
\emph{i.e.} for all $x \in \conf{\bsigma}$, $\{(\pr_\bsigma(s),
\pr_\btau\circ \varphi(s)) \mid s \in x\} \in \ptilde{A}$: the two
maps are \textbf{positively symmetric}.
In that case we say
$\bsigma$ and $\btau$ are \textbf{positively isomorphic}, and write
$\bsigma \simstrat \btau$.
\end{defi}

This means that $\bsigma$ and $\btau$ are the same up to renaming of
their events. This renaming might cause a reindexing of
positive events, but it must keep the copy indices of negative events
unchanged.
Crucially, positive isomorphism is preserved by composition
\cite{cg2}:

\begin{prop}
Consider $\bsigma, \bsigma' : A \vdash B$, $\btau, \btau' : B \vdash
C$ s.t. $\bsigma \simstrat \bsigma'$ and $\btau \simstrat \btau'$.

Then, we have
$\btau \odot \bsigma 
\simstrat
\btau' \odot \bsigma'$.
\end{prop}

The proof is fairly elaborate. Without going into details, it will be
useful to have in mind the first key step: showing that
two (pre)strategies able to synchronize \emph{up to symmetry}, always
also have a synchronization \emph{on the nose}. More precisely, we have
the following:

\begin{prop}\label{prop:sync_sym}
Consider $\bsigma : A \vdash B$ and $\btau : B \vdash C$ two causal
(pre)strategies.

For any $x^\bsigma \in \conf{\bsigma}, x^\btau \in \conf{\btau}$ and
$\theta : x^\bsigma_B \sym_B x^\btau_B$ s.t. the composite
bijection is \emph{secured}:
\[
x^\bsigma \parallel x^\btau_C 
\quad
\stackrel{\pr_{\bsigma}\parallel C}{\simeq}
\quad
x^\bsigma_A \parallel x^\bsigma_B \parallel x^\btau_C 
\quad
\stackrel{A \parallel \theta \parallel C}{\sym}
\quad
x^\bsigma_A \parallel x^\btau_B \parallel x^\btau_C
\quad
\stackrel{A \parallel \pr_{\btau}^{-1}}{\simeq}
\quad
x^\bsigma_A \parallel x^\btau\,,
\]
then there are
$y^\bsigma \in \conf{\bsigma}$ and $y^\btau \in \conf{\btau}$ causally
compatible, along with symmetries
\[
\varphi^\bsigma : y^\bsigma \sym_\bsigma x^\bsigma\,,
\qquad
\qquad
\varphi^\btau : y^\btau \sym_\btau x^\btau\,,
\]
such that $\varphi^\btau_B \circ \theta = \varphi^\bsigma_B$.
\end{prop}

This follows from Lemma 3.23 in \cite{cg2}.
Intuitively, we play $\tilde{\bsigma}$ and $\tilde{\btau}$
against each other. By \emph{$\sim$-receptivity} and \emph{extension}
they adjust their copy indices interactively until reaching
an agreement. This is the first step to
congruence, but not the only one: the requirement that we should get a
global map $\varphi : \btau \odot \bsigma \simstrat \btau' \odot
\bsigma'$ is in tension with the definition of isomorphism families,
which only guarantees a more local bisimulation-like property. The
mismatch is compensated by the uniqueness of extensions granted by
\emph{thin}, without which congruence fails. Details are out of scope
for the present paper \cite{cg2}.

If $\bsigma \simstrat \btau$, there can be in principle multiple
$\varphi : \bsigma \simstrat \btau$. We leave
these isomorphisms to the background, as we have not yet
encountered a computational use for these. If they are retained, then
arenas, causal strategies and positive morphisms form a bicategory
\cite{paquet2020probabilistic}.

\begin{rem}
In Definition \ref{def:pos_iso}, one could ask $\varphi$ to preserve
$\pr$ up to arbitrary symmetry (\emph{weak isomorphism}) or even, to
be itself invertible only up to symmetry (\emph{weak equivalence}). This
changes the mediating morphisms, but not the resulting equivalence
relation between strategies (see Corollary 3.30 in \cite{cg2}). In this
paper we choose positive isomorphism as it seems natural conceptually,
and because the additional positivity constraint is useful.
\end{rem}

\subsubsection{Copycat} So as to complete the categorical structure, it
remains to define \emph{copycat}.

\begin{defi}\label{def:copycat}
For each $-$-arena $A$, the \textbf{copycat strategy} $\cc_A : A \vdash
A$ is defined as:
\[
\begin{array}{rcl}
\ev{\bcc_A} &=& \ev{A \vdash A}\\
\pr_{\bcc_A}(i, a) &=& (i, a)\\
(i, a) \leq_{\bcc_A} (j,a') &\Leftrightarrow& \text{$a <_A a'$; or $a =
a'$ and ($\pol_{A \vdash A}(i, a) = -$ or $\pol_{A^\perp \parallel
A}(j, a') = +$)}\\
(i, a) \conflict_{\bcc_A} (j, a') &\Leftrightarrow& a \conflict_A a'\,,
\end{array}
\]
with symmetries those bijections of the form $\theta_1 \parallel
\theta_2 : x_1 \parallel x_2 \sym_{\bcc_A} y_1 \parallel y_2$ such that
\[
\theta_1 : x_1 \sym_A y_1\,,
\quad
\theta_2 : x_2 \sym_A y_2\,,
\quad
\text{and}
\quad
\theta_1 \cap \theta_2 : x_1 \cap x_2 \sym_A y_1 \cap y_2\,.
\]
\end{defi}

This simplifies the usual definition \cite{cg2},
exploiting the particular shape of arenas. Its immediate causal links
import $\imc_A$ on either side, along with all the $(i, a) \imc_{\bcc_A}
(j,a)$
when $\pol_{A \vdash A}(i, a) = -$ and $\pol_{A\vdash A}(j,a) = +$. In
other words, $\bcc_A$ is an asynchronous forwarder: it is prepared to
play any positive event on one side, under the condition that the
corresponding negative event appears first on the other side. 
Its symmetries are inherited from $\tilde{A\vdash A}$, with the
constraint that they should agree on events already forwarded.

Perhaps the simplest description of copycat is through its
completely forwarded states:

\begin{lem}\label{lem:cc_pcov}
Consider $A$ any $-$-arena. Then, we have:
\[
\begin{array}{rcl}
\confp{\bcc_A} &=& \{x_A \parallel x_A \in \conf{A\parallel A} \mid x_A
\in \conf{A}\}\,\\ 
\tildep{\bcc_A} &=& \{\theta_A \parallel \theta_A \in \tilde{A\parallel
A}\mid \theta_A \in
\tilde{A}\}\,.
\end{array}
\]
\end{lem}
\begin{proof}
Straightforward.
\end{proof}

This foreshadows the link with relational semantics in
Section~\ref{subsec:positional_collapse}: when restricted to $+$-covered
configurations, copycat looks like the identity relation. We may deduce:

\begin{prop}\label{prop:cc_neutral}
Composition is associative up to $\simstrat$ on prestrategies.
For any $\bsigma : A \vdash B$,
\[
\bcc_B \odot \sigma \odot \bcc_A \simstrat \bsigma\,,
\]
so that $-$-arenas and causal strategies form a category.
\end{prop}
\begin{proof}
Associativity follows from Proposition \ref{prop:char_conf_comp} and a
ternary version of causal compatibility -- see also \cite{cg2} for a
detailed proof via the universal property of the interaction
pullback.
For neutrality of copycat, there is an order-isomorphism preserving
display maps
\begin{eqnarray*}
\confp{\bcc_B \odot \bsigma} &\simeq& 
\{(x^{\bcc_B}, x^\bsigma) \in \confp{\bcc_B} \times \confp{\bsigma} \mid 
\text{$x^{\bcc_B}$ and $x^\bsigma$ causally compatible}\}\\
&\simeq& \{(x^{\bcc_B}, x^\bsigma) \in \confp{\bcc_B} \times
\confp{\bsigma} \mid 
\text{$x^{\bcc_B}$ and $x^\bsigma$ matching}\}\\
&\simeq& \{(x^\bsigma_B \parallel x^\bsigma_B, x^\bsigma) \mid x^\bsigma
\in \confp{\bsigma}\}\\
&\simeq& \confp{\bsigma}\,,
\end{eqnarray*}
using first Proposition \ref{prop:comp_pcov};
verifying directly that securedness always holds when
composing with copycat; using Lemma \ref{lem:cc_pcov}. The 
same reasoning can be made with symmetries, concluding that $\bcc_B
\odot \bsigma$ and $\bsigma$ are isomorphic by uniqueness in Proposition
\ref{prop:comp_pcov}. 
\end{proof}

Before we develop further this
categorical structure, we introduce a few useful lemmas.

\subsubsection{Immediate causality in interactions}
\label{subsubsec:caus_conf_int}
Later on, we will
need some tools to reason on the causality in $\btau \inter
\bsigma$ and how it relates to that in $\bsigma$ and $\btau$.

\begin{restatable}{lem}{lemcharcaus}\label{lem:imc_mconf_inter}
For $\bsigma : A \vdash B, \btau : B \vdash C$ causal
prestrategies, for $m, m' \in \ev{\btau \inter \bsigma}$,
if $m \imc_{\btau \inter \bsigma} m'$, then
$m_\bsigma \imc_{\bsigma} m'_\bsigma$, or $m_\btau \imc_{\btau}
m'_\btau$,
where $m_\bsigma, m_\btau$ are defined whenever used.
\end{restatable}

The proof is in Appendix \ref{app:charcaus}.
So in the event-based view of interaction, immediate causal links 
originate in one of the components.
For $\bsigma$ and $\btau$ \emph{strategies}, one can
track down the responsible component via a polarity
analysis. Of course, it is usual in game semantics that
events of $\btau \inter \bsigma$ cannot sensibly be assigned a polarity
in $\{-, +\}$, because $\bsigma$ and $\btau$ disagree on $B$. A more
useful notion of polarity is $\pol_{\btau \inter \bsigma} : \ev{\btau
\inter \bsigma} \to \{-, \labl, \labr\}$ given by:
\[
\begin{array}{rclcl}
\pol_{\btau \inter \bsigma}(m) &=& \labl &\qquad& \text{if
$m_\sigma$ is defined and $\pol_\bsigma(m_\sigma) = +$,}\\
\pol_{\btau \inter \bsigma}(m) &=& \labr   && \text{if
$m_\btau$ is defined and $\pol_\btau(m_\btau) = +$,}\\
\pol_{\btau \inter \bsigma}(m) &=& - &&\text{otherwise}. 
\end{array}
\]

As an example, we show in Figure \ref{fig:ex_interaction} the polarities
arising from this definition. Then:

\begin{lem}\label{lem:caus_int}
Consider $\bsigma : A \vdash B$ and $\btau : B \vdash C$ strategies, and
$m \imc_{\btau \inter \bsigma} m'$. Then,
\[
\begin{array}{rl}
\text{\emph{(1)}}& \text{if $\pol_{\btau \inter \bsigma}(m') = \labl$,
then $m_\bsigma \imc_\bsigma m'_\bsigma$,}\\
\text{\emph{(2)}}& \text{if $\pol_{\btau \inter \bsigma}(m') = \labr$,
then $m_\btau \imc_\btau m'_\btau$,}\\
\text{\emph{(3)}}& \text{if $\pol_{\btau \inter \bsigma}(m') = -$, then
$\pr_{\btau \inter \bsigma}(m) \imc_{A \parallel B \parallel C}
\pr_{\btau \inter \bsigma}(m')$}\,. 
\end{array}
\]
\end{lem}
\begin{proof}
\emph{(1)} By Lemma \ref{lem:imc_mconf_inter}, 
$m_\bsigma, m'_\bsigma$ defined and $m_\bsigma \imc_\bsigma m'_\bsigma$
-- in which case we are done; or $m_\btau, m'_\btau$ defined and
$m_\btau \imc_\btau m'_\btau$. Since
$\pol_{\btau \inter \bsigma}(m') = \labl$,
$\pol_\btau(m'_\btau) = -$. By courtesy,
$\pr_{\btau}(m_\btau) \imc_{B\vdash C} \pr_{\btau}(m'_\btau)$; hence
$m$ occurs in $B$ and $\pr_\bsigma(m_\bsigma) \imc_{A\vdash B}
\pr_\bsigma(m'_\bsigma)$. By Lemma \ref{lem:es_refl_caus}, $m_\bsigma
<_\bsigma m'_\bsigma$, and the causality must be immediate by Lemma
\ref{lem:imc_aux}. \emph{(2)} is symmetric.

\emph{(3)} Assume $m'$ occurs in $A$, the other case is symmetric. In
that case only $m'_\bsigma$ is defined, so Lemma
\ref{lem:imc_mconf_inter} entails that $m_\bsigma$ is defined and
$m_\bsigma \imc_\bsigma m'_\bsigma$. But $\pol_\bsigma(m'_\bsigma) = -$,
so by courtesy $\pr_\bsigma(m_\bsigma) \imc_{A\vdash B}
\pr_\bsigma(m'_\bsigma)$, from which the conclusion follows.
\end{proof}

\subsection{Seely Category} Now, we turn to the different components of
a Seely category.

\subsubsection{Symmetric monoidal category with products}
On $-$-arenas, we keep the definitions for $\NegAlt$, enriched as
in Section~\ref{subsubsec:constr_arenas}.  The tensor of causal
strategies is defined below:

\begin{defi}
For $\bsigma_1 : A_1 \vdash B_1$, $\bsigma_2 : A_2 \vdash B_2$
causal strategies between $-$-arenas, then
\[
\bsigma_1 \tensor \bsigma_2 : A_1 \tensor A_2 \vdash B_1 \tensor B_2
\]
is defined as the ess $\bsigma_1 \parallel \bsigma_2$ along with display map
$\pr_{\bsigma_1 \tensor \bsigma_2}(i, s) = (j, (i, a))$ 
if $\bsigma_i(s) = (j, a)$.
\end{defi}

Bifunctoriality is direct via Proposition \ref{prop:comp_pcov}.
The symmetric monoidal structural isomorphisms are
provided by copycat strategies, only changing display maps:
\[
\begin{array}{rcrcl}
\balpha_{A, B, C} &:& (A \tensor B) \tensor C &\iso& A \tensor (B
\tensor C)\\
\mathbf{s}_{A, B} &:& A \tensor B &\iso & B \tensor A
\end{array}
\qquad
\begin{array}{rcrcl}
\brho_A &:& A \tensor \ees &\iso& A\\
\blambda_A &:& \ees \tensor A &\iso& A
\end{array}
\]
satisfying up to positive iso the expected naturality and coherence laws
\cite{cg2}.
For cartesian products, the \emph{projections}
$\bpi_1 : A_1 \with A_2 \vdash A_1$ and $\bpi_2 : A_1 \with A_2 \vdash
A_2$ are relabeled  copycat strategies,
while the \emph{pairing} of causal strategies is defined similarly to
the tensor:

\begin{defi}
Consider $\bsigma_1 : A \vdash B_1$ and $\btau : A \vdash B_2$ causal
strategies between $-$-arenas. Then,
$\tuple{\bsigma_1, \bsigma_2} : A \vdash B_1 \with B_2$
is defined as having ess $\bsigma_1 \with \bsigma_2$, along with
\[
\begin{array}{rclcl}
\pr_{\tuple{\bsigma_1, \bsigma_2}}(i, s) &=& (1, a)
&\qquad&
\text{if $\pr_{\bsigma_i}(s) = (1, a)$}\\
\pr_{\tuple{\bsigma_1, \bsigma_2}}(i,s) &=& (2, (i, b))
&& \text{if $\pr_{\bsigma_i}(s) = (2, b)$.}
\end{array}
\]
\end{defi}

It follows from Proposition \ref{prop:comp_pcov} and direct
verifications that this yields binary products.

\subsubsection{Monoidal closed structure}\label{subsubsec:cg_mon_clos}
We now describe the monoidal closure.

On objects, the closure is the arrow $A \lin B$ from $\NegAlt$. However,
for now, the strategies on $A\vdash B$ and on $A \lin B$ are \emph{not}
in one-to-one correspondence.
\begin{figure}
\[
\raisebox{60pt}{$
\xymatrix@R=8pt@C=3pt{
\tunit&\vdash & \tunit &\tensor &\tunit\\
&&\qu^-
        \ar@{-|>}[dll]&&
\qu^-   \ar@{-|>}[dllll]\\
\qu^+   \ar@{-|>}[d]\\
\done^- \ar@{.}@/^/[u]
        \ar@{-|>}[drr]
        \ar@{-|>}[drrrr]\\
&&\done^+
        \ar@{.}@/_/[uuu]&&
\done^+ \ar@{.}@/_/[uuu]
}$}
\qquad
\text{vs}
\qquad
\raisebox{60pt}{$
\xymatrix@R=8pt@C=3pt{
\tunit&\lin & \tunit &\tensor &\tunit\\
&&\qu^-
        \ar@{-|>}[dll]&&
\qu^-   \ar@{-|>}[dllll]\\
\qu^+   \ar@{-|>}[d]
        \ar@{.}@/^/[urr]\\
\done^- \ar@{.}@/^/[u]
        \ar@{-|>}[drr]
        \ar@{-|>}[drrrr]\\
&&\done^+
        \ar@{.}@/_/[uuu]&&
\done^+ \ar@{.}@/_/[uuu]
}$}
\quad
\text{and}
\quad
\raisebox{60pt}{$
\xymatrix@R=8pt@C=3pt{
\tunit&\lin & \tunit &\tensor &\tunit\\
&&\qu^-
        \ar@{-|>}[dll]&&
\qu^-   \ar@{-|>}[dllll]\\
\qu^+   \ar@{-|>}[d]
        \ar@{.}@/_/[urrrr]\\
\done^- \ar@{.}@/^/[u]
        \ar@{-|>}[drr]
        \ar@{-|>}[drrrr]\\
&&\done^+
        \ar@{.}@/_/[uuu]&&
\done^+ \ar@{.}@/_/[uuu]
}$}
\]
\caption{Non-uniqueness of the threading pointer}
\label{fig:non_uniq_ptr}
\end{figure}
Indeed strategies in $A\lin B$ include a \emph{pointer} for
initial moves in $A$, while strategies in $A\vdash B$ do not. This
pointer is not always unique, as illustrated in Figure
\ref{fig:non_uniq_ptr}. To cope with this we could have, as for
the play-based strategies of the previous sections, set the
morphisms of our category directly as strategies on $A \lin B$; but
in this causal setting that obfuscates composition.

Instead, we restrict to strategies for which this pointer reconstruction
is \emph{unique}:

\begin{defi}\label{def:caus_pointed}
A causal strategy $\bsigma : A$ on arena $A$ is \textbf{pointed} if for
each $s \in \ev{\bsigma}$ there is a unique event $\init(s) \in
\ev{\bsigma}$ which is minimal for $\leq_\bsigma$ and such that
$\init(s) \leq_\bsigma s$.
\end{defi}

Copycat strategies are pointed (as arenas are forestial),
and pointed strategies are stable under composition and the other
operations on strategies.  From now on, we consider that all
causal strategies are pointed.  We write $\CG$ for the category
having $-$-arenas as objects, and as morphisms from $A$ to $B$, the
pointed causal strategies on $A \vdash B$.

For pointed strategies, the missing pointer can always be recovered
uniquely:

\begin{lem}
Let $A, B$ and $C$ be $-$-arenas. Then, we have a bijection:
\[
\Lambda_{A,B,C}
\quad
:
\quad
\CG(A \tensor B, C) 
\quad
\simeq
\quad
\CG(A, B \lin C)
\]
\end{lem}
\begin{proof}
The bijection only affects the display map, leaving the other components
unchanged.

The non-trivial direction is from left to right. Consider $\bsigma :
A\tensor B \vdash C$. We set:
\[
\pr_{\Lambda(\bsigma)}(s) = 
\left\{
\begin{array}{lcl}
(1, a) &\quad&
        \text{if $\pr_{\bsigma}(s) = (1, (1, a))$,}\\
(2, (2, c))&&
        \text{if $\pr_{\bsigma}(s) = (2, c)$,}\\
(2, (1, (c, b)))&&
        \text{if $\pr_\bsigma(s) = (1, (2, b))$ and
$\pr_\bsigma(\init(s)) = (2, c)$.}
\end{array}\right.
\]

It is a direct verification that this yields a bijection as claimed.
\end{proof}

From this point, we may now easily wrap up the symmetric monoidal closed
structure.

\begin{prop}
The category $\CG$ is symmetric monoidal closed.
\end{prop}
\begin{proof}
First, for any $-$-arenas $A$ and $B$, we have $A\lin B$ and an
\textbf{evaluation}
\[
\evm_{A, B} = \Lambda_{A\lin B, A, B}^{-1}(\cc_{A\lin B}) : (A\lin B)
\tensor A \vdash B\,,
\]
and given $\bsigma : A \tensor B \vdash C$, 
$\evm_{B, C} \odot (\Lambda_{A,B,C}(\bsigma) \tensor B) \simstrat
\bsigma$ follows from a variation over the neutrality of copycat for
composition. From there, the universal property is routine.
\end{proof}

\subsubsection{The exponential}
The first step is to introduce a functor $\oc : \CG \to \CG$.

\begin{defi}
Consider $A$ and $B$ two $-$-arenas, and $\bsigma : A \vdash B$ a causal
strategy.

We define a strategy $\oc \bsigma :\,\oc A \vdash\,\oc B$ with
$\oc \bsigma$ as event structure with symmetry and:
\[
\pr_{\oc \bsigma}(\grey{i}, m) = 
\left\{\begin{array}{lcl}
(1, (\grey{i}, a)) &\quad&
        \text{if $\pr_\bsigma(m) = (1, a)$,}\\
(2, (\grey{i}, b)) &&
        \text{if $\pr_\bsigma(m) = (2, b)$.}
\end{array}\right.
\]
\end{defi}

It is a direct verification that this defines a causal strategy, and
functoriality is proved as for the tensor product of strategies.
To complete the categorical structure, we have:

\begin{prop}
The category $\CG$ is a Seely category.
\end{prop}
\begin{proof}
The structure presented above is completed by structural natural
families of
strategies:
\[
\begin{array}{rcrcl}
\bdig_A &:& \oc A &\to& \oc \oc A\\
\bder_A &:& \oc A &\to & A
\end{array}
\qquad
\qquad
\begin{array}{rcrcl}
\bmon^2_{A, B} &:& \oc A \tensor \oc B &\iso& \oc (A \with B)\\
\bmon^0 &:& \ees &\iso& \oc \ees\,.
\end{array}
\]
making $(\oc, \bdig, \bder)$ a comonad along with the \emph{Seely
isomorphisms}. Those are all relabeled copycat strategies:
for instance, $\bdig_A$ is $\bcc_{\oc \oc A}$ relabeled on the left
hand side following a bijection $\mathbb{N} \times \mathbb{N} \simeq
\mathbb{N}$, $\bder_A$ is $\bcc_A$ relabeled to set events on the
left hand side to copy index $\grey{0}$, \emph{etc}.
The naturality and coherence are easily verified, exploiting again
Proposition \ref{prop:comp_pcov}.
\end{proof}

\subsubsection{Extracting plays} In Section~\ref{subsubsec:gen_nalt}, we
unfolded causal strategies to non-alternating strategies. Here, we show
that this is compatible with the categorical operations. 

First, we extend the definition in Proposition \ref{prop:def_naltstrat}
for causal strategies \emph{from $A$ to $B$}.

\begin{defi}
For $A$ and $B$ two $-$-arenas, and $\bsigma : A \vdash B$ a causal
strategy, we define
\[
\NAltStrat(\bsigma) = \pr_{\Lambda(\bsigma)}(\NAlt(\bsigma)) \in
\NNegAlt(A, B)
\]
exploiting that $\bsigma$ and $\Lambda(\bsigma)$ only differ via their
display map.
\end{defi}

This matches applying Proposition \ref{prop:def_naltstrat}
to $\Lambda(\bsigma) : A \lin B$ obtained by monoidal closure.

\begin{prop}\label{prop:nalt_fonc}
There is a symmetric monoidal closed
$\NAltStrat(-) : \CG \to \NNegAlt$.
\end{prop}
\begin{proof}
For identities, the definition of plays of the
asynchronous copycat (Definition \ref{def:asy_cc}) follows the
characterisation of configurations of $\bcc_A$ found \emph{e.g.} in
Lemma 3.11 in \cite{cg1}.

For composition, take $\bsigma : A \vdash B$ and
$\btau : B \vdash C$. Though $\btau \inter
\bsigma$ is not an esp,
Definition \ref{def:nalt_play} generalizes transparently to
$\NAlt(\btau\inter \bsigma)$. There are two inclusions to check:

$\NAltStrat(\btau \odot \bsigma) \subseteq \NAltStrat(\btau) \odot
\NAltStrat(\bsigma)$: any $s \in \NAltStrat(\btau \odot \bsigma)$ has
the form $\pr_{\Lambda(\btau \odot \bsigma)}(t)$ for $t \in \NAlt(\btau
\odot \bsigma)$, which in turn can be completed to $v \in \NAlt(\btau
\inter \bsigma)$. Then $v$ may be displayed to $u \in \NAltStrat(\btau)
\inter \NAltStrat(\bsigma)$, witnessing $s \in \NAltStrat(\btau) \odot
\NAltStrat(\bsigma)$.

$\NAltStrat(\btau) \odot \NAltStrat(\bsigma) \subseteq \NAltStrat(\btau
\odot \bsigma)$: any $s \in \NAltStrat(\btau) \odot \NAltStrat(\bsigma)$
has a witness $u \in \NAltStrat(\btau) \inter \NAltStrat(\bsigma)$,
projecting to $u \restrict A, B \in \NAltStrat(\bsigma)$ and $u \restrict B,
C \in \NAltStrat(\btau)$. Those are respectively
$\pr_{\Lambda(\bsigma)}(s^\bsigma)$ and $\pr_{\Lambda(\btau)}(s^\btau)$
for $s^\bsigma \in \NAlt(\bsigma)$ and $s^\btau \in \NAlt(\btau)$. 
Then $x^\bsigma := \ev{s^\bsigma}$ and $x^\btau := \ev{s^\btau}$
are causally compatible as
$u$ induces a linearization of the corresponding bijection. By
construction, $x^\btau \inter x^\bsigma$ has a linearization $v$ that
displays to $u$; and its restriction to visible events yields $t \in
\NAlt(\btau \odot \bsigma)$ that displays to $s$.

The preservation of the monoidal structure is direct; the functor is
strict monoidal.
\end{proof}

As $\NNegAlt$ does not handle symmetry, it supports no
equivalence relation corresponding to $\simstrat$.  Nevertheless, this
lets us import the stability of well-bracketing
under composition. We first generalize Definition \ref{def:caus_wb_nalt}
to well-bracketed causal strategies between $-$-arenas:

\begin{defi}\label{def:wb_caus_strat}
Consider $-$-arenas $A, B$, and $\bsigma : A\vdash B$ a causal strategy.

We say that $\bsigma : A \vdash B$ is \textbf{well-bracketed} iff
$\NAltStrat(\bsigma)$ is well-bracketed.
\end{defi}

From Propositions \ref{prop:nalt_smcc} and \ref{prop:nalt_fonc}, there
is an smcc with products $\CGwb$ of $-$-arenas
and well-bracketed causal strategies. Finally, $\oc(-)$ preserves
well-bracketing and the other components for the exponential are
well-bracketed, so $\CGwb$ extends to a Seely category.

\subsection{Interpretation of \texorpdfstring{$\IPA$}{IA//}} We now describe the interpretation
of $\IPA$ in $\CGwb_\oc$. First the \emph{types} of $\IPA$ are the same
as those of $\IA$; their interpretation does not change.

\subsubsection{Interpretation of core $\PCF$}\label{subsubsec:core_pcf}
We focus on the core $\PCF$ primitives, postponing $\mathbf{let}$.

The $\lambda$-calculus primitives are interpreted following the
cartesian closed structure of $\CGwb_\oc$. For constants, we use the
obvious strategies returning the corresponding value. For basic $\PCF$
combinators, there are obvious causal strategies corresponding to
Figures \ref{fig:intr_pcf} and \ref{fig:strat_pred}.  For recursion, we
must first define a partial order on causal strategies: 

\begin{defi}
Consider $A$ an arena, and $\bsigma, \btau : A$ causal (pre)strategies.

We write $\bsigma \cleq \btau$ iff $\conf{\bsigma} \subseteq
\conf{\btau}$ -- so $\ev{\bsigma} \subseteq \ev{\btau}$ as well --
with, additionally:
\[
\begin{array}{rl}
\text{\emph{(1)}} & 
\text{for all $s_1, s_2 \in \ev{\bsigma}$, $s_1 \leq_\bsigma s_2$ iff
$s_1 \leq_\btau s_2$,}\\
\text{\emph{(2)}} & 
\text{for all $s_1, s_2 \in \ev{\bsigma}$, $s_1 \conflict_\bsigma s_2$
iff $s_1 \conflict_\btau s_2$,}\\
\text{\emph{(3)}} &
\text{for all $x, y \in \conf{\bsigma}$ and bijection $\theta : x \simeq
y$, we have $\theta \in \tilde{\bsigma}$ iff $\theta \in
\tilde{\btau}$,}\\
\text{\emph{(4)}} &
\text{for all $s \in \ev{\bsigma}$, $\pr_\bsigma(s) = \pr_\btau(s)$,}
\end{array}
\]
\emph{i.e.} all components compatible with the inclusion.
\end{defi}

Causal strategies on $A$, ordered by $\cleq$, form a
\emph{directed complete partial order}; however without a least
element. Indeed, strategies minimal for $\cleq$ still
have -- by receptivity -- events corresponding to the minimal events of
$A$, but those are \emph{named} arbitrarily. We solve
this as in \cite{cg2}: we choose one minimal causal strategy $\bot_A :
A$ with events \emph{exactly} those negative minimal in $A$; induced
causality, conflict, and isomorphism family; and as display map the
identity. For any $\bsigma : A$, we pick $\bsigma^\flat \iso
\bsigma : A$ such that $\bot_A \cleq \bsigma^\flat$, obtained
by renaming minimal events. We write $\D_A$ for the pointed dcpo of
causal strategies above $\bot_A$.

All operations on strategies examined so far are continuous.
So, the operation
\[
\begin{array}{rcrcl}
F &:& \D_{\emptyar \vdash (A \to A) \to A} &\to& \D_{\emptyar \vdash (A \to A) \to A}\\
&& \sigma &\mapsto& (\lambda f^{A \to A}.\,f\,(\sigma\,f))^\flat\,,
\end{array}
\]
in $\lambda$-calculus syntax following the cartesian closed structure of
$\CGwb_\oc$, is continuous. Thus it has a least fixed point $\Y_A \in
\D_{\emptyar \vdash (A \to A) \to A}$, \emph{i.e.} such that $\Y_A =
F(\Y_A)$.  

\subsubsection{Interpretation of
$\mathbf{let}$}\label{subsubsec:intr_plet}
We give the interpretation of parallel $\mathbf{let}$, as
\[
\intr{\Gamma \vdash \plet{x_1}{N_1}{x_2}{N_2}{M} : \ty}
= \mathsf{plet}_{\tx,\ty} \odot_\oc \tuple{\intr{N_1}, \intr{N_2},
\Lambda^\oc_{\tx\with \tx}(\intr{M})}\,.
\]
where $\mathsf{plet}_{\tx,\ty} \in \CGwb_\oc(\tx \with \tx \with
((\tx\with \tx) \to \ty), \ty)$ first
evaluates its two arguments $\tx$ in parallel. Once they both terminate
on $v$ and $w$, it calls its function argument, with $\tuple{v, w}$.

More formally, we first define a prestrategy ``forcing'' evaluation to
$v, w$, \emph{i.e.}
\[
\force_{v, w} : \oc (\tx \with \tx \with ((\tx\with \tx) \to
\ty)) \vdash \ty
\]
as in Figure \ref{fig:force} -- only a prestrategy, not
receptive to other values.
\begin{figure}
\begin{minipage}{.45\linewidth}
\[
\xymatrix@R=2pt@C=-3pt{
\oc ((\tx &\with& \tx) &\with &(( \tx & \with& \tx) &\to& \ty))
&\vdash& \ty\\
&&&&&&&&&&\qu^-
        \ar@{-|>}[ddllllllllll]
        \ar@{-|>}[ddllllllll]\\~\\
\qu^+_{\grey{0}} 
        \ar@{-|>}[dd]&&
\qu^+_{\grey{1}} 
        \ar@{-|>}[dd]\\~\\
{v^-_\grey{0}}
        \ar@{.}@/^.1pc/[uu]
&&{w^-_\grey{1}}
        \ar@{.}@/^.1pc/[uu]\\~\\~\\~\\~\\~\\
}
\]
\caption{The prestrategy $\force_{v, w}$}
\label{fig:force}
\end{minipage}
\hfill
\begin{minipage}{.45\linewidth}
\[
\xymatrix@R=2pt@C=-3pt{
\oc ((\tx &\with& \tx) &\with &(( \tx & \with& \tx) &\to& \ty))
&\vdash& \ty\\
&&&&&&&&&&\qu^-
        \ar@{-|>}[dllllllllll]
        \ar@{-|>}[dllllllll]\\
\qu^+_{\grey{0}}
        \ar@{-|>}[d]&&
\qu^+_{\grey{1}}
        \ar@{-|>}[d]\\
v^-_{\grey{0}}
        \ar@{.}@/^.1pc/[u]
        \ar@{-|>}[drrrrrrrr]&&
w^-_{\grey{1}}
        \ar@{.}@/^.1pc/[u]
        \ar@{-|>}[drrrrrr]\\
&&&&&&&&\qu^+_{\grey{3}}
        \ar@{-|>}[dllll]
        \ar@{-|>}[dll]
        \ar@{-|>}[d]\\
&&&&\qu^-_{\grey{3, i}}
        \ar@{.}@/^/[urrrr]
        \ar@{-|>}[d]&&
\qu^-_{\grey{3,j}}
        \ar@{.}@/^/[urr]
        \ar@{-|>}[d]&&
u^-_{\grey{3}}
        \ar@{-|>}[drr]\\
&&&&
v^+_{\grey{3, i}}
        \ar@{.}@/^.1pc/[u]&&
w^+_{\grey{3, j}}
        \ar@{.}@/^.1pc/[u]&&&&
u^+     \ar@{.}@/_/[uuuuu]
}
\]
\caption{The strategy $\mathsf{plet}_{\tx, \ty}$}
\label{fig:plet}
\end{minipage}
\end{figure}
Likewise,
\[
\preval_{v, w} : \oc (\tx \with \tx \with ((\tx\with \tx) \to
\ty)) \vdash \ty
\]
is defined by first using the cartesian closed structure of
$\CGwb_\oc$ to obtain a strategy
\[
\lambda \tuple{x_1, x_2, f}.\,f\,\tuple{v, w} \in
\CGwb_\oc(\tx \with \tx \with ((\tx\with \tx) \to \ty), \ty)\,,
\]
evaluating $f$ on $\tuple{v, w}$, from which $\preval_{v, w}$ is
obtained simply by removing the initial $\qu^-$. Up to reindexing, we
assume $\preval_{v, w}$ does not use copy indices
$\grey{0}$ and $\grey{1}$ on the calls to context.

Finally, this lets us define $\mathsf{plet}$ via the expression
(the supremum refers to $\cleq$):
\[
\mathsf{plet}_{\tx, \ty} = \bigvee_{v, w \in \tx} 
\force_{v, w}
\cdot \preval_{v, w} \in \CGwb_\oc(\tx \with \tx \with
((\tx\with \tx)
\to \ty), \ty)
\]
where the concatenation $\cdot$ sets both maximal (negative) events of
$\force_{v, w}$ as dependencies for $\preval_{v, w}$. The
full strategy is represented in Figure \ref{fig:plet}. Unlike for Figure
\ref{fig:force}, the picture in Figure \ref{fig:plet} is to be read as a
symbolic representation. The concrete strategy $\mathsf{plet}_{\tx,
\ty}$ has patterns as in Figure \ref{fig:plet} for all concrete values
for $v, w, u$ and copy indices $\grey{i}$ and $\grey{j}$.

The interpretation of the sequential $\mathsf{let}$ can be
obtained as a simplification.

\subsubsection{Interpretation of interference} \label{subsubsec:intr_inter}
Now, we interpret shared state and semaphores. By and large, it closely
follows the sequential interpretation of Sections~\ref{subsubsec:seq_inter_mem} and \ref{subsubsec:newrs}.

For the primitives interacting with memory and semaphores, we use the
same definitions as in Section~\ref{subsubsec:seq_inter_mem}, with the
obvious causal strategies matching Figures \ref{fig:assign},
\ref{fig:deref}, \ref{fig:grab}, and \ref{fig:release}.

For new references, we regard the alternating strategy
$\cell_n$ from Section~\ref{subsubsec:newrs} as a (sequential) event
structure: its events $\ev{\bcell_n}$ are the non-empty plays, the
causal order is given by prefix, the conflicting pairs are all
non-comparable plays.  The display map
\[
\begin{array}{rcrcl}
\pr_{\bcell_n} &:& \ev{\bcell_n} &\to& \oc \var\\
&& sa &\mapsto& a 
\end{array}
\]
keeps the last move. As configurations are sets of
prefixes of a given play, we set $\tilde{\bcell_n}$ to comprise those
bijections induced by plays $s_1 \sym_A s_2$ symmetric on $A$ (see
Definition \ref{def:alt_play_equiv}).

\begin{figure}
\begin{minipage}{.45\linewidth}
\[
\xymatrix@C=7pt@R=8pt{
&\read^-_{\grey{8}}
        \ar@{~}[rrr]
        \ar@{-|>}[d]&&&
\mwrite{4}^-_{\grey{2}}
        \ar@{~}[rrr]
        \ar@{-|>}[d]&&&
\mwrite{6}^-_{\grey{2}}
        \ar@{-|>}[d]\\
&0^+_{\grey{8}}
        \ar@{.}@/^.1pc/[u]&&&
\ok^+_{\grey{2}}
        \ar@{.}@/^.1pc/[u]
        \ar@{-|>}[dl]
        \ar@{-|>}[d]
        \ar@{-|>}[dr]&&&
\ok^+_{\grey{2}}
        \ar@{.}@/^.1pc/[u]\\
&&&
\read^-_{\grey{8}}
        \ar@{~}[r]
        \ar@{-|>}[d]&
\mwrite{6}^-_{\grey{3}}
        \ar@{~}[r]
        \ar@{-|>}[d]&
\mwrite{7}^-_{\grey{4}}
        \ar@{-|>}[d]\\
&&&4^+_{\grey{8}}
        \ar@{.}@/^.1pc/[u]&
\ok^+_{\grey{3}}
        \ar@{.}@/^.1pc/[u]&
\ok^+_{\grey{4}}
        \ar@{.}@/^.1pc/[u]
}
\]
\caption{Beginning of $\bcell_0$}
\label{fig:caus_cell}
\end{minipage}
\hfill
\begin{minipage}{.45\linewidth}
\[
\xymatrix@C=1pt@R=9pt{
\mgrab^-_{\grey{0}}
        \ar@{~}[rrr]
        \ar@{-|>}[d]&&&
\mgrab^-_{\grey{1}}
        \ar@{-|>}[d]
\\
\ok^+_{\grey{0}}
        \ar@{.}@/^.1pc/[u]&&&
\ok^+_{\grey{1}}
        \ar@{.}@/^.1pc/[u]
        \ar@{-|>}[d]\\
&&&\mrelease^-_{\grey{2}}
        \ar@{-|>}[d]\\
&&&\ok^+_{\grey{2}}
        \ar@{.}@/^.1pc/[u]
}
\]
\caption{Beginning of $\bsem_0$}
\label{fig:caus_sem}
\end{minipage}
\end{figure}

We display in Figure \ref{fig:caus_cell} a few early moves of
$\bcell_0$, with the convention that all moves in the same row are in
pairwise conflict. In this diagram we observe that $\bcell_0$ fails
\emph{courtesy}: indeed, the immediate causal link $\done^+_{\grey{2}}
\imc \read^-_{\grey{8}}$, for instance, would be illegal for a strategy.

Note that although $\bcell_n : \oc \var$ is only a prestrategy and not a
strategy, we have:

\begin{prop}
For all $\bsigma : \oc \var \vdash A$, for all $n \in \mathbb{N}$,
$\bsigma \odot \bcell_n : A$ is a strategy.
\end{prop}
\begin{proof}
First, $\sim$-receptivity is established directly, exploiting that
$\bsigma$ and $\bcell_n$ do not have non-courteous immediate causal
links accross components -- they are \emph{componentwise courteous} in
the sense of \cite{cg2}. Lemma 3.36 from \cite{cg2} ensures that
$\bsigma \odot \bcell_n$ is $\sim$-receptive.

It remains to prove $\bsigma \odot \bcell_n$ receptive and
courteous. But those properties are independent of symmetry. Without
symmetry, strategies are exactly those prestrategies invariant (up to
iso) under composition with copycat \cite{cg1}. But then, we compute:
\begin{eqnarray*}
\cc_A \odot (\bsigma \odot \bcell_n)  &\iso& (\cc_A \odot \bsigma) \odot
\bcell_n\\
&\iso& \bsigma \odot \bcell_n\,,
\end{eqnarray*}
so by Theorem 3.20 of \cite{cg1}, $\bsigma \odot \bcell_n$ is receptive
and courteous as required.
\end{proof}

For semaphores, we obtain a prestrategy $\bsem_n : \oc \sem$ for each
$n \in \mathbb{N}$ using essentially the same recipe, taking as events
all non-empty plays that are prefixes of an even-length play. A
beginning of $\bsem_0$ is displayed in Figure \ref{fig:caus_sem}.  

The interpretation of new references and semaphores is defined
as in Section~\ref{subsubsec:newrs}.

\subsection{Adequacy} Adequacy could be
deduced from the connection with \cite{gm} -- see Section~\ref{subsubsec:collapse_gm}.
Instead we give an independent proof, as we believe it helps
build the operational intuitions for the model. Rather than as usual
relying on logical relations, we follow an alternative route, proving
first adequacy for certain finitary terms in which the correspondence
between operational and game semantics is more concrete.

\subsubsection{Canonical adequacy} The sharpest link between
operational and game semantics holds factoring out recursion,
higher-order, and dynamic generation of semaphores or references.

We temporarily extend $\IPA$ with an explicit $\bot_A : A$ for
every type $A$, interpreted by a minimal strategy with no Player move.
Take $\Sigma$ an \textbf{interference context}, \emph{i.e.} of the
form
\[
\Sigma = \ell_1 : \var, \dots,
\ell_i : \var, \ell_{i+1} : \sem, \dots, \ell_n : \sem\,,
\]
and $\Sigma \vdash M : \tx$ where $\tx \in \{\tunit, \tbool,
\tnat\}$. We say $M$ is in \textbf{canonical form} if it contains
no fixpoint, subterm of higher type, bad references or semaphores,
$\newref$ or $\newsem$, with no store-independent
reductions available (other than the \emph{interfering reductions} of
Figure \ref{fig:operational}).

To help us link operational and denotational semantics for canonical
forms, we introduce a few concepts.  First, if $s$ is a store
with $\dom(s) = \Sigma$ as above, we interpret it as
\[
\intr{s} = (\bigotimes_{1\leq k \leq i} \bcell_{s(\ell_k)}) \tensor
(\bigotimes_{i+1 \leq k \leq n} \bsem_{s(\ell_k)}) :
(\bigotimes_{1\leq k \leq i} \oc \var) \tensor 
(\bigotimes_{i+1\leq k \leq n} \oc \sem) 
\iso\,\oc \intr{\Sigma}\,.
\]

To track evolution of an interaction we use the notion of \emph{residuals}.

\begin{defi}
If $E$ is an ess and $x \in \conf{E}$, the \textbf{residual} $E/x$ has
\[
\begin{array}{rcl}
\ev{E/x} &=& \{e \in \ev{E}\setminus x \mid \forall e' \in x, \neg(e
\conflict_E e')\}\,,\\
e_1 \leq_{E/x} e_2 &\Leftrightarrow& e_1 \leq_E e_2\,,\\
e_1 \conflict_{E/x} e_2 &\Leftrightarrow& e_1 \conflict_E e_2\,,\\
\theta : y \sym_{E/x} z &\Leftrightarrow& \theta \cup \id_x : x \cup y
\sym_{E} x \cup z\,.
\end{array}
\]
\end{defi}

If $E$ has polarities they are preserved as well and for $A$ an
arena, $A/x$ is an arena. In particular for $\oc \intr{\Sigma}$, by
definition and playing Hilbert's hotel, for any $x \in \conf{\oc
\intr{\Sigma}}$ with as many Player as Opponent moves, we have $(\oc
\intr{\Sigma})/x \iso \oc \intr{\Sigma}$.  Residuals also apply to
causal (pre)strategies: if $\bsigma : A$, then for each $x \in
\conf{\bsigma}$, we have $\bsigma/x : A/\pr_\bsigma(x)$ a prestrategy. 

Take $\Sigma \vdash M :\tx$ in canonical form. Necessarily
$\intr{\tx}$ -- also written $\tx$ -- has a unique minimal move $\qu^-$
and $\intr{M}$ has a unique matching minimal move, also written $\qu^-$.
We set
\[
\rintr{M} = \intr{M}/\{\qu^-\} : \Sigma \vdash \rintr{\tx}\,,
\]
where $\rintr{\tx} = \intr{\tx}/\{\qu^-\}$,
yielding a causal prestrategy in the sense of Definition
\ref{def:caus_strat}.

Now, we are equipped to state the most central ingredient of our proof
of adequacy.

\begin{lem}\label{lem:main_can_ad}
Consider $\Sigma \vdash M : \tx$ in canonical form, and $s$ a store with
$\dom(s) = \Sigma$.

Then, there is a (necessarily interfering) one-step reduction
\[
\tuple{M, s} \leadsto \tuple{M', s'}
\]
iff there are matching $x \in \conf{\rintr{M}}$ and $y \in
\conf{\intr{s}}$ with \emph{two elements} each, such that 
\[
\rintr{M}/x \simstrat \rintr{M'} : \oc \intr{\Sigma} \vdash
\rintr{\tx}\,,
\qquad
\qquad
\intr{s}/y \simstrat \intr{s'} : \oc \intr{\Sigma}\,.
\]
\end{lem}
\begin{proof}
For interfering operations $\ell:=v, \oc \ell, \grab(\ell)$ or
$\release(\ell)$, it is a direct verification by definition of the
interpretation. The result then follows by induction on $M$.
\end{proof}

This identifies \emph{store operations} that a canonical $M$ may
perform immediately with store $s$, with the minimal events of 
$\rintr{M} \inter \intr{s}$ that occur in $\Sigma$. It
almost suffices to iterate this to obtain adequacy
for canonical terms; however interfering reductions might yield
non-canonical terms, so state operations must be interleaved with pure reductions.
Write $M \triangleright N$ for \emph{pure} reductions, \emph{i.e.} the
context closure of interference-independent reductions in
Figure \ref{fig:operational} -- these reductions leave
invariant the interpretation as causal strategies. Moreover: 

\begin{lem}\label{lem:canonical1}
Consider $\Sigma \vdash M : \tx$ without fixpoint, subterm of higher
type, bad references or semaphores, $\newref$ or $\newsem$. Then, there
exists $\Sigma \vdash N : \tx$ canonical with $M \triangleright^* N$.

Moreover, for any store $s$ with $\dom(s) =
\Sigma$, there is $\tuple{M, s} \leadsto^* \tuple{\tskip,
s'}$ iff there is $\tuple{N, s} \leadsto^* \tuple{\tskip,
s'}$, and the correspondence preserves the number of interfering
operations.
\end{lem}
\begin{proof}
A routine standardization argument.
\end{proof}

Using this, we can prove adequacy for canonical terms. Intuitively, a
sequence $\tuple{M, s} \leadsto^* \tuple{\tskip, s'}$ can be
reproduced semantically: each interfering reduction yields by Lemma
\ref{lem:main_can_ad} a pair of events in
$\rintr{M} \inter \intr{s}$, while state-free reductions leave the
interpretation unchanged. Reciprocally, a successful interaction in
$\rintr{M} \odot \intr{s}$ is a partially ordered set of memory
operations, which may be linearized by securedness; informing a
reduction sequence.

\begin{prop}\label{prop:adequacy_canonical}
Consider $\Sigma \vdash M : \tunit$ in canonical form, and $s$ a store
with $\dom(s) = \Sigma$.

Then, $\tuple{M, s} \leadsto^* \tuple{\tskip, s'}$ iff $\rintr{M} \odot
\intr{s}$ has a positive move.
\end{prop}
\begin{proof}
\emph{If.} Taking $x \odot y \in \rintr{M} \inter \intr{s}$, we
build the sequence by induction on the size of $x \inter y$.

If $x \inter y$ has exactly one event, it must match
$\done^+$ in $\rintr{\tunit}$. Since $M$ is canonical, by
a case inspection on the shape of $M$, the only case for which
$\rintr{M}$ has a minimal positive event in $\rintr{\tunit}$ is 
a value $\tskip$. Otherwise, since $x$ and $y$ are causally compatible,
there is a sequence 
\[
(\emptyset, \emptyset) \cov (x_1, y_1) \cov \dots \cov (x_n, y_n) = (x,
y)
\]
of one-step extensions (for the product inclusion order) of matching
$x_i \in \conf{\rintr{M}}$ and $y_i \in \conf{\intr{s}}$, obtained by
linearization from the acyclicity of causal compatibility of
$x$ and $y$. In particular, $x_1 = \{m_1^+\}$ and $y_1 =
\{n_1^-\}$ singleton sets. By Proposition \ref{prop:comp_pcov}, $y$ is
$+$-covered, so there is a (unique, by definition of $\intr{s}$) $n_1^-
\imc_{\intr{s}} n_2^+$, and by Lemma \ref{lem:aux_imp_court} there is a
unique matching
$m_1^+ \imc_{\rintr{M}} m_2^-$. Since $n_2$ only depends on $n_1$ and
$m_2$ only depends on $m_1$, \emph{w.l.o.g.} $x_2 = \{m_1,
m_2\}$ and $y_2 = \{n_1, n_2\}$. So, by Lemma \ref{lem:main_can_ad},
there is a one-step 
$\tuple{M, s} \leadsto \tuple{M', s'}$
s.t. $\rintr{M}/x_2 \simstrat \rintr{M'}$ and $\intr{s}/y_2
\simstrat \intr{s'}$. We may now use Lemma \ref{lem:canonical1} to
obtain $M' \triangleright^* M''$ with $M''$ canonical and $\rintr{M'}
\simstrat
\rintr{M''}$. Now, removing $x_2$ and $y_2$ to the sequence above yields
\[
(\emptyset, \emptyset) \cov (x'_3,y'_3) \cov \dots \cov (x'_n, y'_n) =
(x', y')
\]
with $x' = x \setminus \{m_1, m_2\}$ and $y' = y \setminus \{n_1,
n_2\}$, witnessing that matching $x' \in \conf{\rintr{M''}}$ and $y' \in
\conf{\intr{s'}}$ are causally compatible, so that $x' \inter y' \in
\conf{\rintr{M''} \inter \intr{s'}}$. But $x'$ and $y'$ are $+$-covered,
so $x' \odot y' \in \conf{\rintr{M''} \odot \intr{s'}}$, still with a
positive move in $\tunit$. So 
$\tuple{M'', s'} \leadsto^* \tuple{\tskip, s''}$
by induction hypothesis, so $\tuple{M', s'} \leadsto^* \tuple{\tskip,
s''}$ by Lemma \ref{lem:canonical1}, so $\tuple{M, s} \leadsto^*
\tuple{\tskip, s''}$.

\emph{Only if.} By induction on the number of interfering operations in
$\tuple{M, s} \leadsto^* \tuple{\tskip, s''}$.

If $M = \tskip$ it is immediate. Otherwise, consider $\tuple{M, s}
\leadsto \tuple{M', s'} \leadsto^* \tuple{\tskip, s''}$. By Lemma
\ref{lem:main_can_ad}, there are $x \in \conf{\rintr{M}}$ and $y \in
\conf{\intr{s}}$ matching with two elements each, with
\[
\rintr{M}/x \simstrat \rintr{M'}
\qquad
\qquad
\intr{s}/y \simstrat \intr{s'}\,.
\]

Now, to use the induction hypothesis we convert $M' \triangleright^*
M''$ to canonical form; by Lemma \ref{lem:canonical1} we have
$\tuple{M'', s} \leadsto^* \tuple{\tskip, s''}$ with the same number of
interfering operations; and $\rintr{M'} \simstrat \rintr{M''}$. So by
induction hypothesis, there are $+$-covered and causally compatible
\[
x' \in \conf{\rintr{M''}}\,,
\qquad
\qquad
y' \in \conf{\intr{s'}}\,,
\]
where $x'$ has an occurrence of the
positive event of $\tunit$; so up to renaming, $x' \in
\conf{\rintr{M}/x}$
and $y' \in \conf{\intr{s}/y}$. Therefore, by definition of residuals,
one may add back $x$ and $y$ to obtain
\[
x \cup x' \in \conf{\rintr{M}}\,,
\qquad
\qquad
y \cup y' \in \conf{\intr{s}}\,,
\]
$+$-covered, causally compatible with a positive move in $\tunit$, which
concludes the proof.
\end{proof}

\subsubsection{Finitary adequacy} We deduce recursion-free adequacy
from the canonical case.
\begin{figure}
\[
\begin{array}{rcll}
(\newref\,x\!\!:=\!n\,\tin\,M);\,N &\to&
\newref\,x\!\!:=\!n\,\tin\,M;\,N&(x \not \in \fv(N))\\
M;\,(\newref\,x\!\!:=\!n\,\tin\,N) &\to&
\newref\,x\!\!:=\!n\,\tin\,M;\,N&(x \not \in \fv(M))\\
\ite{(\newref\,x\!\!:=\!n\,\tin\,M)}{N_1}{N_2}&\to&
\newref\,x\!\!:=\!n\,\tin\,\ite{M}{N_1}{N_2}&
(x \not \in \fv(N_i))\\
\ite{N_1}{(\newref\,x\!\!:=\!n\,\tin\,M)}{N_2}&\to&
\newref\,x\!\!:=\!n\,\tin\,\ite{N_1}{M}{N_2}
&(x \not \in \fv(N_i))\\
\ite{N_1}{N_2}{(\newref\,x\!\!:=\!n\,\tin\,M)}&\to&\newref\,x\!\!:=\!n\,\tin\,\ite{N_1}{N_2}{M}
&(x \not \in \fv(N_i))\\
\tlet{y}{(\newref\,x\!\!:=\!n\,\tin\,M)}{N}&\to&
\newref\,x\!\!:=\!n\,\tin\,(\tlet{y}{M}{N})&
(x \not \in \fv(N))\\
\tlet{y}{M}{(\newref\,x\!\!:=\!n\,\tin\,N)}&\to&
\newref\,x\!\!:=\!n\,\tin\,(\tlet{y}{M}{N})&
(x \not \in \fv(M))\\
\mathsf{f}\,(\newref\,x\!\!:=\!n\,\tin\,M)&\to&
\newref\,x\!\!:=\!n\,\tin\,(\mathsf{f}\,M) 
\end{array}
\]
where $\mathsf{f} \in \{\tsucc, \pred, \iszero\}$.
\caption{Commutation rules for $\newref$}
\label{fig:congruence}
\end{figure}

To convert terms to canonical form, we perform state-free reductions
while pushing declarations of new references or semaphores outside. The
latter is done by the commutation rules of Figure \ref{fig:congruence},
from which are missing the three rules for the parallel let,
and matching commutations for new semaphores. It is direct
that these rules leave the game semantics invariant, and
preserve and reflect infinite reduction chains in the
operational semantics.

Writing $\equiv$ for the congruence closure of state-free reductions
and commutations above:

\begin{lem}\label{lem:finitary_to_canonical}
Consider $\vdash M : \tx$ a recursion-free term of $\IPA$.
Then, there exists
\[
\begin{array}{l}
M \equiv \new\,x_1\!\!:=\!n_1\,\tin \dots
\new\,x_p\!\!:=\!n_p\,\tin\,N
\end{array}
\]
where each $\new$ is either $\newref$ or $\newsem$, and $\Sigma \vdash N
: \tx$ is canonical.
\end{lem}
\begin{proof}
Consider the reduction $\to$ comprising the (context closure of)
the commutations above with $\beta$-reduction and the state-free
reductions for $\mkvar$ and $\mksem$. Treating $\var$ and $\sem$ as
product types, it is easy to prove from the strong normalization of the
simply-typed $\lambda$-calculus with products that $\to$ terminates.
Moreover, as $M$ has type $\tx \in \{\tunit, \tbool, \tnat\}$, a
$\to$-normal form $M \to^* M'$ has no abstraction, bad variable or
semaphore subterm. Thus
\[
M' = \new\,x_1\!\!:=\!n_1\,\tin \dots \new\,x_p\!\!:=\!n_p\,\tin\,N'
\]
with $\Sigma \vdash N' : \tx$ without recursion, subterm of higher type,
bad reference or semaphores, and reference and semaphore initialization.
Finally, we conclude by Lemma \ref{lem:canonical1}.
\end{proof}

The semantics enjoy \emph{finitary adequacy}:

\begin{prop}\label{prop:finitary_adequacy}
For any recursion-free $\vdash M : \tunit$, we have $M\eval$ iff
$\intr{M}\eval$.
\end{prop}
\begin{proof}
Immediate consequence of Proposition \ref{prop:adequacy_canonical} and
Lemma \ref{lem:finitary_to_canonical}.
\end{proof}

\subsubsection{Deducing adequacy} Finally, we extend the above with
recursion.

\begin{thm}[Adequacy]\label{th:adequacy_ipa}
For $\vdash M : \tunit$ any term of $\IPA$, $M \eval$ iff
$\intr{M}\eval$.
\end{thm}
\begin{proof}
As expected, we simply reason by continuity.
For all type $A$ and $n \in \mathbb{N}$, we set
\[
\Y_A^0 = \lambda f^{A \to A}.\, \bot_A 
\qquad
\qquad
\Y_A^{n+1} = \lambda f^{A \to A}.\, f\,(\Y_A^n\,f)\,,
\]
yielding $\vdash \Y_A^n : (A\to A) \to A$. The \textbf{$n$-th
approximation} $\Gamma \vdash M_n : A$ of any term $\Gamma \vdash M : A$
of $\IPA$ is obtained by replacing each
$\Y\,N$ with $\Y_n\,N$. It is then routine to show that $M\eval$ iff
there is $n \in \mathbb{N}$ such that $M_n \eval$. Likewise, by
definition of the interpretation of recursion (see Section~\ref{subsubsec:core_pcf}) and continuity of the interpretation,
$\intr{M} = \bigvee_{n \in \mathbb{N}} \intr{M_n}$.
Now:
\begin{eqnarray*}
M\eval &\Leftrightarrow& \exists n \in \mathbb{N},~M_n \eval\,,\\
&\Leftrightarrow& \exists n \in \mathbb{N},~\text{$\intr{M_n}$ has a
positive move,}\\
&\Leftrightarrow& \text{$\intr{M}$ has a positive move,}
\end{eqnarray*}
using the above along with Proposition \ref{prop:finitary_adequacy}.
\end{proof}

This continuity argument only works for \emph{may}-convergence: for
extensions such as \emph{must} or \emph{fair}-convergence we would have
to formulate a more complete correspondence between operational and game
semantics -- 
see \emph{e.g.} \cite{DBLP:phd/hal/Castellan17} for an adequacy result
for non-deterministic $\PCF$ \emph{w.r.t.} must-convergence. However,
we leave this out of this paper.

\subsection{Full Abstraction} Now that we have established $\CG$ as an
adequate model of $\IPA$, we explore a bit further its
observable operational content and prove full abstraction.

\subsubsection{Observable behaviour of causal strategies}
\label{subsubsec:obs_caus}
It is clear that without quotient, the interpretation of $\IPA$ in $\CG$
will \emph{not} be fully abstract: the model records very
intensional information that is typically not observable.
\begin{figure}
\begin{minipage}{.5\linewidth}
\[
\xymatrix@R=5pt@C=0pt{
&\tunit&&\to&&\tunit&&\to&&\tunit\\
&&&&&&&&&\qu^-
        \ar@[grey]@{-|>}[dlllllllll]
        \ar@{-|>}[dlllllll]\\
\grey{\qu^+}
        \ar@[grey]@{.}@/^.1pc/[urrrrrrrrr]
        \ar@[grey]@{~}[rr]
        \ar@[grey]@{-|>}[d]&&
\qu^+   \ar@{.}@/^.1pc/[urrrrrrr]
        \ar@{-|>}[d]\\
\grey{\done^-}
        \ar@[grey]@{.}@/^.1pc/[u]&&
\done^- \ar@{.}@/^.1pc/[u]
        \ar@{-|>}[drrr]\\
&&&&&\qu^+
        \ar@{.}@/^/[uuurrrr]
        \ar@{-|>}[d]\\
&&&&&\done^-
        \ar@{.}@/^.1pc/[u]
        \ar@{-|>}[drrrr]\\
&&&&&&&&&\done^+
        \ar@{.}@/_1pc/[uuuuu]
}
\]
\caption{
$\begin{array}{l}
\lambda x^\tunit y^\tunit.\,\newref\,r\!\!:=\!0\,\tin\\
\,\,\,\,\,\,\,\,\,\assrt\,(\iszero\,\oc r);\,x;\,r\!\!:=\!\ttrue\\
\,\,\,\,\parallel \assrt\,(\tnot\,(\iszero\,!r));\,y
\end{array}$
}
\label{fig:ex_must}
\end{minipage}
\hfill
\begin{minipage}{.45\linewidth}
\[
\xymatrix@R=13pt@C=-3pt{
(\tunit&&\to&&\tunit&&\to&&\tunit)&&\to&&\tunit\\
&&&&&&&&&&&&\qu^-
        \ar@{-|>}[dllll]\\
&&&&&&&&\qu^+
        \ar@{.}@/^.1pc/[urrrr]
        \ar@{-|>}[dllllllll]
        \ar@{-|>}[dllll]
        \ar@{-|>}[d]\\
\qu^-   \ar@{.}@/^.1pc/[urrrrrrrr]
        \ar@{-|>}[d]
        \ar@[grey]@{-|>}[drrr]&&&&
\qu^-   \ar@{.}@/^.1pc/[urrrr]
        \ar@[grey]@{-|>}[dl]
        \ar@{-|>}[dr]&&&&
\done^- \ar@{.}@/_.1pc/[u]
        \ar@{-|>}[drrrr]\\
\done^+ \ar@{.}@/^.1pc/[u]&&&
\grey{\done^+}
        \ar@[grey]@{~}[rr]
        \ar@[grey]@{.}@/^.1pc/[ur]&&
\done^+ \ar@{.}@/_.1pc/[ul]&&&&&&&
\done^+ \ar@{.}@/_/[uuu]
}
\]
\caption{
$\begin{array}{l}
\lambda f^{\tunit \to \tunit\to \tunit}.\, \newref\,r\!\!:=\!0\,\tin\\
\,\,\,\,f\,(r:=1)\\
\,\,\,\,\,\,\,\,\,(\tskip \ovee (\assrt\,!r))
\end{array}$
}
\label{fig:ex_caus}
\end{minipage}
\end{figure}
We give two examples in Figures \ref{fig:ex_must} and \ref{fig:ex_caus}.

We introduce a few additional pieces of syntactic
sugar.  First, $\assrt : \tbool \to \tunit$ is $\lambda
x^\tbool.\,\ite{x}\,\tskip\,\bot$ which terminates on $\ttrue$
and diverges otherwise. We also define $\tnot : \tbool \to
\tbool$ as the obvious program. Finally, for any $\Gamma \vdash M
: A$ and $\Gamma \vdash N : A$ of $\IPA$, we set
\[
\begin{array}{l}
M \ovee N = \newref\,x:=0,\,y:=0\,\mathbf{in}\,\\
\hspace{80pt}(x:=1 \parallel y:=!x);\\
\hspace{80pt}\ite{(\iszero\,(!y))}{M}{N}
\end{array}
\]
a non-deterministic sum $\Gamma \vdash M \ovee N : A$,
behaving non-deterministically as $M$ or $N$.

First, Figure \ref{fig:ex_must} represents the semantics of
an encoding of sequential composition via parallel composition
plus shared state.  In Ghica and Murawski's model, this program has the
same interpretation as sequential composition, showing their equivalence
with respect to may-testing. In contrast, $\CG$ also gives account
of the limitation of this encoding: the greyed out branch on the left
corresponds to the bottom read winning the race, causing
divergence\footnote{This shows our model remembers
\emph{some} divergences, though not enough to get adequacy
for \emph{must}: some divergences are lost through hiding. This
can be addressed by tweaking hiding to retain those events dubbed
\emph{essential} that carry divergences, see
\cite{DBLP:conf/fossacs/CastellanCHW18,DBLP:phd/hal/Castellan17} -- but
we shall not take this route in this paper.}.
Likewise, Figure \ref{fig:ex_caus} also has an unobservable branch
greyed out. The model shows that the program may provide values for the
two
arguments of $f$ independently, but it may also provide a value for the
second argument of $f$ \emph{because} $f$ called its first argument. In
a play, an occurrence of $\done^+$ corresponding to a
value for the second argument may be causally explained by either of the
two moves, but the distinction is un-observable.

For full abstraction only the observable behaviour matters, and
$\CG$ clearly records more than necessary. So we ask: what parts of a
concurrent strategy are observable? 

\subsubsection{Non-alternating plays with
pointers}\label{subsubsec:collapse_gm}
We approach this
question in terms which are no surprise to the reader familiar with
Ghica and Murawski's non-alternating games: the observable behaviour is
exactly captured by certain \emph{non-alternating plays with
pointers}:

\begin{defi}
A \textbf{non-alternating play with pointers} on $A$ is a $s_1
\dots s_n$ on $\ev{A}$ s.t.
\[
\begin{array}{ll}
\text{\emph{negative:}} & n\geq 1 \implies \pol(s_1) = -\,,
\end{array}
\]
with, for all $1\leq j \leq n$ s.t. $s_j$ is non-minimal in $A$,
a \textbf{pointer} to some earlier $s_i$ such that $s_i \imc_A s_j$.
We write $\PNAlt(A)$ for the set of non-alternating plays with pointers
on $A$.
\end{defi}

Ghica and Murawski's model is an analogue of $\NNegAlt$ based
on non-alternating plays with pointers. A \textbf{GM-strategy} on
meager $A$ is a non-empty set of well-bracketed non-alternating
plays with pointers (with well-bracketing defined as in Definition
\ref{def:nalt_play_wb}) satisfying conditions analogous to Definition
\ref{def:nalt_strat} (see Definitions 4 and 13 in \cite{gm}) and
additionally \emph{thread-independent} (see Definition 17 of
\cite{gm}). There is a cartesian closed category $\GM$ of meager arenas
and GM-strategies, supporting the interpretation of $\IPA$.

We shall build a functorial bridge between $\CG$ and $\GM$, as in
Proposition \ref{prop:nalt_fonc}, but restricted to the cartesian closed
structure as $\GM$ has no linear decomposition. We use the
\emph{concrete arenas} of Section~\ref{def:concrete-arena}, extended in
the obvious way with Question/Answer labeling.

The proof of Proposition \ref{prop:pointer_plays} applies unchanged to
prove:

\begin{prop}\label{prop:nalt_pointer_plays}
For any concrete arena $A$, there is an injective function
\[
\mathscr{P} : \NAlt(A)/\!\!\sym 
\quad
\to 
\quad
\PNAlt(A^0)
\]
preserving length, prefix, justifiers, and reflecting and preserving
well-bracketing.
\end{prop}

We consider the cartesian closed category $\CG_{\mathsf{c}}$ with
objects concrete arenas, obtained from $\CGwb_\oc$ by replacing all
operations on arenas with those on concrete arenas.

\begin{prop}\label{prop:cg_to_gm}
Consider concrete $-$-arenas $A, B$, and $\bsigma \in \CGwb_\oc(A, B)$.
Then,
\[
\PNStrat(\bsigma) = 
\{\mathscr{P}(s) \mid \text{$s \in \NAltStrat(\bsigma^\dagger)$
well-bracketed}\}
\]
is a GM-strategy on $A^0 \lin B^0$. Moreover, $\PNStrat$ extends to a
cartesian closed functor
\[
\PNStrat : \CG_{\mathsf{c}} \to \GM
\]
preserving the interpretation of $\IPA$.
\end{prop}
\begin{proof}
We detail the two critical points: preservation of symmetry, and
composition.

\emph{Symmetry.} Consider $\bsigma, \btau \in \CGwb_\oc(A, B)$ s.t.
$\bsigma \simstrat \btau$. Then we also have
$\Lambda(\bsigma^\dagger) \simstrat \Lambda(\btau^\dagger)$, \emph{i.e.}
there is an iso $\varphi : \bsigma^\dagger \iso
\btau^\dagger$ s.t. $\pr_{\Lambda(\btau^\dagger)} \circ
\varphi \sim^+ \pr_{\Lambda(\bsigma^\dagger)}$. Consider $\mathscr{P}(s)
\in \PNStrat(\bsigma)$ for some $s \in
\NAltStrat(\Lambda(\bsigma^\dagger))$. This means there is $t \in
\NAlt(\bsigma^\dagger)$ s.t. $s =
\pr_{\Lambda(\bsigma^\dagger)}(t)$. But then, $\varphi(t) \in
\NAlt(\btau^\dagger)$, and from $\pr_{\Lambda(\btau^\dagger)} \circ
\varphi \sim^+ \pr_{\Lambda(\bsigma^\dagger)}$ it is direct that $s =
\pr_{\Lambda(\bsigma^\dagger)}(t) \sym_{A\lin B}
\pr_{\Lambda(\btau^\dagger)}(\varphi(t))$. By Proposition
\ref{prop:nalt_pointer_plays},
$\mathscr{P}(\pr_{\Lambda(\bsigma^\dagger)}(t)) =
\mathscr{P}(\pr_{\Lambda(\btau^\dagger)}(\varphi(t)))$, so that
$\mathscr{P}(s) \in \PNStrat(\btau)$. The other inclusion is
symmetric, so $\PNStrat(\bsigma) = \PNStrat(\btau)$.

\emph{Composition.}
For $\bsigma \in \CGwb_\oc(A, B), \btau \in \CGwb_\oc(B, C)$, we must
show
\[
\PNStrat(\btau \odot_\oc \bsigma) = \PNStrat(\btau) \odot
\PNStrat(\bsigma)\,,
\]
there are two inclusions to prove:

$\subseteq$. Consider $\mathscr{P}(s) \in \PNStrat(\btau \odot_\oc
\bsigma)$. Since $\PNStrat(-)$ preserves positive isomorphism and by the
laws of Seely categories, $\mathscr{P}(s) \in
\PNStrat(\btau^\dagger \odot \bsigma^\dagger)$ for some $s \in
\NAltStrat(\btau^\dagger \odot \bsigma^\dagger)$. By Proposition
\ref{prop:nalt_fonc}, the latter is $\NAltStrat(\btau^\dagger)
\odot \NAltStrat(\bsigma^\dagger)$.
Now, if $u \in \NAltStrat(\btau^\dagger)
\inter \NAltStrat(\bsigma^\dagger)$ is a witness for $s \in
\NAltStrat(\btau^\dagger) \odot \NAltStrat(\bsigma^\dagger)$, then as
$\bsigma^\dagger$ and $\btau^\dagger$ are well-bracketed, it is
direct that $u\restrict A, B \in
\NAltStrat(\bsigma^\dagger)$ and $u\restrict B, C \in
\NAltStrat(\btau^\dagger)$ are well-bracketed as well. Therefore,
$\mathscr{P}(u\restrict A, B) \in \NAltStrat(\bsigma^\dagger)$ and
$\mathscr{P}(u\restrict B, C) \in \NAltStrat(\btau^\dagger)$, hence
$\mathscr{P}(u)$ is a witness in the sense of \cite{gm} for
$\mathscr{P}(s) \in \PNStrat(\btau) \odot \PNStrat(\bsigma)$.

$\supseteq$. Consider now $s \in \PNStrat(\btau) \odot
\PNStrat(\bsigma)$. There is a witness $u$, a sequence with pointers on
$(A' \lin B') \lin C'$, with restrictions $u \restrict A', B' \in
\PNStrat(\bsigma)$ and $u \restrict B', C' \in \PNStrat(\btau)$ -- we
refer to \cite{gm} for the definitions of restrictions of plays with
pointers. Write $u \restrict A', B' = \mathscr{P}(t^\bsigma)$ and $u
\restrict B', C' = \mathscr{P}(t^\btau)$ for $t^\bsigma \in
\NAltStrat(\bsigma^\dagger)$ and $t^\btau \in
\NAltStrat(\btau^\dagger)$. Though $t^\bsigma$ and $t^\btau$ yield plays
with pointers compatible in $B'$, they might not match in $B$ on the
nose. But by Proposition \ref{prop:nalt_pointer_plays}, $t^\bsigma
\restrict B \sym_{\oc B} t^\btau \restrict B$ match \emph{up to
symmetry}. So, writing 
\[
t^\bsigma = \pr_{\Lambda(\bsigma^\dagger)}(v^\bsigma)
\qquad
\qquad
t^\btau = \pr_{\Lambda(\btau^\dagger)}(v^\btau)\,,
\]
for $v^\bsigma \in \NAlt(\bsigma^\dagger)$ and $v^\btau \in
\NAlt(\btau^\dagger)$, writing $x^\bsigma = \ev{v^\bsigma} \in
\conf{\bsigma^\dagger}$ and $x^\btau = \ev{v^\btau} \in
\conf{\btau^\dagger}$ the symmetry $t^\bsigma \restrict B \sym_{\oc B}
t^\btau \restrict B$ induces $\theta : x^\bsigma_B \sym_{\oc B}
x^\btau_B$. Moreover,
\[
x^\bsigma \parallel x^\btau_C 
\quad
\stackrel{\pr_{\bsigma^\dagger}\parallel \oc C}{\simeq}
\quad
x^\bsigma_A \parallel x^\bsigma_B \parallel x^\btau_C 
\quad
\stackrel{\oc A \parallel \theta \parallel \oc C}{\sym}
\quad
x^\bsigma_A \parallel x^\btau_B \parallel x^\btau_C
\quad
\stackrel{\oc A \parallel \pr_{\btau}^{-1}}{\simeq}
\quad
x^\bsigma_A \parallel x^\btau\,,
\]
is secured, as $u$ directly informs a total ordering of its graph
compatible with $\bsigma^\dagger$ and $\btau^\dagger$.
Hence, by Proposition \ref{prop:sync_sym}, there are $\varphi^\bsigma :
x^\bsigma \sym_{\bsigma^\dagger} y^\bsigma$ and $\varphi^\btau : x^\btau
\sym_{\btau^\dagger} y^\btau$ such that $y^\bsigma_B = y^\btau_B$.

Transporting $t^\bsigma$ and $t^\btau$ through $\varphi^\bsigma$ and
$\varphi^\btau$, we get $w^\bsigma \in \NAlt(\bsigma^\dagger)$ and
$w^\btau \in \NAlt(\btau^\dagger)$ s.t. we still have
$\mathscr{P}(w^\bsigma) = u \restrict A', B'$ and $\mathscr{P}(w^\btau)
= u \restrict B', C'$; but this time $w^\bsigma \restrict
B = w^\btau \restrict B$. Zipping them following $u$ we obtain $w \in
\NAlt(\btau^\dagger) \inter \NAlt(\bsigma^\dagger)$ such that
$\mathscr{P}(w\restrict A, C) = s$. But then $w \restrict A, C \in
\NAlt(\btau^\dagger) \odot \NAlt(\bsigma^\dagger)$ by construction, so
in $\NAlt((\btau \odot_\oc \bsigma)^\dagger)$ by Proposition
\ref{prop:nalt_fonc}, hence $s \in \PNStrat(\btau \odot_\oc \bsigma)$ as
required.
\end{proof}

So there is a functorial unfolding from $\CG$ to $\GM$. To further
factor out non-observable behaviour, one can restrict to \emph{complete} plays: 

\begin{defi}
Consider $A$ a meager $-$-arena, and $s \in \PNAlt(A)$.

We say that $s$ is \textbf{complete} iff it is well-bracketed and all
its questions have an answer.
\end{defi}

If $\bsigma \in \CGwb_\oc(A, B)$, $\comp(\bsigma)$ is the
subset of $\PNStrat(\bsigma)$ comprising complete plays only.
The proposition above allows us to deduce (with $\obs$ defined in
Section~\ref{subsubsec:intensional_fa}):

\begin{prop}\label{prop:obs_to_comp}
Consider $A$ concrete and $\bsigma, \btau : A$ well-bracketed causal
strategies. Then,
\[
\comp(\bsigma) = \comp(\btau)
\quad
\implies
\quad
\bsigma \obs \btau\,.
\]
\end{prop}
\begin{proof}
Consider $\balpha \in \CGwb_\oc(A, \tunit)$ s.t.
$\balpha \odot_\oc \bsigma \eval$. By Proposition \ref{prop:cg_to_gm},
$\qu^- \done^+ \in \PNStrat(\balpha) \odot
\PNStrat(\bsigma)$. Considering $u \in \PNStrat(\balpha)
\inter \PNStrat(\bsigma)$, $u \restrict A^0, \tunit \in
\PNStrat(\bsigma)$. From well-bracketing of $\bsigma$ and $\balpha$,
$u \restrict A^0$ is well-bracketed and complete, so
$u \restrict A^0 \in \comp(\btau)$, so $\PNStrat(\balpha)
\odot \PNStrat(\btau) \eval$, so $\balpha \odot_\oc \btau \eval$ by
Proposition \ref{prop:cg_to_gm}.
\end{proof}

To prove the converse and link it to syntactic 
equivalence, we examine \emph{definability}.

\subsubsection{Definability of plays with pointers} As in Sections~\ref{subsubsec:intensional_fa} and \ref{subsubsec:fa_ia}, full
abstraction relies on \emph{definability}.  While definability in $\IA$
rests on definability for finite innocent strategies, Ghica and
Murawski's definability for $\IPA$ gives directly terms realizing
\emph{individual plays}. For conciseness we omit the full development,
but illustrate it on a representative example. 

\begin{figure}
\begin{minipage}{.45\linewidth}
\[
\xymatrix@R=-5pt@C=0pt{
\tunit_1 &\to& (\tunit_2&\to&\tunit_3) &\to& \tunit_4&\to&\tunit_5\\
&&&&&&&&\qu_5^-\\
\qu_1^+
        \ar@{.}@/^/[urrrrrrrr]\\
&&&&\qu_3^+
        \ar@{.}@/^/[uurrrr]\\
\done_1^-
        \ar@{.}@/^/[uu]\\
&&\qu_2^-\ar@{.}@/^/[uurr]\\
&&&&&&\qu_4^+
        \ar@{.}@/^/[uuuuurr]\\
&&&&&&\done_4^-\\
&&\done_2^+
        \ar@{.}@/^/[uuu]
}
\]
\end{minipage}
\hfill
\begin{minipage}{.45\linewidth}
\[
\xymatrix@R=10pt{
&\qu_5^-        \ar@{.}[dl]
        \ar@{.}[d]
        \ar@{.}[dr]\\
\qu_1^+ \ar@{.}[d]&
\qu_3^+ \ar@{.}[d]&
\qu_4^+ \ar@{.}[d]\\
\done_1^-
        \ar@{-|>}[dr]
        \ar@{-|>}[urr]&
\qu_2^- \ar@{.}[d]
        \ar@{-|>}[ur]&
\done_4^-
        \ar@{-|>}[dl]\\
&\done_2^+
}
\]
\end{minipage}
\caption{Example of definability of plays in $\IPA$}
\label{fig:ex_def_ipa}
\end{figure}

Consider $s$ on the left hand side of Figure \ref{fig:ex_def_ipa}.
Naively, we want a term whose only execution is $s$. But
strategies satisfy \emph{courtesy}, so one realizing $s$ must
also realize all plays obtainable by adding asynchronous delays. The
information from $s$ that survives asynchronous delays is that certain
positive moves appear \emph{after} earlier negative moves in $s$.
Together with the static causality from the arena, this yields
a diagram as in the right hand side of Figure
\ref{fig:ex_def_ipa}, very much like a causal strategy. 
It is directly this diagram
that Ghica and Murawski's finite definability
reproduces. First, we ignore $\imc$-links and start with a
term
\[
\lambda x^{\tunit_1} f^{\tunit_2 \to \tunit_3} y^{\tunit_4}.\,(x
\parallel
f\,\tskip \parallel y);\,\bot\,,
\]
that performs all computational events available in $s$ in a maximally
parallel fashion, with only causal dependency enforced by the game.
The $\imc$-links are restored through the memory. For
that we define two helper functions. If $M : \var$, we write
$\mathbf{set}(M) : \tunit$ for $M:=1$, and $\mathbf{test}(M) : \tunit$
for $\assrt(\tnot (\iszero\,!M))$ which converges iff $!M$ is
non-zero. Then:
\[
\lambda x^{\tunit_1} f^{\tunit_2 \to \tunit_3} y^{\tunit_4}.\,
\left(
\left(
\begin{array}{l}
x;\\
\mathbf{set}(\done_1^-)
\end{array}
\right)
\parallel
\left(
\begin{array}{l}
f\,(\mathbf{set}(\qu_2^-);\\
\hspace{12pt} \mathbf{test}(\done_4^-);\\
\hspace{12pt} \mathbf{grab}(\done_2^+);\\
\hspace{12pt} \tskip)\\
\end{array}
\right)
\parallel
\left(
\begin{array}{l}
\mathbf{test}(\done_1^-);\\
\mathbf{test}(\qu_2^-);\\
y;\\
\mathbf{set}(\done_4^-)
\end{array}
\right)
\right);\,\bot
\]
borrows the shape of the first term, signaling the $\imc$-links through
memory. We use one fresh reference (initialized to $0$)
for each Opponent move, which gets set to $1$ when the Opponent move
occurs.  Finally, we use semaphores to ensure that Opponent replications
does not cause a duplication of Player moves by prompting re-evaluation
of the corresponding subterms -- so that we only obtain linearizations
of the diagram on the right hand side of Figure \ref{fig:ex_def_ipa}.

Done systematically for arbitrary plays, this establishes \cite{gm}:

\begin{prop}[Ghica, Murawski]\label{prop:gm_def}
Consider $A$ a type and $s \in \PNAlt(\aintr{A})$ complete.

Then there is $\vdash M_s : A \to \tunit$ such that for all $\vdash N :
A$,
\[
s \in \intr{N}_{\GM} 
\qquad
\Leftrightarrow
\qquad
M_s N \eval
\]
\end{prop}

In particular, as $\intr{N}_{\GM}$ is courteous, any $t \in
\intr{N}_{\GM}$ tracing a successful interaction with $M_s$ can be
converted to $s$ through permutations whose correctness is granted
by courtesy.

\begin{thm}\label{th:fa_ipa}
The model $\CGwb_\oc$ is intensionally fully abstract for $\IPA$.
\end{thm}
\begin{proof}
Consider $\vdash M, N : A$.
If $\intr{M} \obs \intr{N}$, using Theorem \ref{th:adequacy_ipa},
$M \obs N$.
Reciprocally, assume $M \obs N$. Seeking a contradiction, assume
$\intr{M} \not \obs \intr{N}$. By Proposition
\ref{prop:obs_to_comp}, there is \emph{w.l.o.g.} $s \in
\comp(\intr{M})$ where $s \not \in \comp(\intr{N})$. So, by
Proposition \ref{prop:cg_to_gm}, $s \in \intr{M}_{\GM}$ while $s \not
\in \intr{N}_{\GM}$. Finally, by Proposition \ref{prop:gm_def}, we have
$M_s\,M \eval$ while $M_s\,N \div$, contradiction.
\end{proof}

As for Theorem \ref{th:fa_ia}, the resulting quotient is effective and
easily described: for $\bsigma, \btau : A$,
\[
\bsigma \obs \btau \qquad \Leftrightarrow \qquad
\comp(\bsigma) = \comp(\btau)\,.
\]

Observational equivalence is undecidable even for the second-order
fragment and without recursion \cite{DBLP:journals/tcs/GhicaMO06}.
Note also that without semaphores, full abstraction fails
\cite{DBLP:journals/entcs/Murawski10}
as terms are closed under a stuttering behaviour which reduces their
observational power.  

\subsection{Parallel Innocence and Sequentiality}
We resume the discussion left at Section~\ref{subsec:parallelism}: can we find
\emph{parallel innocence} and \emph{sequentiality} disentangling
parallelism and interference?

A subtlety is that while $\IPA$ is non-deterministic, $\PCFpar$ and
$\IA$ are both deterministic -- albeit in different senses.
Non-determinism arises in $\IPA$ from the interaction of parallelism and
interference so, removing either of these causes, determinism has to be
reimposed as well. Accordingly, sequentiality and parallel innocence
will include determinism.

\section{Parallel Innocence}
\label{sec:par_inn}

We capture the causal patterns definable with pure parallel higher-order
programming.

\subsection{Causal determinism} The sense in which $\PCFpar$ is
deterministic is fairly different from Definition
\ref{def:alt_strat}. For instance, after the first Opponent move, the
strategy of Figure \ref{fig:plet} has two available Player moves; but
the order in which these moves are
played does not matter and will eventually reach the same result: the
program is \emph{deterministic up to the choice of the scheduler}.
If $E$ is an event structure, write $\Con_E$ for the set of finite
\textbf{consistent} sets of events, \emph{i.e.} for $X \subseteq_f
\ev{E}$, $X \in \Con_E$ iff for all $e_1, e_2 \in X$, we have $\neg (e_1
\conflict_\bsigma e_2)$.

We use Winskel's definition of \emph{determinism} for concurrent
strategies \cite{DBLP:journals/fac/Winskel12}:

\begin{defi}
A causal strategy $\bsigma : A$ on arena $A$ is \textbf{causally deterministic} if:
\[
\begin{array}{ll}
\text{\emph{causal determinism:}}&
\text{assume $X \subseteq_f \ev{\bsigma}$ is \emph{negatively
compatible},}\\
&\text{\emph{i.e.} $X_- = \{s \in X \mid \pol_\bsigma(s) =
-\} \in \Con_\bsigma$. Then, $X \in \Con_\bsigma$.}
\end{array}
\]
\end{defi}

This ensures that Player branching only spawns parallel threads: only
Opponent may initiate conflict. Copycat is deterministic, and
deterministic strategies compose \cite{DBLP:journals/fac/Winskel12}. All
other constructions in the Seely category structure of $\CG$ preserve determinism. 

\subsection{Parallel Innocence} What causal shapes are distinctive of
pure parallel computation?

\subsubsection{Pre-innocence} Pure parallel programs may spawn
parallel threads, which must remain independent in the
absence of interference. Once they both terminate the program
may take new actions that depend on their results, causally
``merging'' them. A typical causal strategy featuring this behaviour,
for $x: \tunit, y : \tunit \vdash x \parallel y : \tunit$,
appears in Figure \ref{fig:par_comp}. The slogan is: 

\begin{center}
\emph{``Player may merge threads that he himself has spawned''}.
\end{center}

In contrast, both diagrams of Figure \ref{fig:ex_aug2} bear signs of
interference.  In the first, the answer $1^+$ depends on $\qu^-$:  the
program somehow observes if the function has called its argument,
which is only possible if the argument performs some side-effect that the
program observes. In the second, $\done^+$ depends on $\done^-$;
but likewise this can only occur if the termination of the function
triggers a side-effect. In both cases, this is witnessed by Player
``merging'' causal chains which forked at Opponent moves. To ban such
interference, the slogan is:
\begin{center}
\emph{``Player may not merge threads spawned by Opponent''}.
\end{center}

To define parallel innocence, our first step is to introduce a formal
notion of ``thread'':

\begin{defi}
Consider $A$ an arena, and $\bsigma : A$ a causal strategy.

A \textbf{grounded causal chain (gcc)} in $\bsigma$ is
$\rho = \{\rho_1, \dots, \rho_n\} \subseteq \ev{\bsigma}$ forming
\[
\rho_1 \imc_\bsigma \dots \imc_\bsigma \rho_n
\]
a chain with $\rho_1$ minimal with respect to $\leq_\bsigma$.
We write $\gcc(\bsigma)$ for the gccs in $\bsigma$.
\end{defi}

A gcc is just a set, but we write $\rho = \rho_1 \imc_\bsigma \dots
\imc_\bsigma \rho_n \in \gcc(\bsigma)$ to insist on the causal ordering
inherited from $\leq_\bsigma$.
If also $\rho_1 \imc_\bsigma \dots \imc_\bsigma \rho_n \imc_\bsigma m
\in \gcc(\bsigma)$, then we write $\rho \imc m = \rho \cup \{m\}$.
Gccs are not necessarily down-closed:
\begin{figure}
\[
\gcc\left(
\raisebox{40pt}{$
\scalebox{.9}{$
\xymatrix@R=5pt@C=15pt{
(\tunit \ar@{}[r]|\lin&\tunit)\ar@{}[r]|\lin&\tnat\\
&&\qu^-  \ar@{-|>}[dl]\\
&\qu^+   \ar@{-|>}[dl]
        \ar@{-|>}[d]
        \ar@{.}@/^/[ur]\\
\qu^-    \ar@{-|>}[d]
        \ar@{.}@/^/[ur]
        \ar@{-|>}[drr]&
\done^- \ar@{-|>}[dr]
        \ar@{.}@/^/[u]\\
\done^+ \ar@{.}@/^/[u]&&
1^+\ar@{.}@/_/[uuu]
}$}$}\right)
~~\supseteq ~~
\left\{
\raisebox{40pt}{$
\scalebox{.9}{$
\xymatrix@R=5pt@C=15pt{
(\tunit \ar@{}[r]|\lin&\tunit)\ar@{}[r]|\lin&\tnat\\
&&\qu^-  \ar@{-|>}[dl]\\
&\qu^+   \ar@{-|>}[dl]
        \ar@{.}@/^/[ur]
\\
\qu^-    \ar@{-|>}[d]
        \ar@{.}@/^/[ur]
&
\\
\done^+ \ar@{.}@/^/[u]&&
}$}$}~,~
\raisebox{40pt}{$
\scalebox{.9}{$
\xymatrix@R=5pt@C=15pt{
(\tunit \ar@{}[r]|\lin&\tunit)\ar@{}[r]|\lin&\tnat\\
&&\qu^-  \ar@{-|>}[dl]\\
&\qu^+   \ar@{-|>}[dl]
        \ar@{.}@/^/[ur]\\
\qu^-    
        \ar@{.}@/^/[ur]
        \ar@{-|>}[drr]&
\\
&&
1^+\ar@{.}@/_/[uuu]
}$}$}~,~
\raisebox{40pt}{$
\scalebox{.9}{$
\xymatrix@R=5pt@C=15pt{
(\tunit \ar@{}[r]|\lin&\tunit)\ar@{}[r]|\lin&\tnat\\
&&\qu^-  \ar@{-|>}[dl]\\
&\qu^+   
        \ar@{-|>}[d]
        \ar@{.}@/^/[ur]\\
&\done^- \ar@{-|>}[dr]
        \ar@{.}@/^/[u]\\
&&1^+\ar@{.}@/_/[uuu]
}$}$}
\right\}
\]
\caption{Maximal grounded causal chains of a causal strategy}
\label{fig:ex_gcc}
\end{figure}
we show in Figure \ref{fig:ex_gcc} all maximal gccs of a causal
strategy. Of those, the second and third omit some dependencies of
$1^+$.

We may now make formal the idea of a strategy ``only merging
threads forked by Player''.

\begin{defi}\label{def:imc-pre-innocence}
Consider $A$ an arena. A causally deterministic $\bsigma : A$ is
\textbf{pre-innocent} iff
\[
\begin{array}{ll}
\text{\emph{pre-innocence:}} &
\text{If $m^+ \in \ev{\bsigma}$ and $\rho_1 \imc m, \rho_2 \imc m
\in \gcc(\bsigma)$ are distinct,}\\
&\text{then their least distinct moves are positive.}
\end{array}
\]
\end{defi}

As causal strategies are \emph{pointed}, $\rho_1$ and
$\rho_2$ necessarily share the same initial move.
The strategy of Figure \ref{fig:par_comp} is pre-innocent. In
contrast, that of Figure \ref{fig:ex_aug1} is not -- both
augmentations of Figure \ref{fig:ex_aug2} fail pre-innocence.
For instance, the second and third gccs of Figure \ref{fig:ex_gcc}
arrive at $1^+$ but before that, the greatest common event
is $\qu^+$, which is positive: Player is merging (via $1^+$) two gccs
forked by Opponent, which is forbidden by pre-innocence.

\begin{figure}
\begin{minipage}{.46\linewidth}
\[
\xymatrix@R=5pt@C=5pt{
\oc (\tunit &\with & \tunit)& \vdash & \tunit\\
&&&&\qu^-
        \ar@{-|>}[dllll]
        \ar@{-|>}[dll]\\
\qu^+_{\grey{0}}        
        \ar@{-|>}[d]&&
\qu^+_{\grey{1}}        
        \ar@{-|>}[d]\\
\done^-_{\grey{0}}
        \ar@{-|>}[drrrr]
        \ar@{.}@/^/[u]&&
\done^-_{\grey{1}}
        \ar@{-|>}[drr]
        \ar@{.}@/^/[u]\\
&&&&\done^+
        \ar@{.}@/_/[uuu]
}
\]
\caption{A typical pre-innocent strategy}
\label{fig:par_comp}
\end{minipage}
\hfill
\begin{minipage}{.45\linewidth}
\[
\xymatrix@R=5pt@C=10pt{
\grey{\oc (\tunit} &\grey{\with}& \grey{\tunit)} &\grey{\parallel}&
\tunit\\
&&&&\qu^-   
        \ar@[grey]@{-|>}[dll]
        \ar@{-|>}[dllll]\\
{\qu^\labr_{\grey{0}}}
        \ar@{-|>}[drr]
        \ar@[grey]@{-|>}[d]
        \ar@{.}@/^/[urrrr]
&&\grey{\qu^\labr_{\grey{1}}}
        \ar@[grey]@{.}@/^.2pc/[urr]
        \ar@[grey]@{-|>}[d]\\
\grey{\done^\labl_{\grey{0}}}
        \ar@[grey]@{-|>}[drrrr]
        \ar@[grey]@{.}@/^/[u]&&
{\done^\labl_{\grey{1}}}
        \ar@{-|>}[drr]
        \ar@[grey]@{.}@/^/[u]\\
&&&&\done^\labr
        \ar@/_/@{.}[uuu]
}
\]
\caption{Partiality of views}
\label{fig:ex_partial_views}
\end{minipage}
\end{figure}

It will follow later on that the \emph{sequential
pre-innocent} causal strategies exactly match the standard
alternating innocent strategies of Definition \ref{def:innocence}:
\emph{sequentiality} entails that there is no Player branching.
Thus separate branches always correspond to threads spawned by Opponent,
which by pre-innocence cannot interfere. The causal structure
is then a forest, matching that of \emph{P-views} of
Section~\ref{subsec:vis_inn}. We postpone the details to
Section~\ref{subsec:seq_inn_pcf}.

However, it turns out that pre-innocence is incomplete for parallel
strategies.

\subsubsection{Visibility} The problem arises as non-stability of
pre-innocence under composition. A counter-example appears in Figure
\ref{fig:fail_preinn}, examined when proving compositionality of
innocence.

But we can explain the issue intuitively: the definition of
pre-innocence relies on gccs which formalize a notion of \emph{thread}.
If that intuition is to be taken seriously, gccs should be valid
executions of standalone sequential programs. But this is not the case:
\begin{figure}
\[
\raisebox{40pt}{$
\xymatrix@R=5pt@C=15pt{
(\tunit \ar@{}[r]|\lin&\tunit)\ar@{}[r]|\lin&\tnat\\
&&\qu^-  \ar@{-|>}[dl]\\
&\qu^+  \ar@{-|>}[d]
        \ar@{.}@/^/[ur]\\
&
\done^- \ar@{-|>}[dl]
        \ar@{.}@/^/[u]\\
\done^+ \ar@{.}@/^/[u]&&
}$}
\qquad
\in
\qquad
\gcc\left(
\raisebox{40pt}{$
\xymatrix@R=5pt@C=15pt{
(\tunit \ar@{}[r]|\lin&\tunit)\ar@{}[r]|\lin&\tnat\\
&&\qu^-  \ar@{-|>}[dl]\\
&\qu^+   \ar@{-|>}[dl]
        \ar@{-|>}[d]
        \ar@{.}@/^/[ur]\\
\qu^-    \ar@{-|>}[d]
        \ar@{.}@/^/[ur]&
\done^- \ar@{-|>}[dl]
        \ar@{-|>}[dr]
        \ar@{.}@/^/[u]\\
\done^+ \ar@{.}@/^/[u]&&
0^+\ar@{.}@/_/[uuu]
}$}\right)
\]
\caption{A gcc of a non-visible strategy, losing its pointer}
\label{fig:ex_gcc_nonvis}
\end{figure}
Figure \ref{fig:ex_gcc_nonvis} shows a gcc where the last move answers a
question that \emph{was not asked} within this gcc. This could not be a
valid state of a sequential program, because the last move \emph{loses
its pointer}.

\emph{Visible} strategies are simply those such that this does not
happen.

\begin{defi}\label{def:visible}
A causal strategy $\bsigma : A$ is \textbf{visible} if for all $\rho \in
\gcc(\bsigma)$, $\pr_{\bsigma}(\rho) \in \conf{A}$.
\end{defi}

In other words, every move in $\rho$ \emph{points} within $\rho$.  This
phrasing highlights the analogy with Definition \ref{def:pvis},
\emph{i.e.} ``Player always points in the P-view''. It is indeed this
analogy that inspired the name\footnote{This, plus as in traditional
game semantics, visibility is a prerequisite for a working notion of
innocence.}. But one must be wary: the alternating interpretation of
sequential programs with state yields sequential P-visible strategies,
but their causal interpretation (as in Figure \ref{fig:ex_aug1}) is
\emph{not} necessarily visible. Visibility is very restrictive, it is
not clear what would be a sensible programming primitive that would
satisfy visibility but not pre-innocence.

We regard visibility as a key contribution. It
has far-reaching consequences -- some of which will be introduced in the
course of this paper. In fact, visibility is used more than parallel
innocence in further developments in this line of work
\cite{lics18,DBLP:journals/pacmpl/ClairambaultV20}.

The following lemma captures how a gcc may be regarded as a
standalone thread.

\begin{lem}\label{lem:gcc_pview}
Consider $\bsigma : A$ a visible causal strategy.

Then for any $\rho = \rho_1 \imc_\bsigma \dots \imc_\bsigma \rho_n \in
\gcc(\bsigma)$, $\pr_\bsigma(\rho) = \pr_\bsigma(\rho_1) \dots
\pr_\bsigma(\rho_n)$ is a P-view.
\end{lem}
\begin{proof}
By Lemma \ref{lem:app_caus_alt}, $\pr_\bsigma(\rho)$ is
an alternating sequence. By visibility, its prefixes are configurations
of $A$. So, $\pr_\bsigma(\rho) \in \Alt(A)$. By Lemma
\ref{lem:app_aux_caus}, the predecessor of $a^- \in \pr_\bsigma(\rho)$
for $\imc_A$
is its predecessor in $\pr_{\bsigma}(\rho)$, \emph{i.e.} Opponent always
points to the previous move.
\end{proof}

We may now define \emph{parallel innocent} causal strategies,
or just \emph{innocent} for short.

\begin{defi}
Consider $\bsigma : A$ causally deterministic on arena $A$.

It is \textbf{parallel innocent} if it is pre-innocent and \textbf{visible}.
\end{defi}

A standard innocent strategy as in Section~\ref{subsec:vis_inn}, under
its ``causal'' presentation, is a forest of P-views (see Proposition
\ref{prop:causal_inn}), \emph{i.e.} a forest of (displayed) gccs. In
that light the definition of parallel innocent strategies seems natural:
they are generated no longer by a forest of P-views, but by a directed
acyclic graph of P-views with additional conflict relation. This graph
describes how threads are spawned, and then may merge, following the
innocence discipline ensuring that Player may not create interference
between Opponent's threads. 

One of the main hurdles, in traditional game semantics, is to
prove that innocent strategies compose. We now
tackle this problem for parallel innocent strategies.

\subsection{Composition of Visibility} First, we establish
compositionality of visibility.

\subsubsection{Justifiers in strategies} We introduce
some machinery on \emph{justifiers}. Consider $\bsigma : A$ a causal
strategy on some $-$-arena $A$. As for plays, the immediate
causality in $A$ endows moves in $\ev{\bsigma}$ with a notion of
\emph{justifier}. This extends to $\bsigma : A \vdash B$ with $A$ and
$B$ $-$-arenas:

\begin{defi}\label{def:justifier1}
Consider $A$ and $B$ $-$-arenas, and $\bsigma : A \vdash B$. Then, for
all $m, m' \in \ev{\bsigma}$,
\[
\begin{array}{rclcl}
\just(m) &=& m' & \qquad &\text{if $\pr_{\bsigma}(m') \imc_{(A\vdash B)}
\pr_\bsigma(m)$,}\\
\just(m) &=& \init(m) && \text{if $\pr_\bsigma(m)$ minimal in $A$,}
\end{array}
\]
and undefined otherwise.
\end{defi}

This leaves the justifier undefined exactly for moves corresponding to
minimal moves in $B$, the \emph{initial moves}. Note that assigning the
justifier of $m$ minimal in $A$ to $\init(m)$ ensures that the
assignment of justifiers is invariant under currying.
It might be helpful to the reader to observe that a causal strategy
$\bsigma : A \vdash B$ is visible iff for all $\rho \in \gcc(\bsigma)$,
for all $m \in \rho$, $\just(m) \in \rho$ as well: all gccs are closed
under justifiers. We mention in passing this lemma:

\begin{lem}\label{lem:just_pred}
Consider $A, B$ $-$-arenas and $\bsigma : A \vdash B$ a causal strategy. 

Then, for any non-initial $m \in \ev{\bsigma}$, we have $\just(m)
<_\bsigma m$. Moreover, if $\pol_{\bsigma}(m) = -$, then $\just(m)
\imc_\bsigma m$ is its (unique) immediate predecessor. 
\end{lem}
\begin{proof}
As a map of event structures, $\pr_\bsigma$ locally reflects
causality (Lemma \ref{lem:es_refl_caus}), so $\just(m) <_\bsigma m$ if
the first clause of Definition \ref{def:justifier1} applies; for the
other we clearly have $\init(m) <_\bsigma m$.

If $\pol_{\bsigma}(m) = -$, then $\just(m)$ is defined
via the first clause since $A$ is negative, and $\pr_\bsigma(\just(m))
\imc_{A \vdash B} \pr_\bsigma(m)$. Now, $m$ has a predecessor $m'
\imc_\bsigma m$, by courtesy $\pr_\bsigma(m') \imc_{A\vdash B}
\pr_\bsigma(m)$, so $\pr_\bsigma(m') = \pr_\bsigma(\just(m))$ as $A$ is
forestial, and $m' = \just(m)$ by local injectivity.
\end{proof}

\subsubsection{Justifiers in interactions} We extend justifiers to
\emph{interactions} -- consider $A, B$ and $C$ three $-$-arenas, and
$\bsigma : A \vdash B$ and $\btau : B \vdash C$ causal strategies.

\begin{defi}\label{def:just_int}
We define the partial function $\just : \ev{\btau \inter \bsigma} \pto
\ev{\btau \inter \bsigma}$ as $\just(m) = m'$ if:
\[
\begin{array}{ll}
\text{\emph{(1)}}
& \text{$\pr_{\btau\inter \bsigma}(m') \imc_{A\parallel B \parallel
C} \pr_{\btau\inter \bsigma}(m)$, or}\\
\text{\emph{(2)}}
& \text{$\pr_{\btau\inter \bsigma}(m)$ is minimal in $A$ and
$m'_\bsigma = \init(m_\bsigma)$, or}\\
\text{\emph{(3)}} 
& \text{$\pr_{\btau \inter \bsigma}(m)$ is minimal in $B$ and
$m'_\btau = \init(m_\btau)$,} 
\end{array}
\]
and undefined otherwise. We say that $m'$ is the \textbf{justifier} of
$m$ in $\btau\inter \bsigma$.
\end{defi}

This leaves $\just(m)$ undefined exactly if it corresponds to a minimal
move in $C$. Clearly the two notions of justifier are compatible, in the
sense that for all $m \in \ev{\btau \inter \bsigma}$, if $m_\bsigma$ is
defined then $\just(m)_\bsigma$ is defined and equal to
$\just(m_\bsigma)$, and likewise for $\btau$.

\subsubsection{Views of gccs} We introduce the main technical device
on visible causal interactions.

We use polarities in interactions as in Section~\ref{subsubsec:caus_conf_int}, and annotate events accordingly.  We also
write \emph{e.g.} $a^{-, \labr}$ to indicate that $a$ has polarity $-$
\emph{or} $\labr$.  If $\rho \in \gcc(\btau \inter \bsigma)$ with last
event $m$, we say that $\rho$ \textbf{ends in $\bsigma$} if $m_\bsigma$
is defined, and likewise for $\btau$.  We now define \emph{views} of
gccs, used to project a gcc of the interaction to gccs for both
strategies. 

\begin{defi}\label{def:gcc_views}
Consider $\bsigma : A \vdash B$ and $\btau : B \vdash C$ 
with $A$, $B$ and $C$ $-$-arenas.

If $\rho \in \gcc(\btau \inter \bsigma)$ ends in $\bsigma$, we
(partially) define $\pview{\rho}^\bsigma \in \gcc(\bsigma)$ by:
\[
\begin{array}{rclcl}
\pview{\rho_0\!\imc\!\dots\! \imc\! \rho_n\! \imc\! \rho_{n+1}^\labl}^\bsigma
        &\!\!=\!\!& \pview{\rho_0 \!\imc \!\dots\! \imc\! \rho_n}^\bsigma \cup
\{(\rho_{n+1})_\bsigma\}\,,\\
\pview{\rho_0\! \imc\! \dots\! \imc\! \rho_i\! \imc\! \dots\! \imc\!
\rho_{n+1}^{-,\labr}}^\bsigma
        &\!\!=\!\!& \pview{\rho_0\! \imc\! \dots\! \imc\! \rho_i}^\bsigma \cup
\{(\rho_{n+1})_\bsigma\}
        && \!\!\!\!\!\!\!\!\!\!\text{if $\just(\rho_{n+1}) = \rho_i$ in $A$ or
$B$}\,,\\
\pview{\rho_0\! \imc\! \dots\! \imc\! \rho_i\! \imc\! \dots\! \imc\!
\rho_{n+1}^{\labr}}^\bsigma
        &\!\!=\!\!& \{(\rho_{n+1})_\bsigma\} 
        && \!\!\!\!\!\!\!\!\!\!\text{if $\rho_{n+1}$ minimal in $B$}\,,
\end{array}
\]
undefined otherwise. For $\rho \in \gcc(\btau \inter \bsigma)$
ending in $\btau$, we (partially) define $\pview{\rho}^\btau \in
\gcc(\btau)$:
\[
\begin{array}{rclcl}
\pview{\rho_0\! \imc\! \dots\! \imc\! \rho_n\! \imc\! \rho_{n+1}^\labr}^\btau
        &\!\!=\!\!& \pview{\rho_0\! \imc\! \dots\! \imc\! \rho_n}^\btau \cup
\{(\rho_{n+1})_\btau\}\,,\\
\pview{\rho_0\! \imc\! \dots\! \imc\! \rho_i\! \imc\! \dots\! \imc\!
\rho_{n+1}^{-,\labl}}^\btau
        &\!\!=\!\!& \pview{\rho_0\! \imc\! \dots\! \imc\! \rho_i}^\btau \cup
\{(\rho_{n+1})_\btau\}
        && \!\!\!\!\text{if $\just(\rho_{n+1}) = \rho_i$}\,;
\end{array}
\]
when defined we call $\pview{\rho}^\bsigma \in \gcc(\bsigma)$ the
\textbf{$\bsigma$-view} of $\rho$ and $\pview{\rho}^\btau \in
\gcc(\btau)$ the \textbf{$\btau$-view} of $\rho$.
\end{defi}

These definitions almost perfectly follow Definition \ref{def:pview}.
The last clause is only needed for $\pview{-}^\bsigma$ and not
$\pview{-}^\btau$, because an initial event in $C$ must be the first
event of $\rho$ anyway.

That this yields gccs of $\bsigma$ and $\btau$ rests on Lemma
\ref{lem:caus_int}, and courtesy of $\bsigma$ and $\btau$.
The $\bsigma$-view and the $\btau$-view are in principle only partially
defined, because it may be, when attempting to follow the opponent's
pointer, that that justifier lies outside the gcc.
For instance $\pview{\rho}^\btau$, for $\rho$
in Figure \ref{fig:ex_partial_views}, is not well-defined: when
attempting to compute $\pview{\qu^- \qu^\labr_{\grey{0}}
\done^\labl_{\grey{1}}}^\btau$, none of the clauses apply as
$\just(\done^\labl_{\grey{1}}) = \qu^\labr_{\grey{1}}$ is outside
$\rho$.
The bulk of the proof of stability of visibility under composition, is
to show that this cannot happen for visible strategies:

\begin{prop}\label{prop:views_int_gccs}
Let $\bsigma : A \vdash B$ and $\btau : B \vdash C$ be \emph{visible}
causal strategies.

Then, the views of gccs of $\btau \inter \bsigma$ as in Definition
\ref{def:gcc_views} are always well-defined.
\end{prop}
\begin{proof}
We prove by induction on $\rho$ that, for all prefixes of $\rho$,
\[
\begin{array}{ll}
\text{\emph{(1)}} & \text{if $\rho$ ends in $\bsigma$, then
$\pview{\rho}^\bsigma$ is well-defined,}\\
\text{\emph{(2)}} & \text{if $\rho$ ends in $\btau$, then
$\pview{\rho}^\btau$ is well-defined\,.}
\end{array}
\]

Assume $\rho$ finishes in $\btau$. If the last move has polarity
$-$, then either it is initial and there is nothing to prove, or by
Lemma \ref{lem:caus_int} its justifier is its predecessor in
$\rho$, so $\pview{\rho}^\btau \in \gcc(\btau)$ follows immediately by
induction hypothesis (in that case $\rho$ does not end in
$\bsigma$).

If the last move has polarity $\labr$, write $\rho = \rho' \imc m
\imc n^\labr$. By Lemma \ref{lem:caus_int}, $m_\btau \imc_\btau
n_\btau$, so in particular $m$ is in $\btau$. By induction hypothesis,
$\kappa = \pview{\rho' \imc m}^\btau \in \gcc(\btau)$, so
\[
\kappa \imc n_\btau = \pview{\rho}^\btau \in \gcc(\btau)
\]
as well.
But if $\rho$ finishes in $\bsigma$ \emph{and} $\btau$ (\emph{i.e.} in
$B$), we must further prove that $\pview{\rho}^\bsigma \in
\gcc(\bsigma)$. In that case, we observe that since $\pview{\rho}^\btau
\in \gcc(\btau)$ and $\btau$ is visible, it follows that
$\just(n_\btau) \in
\pview{\rho}^\btau$, but this entails $\just(n) \in
\rho$. Hence the second clause of Definition \ref{def:gcc_views}
applies, and we conclude by induction hypothesis.
If $\rho$ finishes in $\bsigma$, the reasoning is symmetric.
\end{proof}

From this, we are now ready to conclude:

\begin{prop}\label{prop:comp_visible}
Let $\bsigma : A \vdash B$ and $\btau : B \vdash C$ be visible causal
strategies.

Then, $\btau \odot \bsigma : A \vdash C$ is also visible.
\end{prop}
\begin{proof}
We prove by induction on $\rho$ that for all $\rho \in
\gcc(\btau\inter \bsigma)$, $\pr_{\btau \inter \bsigma}(\rho) \in
\conf{A\parallel B \parallel C}$. If $\rho$ is empty it is
clear; take $\rho \imc m \in \gcc(\btau \inter \bsigma)$.
By induction hypothesis, $\pr_{\btau \inter \bsigma}(\rho) \in
\conf{A\parallel B \parallel C}$, we only need that the
justifier of $m$ is in $\rho$. We reason by cases on the polarity of
$m$: if it is $\ell$, then by Proposition \ref{prop:views_int_gccs}
$\pview{\rho\imc m}^\bsigma \in \gcc(\bsigma)$. But since
$\bsigma$ is visible, the justifier of $m_\bsigma$ appears in
$\pview{\rho}^\bsigma$; so the justifier of $m$ is in $\rho$. The other cases
are symmetric or trivial.

Now, take $\rho_\odot \in \gcc(\btau \odot \bsigma)$. By
definition of $\btau \odot \bsigma$, there is a (non-necessarily unique)
$\rho_\inter \in \gcc(\btau \inter \bsigma)$ such that $\rho_{\odot}$
comprises exactly those events of $\rho_\inter$ occurring in $A$ or $C$.
By the observation above, $\pr_{\btau \inter \bsigma}(\rho_\inter) \in
\conf{A \parallel B \parallel C}$, hence $\pr_{\btau\odot
\bsigma}(\rho_\odot) \in \conf{A \parallel C}$.
\end{proof}

\subsection{Composition of Innocence} We now address composition
of pre-innocence.

We start this section by showing ``what could go
wrong''.
\begin{figure}
\[
\left(
\raisebox{50pt}{$
\scalebox{.8}{$
\xymatrix@R=10pt@C=5pt{
\tunit &\lin & \tunit& \lin & \tunit\\
&&&&\qu^-
        \ar@{-|>}[dllll]
        \ar@{-|>}[dll]\\
\qu^+
        \ar@{-|>}[d]
        \ar@{.}@/^/[urrrr]&&
\qu^+
        \ar@{-|>}[d]
        \ar@{.}@/^/[urr]\\
\done^-
        \ar@{-|>}[drrrr]
        \ar@{.}@/^/[u]&&
\done^-
        \ar@{-|>}[drr]
        \ar@{.}@/^/[u]\\
&&&&\done^+
        \ar@{.}@/_/[uuu]
}$}$}
\right)
\odot
\left(
\raisebox{50pt}{$
\scalebox{.8}{$
\xymatrix@C=-5pt@R=10pt{
\tunit & \lin & \tunit & \lin & \tunit &\vdash& (\tunit & \lin & \tunit
& \lin & \tunit) & \lin & \tunit\\ 
&&&&&&&&&&&& \qu^-
        \ar@{-|>}[dllllllll]
        \ar@{-|>}[dll]\\
&&&&\qu^+
        \ar@{.}@/^/[urrrrrrrr]
        \ar@{-|>}[dllll]
        \ar@{-|>}[dll]
        \ar@{-|>}[d]
&&&&&&\qu^+
        \ar@{-|>}[dllll]
        \ar@{-|>}[dll]
        \ar@{.}@/^/[urr]\\
\qu^-   \ar@{.}@/^/[urrrr]
        \ar@{-|>}[d]&&
\qu^-   \ar@{.}@/^/[urr]
        \ar@{-|>}[d]&&
\done^- \ar@{.}@/^/[u]
        \ar@{-|>}[drr]&&
\qu^-   \ar@{.}@/^/[urrrr]
        \ar@{-|>}[d]
        \ar@{-|>}[dllllll]&&
\qu^-   \ar@{.}@/^/[urr]
        \ar@{-|>}[dllllll]\\
\done^+ \ar@{.}@/^/[u]&&
\done^+ \ar@{.}@/^/[u]&&&&
\done^+ \ar@{.}@/^/[u]
}$}$}
\right)
=
\left(
\raisebox{50pt}{$
\scalebox{.8}{$
\xymatrix@C=-5pt@R=10pt{
(\tunit & \lin & \tunit & \lin & \tunit) & \lin & \tunit\\
&&&&&&\qu^-
        \ar@{-|>}[dll]\\
&&&&\qu^+
        \ar@{.}@/^/[urr]
        \ar@{-|>}[dllll]
        \ar@{-|>}[dll]\\
\qu^-   \ar@{.}@/^/[urrrr]
        \ar@{-|>}[d]&&
\qu^-   \ar@{.}@/^/[urr]
        \ar@{-|>}[dll]\\
\done^+ \ar@{.}@/^/[u]
}$}$}
\right)
\]
\caption{Failure of preservation of pre-innocence under composition}
\label{fig:fail_preinn}
\end{figure}
In Figure \ref{fig:fail_preinn}, we show a counter-example to the
stability under composition of pre-innocence without visibility,
with the corresponding interaction appearing as Figure
\ref{fig:illegal-merge}. Let us attempt to explain the phenomenon,
calling $\bsigma$ the left hand side strategy (parallel composition)
and $\btau$ the right hand side one -- observe that the dotted lines
include the justifications relations from Definition
\ref{def:justifier1} rather than just those coming from the arena.
Imagine that $\btau$ wants to perform an illegal causal merge between
the two argument calls of its argument of type $\tunit \lin \tunit \lin
\tunit$. By pre-innocence it cannot do so directly. However, it can
outsource the merge to $\bsigma$ by linking
(legally with respect to pre-innocence, but illegally with respect to
visibility) the arguments of the parallel composition to those that it
wants to merge.

We shall prove that this cannot happen in the presence of visibility.
Let us fix, until the end of the section, two visible causal strategies
$\bsigma : A \vdash B$ and $\btau : B \vdash C$.

\subsubsection{The ``forking lemma''}
\begin{figure}
\begin{minipage}{.45\linewidth}
\[
\xymatrix@C=-3pt@R=8pt{
\tunit & \lin & \tunit & \lin & \tunit &\vdash& (\tunit &
\lin & \tunit & \lin &
\tunit) & \lin & \tunit\\
&&&&&&&&&&&& \qu^-
        \ar@[grey]@{-|>}[dllllllll]
        \ar@{-|>}[dll]\\
&&&&\grey{\qu^\tau}
        \ar@[grey]@{.}@/^/[urrrrrrrr]
        \ar@[grey]@{-|>}[dllll]
        \ar@[grey]@{-|>}[dll]
&&&&&&\qu^\tau
        \ar@{-|>}[dllll]
        \ar@{-|>}[dll]
        \ar@{.}@/^/[urr]\\
\grey{\qu^\sigma}
        \ar@[grey]@{.}@/^/[urrrr]
        \ar@[grey]@{-|>}[d]&&
\grey{\qu^\sigma}
        \ar@[grey]@{.}@/^/[urr]
        \ar@[grey]@{-|>}[d]&&&&
\qu^-   \ar@{.}@/^/[urrrr]
        \ar@{-|>}[dllllll]&&
\qu^-   \ar@{.}@/^/[urr]
        \ar@{-|>}[dllllll]\\
\done^\tau 
        \ar@[grey]@{.}@/^/[u]
        \ar@{-|>}[drrrr]&&
\done^\tau 
        \ar@[grey]@{.}@/^/[u]
        \ar@{-|>}[drr]&&&&\\
&&&&\done^\sigma
        \ar@[grey]@{.}@/_/[uuu]
        \ar@{-|>}[drr]&&&\\
&&&&&&\done^\tau
        \ar@{.}@/_/[uuu]
}
\]
\caption{Illegal causal merge}
\label{fig:illegal-merge}
\end{minipage}
\hfill
\begin{minipage}{0.5\linewidth}
\bigskip
\[
\scalebox{.6}{$
\xymatrix@C=8pt{
&&&&&&&&&&\bullet
        \ar@[red]@/_.1pc/[dl]
        \ar@{-|>}[r]&
~       \ar@{.}[rrrrrrr]&
~       \ar@[blue]@/_1pc/[ll]&
~       \ar@[red]@/_.1pc/[l]&&
~       \ar@[blue]@/_1pc/[ll]&
~       \ar@[red]@/_.1pc/[l]&&
~       \ar@[blue]@/_1pc/[ll]
        \ar@{-|>}[r]&
\bullet \ar@[red]@/_.1pc/[l]\\
\bullet \ar@{-|>}[r]&
~       \ar@{.}[rrrrrrr]
        \ar@[red]@/_.1pc/[l]
&
&
~       \ar@[blue]@/_1pc/[ll]&
~       \ar@[red]@/_.1pc/[l]&&
~       \ar@[blue]@/_1pc/[ll]&
~       \ar@[red]@/_.1pc/[l]&
~       \ar@{-|>}[r]&
\bullet \ar@{-|>}[ur]
        \ar@{-|>}[dr]
        \ar@[blue]@/_1pc/[ll]\\
&&&&&&&&&&\bullet
        \ar@[red]@/^.1pc/[ul]
        \ar@{-|>}[r]&
~       \ar@{.}[rrrrrrr]&
~       \ar@[blue]@/^1pc/[ll]&
~       \ar@[red]@/^.1pc/[l]&&
~       \ar@[blue]@/^1pc/[ll]&
~       \ar@[red]@/^.1pc/[l]&&
~       \ar@[blue]@/^1pc/[ll]
        \ar@{-|>}[r]&
\bullet \ar@[red]@/^.1pc/[l]
}$}
\]
\bigskip
\bigskip
\[
\scalebox{.6}{$
\xymatrix@C=8pt{
&&&&&&&&&&\bullet
        \ar@{-|>}[r]&
~       \ar@{.}[rrrrrrr]&
~       \ar@[blue]@/_3pc/[dlllll]&
~       \ar@[red]@/_.1pc/[l]&&
~       \ar@[blue]@/_1pc/[ll]&
~       \ar@[red]@/_.1pc/[l]&&
~       \ar@[blue]@/_1pc/[ll]
        \ar@{-|>}[r]&
\bullet \ar@[red]@/_.1pc/[l]\\
\bullet \ar@{-|>}[r]&
~       \ar@{.}[rrrrrrr]
        \ar@[red]@/_.1pc/[l]&&
~       \ar@[blue]@/_1pc/[ll]
&
~       \ar@[red]@/_.1pc/[l]
&
&
~       \ar@[blue]@/_1pc/[ll]
&
~       \ar@[red]@/_.1pc/[l]&
~       \ar@{-|>}[r]&
\bullet \ar@{-|>}[ur]
        \ar@{-|>}[dr]
        \ar@[blue]@/^2pc/[lllll]
\\
&&&&&&&&&&\bullet
        \ar@[red]@/^.1pc/[ul]
        \ar@{-|>}[r]&
~       \ar@{.}[rrrrrrr]&
~       \ar@[blue]@/^1pc/[ll]&
~       \ar@[red]@/^.1pc/[l]&&
~       \ar@[blue]@/^1pc/[ll]&
~       \ar@[red]@/^.1pc/[l]&&
~       \ar@[blue]@/^1pc/[ll]
        \ar@{-|>}[r]&
\bullet \ar@[red]@/^.1pc/[l]
}$}
\]
\bigskip
\caption{Merging paths in $G$}
\label{fig:illustr_G}
\end{minipage}
\end{figure}
Taking a closer look at Figure \ref{fig:illegal-merge}, we highlight the
two illegally merging gccs in the interaction: while $\bsigma$ is
responsible for the merge, the point where these gccs forked is
external, outside the scope of $\bsigma$!  The next lemma, dubbed the
``forking lemma'', forbids this: it implies that visible strategies
cannot unknowingly close an Opponent fork.

If $\rho = \rho_1 \imc \dots
\imc \rho_n$ is a gcc and $1\leq i \leq n$, 
$\rho_{\leq i}$ is the gcc $\rho_1 \imc \dots \imc \rho_i$. Two
gccs $\rho, \kappa$ are \textbf{forking} iff $\rho \cap \kappa \neq
\emptyset$, and for all $i, j$, if $\rho_i = \kappa_j$ then $\rho_{\leq
i} = \kappa_{\leq j}$.  If $\rho, \kappa$ are two forking gccs, we write
$\gce(\rho, \kappa)$ for their \textbf{greatest common event}. Notice
that despite the terminology, two forking gccs can be prefix of one
another and never truly go separate ways.

\begin{lem}[Forking lemma]\label{lem:forking}
Let $\rho, \kappa \in \gcc(\btau \inter \bsigma)$ be forking gccs ending
in $\bsigma$, s.t. $\pview{\rho}^\bsigma \cap
\pview{\kappa}^\bsigma \neq \emptyset$ and $\gce(\pview{\rho}^\bsigma,
\pview{\kappa}^\bsigma)$ negative (the least distinct
events, if any, are positive).

Then, $\gce(\pview{\rho}^\bsigma, \pview{\kappa}^\bsigma) = \gce(\rho,
\kappa)_\bsigma$. Moreover, the symmetric property holds for $\btau$.
\end{lem}
\begin{proof}
We only detail the proof for $\bsigma$, the proof for $\btau$ is exactly
the same.  We build a directed graph $G$ with vertices $\rho \cup
\kappa$, and edges the (disjoint) union of the sets:
\[
\begin{array}{rcl}
\oedges &=& \{(m_1, m^\labl_2)\mid \just(m_1) = m_2\}\\
\pedges &=& \{(m_1^\labl, m_2) \mid m_2 \imc_{\btau \inter \bsigma} m_1\}
\end{array}
\]
where the annotation $m_i^\labl$ indicates the polarity.
Each vertex is source of at most one edge, and
following edges consists exactly in computing the $\bsigma$-view. If
$\rho$ and $\kappa$ have the same final move, then $\rho = \kappa$.
Otherwise, consider the two paths in $G$
starting with these distinct final moves. Since $\pview{\rho}^\bsigma
\cap \pview{\kappa}^\bsigma \neq \emptyset$, these two paths must
intersect -- Figure \ref{fig:illustr_G} represents a typical $G$ with
$O$-edges in blue and $P$-edges in red with the two typical cases.

These paths meet at a vertex of incoming degree at least $2$;
but vertices receive only $O$-edges, or only $P$-edges. For the
former (as in the bottom of Figure \ref{fig:illustr_G}), then
$\gce(\pview{\rho}^\bsigma, \pview{\kappa}^\bsigma)$ is positive, which
contradicts the hypothesis.  For the latter (as in the top of
Figure \ref{fig:illustr_G}), we remark that $P$-edges are immediate
causal links in $\btau \inter \bsigma$; and there is at most one event
in $\rho \cup \kappa$ causing two distinct events: if it exists, it must
be $\gce(\rho, \kappa)$. 
\end{proof}

This provides the core argument for the compositionality of
pre-innocence: intuitively, if a pre-innocent strategy
merges two threads, by pre-innocence its \emph{views} of these two
threads fork positively. But then the forking lemma ensures that this
strategy sees the actual forking point for these threads -- which
therefore cannot be due to the external Opponent.

\subsubsection{Stability of pre-innocence}
Now, much of the proof consists in restricting the causal shapes in
$\btau \inter \bsigma$ corresponding to a causal merge in $\btau \odot
\bsigma$, so that the forking lemma applies.

\begin{prop}\label{prop:rig_comp}
Consider $\bsigma : A \vdash B$ and $\btau : B \vdash C$ visible causal
strategies.

If $\bsigma : A \vdash B$ and $\btau : B \vdash C$ are
pre-innocent, then so is $\btau \odot \bsigma$.
\end{prop}
\begin{proof}
Consider $m \in \ev{\btau \odot \bsigma}$ and distinct $\rho^1 \imc m,
\rho^2 \imc m \in \gcc(\btau \odot \bsigma)$.  \emph{W.l.o.g.} assume
that whenever $\rho^1_i = \rho^2_j$, $\rho^1_{\leq i} = \rho^2_{\leq j}$
-- or we can change $m$ and $\rho^i$ keeping the same least distinct
events, but satisfying this property. Likewise, since $\rho^1, \rho^2$
are distinct, we assume \emph{w.l.o.g.} that their last moves $m_1 \in
\rho^1, m_2 \in \rho^2$ are distinct -- or we may replace $m$ with an
earlier causal merge. These two causal chains $\rho^1$ and $\rho^2$ may
be completed to $\kappa^1 \imc m, \kappa^2 \imc m \in \gcc(\btau \inter
\bsigma)$ such that $\rho^i$ consists exactly of the events of
$\kappa^i$ occurring in $A$ or $C$.  Necessarily, the greatest visible
events of $\kappa^1$ and $\kappa^2$ are $m_1$ and $m_2$ respectively.
Call $n$ the least common event of $\kappa^1 \imc m$ and $\kappa^2 \imc
m$ above $m_1$ and $m_2$ (which might not be $m$). The situation is:
\[
\xymatrix@R=10pt@C=5pt{
&&&&&&&
&
~        \ar@{.}[r]
&
~       \ar@{-|>}[r]&
m_1     \ar@{-|>}[r]&
~       \ar@{.}[r]&
~       \ar@{-|>}[r]&
\kappa^1_n
        \ar@{-|>}[dr]\\
&
~
&&&
~ 
&
&
&&&&&&&&
n      \ar@{-|>}[r]&
~       \ar@{.}[rrr]&&&
~       \ar@{-|>}[r]&
m\\
&&&&&&&
&
~       \ar@{.}[r]&
~       \ar@{-|>}[r]&
m_2     \ar@{-|>}[r]&
~       \ar@{.}[r]&
~       \ar@{-|>}[r]&
\kappa^2_p
        \ar@{-|>}[ur]
}
\]
with $m_1, m_2$ and $m$ visible, and no one visible in between. We
reason on the polarity of $n$ in $\btau \inter \bsigma$. By Lemma
\ref{lem:caus_int} and since arenas are forestial, it cannot be
negative, so its polarity is either $\labl$ or $\labr$. Assume it
is $\labl$ -- the other case is symmetric. We may compute the
\emph{$\bsigma$-views}:
\[
\pview{\kappa^1_{\leq n} \imc n}^\bsigma,
\qquad
\pview{\kappa^2_{\leq p} \imc n}^\bsigma
\quad
\in 
\quad
\gcc(\bsigma)\,,
\]
respectively $\pview{\kappa^1_{\leq n}}^\bsigma
\imc n_\bsigma$ and $\pview{\kappa^2_{\leq p}}^\bsigma \imc
n_\bsigma$, with $\pview{\kappa^1_{\leq n}}^\bsigma$ and
$\pview{\kappa^2_{\leq
p}}^\bsigma$ distinct as they respectively contain $\kappa_n^1$ and
$\kappa_p^2$. Since $\bsigma$ is pointed,
$\pview{\kappa^1_{\leq n}}^\bsigma \cap \pview{\kappa^2_{\leq
p}}^\bsigma \neq
\emptyset$, so $\kappa^1 \cap \kappa^2 \neq \emptyset$ as well.
Call $m'$ the greatest common event of $\kappa^1$ and $\kappa^2$,
necessarily below $m_1$ and $m_2$, then:
\[
\xymatrix@R=10pt@C=5pt{
&&&&&&&\kappa^1_{i+1}
        \ar@{-|>}[r]&
~       \ar@{.}[r]&
~       \ar@{-|>}[r]&
m_1     \ar@{-|>}[r]&
~       \ar@{.}[r]&
~       \ar@{-|>}[r]&
\kappa^1_n
        \ar@{-|>}[dr]\\
\kappa^1_1\ar@{-|>}[r]&~
        \ar@{.}[rrr]&&&
~       \ar@{-|>}[r]&
\kappa^1_{i-1}  
        \ar@{-|>}[r]&
m'\ar@{-|>}[ur]
        \ar@{-|>}[dr]&&&&&&&&
{n}^\labl      \ar@{-|>}[r]&
~       \ar@{.}[rrr]&&&
~       \ar@{-|>}[r]&
m\\
&&&&&&&\kappa^2_{i+1}
        \ar@{-|>}[r]&
~       \ar@{.}[r]&
~       \ar@{-|>}[r]&
m_2     \ar@{-|>}[r]&
~       \ar@{.}[r]&
~       \ar@{-|>}[r]&
\kappa^2_p
        \ar@{-|>}[ur]
}
\]
assuming \emph{w.l.o.g.} that $\kappa^1_{\leq i} =
\kappa^2_{\leq i}$ (changing the beginning of $\kappa^2$ if required).
Summing up some properties, $\xi^1 = \kappa^1_{i+1} \dots \kappa^1_n$
and $\xi^2 = \kappa^2_{i+1} \dots \kappa^2_p$ are disjoint. This entails
the $\bsigma$-views
\[
\pview{\kappa^1_{\leq n}}^\bsigma\,,
\quad
\pview{\kappa^2_{\leq p}}^\bsigma
\quad
\in
\quad
\gcc(\bsigma)
\]
are \emph{forking}: they coincide on a
prefix and disjoint afterwards. Indeed, since $\xi^1$ and $\xi^2$
are disjoint, any common event appears in $\kappa^1_1 \imc \dots \imc
\kappa^1_i$, before which the $\bsigma$-views coincide.

Now, since $\bsigma$ is pre-innocent, the least distinct moves
$m'_1$ and $m'_2$ of $\pview{\kappa^1_{\leq n}}^\bsigma$ and
$\pview{\kappa^2_{\leq p}}^\bsigma$ are positive.
Thus their common immediate predecessor is negative -- but it is also
their greatest common event, since $\pview{\kappa^1_{\leq n}}^\bsigma$
and $\pview{\kappa^2_{\leq p}}^\bsigma$ are forking.
So, by Lemma \ref{lem:forking},
\[
\gce(\pview{\kappa^1_{\leq n}}^\bsigma, \pview{\kappa^2_{\leq
p}}^\bsigma) = \gce(\kappa^1_{\leq n}, \kappa^2_{\leq p})_\bsigma = m'\,,
\]
so $m'$ is negative for $\bsigma$. In $\btau \inter \bsigma$, $m'$ is
negative or in $B$ -- in  both cases the
least visible events in $\xi^1$ and $\xi^2$ are positive, but those are
our least distinct events of $\rho_1$ and $\rho_2$.
\end{proof}

\begin{prop}
There is $\CGwbinn$, a sub-Seely category of $\CGwb$ having the same
objects and morphisms restricted to parallel innocent causal strategies.
\end{prop}
\begin{proof}
Propositions \ref{prop:comp_visible} and \ref{prop:rig_comp} ensure that
parallel innocent strategies compose. Stability under tensor and pairing
are immediate. It remains that structural morphisms are innocent.

We detail it for copycat. Consider $A$ a $-$-arena. We show that
any $(i, a) \in \ev{\cc_A}$, $(i, a)$ is minimal or has exactly
one predecessor for $\imc_{\cc_A}$. Assume first $\pol(i, a) =
-$. If it is not minimal, take $(j ,a') \imc_{\cc_A} (i, a)$.
Necessarily $i = j$ and $a' \imc_A a$, so uniqueness follows from
$A$ forestial. If $\pol(i, a) = +$, its unique immediate
predecessor is $(2-i, a)$.  So, $\cc_A$ is forestial.

Now, for visibility, consider $\rho \in \gcc(\cc_A)$. But since $\cc_A$
is forestial, $\rho \in \conf{\cc_A}$, so clearly $\pr_{\cc_A}(\rho) =
\rho \in \conf{A \vdash A}$. For parallel innocence, consider $\rho_1
\imc_{\cc_A} (i, a), \rho_2 \imc_{\cc_A} (i, a) \in \gcc(\cc_A)$
distinct. But since $\cc_A$ is forestial, $\rho_1 = \rho_2$,
contradicting their distinctness.
\end{proof}

\subsubsection{Interpretation of $\PCFpar$}
We complete the interpretation. For all sequential primitives of
$\PCFpar$, the corresponding strategy is forestial: as for copycat,
parallel innocence follows. Finally, the strategy $\mathsf{plet}_{\tx,
\ty}$ is shown parallel innocent by direct inspection (see Figure
\ref{fig:plet}). Altogether, this yields an interpretation of $\PCFpar$
in $\CGwbinn_\oc$, and we have:

\begin{cor}[Adequacy]
For $\vdash M : \tunit$ any term of $\PCFpar$, $M \eval$ iff
$\intr{M}_{\CGwbinn_\oc}\eval$.
\end{cor}
\begin{proof}
Consequence of Theorem \ref{th:adequacy_ipa}, as $\PCFpar$ is a
sub-language of $\IPA$.
\end{proof}

By the end of the paper, we will have established that $\CGwbinn$ is
intensionally fully abstract for $\PCFpar$. However the corresponding
technical development is left for last.

\section{Sequentiality and Causal Full Abstraction for $\IA$ and $\PCF$}
\label{sec:fa_ia_pcf}

In this section, we shall define \emph{sequentiality} on $\CGwb$, then
prove full abstraction results
\[  
\begin{array}{rcl}
\CGwb_\oc + \text{{sequentiality}} &\text{is
fully abstract for} & \IA\,,\\
\CGwb_\oc + \text{{parallel innocence}} +
\text{{sequentiality}} & \text{is fully abstract for} & \PCF\,, 
\end{array}
\]
established by linking with Theorems \ref{th:fa_pcf} and \ref{th:fa_ia}
for alternating strategies.

\subsection{Sequentiality} \label{subsec:seq_ia}
We construct the Seely category of \emph{sequential} causal strategies.

\subsubsection{Definition} Intuitively, a causal strategy $\bsigma : A$
is \emph{sequential} if it unfolds gracefully to a (deterministic)
alternating strategy. That does not mean that Player never throws
parallel threads, or always acts deterministically: for instance,
the strategy in Figure \ref{fig:ex_aug1} should be sequential and yet
has certain of its configurations enabling two parallel or two
conflicting Player moves.
But: as long as Opponent follows an alternating discipline, so should
Player.  

A play $s \in \NAlt(A)$ is \textbf{alternating} if $s \in \Alt(A)$. We
shall often use $\NAlt(-)$ or $\Alt(-)$ on causal 
strategies -- recall that $\Alt(\bsigma)$ and $\NAlt(\bsigma)$ are
sequences of $\ev{\bsigma}$, and must not be confused with
$\NAltStrat(\bsigma)$ that includes a move-by-move projection via the
display map.
We start by giving the
definition of (deterministic) \emph{sequentiality}:

\begin{defi}\label{def:seq_nalt_strat}
Consider $A$ an arena. A causal (pre)strategy $\bsigma : A$ is
\textbf{sequential} if:
\[
\begin{array}{rl}
\text{\emph{reachable sequentiality:}} & \text{for $t n^+ \in
\NAlt(\bsigma)$, if $t \in \Alt(\bsigma)$ then $tn \in
\Alt(\bsigma)$.}\\
\text{\emph{sequential determinism:}} & \text{for $t n_1^+, t n_2^+
\in \Alt(\bsigma)$, then $n_1 = n_2$;}\\ 
\text{\emph{sequential visibility:}} & \text{every alternating $s \in
\NAltStrat(\bsigma)$ is P-visible.}
\end{array}
\]
\end{defi}

For \emph{sequential determinism}, more than merely asking that
$\NAltStrat(\bsigma)$ acts deterministically on alternating plays, this
condition imposes that no \emph{internal} non-deterministic choice is
alternatingly reachable, even when this choice would yield no observable
non-deterministic behaviour (this is required for the forthcoming
alternating projection to preserve symmetry).

\emph{Sequential visibility} is perhaps puzzling, as P-visibility
is usually associated not to sequentiality, but to the absence of
higher-order state \cite{DBLP:conf/lics/AbramskyHM98}.  From a given
control point, a P-visible strategy may only call a procedure bound
within the branch of the syntax tree leading to that control point.  In
contrast, with higher-order state, a program may call a procedure stored
in the memory, originating from a remote program phrase outside the current branch.
This phenomenon is independent of sequentiality, but in $\CG$ the
causality due to the syntax tree blurs with that due to interference.
So strategies arising from the interpretation of $\IPA$ are ``morally''
P-visible but formalizing this is nontrivial\footnote{This was done by
Laird in the first interleaving games model
\cite{DBLP:journals/entcs/Laird01}, via explicit threading
information.}. For us it is not worth the trouble as P-visibility is not
required for full abstraction for $\IPA$\footnote{Proposition
\ref{prop:gm_def} shows that non-visible behaviour characteristic of
higher-order state can be mimicked by running several threads in
parallel and using signaling via \emph{interference} to jump control
between them.}. Consequently, it suffices to reinstate it once we
restrict to sequential strategies.

\subsubsection{Alternating projection}
We start by defining the \emph{alternating projection}.

\begin{defi}
Consider $A, B$ $-$-arenas, and $\bsigma : A \vdash B$ a sequential
causal (pre)strategy.

Then, we define $\AltStrat(\bsigma) = \{\pr_{\Lambda(\bsigma)}(t) \mid t
\in \Alt(\bsigma)\}$.
\end{defi}

We shall prove that $\AltStrat(\bsigma)$ is an alternating (pre)strategy
on $A \lin B$. The two subtle points are that it is deterministic, and
uniform -- which both rest on the observation:

\begin{lem}\label{lem:uniq_caus_wit}
Consider $A, B$ $-$-arenas, and $\bsigma : A \vdash B$ a sequential
(pre)strategy.

For any $s \in \AltStrat(\bsigma)$, there is a \emph{unique} $t \in
\Alt(\bsigma)$ such that $s = \pr_{\Lambda(\bsigma)}(t)$.
\end{lem}
\begin{proof}
Immediate by induction, using receptivity and \emph{sequential
determinism}.
\end{proof}

From this it follows that $\AltStrat(\bsigma)$ satisfies
determinism. For uniformity, we prove:

\begin{lem}\label{lem:alt_proj_pres_sym}
Consider $A$ an arena and $\bsigma, \btau : A \vdash B$ sequential
causal (pre)strategies.

If $\bsigma \simstrat \btau$, then $\AltStrat(\bsigma) \simstrat
\AltStrat(\btau)$.
\end{lem}
\begin{proof}
If $\bsigma \simstrat \btau$, by Definition \ref{def:pos_iso} there is
an isomorphism $\varphi : \bsigma \iso \btau$ of ess satisfying
\[
\pr_\btau \circ \varphi \sim^+ \pr_\bsigma\,.
\]

We prove by induction on $s^\bsigma \in \AltStrat(\bsigma)$ that
if $s^\bsigma \sym_{A\vdash B} s^\btau$ with
$s^\btau \in \AltStrat(\btau)$, then taking $t^\bsigma \in
\Alt(\bsigma)$ and $t^\btau \in \Alt(\btau)$ s.t. $s^\bsigma =
\pr_{\Lambda(\bsigma)}(t^\bsigma)$ and $s^\btau =
\pr_{\Lambda(\btau)}(t^\btau)$ from Lemma
\ref{lem:uniq_caus_wit},
we have $\varphi(t^\bsigma) \sym_\btau t^\btau$
as in (the obvious generalization of) Definition
\ref{def:alt_play_equiv}.

For $s^\bsigma$ empty it is clear.
For $s^\bsigma m_1^- \in \AltStrat(\bsigma)$ and $s^\btau m_2^- \in
\AltStrat(\btau)$ with $s^\bsigma m_1^- \sym_{A\vdash B} s^\btau m_2^-$,
there is a unique matching $t^\bsigma n_1^- \in \Alt(\bsigma)$. So
$\varphi(t^\bsigma) \varphi(n_1) \in \Alt(\btau)$, and by induction
hypothesis $\varphi(t^\bsigma) \sym_\btau t^\btau$. Moreover, from
$s^\bsigma m_1^- \sym_A s^\btau m_2^-$ and $\pr_{\Lambda(\btau)} \circ
\varphi 
\sim \pr_{\Lambda(\bsigma)}$, 
\[
\pr_{\Lambda(\btau)}(\varphi(t^\bsigma) \varphi(n_1)) \sym_{A\vdash B}
\pr_{\Lambda(\btau)}(t^\btau
n_2)\,,
\]
so $\varphi(t^\bsigma) \varphi(n_1) \sym_\btau t^\btau n_2$ by
\emph{$\sim$-receptivity} of $\btau$.
Now, for $s^\bsigma m_1^+ \in \AltStrat(\bsigma)$ and $s^\btau m_2^+ \in
\AltStrat(\btau)$ with $s^\bsigma m_1^+ \sym_{A\vdash B} s^\btau m_2^+$,
take $t^\bsigma n_1 \in \Alt(\bsigma)$ and $t^\btau n_2
\in \Alt(\btau)$ s.t. $\pr_{\Lambda(\bsigma)}(t^\bsigma n_1) =
s^\bsigma m_1$ and $\pr_{\Lambda(\btau)}(t^\btau n_2) = s^\btau m_2$. By
induction hypothesis, $\varphi(t^\bsigma) \sym_\btau t^\btau$.
As the former extends with $\varphi(n_1)$, by \emph{extension} for
$\tilde{\btau}$ there is some (positive) $n'_2$ s.t.
$\varphi(t^\bsigma) \varphi(n_1) \sym_\btau t^\btau n'_2$
but now, by \emph{reachable determinism} of $\btau$, $n_2 = n'_2$.

The lemma follows immediately from this auxiliary proposition.
\end{proof}

In particular, if $\bsigma : A\vdash B$ is sequential,
then $\bsigma \simstrat \bsigma$ via the identity isomorphism,
consequently $\AltStrat(\bsigma) \simstrat \AltStrat(\bsigma)$,
\emph{i.e.} $\AltStrat(\bsigma)$ is uniform -- from here we conclude:

\begin{prop}
Consider $A, B$ $-$-arenas, and $\bsigma : A \vdash B$ a sequential
causal strategy.

Then, $\AltStrat(\bsigma) : A \lin B$ is a P-visible alternating
strategy.
\end{prop}
\begin{proof}
\emph{Prefix-closure} and \emph{receptivity} are clear.
\emph{Determinism} is immediate from Lemma \ref{lem:uniq_caus_wit} and
\emph{sequential determinism}, and \emph{uniformity} is immediate from
Lemma \ref{lem:alt_proj_pres_sym}.

Finally, \emph{P-visibility} is immediate by definition of
\emph{sequential visibility}.
\end{proof}

\subsubsection{Composition}\label{subsubsec:comp_seq_caus}
We now focus on composition of sequential causal strategies.

We introduce some terminology. If $A, B$ and $C$
are $-$-arenas and $\bsigma : A \vdash B$, $\btau : B \vdash C$ are
causal (pre)strategies, then we define $\NAlt(\btau \inter \bsigma)$ as
in Definition \ref{def:nalt_play}, referring to polarities of events in
$\{-, \labl, \labr\}$. If $u \in \NAlt(\btau \inter \bsigma)$, we write
$u_\bsigma \in \NAlt(\bsigma)$, $u_\btau \in \NAlt(\btau)$, and $u_\odot
\in \NAlt(\btau \odot \bsigma)$ for the obvious restrictions. If a play
is alternating, it is \textbf{in state $O$} if it has
even-length, and \textbf{in state $P$} if it has odd length.

Many properties of causal sequentiality follow from the following
crucial observation:

\begin{lem}\label{lem:state-diagram}
Consider $\bsigma : A \vdash B$ and $\btau : B \vdash C$ sequential
causal strategies. 

Then, for any $u \in \NAlt(\btau \inter \bsigma)$ such that $u_\odot \in
\Alt(\btau \odot \bsigma)$, we have $u_\bsigma \in \Alt(\bsigma)$
and $u_\btau \in \Alt(\btau)$, and we are in one of the
following three cases:
\[
\begin{array}{ll}
\text{\emph{(1)}}&
\text{$u_\bsigma, u_\btau, u_\odot$ are respectively in state $O, O,
O$,}\\
\text{\emph{(2)}}&
\text{$u_\bsigma, u_\btau, u_\odot$ are respectively in state $O, P,
P$,}\\
\text{\emph{(3)}}&
\text{$u_\bsigma, u_\btau, u_\odot$ are respectively in state $P, O,
P$.}
\end{array}
\]
\end{lem}
\begin{proof}
By induction on $u$. If $u$ is empty, this is clear.
Consider $u m \in \NAlt(\btau \inter \bsigma)$.
By induction hypothesis $u$ is in one of cases \emph{(1)},
\emph{(2)} and \emph{(3)}. We distinguish cases:

\emph{(1)} Seeking a contradiction, assume $m$
occurs in $B$. Then, one of $m_\bsigma$ or $m_\btau$ is
positive -- say \emph{w.l.o.g.} the former. By induction
hypothesis, $u_\bsigma$ is alternating in state $O$, so ends
with a Player move. But so, $u_\bsigma m_\bsigma^+ \in \NAlt(\bsigma)$
with $u_\bsigma \in \Alt(\bsigma)$, so $u_\bsigma m_\bsigma \in
\Alt(\bsigma)$ since $\bsigma$ satisfies reachable sequentiality,
contradiction. So, $m$ occurs in $A$ or $C$ -- assume
\emph{w.l.o.g.} in $A$. Since $u_\odot$ is in state $O$, 
$m$ is negative -- then, it is direct that $u m$
satisfies \emph{(3)}.

\emph{(2)} First assume that $m$ occurs in $A$. Since
$u_\odot \in \Alt(\btau \inter \bsigma)$ is in state $P$, then $m$ is
positive; then $m_\bsigma$ is positive, contradicting reachable
sequentiality of $\bsigma$ with the fact that $u_\bsigma$ is in
state $O$. Similarly, if $m$ occurs in $B$ it has polarity
$\labr$ and we transition to \emph{(3)}, and if $m$ occurs in $C$ it has
polarity $\labr$ and we transition to \emph{(1)}. \emph{(3)}
Symmetric to \emph{(2)}.
\end{proof}

In other words, as long as the external Opponent respects the
alternation discipline, interactions follow the familiar \emph{state
diagram} of interactions in alternating game semantics, shown in Figure
\ref{fig:state-diagram}. 
\begin{figure}
\begin{center}
\includegraphics[scale=.5]{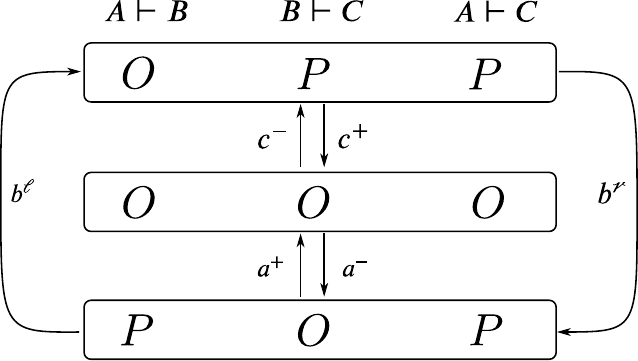}
\end{center}
\caption{State diagram for sequential interactions}
\label{fig:state-diagram}
\end{figure}
None of the interacting agents can be the first to break alternation, so
the interaction ends up fully alternating. It follows that
$\btau \odot \bsigma$ satisfies reachable sequentiality:

\begin{lem}\label{lem:comp_reach_seq}
Consider $\bsigma : A \vdash B$, $\btau : B \vdash C$ sequential causal
(pre)strategies. 

Then, $\btau \odot \bsigma$ satisfies \emph{reachable sequentiality}.
\end{lem}
\begin{proof}
Consider $tn^+ \in \NAlt(\btau \odot \bsigma)$ s.t. $t \in
\Alt(\btau \odot \bsigma)$. We have either $n_\bsigma$ or $n_\btau$
defined and positive, say the former \emph{w.l.o.g.}. Let us complete $t
n$ to $u n \in \NAlt(\btau \inter \bsigma)$. We have $u_\odot = t \in
\Alt(\btau \odot \bsigma)$, so by may distinduish along the three cases
of Lemma \ref{lem:state-diagram}:

\emph{(1)} If $u_\bsigma, u_\btau, u_\odot$ are in state $O, O, O$. By
hypothesis, $u_\bsigma n^+_\bsigma \in \NAlt(\bsigma)$. Since
$\bsigma$ satisfies \emph{reachable sequentiality}, $u_\bsigma n_\bsigma
\in \Alt(\bsigma)$ as well, contradicting $u_\bsigma$ in state $O$.

\emph{(2)} If $u_\bsigma, u_\btau, u_\odot$ are in state $O, P, P$, 
as in \emph{(1)} this contradicts \emph{reachable sequentiality}.

\emph{(3)} If $u_\bsigma, u_\btau, u_\odot$ are in state $P, O, P$.
With $u_\odot = t$ in state $P$, $t n \in \Alt(\btau \odot \bsigma)$.
\end{proof}

It also follows that $\btau \odot \bsigma$ satisfies \emph{sequential
determinism}: Lemma \ref{lem:state-diagram} expresses that in an
alternatingly reachable interaction, only one agent has control at any
point. So any alternatingly reachable non-deterministic choice in $\btau
\odot \bsigma$ can be attributed to $\bsigma$ and $\btau$: 

\begin{lem}\label{lem:comp_seq_det}
Consider $\bsigma : A \vdash B$ and $\btau : B \vdash C$ sequential
causal (pre)strategies.

Then, $\btau \odot \bsigma$ satisfies \emph{sequential determinism}.
\end{lem}
\begin{proof}
Consider $t n_1^+, t n_2^+ \in \Alt(\btau\odot \bsigma)$ completed to
$u_1 n_1, u_2 n_2 \in \NAlt(\btau \inter \bsigma)$, and $u'$
the greatest common prefix of $u_1 n_1$ and $u_2 n_2$, with $u' m_1
\prefix u_1 n_1$ and $u' m_2 \prefix u_2 n_2$. Necessarily, the visible
restriction of $u'$ is a prefix $t' \prefix t$, so in particular $t' \in
\Alt(\btau \odot \bsigma)$.

We distinguish cases on Lemma \ref{lem:state-diagram} applied to $u'$.
For \emph{(1)}, $m_1$ and $m_2$ must both be the next
negative event appearing in $t$, so $m_1 = m_2$, contradiction. For
\emph{(2)}, $m_1$ and $m_2$ both have polarity $\labr$ and we
have $u'_\btau \in \Alt(\btau)$, $u'_\btau (m_1)_\btau^+, u'_\btau
(m_2)_\btau^+ \in \Alt(\btau)$. Hence $(m_1)_\btau = (m_2)_\btau$ since
$\btau$ satisfies sequential determinism -- so $m_1 = m_2$
by local injectivity of the projection $\Pi_\btau$; contradiction.
Finally, for \emph{(3)}, then the reasoning is symmetric.
\end{proof}

To conclude compositionality, we link
with composition of the alternating projections:

\begin{lem}\label{lem:alt_proj_comp}
Consider $\bsigma : A \vdash B$ and $\btau : B \vdash C$ sequential
causal (pre)strategies.

Then, $\AltStrat(\btau \odot \bsigma) = \AltStrat(\btau) \odot
\AltStrat(\bsigma)$.
\end{lem}
\begin{proof}
$\supseteq.$ If $s \in \AltStrat(\btau) \odot \AltStrat(\bsigma)$, then
it is in $\NAltStrat(\btau) \odot \NAltStrat(\bsigma)$, thus in
$\NAltStrat(\btau \odot \bsigma)$ by Proposition \ref{prop:nalt_fonc}.
As $s$ is alternating, we also have $s \in \AltStrat(\btau\odot
\bsigma)$.

$\subseteq.$ If $s \in \AltStrat(\btau\odot \bsigma)$, there is $t \in
\Alt(\btau \odot \bsigma)$ s.t. $s = \pr_{\Lambda(\btau \odot
\bsigma)}(t)$ completed to $v \in \NAlt(\btau \inter \bsigma)$.
We display $v$ to $u \in \NAltStrat(\btau) \inter
\NAltStrat(\bsigma)$ s.t. $u \restrict A, C = s$.  But then, by
Lemma \ref{lem:state-diagram}, $v_\bsigma \in \Alt(\bsigma)$ and
$v_\btau \in \Alt(\btau)$, so $u \restrict A, B$ and $u \restrict B, C$
are actually alternating, and $u \in \AltStrat(\btau) \inter
\AltStrat(\bsigma)$. Thus, $s \in \AltStrat(\btau) \odot
\AltStrat(\bsigma)$.
\end{proof}

\begin{prop}
There is a category $\CGseq$ of $-$-arenas and sequential strategies.

Moreover, there is a functor
$\AltStrat(-) : \CGseq \to \NegAltvis$, preserving $\simstrat$.
\end{prop}
\begin{proof}
\emph{Category.} It is straightforward that copycat is sequential.
Consider $\bsigma : A \vdash B$ and $\btau : B \vdash C$ sequential
causal strategies. By Lemma \ref{lem:comp_reach_seq}, $\btau \odot
\bsigma$ satisfies \emph{reachable sequentiality}. By Lemma
\ref{lem:comp_seq_det}, it satisfies \emph{sequential determinism}. By
Lemma \ref{lem:alt_proj_comp} and preservation of P-visible alternating
strategies under composition, it satisfies \emph{sequential visibility}.

\emph{Functorial projection.} Preservation of copycat is a direct
verification. Preservation of composition is Lemma
\ref{lem:alt_proj_comp}. Preservation of symmetry is Lemma
\ref{lem:alt_proj_pres_sym}.
\end{proof}

\subsubsection{Seely category}\label{subsubsec:seq_seely}
We now show that sequential causal strategies form a Seely category.

We use the
following notations, for $\bsigma : A \vdash B$ and $\btau : C \vdash
D$: if $s \in \NAlt(\bsigma \tensor \btau)$, then $s_\bsigma \in
\NAlt(\bsigma)$ and $s_\btau \in \NAlt(\btau)$ are the corresponding
restrictions, defined in the obvious way.
Stability of sequentiality under tensor uses a state
analysis as in Lemma \ref{lem:state-diagram}:

\begin{lem}
Consider $\bsigma : A \vdash B$ and $\btau : C \vdash D$ sequential
(pre)strategies.
For any $s \in \Alt(\bsigma \tensor \btau)$ then $s_\bsigma \in
\Alt(\bsigma)$ and $s_\btau \in \Alt(\btau)$. Moreover, we are in one
of:
\[
\begin{array}{ll}
\text{\emph{(1)}} &
\text{$s_\bsigma, s_\btau, s$ are respectively in state $O,O,O$,}\\
\text{\emph{(2)}} &
\text{$s_\bsigma, s_\btau, s$ are respectively in state $O, P, P$,}\\
\text{\emph{(3)}} & 
\text{$s_\bsigma, s_\btau, s$ are respectively in state $P, O, P$.}
\end{array}
\]
\end{lem}
\begin{proof}
Straightforward by induction on $s$, using reachable sequentiality of
$\bsigma$ and $\btau$.
\end{proof}

Again, this is the familiar state diagram for alternating plays on a
tensor of strategies, see \emph{e.g.} \cite{harmer2004innocent}.
Finally, there is a state diagram for the functorial
action of the exponential -- as for tensor, if $s \in \NAlt(\oc
\bsigma)$, we write $s_{\grey{i}}$ for its restriction on copy index
$\grey{i}$.

\begin{lem}
Consider $\bsigma : A \vdash B$ be a sequential (pre)strategy.

For any $s \in \Alt(\oc \bsigma)$, then for any $\grey{i} \in
\mathbb{N}$, $s_{\grey{i}} \in \Alt(\bsigma)$; and we are in one of:
\[
\begin{array}{ll}
\text{\emph{(1)}} & 
\text{$s$ has state $O$, and for all $\grey{i} \in \mathbb{N}$,
$s_{\grey{i}}$ has state $O$,}\\
\text{\emph{(2)}} &
\text{$s$ has state $P$, and there exists a unique $\grey{i} \in
\mathbb{N}$ such that $s_{\grey{i}}$ has state $P$.}
\end{array}
\]
\end{lem}
\begin{proof}
Straightforward by induction on $s$, using reachable sequentiality of
$\bsigma$.
\end{proof}

As for composition, the preservation of \emph{reachable sequentiality}
and \emph{sequential determinism} under composition are immediate
applications. We omit the easy verifications that the alternating
projection is compatible with tensor and bang; from which -- as for
composition -- it follows that sequential visibility is preserved. All
structural strategies involved in the Seely category structure, being
variants of copycat, are easily proved sequential. Overall, we have: 

\begin{prop}\label{prop:func_seq}
There is a Seely category $\CGseq$ of $-$-arenas and sequential 
strategies.
Moreover,
$\AltStrat(-) : \CGseq \to \NegAltvis$ preserves $\simstrat$ and the
Seely structure.
\end{prop}

\subsection{Full Abstraction for $\IA$} Next, we fine-tune $\CGseq$ to
get full abstraction for $\IA$.

\subsubsection{Interpretation of $\IA$} It only remains to prove that
the interpretation of the primitives of $\IA$, \emph{i.e.} all
primitives of $\IPA$ except for the parallel let, are sequential. We
have:

\begin{lem}
The strategies $\mathbf{seq}, \mathbf{succ}, \mathbf{if},
\mathbf{iszero}, \mathbf{let}, \mathbf{assign}, \mathbf{deref},
\mathbf{grab}, \mathbf{release}$ and the prestrategies
$\bcell_n$ and $\bsem_n$ are sequential.
Moreover, for each of those (pre)strategies $\bsigma$,
$\AltStrat(\bsigma)$ is the corresponding alternating strategy from
Sections~\ref{subsubsec:intr_pcf} and \ref{subsubsec:newrs}.
\end{lem}
\begin{proof}
Routine verification.
\end{proof}

It follows that we have an adequate interpretation of $\IA$ as
sequential strategies, and:

\begin{prop}\label{prop:alt_pres_intr}
For any $\Gamma \vdash M : A$ in $\IA$, we have
$\intr{M}_{\NegAltvis_\oc} = \Alt(\intr{M}_{\CGseq_\oc})$.
\end{prop}
\begin{proof}
Straightforward by induction on the derivation $\Gamma \vdash M : A$.
\end{proof}

The two adequate models of $\IA$, $\NegAltvis_\oc$ and
$\CGseq_\oc$, differ in crucial ways: $\CGseq_\oc$ is
much more expressive, and records intensional \emph{causal}
information. Secondly, in $\CGseq_\oc$ one can also
read the behaviour of the program under contexts outside
$\IA$.

However, $\CGseq_\oc$ is not fully abstract for $\IA$ -- we need
to deal with well-bracketing.

\subsubsection{Well-bracketing} \label{sec:wb_ia}
Ideally, we would hope for an interpretation-preserving functor
\[
\Alt(-) : \CGwbseq_\oc \to \NegAltviswb_\oc\,,
\]
\begin{figure}
\[
\xymatrix@R=-7pt@C=5pt{
(\tunit&\to&\tunit)&\to&(\tunit& \to& \tunit) &\to& \tunit\\
&&&&&&&&\qu^-\\
&&\qu^+ \ar@{.}@/^/[urrrrrr]\\
\qu^-   \ar@{.}@/^/[urr]\\
&&&&&&\qu^+
        \ar@{.}@/^/[uuurr]\\
&&&&\qu^-
        \ar@{.}@/^/[urr]\\
\done^+ \ar@{.}@/^/[uuu]
}
\]
\caption{Definition \ref{def:nalt_play_wb} is weaker than Definition
\ref{def:alt_play_wb}}
\label{fig:count_ex_alt_wb}
\end{figure}
but Definition \ref{def:nalt_play_wb} (for non-alternating
plays) is weaker than Definition \ref{def:alt_play_wb} (for
alternating plays) when applied on alternating plays, as illustrated in
Figure \ref{fig:count_ex_alt_wb}.
While Definition \ref{def:alt_play_wb} closely follows
the operational idea that calls and returns are handled by
a single stack, Definition \ref{def:nalt_play_wb}
restricts the hierarchical relationship between calls and returns.

Fortunately,
from a distinguishing test in $\CGwbseq$ one can extract a
characteristic \emph{complete} (see Section~\ref{subsubsec:fa_ia})
alternating play, which -- as we shall see -- is well-bracketed
as in Definition \ref{def:alt_play_wb}.  This is due to the following
lemma, a well-known observation:

\begin{lem}\label{lem:seg_pointer}
Let $s \in \Alt(A)$ be P- and O-visible. Assume that $s$ has the form
\[
s = \xymatrix@C=5pt{\dots & s_i &\dots & s_j\ar@{.}@/_/[ll] & \dots}
\]
where no further move points to $s_j$. Then, no move after $s_j$ can
point within
$s_i \dots s_j$.
\end{lem}
\begin{proof}
By P- or O-visibility, $s_{j+1}$ points strictly before $s_i$. 
Then no view can ever see $s_{i+1} \dots s_j$ -- so no move can point
there. Besides, $s_i$ can only be seen by the
player responsible for it, so no move can point to $s_i$. A
proof appears in \cite[Lemma 5]{DBLP:journals/apal/ClairambaultH10}.
\end{proof}

From this, we may easily deduce the following:

\begin{lem}\label{lem:complete_wb}
Consider $s \in \Alt(A)$ P- and O-visible where every question is
answered.

Then, it is well-bracketed in the sense of Definition
\ref{def:alt_play_wb}, and thus it is complete.
\end{lem}
\begin{proof}
Consider $s \in \Alt(A)$ complete but with a well-bracketing failure,
\emph{i.e.} as in:
\[
s = \xymatrix@C=5pt{\dots & \qu_1^\Qu& \dots &\qu_2^\Qu
&
\dots&
a^\An   \ar@{.}@/_1pc/[llll]& \dots}
\]
with $\qu_2$ unanswered when playing $a$. By \emph{answer-closing}, no
further move can point to $a$. Thus by Lemma \ref{lem:seg_pointer}, no
further move can point to $\qu_2$, contradicting that $\qu_2$ is
answered.
\end{proof}

The play of Figure \ref{fig:count_ex_alt_wb} is well-bracketed as in
Definition \ref{def:nalt_play_wb}, but it cannot be extended to a
complete play: the two questions covered will not be adressed ever
again.

\subsubsection{Full abstraction} From that, we may finally conclude:

\begin{thm}\label{th:naltwbseq_fa_ia}
The model $\CGwbseq$ is intensionally fully abstract for $\IA$.
\end{thm}
\begin{proof}
Let $\vdash M, N : A$ be terms in $\IA$, and assume that $\intr{M} \not
\obs \intr{N}$, \emph{i.e.} there is a test $\balpha \in
\CGwbseq_\oc(\intr{A}, \intr{\tunit})$ such that $\balpha \odot_\oc
\intr{M} \neq \balpha \odot_\oc \intr{N}$ -- assume \emph{w.l.o.g.} that
$\balpha \odot_\oc \intr{M}$ converges while $\balpha \odot_\oc \intr{M}$
diverges. Writing $\alpha' = \AltStrat(\alpha)$, it follows that
\[
\alpha' \odot_\oc \intr{M}_{\NegAltviswb_\oc}\eval
\qquad
\qquad
\alpha' \odot_\oc \intr{N}_{\NegAltviswb_\oc}\div\,,
\]
but $\alpha'$ may not be well-bracketed as in Definition \ref{def:alt_strat_wb}.
Consider $s \in \alpha'$ involved in $\alpha' \odot_\oc \intr{M} \eval$
-- until the rest of the proof, $\intr{M}$ is the
interpretation in $\NegAltviswb_\oc$. The initial question of $s$ has an
answer, thus as $\intr{M}$ and $\alpha'$ are well-bracketed in the sense
of \ref{def:nalt_strat_wb}, all its questions are answered. It is
P-visible and O-visible since both $\intr{M}$ and $\alpha'$ are
P-visible. Hence, it is complete by Lemma \ref{lem:complete_wb} (in
particular well-bracketed as in Definition \ref{def:alt_strat_wb}).

Consider $\alpha'$ restricted to (plays symmetric to) prefixes of
$s$. Now $\alpha'$ is well-bracketed as in Definition
\ref{def:alt_play_wb}, and it distinguishes $\intr{M}$ and $\intr{N}$,
hence $M \not \obs N$ by Theorem \ref{th:fa_ia}.
\end{proof}

\subsection{Sequential Innocence} \label{subsec:seq_inn_pcf}
A causal strategy $\bsigma : A$ is \textbf{sequential
innocent} if it is both \emph{sequential} and \emph{parallel innocent};
those form a Seely category $\CGwbseqinn$.

We already know that $\CGwbseqinn$ supports an adequate
interpretation of $\PCF$. For (intensional) full abstraction, we shall
prove that $\AltStrat(-)$ sends sequential innocent strategies to
innocent alternating strategies as in Definition \ref{def:innocence},
and rely on Theorem \ref{th:fa_pcf}. 

\subsubsection{Causal analysis of sequential innocence} First, we shall
see that sequential innocent causal strategies are really
representations of the \emph{P-view forests} of Section~\ref{subsec:vis_inn}.

\begin{lem}\label{lem:seqinn_forest}
Consider $\bsigma : A$ a
\emph{sequential}, \emph{parallel innocent} causal strategy.

Then, $\bsigma$ is an O-branching alternating forest.
\end{lem}
\begin{proof}
First, we prove that for all $m \in \ev{\bsigma}$, its set of
dependencies $[m]_\bsigma$ is a total order.

Seeking a contradiction, take
$m' \in \ev{\bsigma}$ minimal with $m' \imc_\bsigma m_1$
and $m' \imc_\bsigma m_2$ distinct, all within $[m]_\bsigma$.
By minimality, $[m']_\bsigma$ is a total order, \emph{i.e.} a gcc.  By
Lemma \ref{lem:app_caus_alt}, $m_1$ and $m_2$ have the same
polarity, opposite of $m'$. Consider $\rho_1 \in \gcc(\bsigma)$
a gcc for $m$ passing through $m' \imc_\bsigma m_1$, and
$\rho_2 \in \gcc(\bsigma)$ a gcc for $m$ passing through $m'
\imc_\bsigma m_2$. Then $\rho_1$ and $\rho_2$ have least distinct events 
$m_1$ and $m_2$; hence by pre-innocence $m_1$ and $m_2$ are positive.

Now, $m'$ is the only immediate dependency of $m_1$ and
$m_2$; indeed if there was $m'' \imc_\bsigma m_i$, then considering
$\rho' \imc m_i \in \gcc(\bsigma)$ passing through $m''$,
$\rho$ and $\rho'$ would fork at some event smaller than $m'$,
contradicting its minimality. Hence, $[m'] \cup \{m_i\} \in
\conf{\bsigma}$ for $i \in \{1, 2\}$.

Also writing $[m']$ for the play in $\Alt(\bsigma)$ with events in the
same order, we have
\[
[m'], ~[m'] m_1,~[m'] m_2 ~ \in ~ \NAlt(\bsigma)\,,
\]
but by Lemma \ref{lem:app_caus_alt}, $[m'] m_1^+$ and $[m']
m_2^+$ are alternating. By \emph{sequential determinism} of $\bsigma$,
it follows that $m_1 = m_2$, contradiction. So, for all $m \in
\ev{\bsigma}$, $[m]_\bsigma$ is a total order.

Thus $(\ev{\bsigma}, \leq_\bsigma)$ is a forest. Likewise, if $m^- \imc
m_1^+$ and $m^- \imc m_2^+$ in $\bsigma$, by
\emph{sequential determinism} and the same reasoning as above, $m_1
= m_2$, so $\bsigma$ is $O$-branching. Finally, as for any causal
strategy $\bsigma : A$ on $A$ alternating, we have
$\imc_\bsigma$ is alternating as well.
\end{proof}

Let us call a \textbf{branch} of sequential innocent $\bsigma :
A$ a $s = m_1 \dots m_n \in \Alt(\bsigma)$ s.t.
\[
m_1 \imc_\bsigma \dots \imc_\bsigma m_n \in \gcc(\bsigma)\,.
\]

Then, $\pr_\bsigma(s) \in \Alt(A)$, but there is more: by courtesy of
$\bsigma$, if $m_i^+ \imc_\bsigma m_{i+1}^-$ then $\pr_\bsigma(m_i)
\imc_A \pr_\bsigma(m_{i+1})$, \emph{i.e.} $\pr_\bsigma(m_{i+1})$
\emph{points to} $\pr_\bsigma(m_i)$ in $\pr_\bsigma(s)$. In other words,
$\pr_\bsigma(s)$ is actually a \emph{P-view}, \emph{i.e.} an alternating
play where Opponent always points to the previous move. One can
\emph{display} a sequential innocent causal strategy to a forest of
P-views as in Section~\ref{subsubsec:def_inn}: we have recovered, as the
causal structure of sequential innocent causal strategies, the
\emph{forest of P-views} $\pviews{\sigma}$, the ``causal presentation''
of alternating innocent strategies from
Proposition \ref{prop:causal_inn}. This connection confirms that the
\emph{causal strategies} of Section~\ref{sec:par_inn} are
generalizations of the sets of \emph{P-views} of traditional
innocence\footnote{This shows that the causal reasoning permitted by
traditional innocence, one of the main tools of traditional
game semantics, is \emph{not} inherently restricted to
innocence. This is a powerful observation, and much of the subsequent
line of work in concurrent games has consisted in exploring its
implications.}.
\begin{figure}
\begin{minipage}{.46\linewidth}
\[
\scalebox{.9}{$
\xymatrix@R=4pt@C=4pt{
&&~^\red{1}\qu^-
        \ar@{-|>}[d]\\
&&~^\red{2}\qu_{\grey{0}}^+
        \ar@{-|>}[dll]
        \ar@{-|>}[d]
        \ar@{-|>}[drr]
        \ar@{.}@/_.2pc/[u]
\\
~^\red{3}\qu_{\grey{0,0}}^-
        \ar@{-|>}[d]
        \ar@{.}@/^.2pc/[urr]
&&
~^\red{5}\qu_{\grey{0,1}}^-
        \ar@{.}@/_.2pc/[u]
        \ar@{-|>}[d]&&
~^\red{7}\done_{\grey{0}}^-
        \ar@{-|>}[d]
        \ar@{.}@/_.2pc/[ull]
\\
~^\red{4}\done_{\grey{0,0}}^+
        \ar@{.}@/^.2pc/[u]&&
~^\red{6}\done_{\grey{0,1}}^+
        \ar@{.}@/_.2pc/[u]&&
~^\red{8}\qu^+_{\grey{1}}
        \ar@{-|>}[dl]
        \ar@{-|>}[d]
        \ar@{-|>}[dr]
        \ar@{.}@/_1pc/[uuull]\\
&&&~^\red{9}\qu_{\grey{1,0}}^-
        \ar@{-|>}[d]
        \ar@{.}@/^.2pc/[ur]&
~^\red{11}\qu_{\grey{1,1}}^-
        \ar@{-|>}[d]
        \ar@{.}@/_.2pc/[u]&
~^\red{13}\done_{\grey{1,1}}^-
        \ar@{-|>}[d]
        \ar@{.}@/_.2pc/[ul]\\
&&&~^\red{10}\done_{\grey{1,0}}^+
        \ar@{.}@/_.2pc/[u]&
~^\red{12}\done_{\grey{1,1}}^+
        \ar@{.}@/_.2pc/[u]&
~^\red{14}\ttrue^+\ar@{.}@/_5.5pc/[uuuuulll]
}$}
\]
\caption{Innocence witness for Figure \ref{fig:plays_pointers}}
\label{fig:wit_caus_inn}
\end{minipage}
\hfill
\begin{minipage}{.45\linewidth}
\[
\scalebox{.9}{$
\xymatrix@R=8pt@C=4pt{
~\\
&&&\qu^-
        \ar@{.}[dll]
        \ar@{.}[d]
        \ar@{.}[drr]
        \ar@{-|>}@/_.1pc/[dll]\\
&\qu^+_{\grey{0}}       
        \ar@{.}[dl]
        \ar@{.}[d]
        \ar@{.}[dr]
        \ar@{-|>}@/_.1pc/[dl]
        \ar@{-|>}@/^.1pc/[d]
        \ar@{-|>}@/^.1pc/[dr]
&&\ttrue^+&&\qu^+_{\grey{1}}
        \ar@{.}[dl]
        \ar@{.}[d]
        \ar@{.}[dr]
        \ar@{-|>}@/_.1pc/[dl]
        \ar@{-|>}@/^.1pc/[d]
        \ar@{-|>}@/^.1pc/[dr]\\
\qu^-_{\grey{{0,0}}}
        \ar@{.}[d]
        \ar@{-|>}@/^.1pc/[d]
&\qu^-_{\grey{0,1}}
        \ar@{.}[d]
        \ar@{-|>}@/^.1pc/[d]
&\done^-_{\grey{0}}
        \ar@{-|>}[urrr]
        &&\qu^-_{\grey{1,0}}
        \ar@{-|>}@/^.1pc/[d]
        \ar@{.}[d]
&\qu^-_{\grey{1,1}}
        \ar@{-|>}@/^.1pc/[d]
        \ar@{.}[d]&\done^-_{\grey{1}}
        \ar@{-|>}[ulll]\\
\done^+_{\grey{0,0}}&\done^+_{\grey{0,1}}&&&
\done^+_{\grey{1,0}}&\done^+_{\grey{1,1}}\\~\\
}$}
\]
\caption{An alternative presentation}
\label{fig:wit_caus_inn_alt}
\end{minipage}
\end{figure}

\subsubsection{P-views are the causal dependency}
This gives two ways to get \emph{plays} from $\bsigma
: A$ sequential innocent: following Section~\ref{subsubsec:def_inn}, by
selecting those plays whose P-views appear in the causal representation;
and following Section~\ref{sec:par_inn}, as displays of alternating
plays over $\bsigma$ -- we must prove them identical.
Figure \ref{fig:wit_caus_inn} shows the
augmentation explored in the play of
\[
\intr{\lambda f^{\tunit \to \tunit}.\,f\,\tskip;\,f\,\tskip;\,\ttrue} :
\intr{(\tunit
\to \tunit) \to \tbool}
\]
in Figure \ref{fig:plays_pointers} -- the numbers in red
correspond to the order in which that augmentation is explored. The
reader should take some time to digest this picture, and in particular
observe that for each prefix of the play in Figure
\ref{fig:plays_pointers}, the P-view is exactly (the display of) the
branch leading to the corresponding move in the configuration. Opponent
could explore the same configuration in a different order, corresponding
to a different play with the same P-views. On the other hand, only
Opponent has any degree of freedom in this exploration: Player has ever
at most one possible move, that immediately caused by the last Opponent
move\footnote{Different explorations of the same augmentation may be
related by permuting contiguous OP pairs of moves. Deterministic
innocent strategies may be defined as those stable under the
permutations of OP pairs permitted by the arena: this is the idea behind
Melli\`es' presentation of innocence
\cite{DBLP:conf/concur/Mellies04}.}. 

As an aside, in Figure \ref{fig:wit_caus_inn_alt}, we give an
alternative presentation of the same augmentation, making explicit how
an augmentation of a sequential innocent strategy consists is the
underlying configuration (here, Figure \ref{fig:exp_conf}), enriched
with immediate causal links. The set of (isomorphism classes of)
configurations reached is essentially the information recorded by the
relational model; so this presentation shows plainly that innocent game
semantics consist in enriching the relational model with explicit causal
/ temporal information.

Back to the technical development, we must prove that if $\bsigma
: A$ is sequential innocent, then $\AltStrat(\bsigma)$ is innocent as in
Definition \ref{def:innocence}. This relies on a link between
P-views and the causal structure of $\bsigma : A$, which should be
expected in the light of Figure \ref{fig:wit_caus_inn}.

\begin{lem}\label{lem:pview_dep}
Consider $A, B$ $-$-arenas and $\bsigma : A \vdash B$ sequential
innocent.

Then, for any $t m \in \Alt(\bsigma)$, we have
$\pview{\pr_{\Lambda(\bsigma)}(t m)} =
\pr_{\Lambda(\bsigma)}([m]_\bsigma)$.
\end{lem}
\begin{proof}
In the statement above, we treat $[m]_\bsigma$ as the sequence
induced by its total ordering.

The crucial observation is that is $t m^- n^+ \in \Alt(\bsigma)$, then
necessarily $m^- \imc_\bsigma n^+$.
To prove that, we prove by induction on $t$ that for any $t \in
\Alt(\bsigma)$:
\emph{(1)} if $t$ has even length, then all maximal events of $\ev{t}
\in \conf{\bsigma}$ are positive; and \emph{(2)} if $t$ has odd length,
then $\ev{t} \in \conf{\bsigma}$ has \emph{exactly one} maximal negative
event. Indeed, for $t m^- \in \Alt(\bsigma)$, then $t$ has even length,
so $\ev{t}$ has all its maximal events positive. But then $\ev{t m^-}$
has exactly one maximal negative event, namely $m^-$. Likewise, for $t
m^+ \in \Alt(\bsigma)$, then $\ev{t}$ has exactly one maximal negative
event. Now, the immediate predecessor of $m$ must be negative. But if it
is not maximal in $\ev{t}$, this contradicts Lemma
\ref{lem:seqinn_forest}, and in particular the fact that $\bsigma$ is
$O$-branching. Therefore, the predecessor of $m$ must be the unique
maximal negative event of $\ev{t}$, and $\ev{t m}$ has all maximal
events positive as required.
Now, if $t m^- n^+$, then $\ev{t m^-}$ has exactly one maximal negative
event (namely $m^-$); while the maximal events of $\ev{t m^- n^+}$ are
all positive (and comprise $n^+$). Hence, $m^- \imc_\bsigma n^+$ as
required.
Likewise, if $t_1 m^+ t_2 n^- \in \Alt(\bsigma)$ s.t.
$\pr_{\Lambda(\bsigma)}(m) \imc_{A\lin B} \pr_{\Lambda(\bsigma)}(n)$ -- so
$\pview{\pr_{\Lambda(\bsigma)}(t_1 m t_2 n)} =
\pview{\pr_{\Lambda(\bsigma)}(t_1)}
\pr_{\Lambda(\bsigma)}(m) \pr_{\Lambda(\bsigma)}(n)$ then we must have
$m
\imc_\bsigma n$ by Lemma \ref{lem:app_aux_caus}. From these two
facts, the lemma is a direct verification by induction on $t$.
\end{proof}

\subsubsection{Innocent alternating unfolding}
From the above, we may now deduce:

\begin{prop}\label{prop:proj_pres_inn}
Consider $A, B$ $-$-arenas, and $\bsigma : A \vdash B$ a sequential
causal strategy.

If $\bsigma$ is parallel innocent, then $\AltStrat(\bsigma)$ is innocent
as
in Definition \ref{def:innocence}.
\end{prop}
\begin{proof}
First, we must show that $\AltStrat(\bsigma)$ is P-visible. In other
words, we must prove that for all $s \in \AltStrat(\bsigma)$, $\pview{s}
\in \Alt(A\lin B)$. For the empty play there is nothing to prove; so
consider $s a \in \AltStrat(\bsigma)$ and $t m \in \Alt(\bsigma)$ such
that $sa = \pr_{\Lambda(\bsigma)}(tm)$. Now, by Lemma
\ref{lem:pview_dep}, we have $\pview{sa} =
\pr_{\Lambda(\bsigma)}([m]_\bsigma)$, as plays -- therefore $\pview{sa}
\in \Alt(A\lin B)$ as required.

Now, we prove innocence. Let $s a^+, s' \in \AltStrat(\bsigma)$ such
that $\pview{s} = \pview{s'}$. By definition, there is $t m^+ \in
\Alt(\bsigma)$ and $t' \in \Alt(\bsigma)$ such that $s a =
\pr_{\Lambda(\bsigma)}(t m)$ and $s' = \pr_{\Lambda(\bsigma)}(t')$.

Now, by Lemma \ref{lem:pview_dep}, $\pview{s a} =
\pr_{\Lambda(\bsigma)}([m]_\bsigma)$ as plays, hence also
$\pview{s} =
\pr_{\Lambda(\bsigma)}([n]_\bsigma)$ for $n \imc_\bsigma m$. Again by
Lemma \ref{lem:pview_dep}, $\pview{s'} =
\pr_{\Lambda(\bsigma)}([n']_\bsigma)$ for
$n'$ the last move of $t'$. Since $\pview{s} = \pview{s'}$, by Lemma
\ref{lem:uniq_caus_wit}, $[n]_\bsigma = [n']_\bsigma$. So,
$\ev{t'} \cup \{m\}$ is down-closed. Finally, $\ev{t'} \cup \{m\}$ is
negatively compatible since $\ev{t'} \in \conf{\bsigma}$ and $m$ is
positive, hence $\ev{t'} \cup \{m\}$ is compatible as $\bsigma$
satisfies \emph{causal determinism}. Therefore, $t' m \in
\Alt(\bsigma)$, and $s' m = \pr_{\Lambda(\bsigma)}(t' m) \in
\AltStrat(\bsigma)$.
\end{proof}

Altogether, we have proved the following proposition:

\begin{prop}
There is a Seely category $\CGseqinn$ of $-$-arenas and sequential
strategies. Moreover, the alternating unfolding preserves
$\simstrat$ and the Seely category structure:
\[
\AltStrat(-) : \CGseqinn \to \NegAltinn\,.
\]
\end{prop}
\begin{proof}
It suffices to prove that the functor of Proposition \ref{prop:func_seq}
sends sequential innocent causal strategies to innocent alternating
strategies, which we know by Proposition \ref{prop:proj_pres_inn}.
\end{proof}

We are now equipped to show full abstraction for $\PCF$.

\subsubsection{Full abstraction for $\PCF$}\label{subsubsec:caus_fa_pcf}
We show that $\CGwbseqinn$ is fully abstract for $\PCF$.

\begin{thm}
The model $\CGwbseqinn$ is intensionally fully abstract for $\PCF$.
\end{thm}
\begin{proof}
Let $\vdash M, N : A$ be terms in $\PCF$ s.t.
$\intr{M} \not \obs \intr{N}$, \emph{i.e.} there is $\balpha :
\oc \intr{A} \vdash \tunit$ sequential innocent and well-bracketed
such that, \emph{w.l.o.g.}, $\balpha \odot_\oc \intr{M} \eval$ and
$\balpha \odot_\oc \intr{N} \div$. Then, it follows 
\[
\AltStrat(\balpha) \odot_\oc \intr{M}_{\NegAlt_\oc} \eval\,,
\qquad
\qquad
\AltStrat(\balpha) \odot_\oc \intr{N}_{\NegAlt_\oc} \div\,,
\]
with $\AltStrat(\balpha)$ well-bracketed by Proposition
\ref{prop:inn_wb}. Hence, $M \not \obs N$ by Theorem \ref{th:fa_pcf}.
\end{proof}

\section{Finite Definability and Full Abstraction for \texorpdfstring{$\PCFpar$}{PCF//}}
\label{sec:definability}

We have now established the following intensional full abstraction
results:
\[  
\begin{array}{rcl}
\CGwb_\oc & \text{is fully abstract for} & \IPA\,,\\
\CGwb_\oc + \text{\emph{sequentiality}} &\text{is
fully abstract for} & \IA\,,\\
\CGwb_\oc + \text{\emph{parallel innocence}} +
\text{\emph{sequentiality}} &
\text{is fully
abstract for} & \PCF\,,
\end{array}
\]
and we are left with the one outstanding objective:
\[  
\begin{array}{rcl}
\CGwb_\oc + \text{\emph{parallel innocence}} & \text{is
fully abstract for} & \PCFpar\,.
\end{array}
\]

Unfortunately, this is also the most challenging of our full abstraction
results: whereas for the others we could leverage earlier
work, we must prove finite definability from scratch.

Proving finite definability for parallel innocent strategies involves
many steps. In Section~\ref{subsec:positional_collapse}, we introduce a more convenient
equivalence between parallel innocent strategies, \emph{positional
equivalence}. In Section~\ref{subsec:wb_pruning} we show it suffices
to consider tests that satisfy a stronger, \emph{causal}, form of
well-bracketing useful for definability. In Section~\ref{subsec:finite_tests} we introduce a notion of finiteness, and
show that finite tests suffice.
In Section~\ref{subsec:factor}, we show a factorization result, reducing
finite definability to that for first-order strategies.  In Section~\ref{subsec:fin_def}, we conclude the proof of finite definability.
Finally, in Section~\ref{subsec:fa_pcfpar}, we prove intensional full
abstraction for $\PCFpar$, concluding the technical contents of the
paper.

\subsection{The Positional Collapse} \label{subsec:positional_collapse}
Definability will hold only with respect to \emph{positional
equivalence}, a congruence amounting to an equal projection in the
relational model. 

\subsubsection{Positions of arenas} We will observe strategies only on
certain \emph{positions} of arenas.

\begin{defi}
Let $A$ be an arena, and $x \in \conf{A}$.

We say that $x$ is \textbf{complete} iff every question in $x$ has an
answer in $x$.
\end{defi}

\emph{Complete} configurations mirror the \emph{complete plays} of
Section~\ref{subsubsec:fa_ia} and onwards. In both cases,
all function calls have returned. The
difference, however, is that complete plays are sequential whereas
complete configurations are not: they are a ``static'' snapshot
presenting all calls and returns and their hierarchical relationships,
with no temporal information.

Besides, complete configurations 
also comprise the ad-hoc choice of copy indices for all replicable
moves. So we quotient them out via the following variation.

\begin{defi}
Let $A$ be an arena. The set of \textbf{positions} on $A$, ranged over
by $\x, \y,
\dots$, is:
\[
\coll A \quad=\quad \{x \in \conf{A} \mid \text{$x$
complete}\}/\sym_A\,.
\]

If $x \in \conf{A}$ is complete, we write $[x]_\sym \in \coll A$ for
its symmetry class.
\end{defi}

For instance, the configuration of Figure \ref{fig:exp_conf} is 
complete. The matching position is described intuitively by
removing the grey subscripts, \emph{i.e.} the copy indices.

\subsubsection{Positions of strategies} The \emph{positional collapse}
of a strategy is an explicit \emph{desequentialization}, obtained by
forgetting the chronological ordering of complete plays.

\begin{defi}\label{def:pos_strat}
Consider $A$ an arena, and $\bsigma : A$ a causal strategy.
The \textbf{positions} of $\bsigma$ are
\[
\coll \bsigma = \{\x \in \coll A \mid \exists x \in \conf{\bsigma},~\x =
[\pr_{\bsigma}(x)]_\sym\}\,.
\]
\end{defi}

For $\bsigma : A \vdash B$, positions are symmetry classes of parallel
compositions $x_A \parallel x_B$, also written $\x_A \parallel \x_B$ for
$\x_A = [x_A]_\sym$, $\x_B = [x_B]_\sym$. They correspond to pairs
$(\x_A, \x_B) \in (\coll A) \times (\coll B)$ -- so $\coll \bsigma$
gives a \emph{relation} from $\coll A$ to $\coll B$; accordingly
we also write $(\x_A, \x_B) \in \coll \bsigma$ for $\x_A \parallel
\x_B \in \coll \bsigma$.
\begin{figure}
\[
\int
\left(
\raisebox{45pt}{$
\scalebox{.8}{$
\xymatrix@R=2pt@C=5pt{
\tunit &\lin& \tunit &\lin& \tunit\\
&&&&\qu^-
        \ar@{-|>}[dllll]\\
\qu^+
        \ar@{.}@/^/[urrrr]
        \ar@{-|>}[d]\\
\done^-
        \ar@{-|>}[drr]\\
&&\qu^+
        \ar@{-|>}[d]
        \ar@{.}@/^/[uuurr]\\
&&\done^-
        \ar@{-|>}[drr]\\
&&&&\done^+
        \ar@{.}@/_/[uuuuu]
}$}$}\right)
=
\int
\left(
\raisebox{45pt}{$
\scalebox{.8}{$
\xymatrix@R=10pt@C=5pt{
\tunit &\lin& \tunit &\lin& \tunit\\
&&&&\qu^-
        \ar@{-|>}[dll]
        \ar@{-|>}[dllll]\\
\qu^+   \ar@{.}@/^/[urrrr]
        \ar@{-|>}[d]
&&\qu^+ \ar@{.}@/^/[urr]
        \ar@{-|>}[d]\\
\done^-
        \ar@{-|>}[drrrr]
        \ar@{.}@/^/[u]&&
\done^-
        \ar@{-|>}[drr]
        \ar@{.}@/^/[u]\\
&&&&\done^+
        \ar@{.}@/_/[uuu]
}$}$}\right)
=
\left\{
\raisebox{40pt}{$
\scalebox{.8}{$
\xymatrix@R=18pt@C=5pt{
\tunit &\lin& \tunit &\lin& \tunit\\
&&&&\qu^-\\
\qu^+   \ar@{.}@/^/[urrrr]
&&\qu^+ \ar@{.}@/^/[urr]&&
\done^+ \ar@{.}@/_/[u]
\\
\done^- \ar@{.}@/^/[u]&&
\done^- \ar@{.}@/^/[u]
}$}$}\right\}
\]
\caption{Quotienting out evaluation order}
\label{fig:quotient_order}
\end{figure}

We illustrate the construction in Figure \ref{fig:quotient_order}:
the example illustrates how, by only keeping complete positions, the
collapse forgets the evaluation order. We define:

\begin{defi}\label{def:pos_eq}
Two causal strategies $\bsigma, \btau : A$ are \textbf{positionally
equivalent} iff $\coll \bsigma = \coll \btau$.

We write $\bsigma \equiv \btau$ to denote the fact that $\bsigma$ and
$\btau$ are positionally equivalent.
\end{defi}

This is a drastic quotient, identifying sequential and parallel
evaluation.  It will help tremendously in our definability procedure,
which will not respect the evaluation order. Of course, the more drastic
the quotient, the more challenging the corresponding proof obligation
that it is preserved by operations on strategies -- and in particular by
composition. 

\subsubsection{Deadlocks} Stability of positional
equivalence by composition boils down to
\[
\coll(-): \CG \to \Rel
\]
being functorial, where $\Rel$ is the usual category of \emph{sets and
relations}.

For morphisms $\bsigma \in \CG(A, B)$, we have defined $\coll \bsigma
\in \Rel(\coll A, \coll B)$.  But for now, this
operation has no reason to preserve composition! In fact neither
inclusion holds: firstly, $\coll (\btau \odot \bsigma) \subseteq (\coll
\btau) \odot (\coll \bsigma)$ may fail as a complete configuration may
arise through an interaction involving a non-complete configuration on
$B$. We shall see later on that this can be salvaged by
restricting to \emph{well-bracketed} causal strategies, ensuring
that an interaction producing a complete configuration on $A \vdash C$
will only involve a complete configuration on $B$.

However the other direction also fails, and the diagnosis is more
serious.
\begin{figure}
\[
\left(
\raisebox{35pt}{$
\xymatrix@R=2pt@C=2pt{
\tunit & \lin &\tunit &\vdash &\tnat\\
&&&&\qu^-  \ar@{-|>}[dll]\\
&&\qu^+   \ar@{-|>}[dll]
        \ar@{-|>}[d]\\
\qu^-    \ar@{-|>}[d]
        \ar@{.}@/^/[urr]&&
\done^- \ar@{-|>}[dll]
        \ar@{-|>}[drr]
        \ar@{.}@/^/[u]\\
\done^+ \ar@{.}@/^/[u]&&&&
0^+\ar@{.}@/_/[uuu]
}$}
\right)
\quad
\odot
\quad
\left(
\raisebox{35pt}{$
\xymatrix@R=2pt@C=2pt{
\tunit & \lin & \tunit\\
&&\qu^-
        \ar@{-|>}[dll]\\
\qu^+   \ar@{.}@/^/[urr]
        \ar@{-|>}[d]\\
\done^- \ar@{-|>}[drr]
        \ar@{.}@/^/[u]\\
&&\done^+
        \ar@{.}@/_/[uuu]
}$}
\right)
\quad
=
\quad
\left(
\raisebox{35pt}{$
\xymatrix@R=2pt@C=2pt{
\tnat\\
\qu^-   
}$}
\right)
\]
\caption{Deadlocking composition of causal strategies}
\label{fig:deadlock}
\end{figure}
To construct $(\x_A, \x_C) \in (\coll \btau) \odot (\coll
\bsigma)$, one provides $\x_B \in \coll B$ mediating for
relational composition, so 
\[
(\x_A, \x_B) \in \coll \bsigma\,,
\qquad
\qquad
(\x_B, \x_C) \in \coll \btau\,,
\]
\emph{i.e.} with $x^\bsigma \in \conf{\bsigma}$ and
$x^\btau \in \conf{\btau}$ s.t. writing $\pr_\bsigma x^\bsigma =
x_A^\bsigma \parallel x^\btau_B$ and $\pr_\btau x^\btau = x^\btau_B
\parallel x^\btau_C$, $x^\bsigma_A \in \x_A$, $x^\btau_C \in
\x_C$, and $x^\bsigma_B, x^\btau_B \in \x_B$, so $x^\bsigma_B
\sym_B x^\btau_B$ match \emph{up to symmetry}. In other words, we must
provide a \emph{witness position} that both strategies agree (up
to symmetry) is \emph{reachable}.

On the other hand, composition of strategies is more rigid: not only
should the projections of $x^\bsigma$ and $x^\btau$ on $B$ match,
they should also arrive at this position \emph{in
the same chronological ordering}. This is not always possible: these two
notions of composition differ when interaction triggers a \emph{causal
deadlock}, \emph{i.e.} pairs of configurations that are \emph{matching}
but not \emph{secured} as in Definition \ref{def:caus_comp}.
Figure \ref{fig:deadlock} displays an example: the strategy
obtained by composition has no response to the initial Opponent move,
while relational composition authorizes $0^+$.

This strikes at the heart of the difference between game
and relational semantics: the former is dynamic hence
sensitive to deadlocks, while the latter is static.
This of course is what lets game semantics model languages
with non-commutative effects, but for us, very concretely, it means that
positional equivalence is in general not a congruence.

\subsubsection{The deadlock-free lemma} Our \emph{deus ex machina} is
visibility.  A powerful -- and at first unexpected -- consequence of
visibility is that any interaction between visible strategies is always
deadlock-free.  The consequence of visibility that our proof will
exploit repeatedly is:

\begin{lem}\label{lem:locality}
Consider $A, B$ $-$-arenas, $\bsigma : A \vdash B$  visible, and $m, m'
\in \ev{\bsigma}$ s.t. $m <_\bsigma m'$.

Then, $\just(m')$ is comparable with $m$ with respect to $\leq_\bsigma$.
\end{lem}
\begin{proof}
Since $m <_\bsigma m'$, there is $\rho \imc m' \in \gcc(\bsigma)$ s.t.
$m \in \rho$. If $\pr_\bsigma(m')$ is minimal in $A$,
$\pr_\bsigma(\just(m'))$ is minimal in $B$, so $\just(m')$ is
minimal for $\leq_\bsigma$ by courtesy. But since $\bsigma$ is pointed,
$\just(m')$ is the initial move of $\rho$, obviously comparable
with $m$ as $\rho$ is totally ordered.

Else, by visibility $\just(m') \in \rho$. But $\rho$ is
totally ordered, so $m, \just(m')$ comparable.
\end{proof}

We shall prove that the composition of visible causal
strategies is deadlock-free. But first, we recall
the basic mechanisms of interactions between causal strategies.
Consider $\bsigma : A \vdash B$ and $\btau : B \vdash C$, and
configurations $x^\bsigma \in \conf{\bsigma}$, and $x^\btau \in
\conf{\btau}$ such that, writing $\pr_\bsigma x^\bsigma = x^\bsigma_A
\parallel x^\bsigma_B$ and $\pr_\btau x^\btau = x^\btau_B \parallel
x^\btau_C$, we have $x^\bsigma_B = x^\btau_B = x_B$, \emph{i.e.}
$x^\bsigma$ and $x^\btau$ are \emph{matching}. Then, recall from
Definition \ref{def:caus_comp} the bijection arising from their
synchronization:
\[
\varphi
\quad
:
\quad
x^\bsigma \parallel x^\btau_C 
\quad
\stackrel{\pr_\bsigma \parallel x^\btau_C}{\simeq} 
\quad
x^\bsigma_A \parallel x_B \parallel x^\btau_C 
\quad
\stackrel{x^\bsigma_A \parallel \pr_\btau^{-1}}{\simeq}
\quad
x^\bsigma_A \parallel x^\btau\,,
\]
whose graph is equipped with a relation importing the causal
dependencies from $\bsigma$ and $\btau$:
\[  
(l, r) \vartriangleleft (l', r') 
\qquad
\Leftrightarrow
\qquad
l <_{\bsigma \parallel C} l'
\quad
\vee
\quad
r <_{A \parallel \btau} r'\,.
\]

We saw in Definition \ref{def:caus_comp} and Proposition
\ref{prop:main_interaction} that $(x^\bsigma, x^\btau)$
corresponds to a configuration of the interaction $\btau\inter \bsigma$
exactly when this bijection is \emph{secured}, \emph{i.e.}
$\vartriangleleft$ is acyclic.

If $\bsigma : A \vdash B$ and $\btau : B \vdash C$ are visible, 
we claim that this is always the case. We reason by
contradiction: starting with a putative deadlock, we 
repeatedly push it down the causal dependency of the arena, until it
reaches a minimal event -- but those cannot appear in a cycle. Before
giving the formal proof, we showcase the reasoning on a simplified case.

Consider a \textbf{simple deadlock} in $\varphi$, given by
$p_1 = (l_1, r_1)$ and $p_2 = (l_2, r_2) \in \varphi$ such that
\[
l_1 <_{\bsigma \parallel C} l_2\,,
\qquad
r_2 <_{A \parallel \btau} r_1\,,
\]
an immediate causal incompatibility between $p_1$ and $p_2$. In other
words we have $p_1 
\vartriangleleft p_2$ and $p_2 \vartriangleleft p_1$, and we use
$p_1 \vartriangleleft_\bsigma p_2$ and $p_2
\vartriangleleft_\btau p_1$ to indicate the origin of the causal
constraint. Finally, we apply the same conventions for polarity of
elements of $\varphi$ as in Section~\ref{subsubsec:caus_conf_int}.

The first observation (skipped here) is that \emph{w.l.o.g.}, the polarities are as in
\[
\xymatrix{
p_1^\labr
        \ar@/^1pc/@{->}[rr]|{\vartriangleleft_\bsigma}&&
p_2^\labl
        \ar@/^1pc/@{->}[ll]|{\vartriangleleft_\btau}
}\,,
\]
where both occur in $B$ but not minimal in $B$ -- so we may
take $\just(p_i) = (\just(l_i), \just(r_i))$. By Lemma
\ref{lem:locality}, $l_1$ and $\just(l_2)$ are comparable
for $\bsigma$; while $r_2$ and $\just(r_1)$ are comparable for $\btau$.
If $p_1 \vartriangleleft_\bsigma \just(p_2)$ or $p_2
\vartriangleleft_\btau \just(p_1)$, then we respectively have one of the
cycles:
\[
\raisebox{18pt}{$
\xymatrix@R=0pt@C=30pt{
&\just(p_2)^\labr
        \ar@/^/@{->}[dr]|{\vartriangleleft_\btau}\\
p_1^\labr
        \ar@/^/@{->}[ur]|{\vartriangleleft_\bsigma}&&
p_2^\labl
        \ar@/^1.5pc/@{->}[ll]|{\vartriangleleft_\btau}
}$}
\qquad
\text{\emph{or}}
\qquad
\xymatrix@R=0pt@C=30pt{
p_1^\labr
        \ar@/^1.5pc/@{->}[rr]|{\vartriangleleft_\bsigma}&&
p_2^\labl
        \ar@/^/@{->}[dl]|{\vartriangleleft_\btau}\\
&\just(p_1)^\labl
        \ar@/^/@{->}[ul]|{\vartriangleleft_\bsigma}
}
\]
so simple deadlocks between $p_1$ and $\just(p_2)$; or between
$\just(p_1)$ and $p_2$. The cumulative \emph{depth} in $B$
has decreased. The case $p_1 = \just(p_2)$ or $p_2 =
\just(p_1)$ is easily discarded.

The last case has $\just(p_2)
\vartriangleleft_\bsigma p_1$ and $\just(p_1) \vartriangleleft_\btau
p_2$. But $p_1$ has polarity $\labr$, so by Lemma \ref{lem:just_pred}
the only immediate dependency in $\bsigma$ of $l_1^-$ is $\just(l_1)$.
So $\just(p_2) \vartriangleleft_\bsigma p_1$ factors as $\just(p_2)
\vartriangleleft_\bsigma \just(p_1) \vartriangleleft_\bsigma p_1$.
Symmetrically $\just(p_1) \vartriangleleft_\btau \just(p_2)$, so we
have:
\[
\xymatrix@R=0pt{
\just(p_1)
        \ar@/^1.5pc/@{->}[rr]|{\vartriangleleft_\btau}&&
\just(p_2)
        \ar@/^1.5pc/@{->}[ll]|{\vartriangleleft_\bsigma}
}\,,
\]
closer to the root of the arena. Repeating this we eventually
hit an impossible simple deadlock with a minimal event in $B$,
finally exposing the contradiction. So visibility \emph{structures}
the interaction around the dependency of the arena, giving us an
effective reasoning principle.

The proof of the deadlock-free lemma is the same in
essence, but challenging in form. Firstly, cycles in
$\vartriangleleft$ in Definition \ref{def:caus_comp} may have arbitrary
length. Secondly, in relational composition strategies synchronize on
\emph{symmetry classes} of configurations rather than concrete
configurations; so we must account for synchronization through symmetry.

\begin{lem}\label{lem:deadlock_free}
Consider $A, B, C$ $-$-arenas, $\bsigma : A \vdash B$ and $\btau : B
\vdash C$ visible causal strategies, $x^\bsigma \in \conf{\bsigma}$ and
$x^\btau \in \conf{\btau}$ with a symmetry $\theta : x^\bsigma_B \sym_B
x^\btau_B$.
Then, the composite bijection
\[
\varphi
\quad
:
\quad
x^\bsigma \parallel x^\tau_C 
\quad
\stackrel{\pr_\bsigma \parallel x^\btau_C}{\simeq} 
\quad
x^\bsigma_A \parallel x^\bsigma_B \parallel x^\btau_C 
\quad
\stackrel{x^\bsigma_A \parallel \theta \parallel x^\btau_C}{\simeq}
\quad
x^\bsigma_A \parallel x^\btau_B \parallel x^\btau_C
\quad
\stackrel{x^\bsigma_A \parallel \pr_\btau^{-1}}{\simeq}
\quad
x^\bsigma_A \parallel x^\btau\,,
\]
is secured, in the sense that the relation $\vartriangleleft$, defined
on the graph of $\varphi$ with
\[
(l, r) \vartriangleleft (l', r')
\]
whenever $l \mathrel{(<_\bsigma \parallel <_C)} l'$ or $r \mathrel{(<_A
\parallel <_\btau)} r'$, is acyclic\footnote{For $\theta$ an
identity, this exactly means that $x^\bsigma$ and $x^\btau$ satisfy
the
\emph{secured} condition of Definition \ref{def:caus_comp}.}.
\end{lem}
\begin{proof}
We use \emph{polarities} $\labl, \labr$ or $-$ for elements of $\varphi$
(\emph{i.e.} pairs $(l,r)$) as in
Section~\ref{subsubsec:caus_conf_int}. We say $(l ,r)$ \emph{occurs
in  $A$, $B$ or $C$} in the obvious sense. We use a notion of
\emph{justifier} of a pair $(l ,r)$ non-minimal in $B$: 
as  $\theta$ is an order-isomorphism, $\pr_\bsigma(l)$ is minimal
in $B$ iff $\pr_\btau(r)$ is. If not, then $\just(l)$
and $\just(r)$ also match up to $\theta$ and $(\just(l), \just(r))$ must
be in $\varphi$ as well -- we write it $\just(l, r)$.  Suppose now 
$\vartriangleleft$ is \emph{not} secured, \emph{i.e.} there is $((l_1,
r_1),  \ldots , (l_n, r_n))$ with
\[
(l_1, r_1) \triangleleft  (l_2, r_2) \triangleleft   \ldots
\triangleleft  (l_n, r_n)
\triangleleft (l_1, r_1)\,,
\]
written $p_1 \triangleleft \dots \triangleleft p_n \triangleleft p_1$ --
the \textbf{length} of this cycle is $n$. First, \emph{w.l.o.g.} the cycle occurs
entirely in $B$. Assume it has minimal length. If it occurs
entirely in $A$ or $C$, then $(l_i)_{1\leq i \leq n}$ (resp. $(r_i)_{1
\leq i \leq n}$) is a cycle in $\bsigma$ (resp. $\btau)$, absurd. So,
it passes through $B$. Next, if \emph{e.g.}
\[
p_i^{(B)} \triangleleft p_{i+1}^{(C)} \triangleleft \dots \triangleleft
p_{j-1}^{(C)} \triangleleft p_j^{(B)}\,, 
\]
then it is easy to prove that $r_i < r_{i+1} < \dots <
r_{j-1} < r_j$, so that $p_i^{(B)} \triangleleft p_j^{(B)}$ and
the cycle can be shortened, contradicting its minimality -- the same
argument holds for $A$.

We restrict to cycles in $B$. The \textbf{depth} of $(l,r)$
is the length of the chain of justifiers to $(l_0, r_0)$ minimal in $B$
-- the depth of $(l_0, r_0)$ minimal in $B$ is $0$. The
\emph{depth} of the cycle is
\[
d = \sum_{1\leq i \leq n} \depth(l_i,r_i)\,,
\]
and we assume \emph{w.l.o.g.} the cycle minimal for the product order on
pairs $(n, d)$. In this proof, all arithmetic computations on indices
are done modulo $n$ (the length of the cycle).

Next, let us write $p_i \triangleleft_\bsigma p_{j}$ if $l_i
\mathrel{(<_\bsigma \parallel <_C)} l_j$ and $p_i \triangleleft_\btau
p_j$ symmetrically.
We notice that $\triangleleft_\bsigma$ and $\triangleleft_\btau$
alternate
-- if not we shorten the cycle by transitivity, contradicting
minimality. We assume \emph{w.l.o.g.} that $p_{2k} \triangleleft_\bsigma
p_{2k+1}$ and $p_{2k+1} \triangleleft_\btau p_{2k+2}$ for all $k$. But
then, $\pol(p_{2k}) = \labr$ and $\pol(p_{2k+1}) = \labl$ so that polarity
in the cycle is alternating as well. Indeed, assume \emph{e.g.}
\[
p_{2k+1} \triangleleft p_{2k+2}^\labl \triangleleft p_{2k+3}
\]
with $p_{2k+1} \triangleleft_\btau p_{2k+2}$ and $p_{2k+2}
\triangleleft_\bsigma p_{2k+3}$. Then,
$r_{2k+1} <_{A \parallel \btau} r_{2k+2}^-$. From its polarity,
$r_{2k+2}^-$
cannot be minimal in $B$. By Lemma \ref{lem:just_pred}, it has a
unique predecessor $\just(r_{2k+2})
\imc_{A\parallel \btau} r_{2k+2}$, so $r_{2k+1} <_{A \parallel \btau} r_{2k+2}^-$
factors as $r_{2k+1} <_{A \parallel \btau} \just(r_{2k+2}) \imc_{A
\parallel \btau} r_{2k+2}^-$.
Accordingly, $p_{2k+1} \triangleleft_\btau \just(p_{2k+2})
\triangleleft_\btau p_{2k+2}$ -- but dependencies in the game are
respected by both strategies, so $\just(p_{2k+2}) \triangleleft_\bsigma
p_{2k+2}$. So $\just(p_{2k+2}) \triangleleft_\bsigma
p_{2k+3}$, and we can replace the cycle fragment with
\[
p_{2k+1} \triangleleft \just(p_{2k+2}) \triangleleft p_{2k+3}
\]
which is still in $B$, has the same length but strictly smaller depth,
contradiction. The symmetric argument applies for $\bsigma$, so any
$p_{2k+1}$ has polarity $\labl$ and any $p_{2k+2}$ has polarity $\labr$.

Now we show the cycle cannot have an event minimal in $B$.
Seeking a contradiction, if
\[
p_{2k+1}^\labl \triangleleft_\btau p_{2k+2}^\labr \triangleleft_\bsigma
p_{2k+3}^\labl
\]
with $p_{2k+2}$ minimal in $B$, then $l_{2k+2} <_{\bsigma \parallel C}
l_{2k+3}$ with $l_{2k+2}$ minimal in $B$, but then $\pr_{\bsigma
\parallel C}(l_{2k+2}) <_{A\parallel B \parallel C} \pr_{\bsigma
\parallel C}(l_{2k+3})$. Indeed, if $\pr_{\bsigma \parallel
C}(l_{2k+2})$ is
minimal in $B$, $l_{2k+2}$ is (by courtesy) minimal in $\bsigma
\parallel C$. Likewise, since $l_{2k+3}$ occurs in $B$, $\pr_{\bsigma
\parallel C}(l_{2k+3})$ depends (for $\leq_{A\parallel B \parallel C}$)
on a unique $\pr_{\bsigma \parallel C}(l)$ minimal in $B$, where $l$
must
also be minimal in $\bsigma$. But since $\bsigma$ is pointed, $l_{2k+3}$
has a unique minimal dependency, hence $l = l_{2k+2}$ and $\pr_{\bsigma
\parallel C}(l_{2k+2}) <_{A\parallel B \parallel C} \pr_{\bsigma
\parallel C}(l_{2k+3})$ as claimed. But then, $r_{2k+2} <_{A\parallel
\btau} r_{2k+3}$, so
$p_{2k+1}^\labl \triangleleft_\btau p_{2k+2}^\labr \triangleleft_\btau
p_{2k+3}^\labl$
and again the cycle can be shortened by transitivity, contradicting its
minimality.

Now, we have proved a minimal cycle has a canonical form where the
strategies alternate, polarity alternates, all events are in $B$ and
non-minimal. 
Since $p_{2k}^\labr
\triangleleft_\bsigma p_{2k+1}^\labl$, writing $p = (l,r) =
\just(p_{2k+1})$, we have that $l = \just(l_{2k+1})$ as well. From Lemma
\ref{lem:locality}, we know that $l = \just(l_{2k+1})$ is comparable
with $l_{2k}$ in $\bsigma \parallel C$ (by visibility of $\bsigma$). If
$\just(l_{2k+1}) = l_{2k}$, then $r_{2k} \triangleleft_\btau r_{2k+1}$
as well. This gives $p_{2k-1} \triangleleft_\btau p_{2k+2}$,
contradicting minimality of the cycle. So $\just(l_{2k+1})\neq l_{2k}$.
Similarly, $\just(r_{2k+2})$ is comparable with $r_{2k+1}$ in $A
\parallel \btau$, but distinct.

Assume that we have $p_{2k} \triangleleft_\bsigma \just(p_{2k+1})$ for
some $k$. Since $\just(p_{2k+1}) \triangleleft_\btau p_{2k+1}
\triangleleft_\btau p_{2k+2}$ we can replace the cycle fragment $p_{2k}
\triangleleft p_{2k+1} \triangleleft p_{2k+2}$ with the cycle fragment
\[
p_{2k} \triangleleft \just(p_{2k+1}) \triangleleft p_{2k+2}
\]
which has the same length but smaller depth, absurd.
So, $\just(p_{2k+1}) \triangleleft_\bsigma p_{2k}$ for all $k$
(symmetrically, $\just(p_{2k+2}) \triangleleft_\btau p_{2k+1}$ for all
$k$). In particular,
$\just(l_{2k+1}) <_{\bsigma \parallel C} l_{2k}^-$
but by Lemma \ref{lem:just_pred}, $l_{2k}^-$ has a unique
immediate predecessor $\just(l_{2k})$. So, 
$\just(p_{2k+1}) \triangleleft_\bsigma \just(p_{2k})$
for all $k$; and likewise $\just(p_{2k+2}) \triangleleft_\btau
\just(p_{2k+1})$ for all $k$.
So we can replace the full cycle with
\[
\just(p_n) \triangleleft \just(p_{n-1}) \triangleleft  \ldots
\triangleleft \just(p_1) \triangleleft \just(p_n)
\]
which has the same length but smaller depth, absurd.
\end{proof}
 
Despite its relatively discreet role in the development, we regard the
deadlock-free lemma as one of our main contributions. It is a
powerful observation with far-reaching consequences in linking game
semantics and relational models. It also gives a lot of weight to the
notion of visibility, as a simple, well-behaved and fairly general
under-approximation of innocence.

\subsubsection{Preservation of composition}
For preservation of the positional collapse by composition, we
need one further lemma:
that any complete position is reachable by a \emph{well-bracketed} play.

\begin{lem}\label{lem:reach_inn_wb}
Take $\bsigma : A$ visible well-bracketed on $-$-arena $A$, 
$x^\bsigma \in \conf{\bsigma}$ with $\pr_\bsigma(x)$ complete.
Then, there is $t \in \NAlt(\bsigma)$ such that $\ev{t} = x^\bsigma$ and
$\pr_{\bsigma}(t)$ is well-bracketed.
\end{lem}
\begin{proof}
The idea is simple: since $\bsigma$ is well-bracketed, it suffices to
show that $\pr_\bsigma(x)$ is reachable by a well-bracketed Opponent.
But we can
set up causal constraints forcing Opponent to be well-bracketed,
formulated as a visible causal strategy, and apply the deadlock-free
lemma.

For $x \in \conf{\bsigma}$ s.t. $\pr_\bsigma(x)$ complete,
consider $x^\bsigma_A$ as an arena with trivial symmetry. We
build $\btau : x^\bsigma_A \vdash \tunit$ as
$\ev{\btau} = \ev{x^\bsigma_A \vdash \tunit}$, and $\leq_\btau$ as the
order of the arena enriched with:
\[  
\begin{array}{rclcl}
(2, \qu)^- &\imc_\btau& (1, q)^{+,\Qu} &\qquad& \text{if $q$ is an
initial question in $A$},\\
(1, a)^{-,\An} &\imc_\btau& (2, \done)^+ && \text{if $a$ answers an
initial question in $A$},\\
(1, a_1)^{-,\An} &\imc_\btau& (1, a_2)^{+,\An} && \text{if $\just(a_2)
\imc_A \just(a_1)$,} 
\end{array}
\]
resulting in $\btau$ visible well-bracketed. By Lemma
\ref{lem:deadlock_free}, there is a linear ordering of $\ev{x^\bsigma_A
\vdash \tunit}$ compatible with the constraints of both $\bsigma$ and
$\btau$. As both are well-bracketed, its projection on the left gives $s
\in \NAltStrat(\bsigma)$ such that $\ev{s} = x^\bsigma_A$ and $s$
well-bracketed as required.
\end{proof}                          

\begin{prop}\label{prop:vis_coll}
Consider $\bsigma : A \vdash B$ and $\btau : B \vdash C$ causal strategies.

If $\bsigma$ and $\btau$ are well-bracketed and visible, then
$\coll (\btau \odot \bsigma) = (\coll \btau) \odot (\coll \bsigma)$.
\end{prop}
\begin{proof}
$\subseteq$. Consider $(\x_A, \x_C) \in \coll (\btau \odot \bsigma)$.
By definition, $\x_A \parallel \x_C = \coll (\pr_{\btau \odot
\bsigma}(x^\btau \odot x^\bsigma))$ for $x^\btau \odot x^\bsigma \in
\conf{\btau \odot \bsigma}$ with $\pr_{\btau \odot \bsigma}(x^\btau
\odot x^\bsigma) = x_A \parallel x_C$ complete, and $\x_A = \coll x_A,
\x_C = \coll x_C$. By Lemma \ref{lem:reach_inn_wb}, there is $s \in
\NAltStrat(\btau \odot \bsigma)$ well-bracketed s.t. $\ev{s} =
\pr_{\Lambda(\btau \odot \bsigma)}(x^\btau \odot x^\bsigma)$. By
Proposition \ref{prop:nalt_fonc}, there is $u \in \NAltStrat(\btau)
\inter \NAltStrat(\btau)$ s.t. $u \restrict A, C = s$. Now, since
$\bsigma$ and $\btau$ are well-bracketed and $s$ is well-bracketed, it
is direct by induction that $u$ is well-bracketed.
Since $x_A \parallel x_C = \pr_{\btau \odot \bsigma}(x^\btau\odot
x^\bsigma)$ is complete, in particular the initial question has an
answer -- but as $u$ is well-bracketed, all questions in
$u$ are answered.  Writing $\pr_\bsigma(x^\bsigma) = x_A \parallel x_B$
and $\pr_\btau(x^\btau) = x_B \parallel x_C$, $x_B$ is
complete as well. But then, $[x_B]_\sym = \x_B \in \coll B$, so
$x^\bsigma$ witnesses $(\x_A, \x_B) \in \coll \bsigma$ and $x^\btau$
witnesses $(\x_B, \x_C) \in \coll \btau$; so $(\x_A, \x_C) \in (\coll
\btau) \odot (\coll \bsigma)$.

$\supseteq$.
Assume we have symmetry classes of complete configurations $\x_A,
\x_B$ and $\x_C$ s.t.
\[
(\x_A, \x_B) \in \coll \bsigma
\qquad
\qquad
(\x_B, \x_C) \in \coll \btau\,,
\]
so there are $x^\bsigma \in \conf{\bsigma}$ and $x^\btau
\in \conf{\btau}$ with $\pr_\bsigma(x^\bsigma) = x^\bsigma_A \parallel
x^\bsigma_B$, $\pr_\btau(x^\btau) = x^\btau_B \parallel x^\btau_C$, with
$x^\bsigma_A \in \x_A$, $x^\bsigma_B, x^\btau_B \in \x_B$, and
$x^\btau_C \in \x_C$. In particular, there is $\theta :
x^\bsigma_B \sym_B x^\btau_B$. Now,
\[
\varphi
\quad
:
\quad
x^\bsigma \parallel x^\tau_C 
\quad
\stackrel{\pr_\bsigma \parallel x^\btau_C}{\simeq} 
\quad
x^\bsigma_A \parallel x^\bsigma_B \parallel x^\btau_C 
\quad
\stackrel{x^\bsigma_A \parallel \theta \parallel x^\btau_C}{\simeq}
\quad
x^\bsigma_A \parallel x^\btau_B \parallel x^\btau_C
\quad
\stackrel{x^\bsigma_A \parallel \pr_\btau^{-1}}{\simeq}
\quad
x^\bsigma_A \parallel x^\btau\,,
\]
is secured by Lemma \ref{lem:deadlock_free}. They only match up to
symmetry, but by Proposition
\ref{prop:sync_sym}, there are $y^\bsigma \sym_\bsigma x^\bsigma$ and
$y^\btau \sym_\btau x^\btau$ such that $\pr_\bsigma(y^\bsigma) =
y^\bsigma_A \parallel y^\bsigma_B$, $\pr_\btau(y^\btau) = y^\btau_B
\parallel y^\btau_C$ with $y^\bsigma_B = y^\btau_B$ -- so $y^\btau
\odot y^\bsigma \in \conf{\btau \odot \bsigma}$. But since
$y^\bsigma \sym_\bsigma x^\bsigma$ and $y^\btau \sym_\btau x^\btau$, we
also know that $y^\bsigma_A \in \x_A$ and $y^\btau_C \in \x_C$ still, so
$y^\btau \odot y^\bsigma$ witnesses $(\x_A, \x_C) \in \coll (\btau \odot
\bsigma)$ as required.
\end{proof}

So positional equivalence of well-bracketed innocent causal
strategies is preserved under composition. The other constructions pose
no challenge. Altogether, we get:

\begin{cor}\label{cor:equiv_cong}
There is a Seely category $\CGwbvis/\equiv$ of $-$-arenas and
positional equivalence classes of visible well-bracketed strategies.
Finally, the canonical functor
\[
\CGwbvis \to \CGwbvis/\equiv
\]
preserves the interpretation of $\PCFpar$.
\end{cor}

\subsubsection{On the relational model}\label{subsubsec:relational}
The above yields a functor to $\Rel$:

\begin{prop}\label{prop:coll_main}
The positional collapse defines a functor $\coll(-) : \CGwbvis \to
\Rel$.
\end{prop}
\begin{proof}
Composition is Proposition \ref{prop:vis_coll}, while for
identity it is a direct verification.
\end{proof}

Unfortunately, this functor is not compatible with the Seely category
structure. For negative
arenas $A, B$ we do have $\coll (A \tensor B) \iso (\coll A) \times
(\coll B)$; but for instance $\coll (A \with B) \not \iso (\coll A) +
(\coll B)$ because $\coll (A\with B)$ includes the empty position, which
we do not have enough information to send to the left or to the right.
Likewise, $\coll (\oc A) \not \iso \mathcal{M}_f(\coll A)$.  Considering
only non-empty configurations does not solve the issue, as we lose
$\coll (A \tensor B) \iso (\coll A) \times (\coll B)$.

This mismatch can be mitigated by focusing on the cartesian closed
Kleisli categories. Say that an $-$-arena $A$ is \textbf{strict} if all
its minimal (necessarily negative) events are in pairwise conflict --
all types and contexts of $\PCF$ are interpreted as strict $-$-arenas.
If $A$ is an arena, we write $\collo A$ for the set of non-empty
complete positions of $A$. Then: 

\begin{lem}\label{lem:rel_conf}
For $A, B, C, D$ $-$-arenas with $C$ strict there are bijections:
\[
\begin{array}{rcl}
\coll (A \tensor B) &\iso& (\coll A) \tensor (\coll B)\\
\collo (A \lin C) &\iso& (\coll A) \times (\collo C)
\end{array}
\qquad
\begin{array}{rcl}
\collo(A \with B) &\iso& \collo A + \collo B\\
\coll (\oc C) &\iso& \mathcal{M}_f(\collo C)\,.
\end{array}
\]
\end{lem}

For $B, C$ strict it follows that $\collo (\oc B \lin C) \iso
\mathcal{M}_f(\collo B) \times (\collo C)$, matching relational
semantics. For any $\PCF$ type $A$ we obtain $\collo \intr{A} \iso \intr{A}_\Rel$
-- such a bijection is easily established for ground types -- so
for the interpretation of $\PCF$ types, points of the web in the
relational model exactly correspond to non-empty complete symmetry
classes in the game semantics. This can be extended to an
interpretation-preserving functor from games to relations, but we leave
the details to a forthcoming paper on the quantitative collapse. 

\subsection{Well-Bracketed Pruning}
\label{subsec:wb_pruning}
We need a sufficiently constrained description of the causal shape of
strategies so that we might replicate it syntactically. While innocence
is a causal notion, well-bracketing is not, and it leaves the causal
shape too liberal. We shall now see how via a well-bracketed pruning a
more causal form of well-bracketing can be enforced.

\subsubsection{Causal well-bracketing}
Recall the idea of Theorem \ref{th:definability}:
each Player question $\qu^+$ corresponds to a call to a variable $x$.
Opponent questions pointing to $\qu^+$ correspond to $x$ calling an
argument; and Opponent answers to $\qu^+$ correspond to $x$ evaluating
to a return value. The crux of the definability argument is that these
subsequent possibilities of Opponent moves pointing to $\qu^+$
\emph{split} (the causal representation of) the innocent strategy under
scrutiny into sub-strategies independent from each other -- as they must
be, if they are to correspond to distinct branches of the desired syntax
tree.

Parallel innocent strategies enjoy the same ``splitting'' property, to
an extent:

\begin{lem}\label{lem:arguments}
Consider $A$ an arena, $\bsigma : A$ an innocent causal strategy,
and $q^+ \in \ev{\bsigma}$.

If $q^+ \imc_\bsigma m_1^-, q^+ \imc_\bsigma m_2^-$ distinct,
$\{m \in \ev{\bsigma} \mid m\geq_\bsigma m_1\}$, $\{m \in
\ev{\bsigma} \mid m \geq_\bsigma m_2\}$ disjoint.
\end{lem}
\begin{proof}
Obvious by pre-innocence.
\end{proof}

Unfortunately, this is too weak. Indeed $q^+$ in the statement above
should correspond to a syntactic variable occurrence $x$, so that the
overall term has form $C[x\,M_1\,\dots\,M_n]$ where $x$ has arity $n$,
with the $M_i$ written in $\PCFpar$. As standalone pieces of syntax in a
language without interference they are indeed independent from each
other (as guaranteed semantically by Lemma \ref{lem:arguments}), but are
also independent from the context $C$ -- which Lemma \ref{lem:arguments}
does not capture. In fact, the current conditions on strategies do not
ensure this.

\begin{figure}
\[
\scalebox{.9}{$
\xymatrix@R=2pt@C=0pt{
(\tunit &\lin& \tunit) &\lin & \tunit & \lin & \tunit & \lin & \tunit\\
&&&&&&&&\grey{\qu^-}
        \ar@[grey]@{-|>}[dllllll]
        \ar@[grey]@{-|>}[dll]\\
&&\qu^+_1
        \ar@[grey]@{.}@/^/[urrrrrr]
        \ar@{-|>}[dll]&&&&
\grey{\qu^+}
        \ar@[grey]@{.}@/^/[urr]
        \ar@[grey]@{-|>}[d]\\
\qu^-   \ar@{.}@/^/[urr]
        \ar@{-|>}[drrrr]&&&&&&
\grey{\done^-}
        \ar@[grey]@{.}@/^/[u]
        \ar@[grey]@{-|>}[dll]\\
&&&&\qu^+_2
        \ar@[grey]@{.}@/^1.5pc/[uuurrrr]
}$}
\]
\caption{Non-locality of argument sub-strategies}
\label{fig:non_locality}
\end{figure}

In Figure \ref{fig:non_locality}, we show a counter-example: this
satisfies all the conditions for an augmentation of a well-bracketed
parallel innocent strategy.  Here, $\qu^+_1$ should correspond to a
variable call, and $\qu^-$ should initiate the exploration of its
argument, and independent sub-term.  But the subsequent Player move
$\qu^+_2$, the ``head variable occurrence'' of the argument sub-term,
also depends on a parallel thread. This behaviour is not realizable in
$\PCFpar$ and should be banned before any definability attempt. 

We now define which causal behaviour is deemed acceptable for $\PCFpar$.

\begin{defi}\label{def:caus_wb}
Consider $A$ an arena, and $\bsigma : A$ parallel innocent.

We say that $\bsigma$ is \textbf{causally well-bracketed} iff it
satisfies the two conditions below:
\[
\begin{array}{rl}
\text{\emph{wb-threads:}} & \text{for every $\rho \in \gcc(\bsigma)$,
$\rho$ is well-bracketed in the sense of Definition
\ref{def:alt_play_wb},}\\ 
\text{\emph{globular:}} & \text{for a diagram in $\bsigma$ with
$X = \{m_1, \dots, m_n\}$ and $Y = \{n_1, \dots, n_p\}$ disjoint}\\ 
&\qquad\qquad\qquad\qquad
$
\xymatrix@R=0pt@C=5pt{
&m_1^+  \ar@{-|>}[r]&
~       \ar@{.}[r]&
~       \ar@{-|>}[r]&
m_n^-   \ar@{-|>}[dr]\\
m_0^-   \ar@{-|>}[ur]
        \ar@{-|>}[dr]&&&&&
m^+\\
&n^+_1  \ar@{-|>}[r]&
~       \ar@{.}[r]&
~       \ar@{-|>}[r]&
n_p^-
        \ar@{-|>}[ur]
}
$\\
&\text{then every question in $X$ (resp. $Y$) is answered in $X$ (resp.
$Y$).} 
\end{array}
\]
\end{defi}

The Question/Answer labeling is implicitly imported from $\ev{A}$ to
$\ev{\bsigma}$ via $\pr_\bsigma$ -- likewise, for $m^\Qu, n^\An \in
\ev{\bsigma}$, we say that $n$ \textbf{answers} $m$ if $\pr_\bsigma(m)
\imc_A \pr_\bsigma(n)$, \emph{i.e.} $\pr_\bsigma(n)$ answers
$\pr_\bsigma(m)$.

We call such a diagram a \emph{globule}. 
This bans directly behaviours as in Figure \ref{fig:non_locality}.
Only Player can merge parallel threads, and only if he
is responsible for the fork (by parallel innocence). So, polarities
in the definition of globules are not restrictive. One can further
observe that globules always have Question/Answer assignments as in
\[
\xymatrix@R=0pt@C=10pt{
&m_1^{+,\Qu}    
        \ar@{-|>}[r]&
m_2^{-,\An}     
        \ar@{.}@/_1pc/[l]
        \ar@{-|>}[r]&
~       \ar@{.}[r]&
~       \ar@{-|>}[r]&
m_{2n-1}^{+,\Qu}
        \ar@{-|>}[r]&
m_{2n}^{-,\An}
        \ar@{.}@/_1pc/[l]
        \ar@{-|>}[dr]\\
m_0^-   \ar@{-|>}[ur]
        \ar@{-|>}[dr]&&&&&&&
m^+\\
&n_1^{+,\Qu}
        \ar@{-|>}[r]&
n_2^{-,\An}
        \ar@{-|>}[r]
        \ar@{.}@/^1pc/[l]&
~       \ar@{.}[r]&
~       \ar@{-|>}[r]&
n_{2p-1}^{-,\Qu}
        \ar@{-|>}[r]&
n_{2p}^{+,\An}
        \ar@{.}@/^1pc/[l]
        \ar@{-|>}[ur]
}
\]

Indeed, if $m_1$ was an answer, it would be maximal in
$\bsigma$ (by Lemma \ref{lem:app_aux_max_pos}, as answers are maximal in
$A$) and the merge would be impossible.  So $m_1$
is a question, and by \emph{globular} it has an answer in $\{m_1, \dots,
m_{2n}\}$. By courtesy this answer depends immediately on $m_1$ in
$\bsigma$ and so must be $m_2$. Repeating this we get the description
above.  Hence in causally well-bracketed strategies, parallel threads
that might merge follow a strict call/return discipline.

A variant of the proof of Proposition \ref{prop:rig_comp}, relying on
Lemma \ref{lem:forking}, shows that causal well-bracketing composes.
This was done in \cite{lics15} (see also
\cite{DBLP:phd/hal/Castellan17}), but it is not our route here: we must
prove full abstraction with respect to the \emph{same} well-bracketing
condition used before. We shall see that the situation is
similar as in Section~\ref{sec:wb_ia}: in \emph{complete}
configurations, which suffice for tests, the weaker well-bracketing
implies the stronger. 

\subsubsection{Strengthening well-bracketing} We will show that if
$\bsigma : A$ is innocent and well-bracketed (as in Definition
\ref{def:wb_caus_strat}) and $x^\bsigma \in \conf{\bsigma}$ is complete
(\emph{i.e.} all its questions have an answer within $x^\bsigma$), then
the corresponding augmentation is also \emph{causally} well-bracketed --
so restricting an innocent well-bracketed $\bsigma : A$ to its
completable part yields a positionally equivalent, causally
well-bracketed strategy; a procedure we call ``well-bracketed pruning''.

On arenas arising from types, gccs of innocent strategies are
\emph{already} well-bracketed -- this holds even without
well-bracketing.  The proof is morally as for Proposition
\ref{prop:inn_wb}.

\begin{lem}\label{lem:gcc_wb_aut}
Consider $A$ a type, $\bsigma : \intr{A}$ an innocent causal
strategy, and $\rho \in \gcc(\bsigma)$.

Then, $\rho$ is well-bracketed in the sense of Definition
\ref{def:alt_play_wb}.
\end{lem}
\begin{proof}
Consider $\rho = \rho_1 \imc \dots \imc \rho_n \in \gcc(\bsigma)$ and
assume, seeking a contradiction, that $\rho_n$ is an answer not pointing
to the pending question. It follows that $\rho$ has the form:
\[
\xymatrix@C=0pt{
\rho_1 &\imc& \dots& \imc &\rho_i^{-, \Qu}& \imc &\dots &\imc
&\rho_j^{+,\Qu} &\imc&
\rho_{j+1}^{-,\Qu} &\imc& \dots& \imc& \rho_n^{+,\An} 
        \ar@{.}@/_1pc/[llllllllll]
}
\]
where $\rho_{j+1}^{-, \Qu}$ was pending. As $\intr{A}$
is the interpretation of a type, $\rho_{j+1}^{-,\Qu}$ is
sibling to countably many symmetric copies, \emph{i.e.} we may write it
as $b_{\grey{k}}$ with some copy index $\grey{k} \in \mathbb{N}$, and
consider a copy $b_{\grey{k+1}} \in \ev{\bsigma}$ with both
$b_{\grey{k}}$ and $b_{\grey{k+1}}$ pointing to $\rho_j$. Write also $m
= \rho_n$.

Consider $x^\bsigma \in \conf{\bsigma}$ with $\rho$ in $x^\bsigma$,
\emph{w.l.o.g.} we can assume that (the augmentation) $x^\bsigma$ has
top element $m$. Consider any $t \in \NAlt(x^\bsigma)$ s.t. $\ev{t}
= x^\bsigma$. Then $t$ may be written as $t_1 \cdot b_{\grey{k}} \cdot
t_2 \cdot m$ where without loss of generality, we may assume that all
events after $b_{\grey{k}}$ depend on it. Because $\rho$ is in
$x^\bsigma$, $m$ points in $t_1$. Now, by uniformity of
$\bsigma$, there is also some
\[
t' = t_1 \cdot b_{\grey{k+1}} \cdot t'_2 \cdot m' \in \NAlt(\bsigma)\,,
\]
such that $t \sym_\bsigma t'$, where (therefore) $m'$ points to the same
move as $m$ in $t_1$. There is some $y^\bsigma \in \conf{\bsigma}$ such
that $t' \in \NAlt(y^\bsigma)$, \emph{w.l.o.g.} assume $\ev{t'} =
y^\bsigma$. By innocence, any causal branching in $x^\bsigma$ and
$y^\bsigma$ is due to Player. Next, we observe that $x^\bsigma$ and
$y^\btau$ are negatively compatible. Indeed if $n_1^- \in
\ev{x^\bsigma}$, $n_2^- \in \ev{y^\bsigma}$ are in minimal conflict in
$\bsigma$, then by Lemma \ref{lem:app_pres_neg_imm_confl},
$\pr_\bsigma(n_1)$ and
$\pr_\bsigma(n_2)$ are in minimal conflict in $\intr{A}$. But by
property of arenas originating from types, this implies that $n_1$ and
$n_2$ have the same justifier $n$. As moves in $t_2 \cdot m$ depend on
$b_{\grey{k}}$ and moves in $t'_2 \cdot m'$ depend on $b_{\grey{k+1}}$,
by pre-innocence (Lemma \ref{lem:arguments}) these must be disjoint, so
$n$ appears in $t_1$. But again by pre-innocence, since $x^\bsigma$ has
top element $m$ and $y^\bsigma$ has top element $m'$, only one Opponent
move can point to $n$ in $x^\bsigma$; and likewise for $y^\bsigma$. As
$t \sym_\bsigma t'$ preserves pointers, that means that $n_1$ and $n_2$
must arrive at the same (chronological) index in $t, t'$, and so their
display in $\intr{A}$ are related by a symmetry in $\intr{A}$. But in
arenas arising from $\PCF$ types, no conflicting events can be related
by a symmetry, contradiction.
\end{proof}

So condition \emph{wb-threads} of Definition \ref{def:caus_wb} is
automatic, even without assuming well-bracketing.  However, condition
\emph{globular} is \emph{not} automatic as illustrated by Figure
\ref{fig:non_locality}.

We shall now prove \emph{globular} on \emph{complete} augmentations of
well-bracketed innocent strategies. Intuitively, the reason is simple.
Consider a globule
\[
\xymatrix@R=0pt@C=5pt{
&a_1^+  \ar@{-|>}[r]&
~       \ar@{.}[r]&
~       \ar@{-|>}[r]&
a_n^-   \ar@{-|>}[dr]\\
a_0^-   \ar@{-|>}[ur]
        \ar@{-|>}[dr]&&&&&
a^+\\
&b^+_1  \ar@{-|>}[r]&
~       \ar@{.}[r]&
~       \ar@{-|>}[r]&
b_p^-
        \ar@{-|>}[ur]
}
\]
in $x^\bsigma$. As $x^\bsigma$ is complete, every question is eventually
answered. But after the merge, it is too late for questions in the
$a_i$s and $b_j$s: if some $b$ such that $a \geq b$ answers $a_i$ then
by visibility $a_i$ must appear in all its gccs; but a gcc to $b$ may go
through the $b_j$s and avoids $a_i$ entirely. 
Of course, turning this idea into a proof takes some work. First we
prove:

\begin{lem}\label{lem:key_caus_wb}
Consider a $-$-arena $A$,  $\bsigma : A$ visible well-bracketed, and
$x^\bsigma \in \conf{\bsigma}$ complete.

For all $q^{-, \Qu}_1 \in x^\bsigma$, for all $q^{+,\Qu}_2 \in
x^\bsigma$ such that $\pr_\bsigma(q_1) \imc_A \pr_\bsigma(q_2)$, we
have:
\[
\xymatrix@R=-7pt{
&q_1^{-, \Qu}\\
q_2^{+,\Qu}
        \ar@{.}@/^/[ur]\\
a_2^{-,\An}
        \ar@{.}@/^/[u]
        \ar@{}[dr]|{\rotatebox{340}{$\leq_\bsigma$}}\\
&a_1^{+,\An}
        \ar@{.}@/_/[uuu]
}
\]
with $a_1$ and $a_2$ the (unique) answers to $q_1$ and $q_2$.
\end{lem}
\begin{proof}
First, $a_1$ and $a_2$ exist as $x^\bsigma$ is complete, and are
unique by \emph{answer-linear}.

Seeking a contradiction, assume $a_2 \not \leq_\bsigma a_1$, \emph{i.e.}
$a_2 \not \in [a_1]_\bsigma$. As $a_2 \not
\in [q_2]_\bsigma$, writing $y = [a_1]_\bsigma \vee
[q_2]_\bsigma$, $y \in \conf{\bsigma}$ with $a_1, q_2 \in y$
and $a_2 \not \in y$. By Lemma \ref{lem:reach_inn_wb}, there is $t \in
\NAlt(y)$ well-bracketed s.t. $\ev{t} = y$.
So by \emph{wait} of Definition \ref{def:nalt_play_wb},
$q_2$ is answered in $\ev{t} = y$. But since $a_2 \not \in y$,
this means $q_2$ is answered twice in $x^\bsigma$, which
contradicts \emph{answer-linear}.
\end{proof}

Using that, we may now formalize and complete the intuitive argument
above.

\begin{lem}\label{lem:aug_comp_glob}
For $\bsigma : \intr{A}$ well-bracketed innocent and
$x^\bsigma \in \conf{\bsigma}$ complete, $x^\bsigma$ is
\emph{globular}.
\end{lem}
\begin{proof}
Consider a \emph{globule} in $x^\bsigma$, \emph{i.e.} a diagram 
\[
\xymatrix@R=-7pt@C=5pt{
&m_1^+  \ar@{-|>}[r]&
~       \ar@{.}[r]&
~       \ar@{-|>}[r]&
m_n^-   \ar@{-|>}[dr]\\
m_0^-   \ar@{-|>}[ur]
        \ar@{-|>}[dr]&&&&&
m^+\\
&n^+_1  \ar@{-|>}[r]&
~       \ar@{.}[r]&
~       \ar@{-|>}[r]&
n_p^-
        \ar@{-|>}[ur]
}
\]
with $X = \{m_1, \dots, m_n\}$ and $Y = \{n_1, \dots, n_p\}$ disjoint.
Seeking a contradiction, consider $m_i^{\Qu} \in X$ unanswered in $X$.
First, we consider $q^- = \just(m^+)$. Since $\bsigma$ is visible,
$q^-$ appears in any gcc of $m$, so $q^-
\leq_\bsigma m_0^-$. But hence $m$ is a question, or a gcc
like
\[
\xymatrix@C=-5pt{
\dots &\imc& \dots &\imc &q^{-,\Qu}& \imc &\dots & \imc& m_0^- &\imc&
m_1^+ & \imc & \dots & \imc & m_n^-& \imc & m^{+,\An}
        \ar@{.}@/_1pc/[llllllllllll]
}
\]
fails well-bracketing as $m_i$ is unanswered, forbidden by Lemma
\ref{lem:gcc_wb_aut}. So $m^+$ has an answer $a^-$ in $x^\bsigma$. But
$q^-$ also has an answer $b^+$ in $x^\bsigma$, and by
Lemma \ref{lem:key_caus_wb}, $a^- \leq_\bsigma b^+$. So
altogether
\[
\xymatrix@R=10pt@C=-5pt{
\dots &\imc& \dots &\imc &q^{-,\Qu}& \imc &\dots & \imc& m_0^- &\imc&
m_1^+ & \imc & \dots & \imc & m_n^-& \imc & m^{+,\Qu}
        \ar@{.}@/_1pc/[llllllllllll]&
\imc& a^{-,\An} & \imc & \dots & \imc & b^{+,\An}
        \ar@{.}@/^1pc/[llllllllllllllllll]
}
\]
well-bracketed by Lemma \ref{lem:gcc_wb_aut}. So in  $m^+ \imc a^- \imc
\dots \imc b^+$, some move must answer $m_i$; and in particular point to
$m_i$. But $m_i$ does not appear in the gcc $\dots \imc m_0 \imc n_1
\imc \dots \imc n_p \imc m \imc a \imc \dots \imc b$, contradicting
visibility. Therefore, $X$ is complete.  
\end{proof}

\subsubsection{Well-bracketed pruning}
For $x \in \conf{\bsigma}$, write $x^+ \in \conf{\bsigma}$ for the
greatest $+$-covered configuration s.t. $x^+ \subseteq x$, obtained
by removing trailing Opponent moves.

\begin{prop}\label{prop:wb_pruning}
For $A$ a $-$-arena and $\bsigma : A$ a well-bracketed innocent
strategy, set 
\[
\ev{\comp(\bsigma)} = \cup \{x \in \conf{\bsigma} \mid x^+ \subseteq y
\in \conf{\bsigma} \text{ complete}\}\,,
\]
with all other components directly inherited from $\bsigma$.

Then $\comp(\bsigma) \equiv \bsigma : A$ is innocent, well bracketed, \emph{causally
well-bracketed}\footnote{A strategy may well be causally
well-bracketed without being well-bracketed: an example of that is a
strategy $\sigma : \tunit \vdash \tunit$ that simultaneously calls its
argument (but does nothing with the result) and returns.}.
\end{prop}
\begin{proof}
Most conditions are immediate consequences of those
from $\bsigma$. The only non-trivial property is that $\tilde{\bsigma}$
restricted to $\ev{\comp(\bsigma)}$ is still an isomorphism family.

First, we prove that for any $\theta : x_1
\sym_\bsigma x_2$, $x_1 \in \conf{\comp(\bsigma)}$ iff $x_2 \in
\conf{\comp(\bsigma)}$. Indeed, assume $x_1^+ \subseteq y_1 \in
\conf{\bsigma}$ complete. Then, since $\theta$ is an order-iso
that preserves polarities, by \emph{restriction} it restricts to
$\theta' : x_1^+ \sym_\bsigma x_2^+$. Now, by \emph{extension},
$\theta'$ extends to $\theta'' : y_1 \sym_\bsigma y_2$ for some $x_2^+
\subseteq y_2$. But since $\theta''$ preserves the Question/Answer
labeling, $y_2 \in \conf{\bsigma}$ is complete; hence $x_2 \in
\conf{\comp(\bsigma)}$ as required. From that, it is straightforward
that $\tilde{\comp(\bsigma)}$ comprising symmetries between
configurations of $\comp(\bsigma)$ is an isomorphism family.

For causal well-bracketing, $\comp(\bsigma)$ satisfies
\emph{wb-threads} by Lemma \ref{lem:gcc_wb_aut}. For \emph{globular},
taking a diagram as in Definition \ref{def:caus_wb}, $m^+$ appears in a
$+$-covered configuration of $\comp(\bsigma)$; hence in a complete
configuration of $\bsigma$. Thus, the condition follows by Lemma
\ref{lem:aug_comp_glob}.

Finally, we must show $\bsigma \equiv \comp(\bsigma)$ -- in fact, we
show both have the same complete configurations. For that, any complete
$x \in \conf{\bsigma}$ must also be $+$-covered: take $x \in
\conf{\bsigma}$ complete. If $x$ has a maximal negative event $m^-$,
since $x$ is complete, $m$ is an answer. But $\just(\just(m))$ is
a Question answered by some $a^{+,\An}$ in $x$ -- but then
$m\leq_\bsigma a$ by Lemma \ref{lem:key_caus_wb}, contradicting
maximality of $m$. Using this we conclude: clearly, any complete
configuration of $\comp(\bsigma)$ is a complete configuration of
$\bsigma$. Reciprocally, a complete $x \in
\conf{\bsigma}$ is $+$-covered, and thus also a (complete)
configuration of $\comp(\bsigma)$. So, $\coll \bsigma =
\coll \comp(\bsigma)$.
\end{proof}

By construction, we also have that $\comp(\bsigma)$ is
\textbf{complete}, in the following sense:

\begin{defi}
Consider $A$ an arena, and $\bsigma : A$ a causal strategy.

We say that $\bsigma$ is \textbf{complete}, if for any $x \in
\confp{\bsigma}$ there is $x \subseteq y \in \conf{\bsigma}$ complete.
\end{defi}

\subsection{Meager Form} \label{subsec:finite_tests}
As for sequential strategies, definability applies for \emph{finite}
strategies, defined through a notion of \emph{meager form}. But
to define meager forms we will first need to restrict to \emph{concrete
arenas} in the sense of Section~\ref{def:concrete-arena}, which we must
first update.  

\subsubsection{Updating concrete arenas}\label{subsubsec:concrete_pol_sym}
We enrich and update Definition \ref{def:concrete-arena}.

\begin{defi}\label{def:concrete-arena2}
A \textbf{concrete arena} is $(A, A^0, \lbl, \ind)$ with $A$ an arena,
$A^0$ \emph{meager arena},
\[
\lbl : \ev{A} \to \ev{A^0}\,,
\qquad
\qquad
\ind : \ev{A} \to \mathbb{N}
\]
two functions, satisfying, additionally to the conditions of Definition
\ref{def:concrete-arena}, the conditions:
\[
\begin{array}{rl}
\text{\emph{jointly injective:}} & \text{for $a_1, a_2 \in \ev{A}$, if
$\lbl(a_1) = \lbl(a_2)$, $\ind(a_1) = \ind(a_2)$,}\\
&\text{and $\pred(a_1) = \pred(a_2)$, then $a_1 = a_2$.}\\
\text{\emph{$\Qu$-wide:}} &
\text{for any $q_1^\Qu \in \ev{A}$ non-minimal, for any $\grey{n} \in
\mathbb{N}$, there is $q_2^{\Qu} \in \ev{A}$}\\
&\text{such that $\lbl(q_1) = \lbl(q_2)$, $\pred(q_1) =
\pred(q_2)$ and $\ind(q_2) = \grey{n}$.}\\
\text{\emph{$\An$-narrow:}} &
\text{for any $a^- \in \ev{A}$ minimal or $a^\An \in \ev{A}$, $\ind(a) =
\grey{0}$,}\\
\text{\emph{$\An$-conflicting:}} &
\text{if $a_1, a_2 \in \ev{A}$ are distinct, they are in minimal
conflict iff}\\
&\text{they are both minimal with the same polarity and copy index,
or}\\
&\text{they are both answers to the same question.}\\
\text{\emph{$+$-transparent:}} & \text{for $\theta : x \sym_A y$,
then $\theta \in \ptilde{A}$ iff for all $a^- \in x$, $\ind(\theta(a)) =
\ind(a)$.}\\
\text{\emph{$-$-transparent:}} & \text{for $\theta : x \sym_A y$,
then $\theta \in \ntilde{A}$ iff for all $a^+ \in x$, $\ind(\theta(a)) =
\ind(a)$.}
\end{array}
\]

We call $\ind(a)$ the \textbf{copy index} of $a$, and $\apred(-)$
is the (unique) \emph{immediate predecessor}.
\end{defi}

In arenas for ground types, all moves have copy index $\grey{0}$.
For $A\with B$, the copy index function is simply inherited.
For $A \to B$, the copy index of an initial $(1, (\grey{i},
a))$ in $A$ is simply $\grey{i}$ -- in all other cases it is inherited.
It is direct that all requirements are met.

\subsubsection{Meager innocent strategies}
We now introduce the causal counterpart of the meager alternating
innocent strategies of Section~\ref{subsubsec:def_inn}. Those are
parallel innocent strategies in the sense of Section~\ref{sec:par_inn},
but on a restricted arena authorizing only Player replications: 

\begin{prop}
Consider $A$ a concrete arena. Then, setting events:
\[
\ev{A^+} = \{a' \in \ev{A} \mid \forall a^- \leq_A a',~\ind(a) =
\grey{0}\}\,,
\]
with other components inherited from $A$, yields an arena $A^+$.
\end{prop}
\begin{proof}
All verifications are straightforward. For symmetry, if $\theta : x
\sym_A y$ with $x, y \in \conf{A^+}$, then by Definition
\ref{def:concrete-arena2}, $\theta \in \ptilde{A}$. But likewise 
by Definition \ref{def:concrete-arena2}, for $\theta : x \sym_A^+ y$, 
$x \in \conf{A^+}$ iff $y \in \conf{A^+}$. Together those imply that the
restriction of $\tilde{A}$ to $A^+$ satisfies \emph{extension} -- other
axioms are easy. Finally, from Definition \ref{def:concrete-arena2}
again the polarized isomorphism families are $\ptilde{A^+} =
\tilde{A^+}$, and $\ntilde{A^+}$ restricted to identities. 
\end{proof}

So $A^+$ is $A$ where Opponent has only access to copy index $\grey{0}$.
This lets us define: 

\begin{defi}
Consider $A$ a concrete arena.

A \textbf{meager causal (pre)strategy} on $A$ is a causal (pre)strategy
$\bsigma : A^+$.
\end{defi}

As intended, this eliminates the infinity originating from Opponent
repetitions. As a side-effect, the isomorphism family of $\bsigma$
becomes trivial: as the only non-trivial symmetries of $A^+$ are
positive, it follows by condition \emph{thin} (see Lemma 3.28 of
\cite{cg2}) that symmetries of $\bsigma$ are reduced to identities -- so
a meager strategy is really a plain event structure.

Any causal strategy $\bsigma : A$ yields a meager strategy $\rf{\bsigma}
: A^+$, simply by restriction:

\begin{prop}
Consider $A$ a concrete arena, and $\bsigma : A$ any causal strategy.
Setting
\[
\ev{\rf{\bsigma}} = \{m' \in \ev{\bsigma} \mid \forall m^- \leq_\bsigma
m',~\ind(\pr_\bsigma(m)) = \grey{0}\}
\]
with other components inherited from $\bsigma$, yields a meager causal
strategy $\rf{\bsigma} : A^+$. 
\end{prop}
\begin{proof}
All conditions are straightforward verifications.
\end{proof}

We call $\rf{\bsigma} : A^+$ the \textbf{meager form} of $\bsigma : A$.
As $A^+$ is closed under positive symmetry, it is immediate that
$\rf{-}$ preserves positive isomorphism: if $\bsigma \simstrat
\bsigma'$, then $\rf{\bsigma} \simstrat \rf{\bsigma'}$. As illustration,
we show in Figure \ref{fig:meager_strat} the meager form of
$\mathbf{plet}_{\tunit,\tunit}$ as defined in Section~\ref{subsubsec:intr_plet}. Unlike Figure \ref{fig:plet}, this is now not
merely a symbolic representation, but an exhaustive display of the full
event structure of the meager form. We shall see that as for
alternating innocent strategies, meager forms of parallel innocent,
causally well-bracketed strategies provide a way to give complete formal
descriptions of the full infinite strategy. 

Without parallel innocence, taking the meager form is a lossy operation.
In Figure \ref{fig:non_meager} we show a typical augmentation of the
causal strategy obtained as the interpretation of 
\[
\vdash 
\newref\,x\!\!:=\!0\,\tin\,\lambda f^{\tunit \to
\tunit}.\,f\,(\tlet{v}{!x}{x\!\!:=\!1;\,\assrt\,(v\,=_\tnat\,1)})
: (\tunit \to \tunit) \to \tunit\,,
\]
displaying behaviour that is lost when taking the meager form: any
interference between different copies of the same branch -- a behaviour
typically banned by parallel innocence.  

\begin{figure}
\begin{minipage}{.45\linewidth}
\[
\xymatrix@R=2pt@C=-4pt{
\oc ((\tunit &\with& \tunit) &\with &(( \tunit & \with& \tunit) &\to&
\tunit))
&\vdash& \tunit\\
&&&&&&&&&&\qu^-
        \ar@{-|>}[dllllllllll]
        \ar@{-|>}[dllllllll]\\
\qu^+_{\grey{0}}
        \ar@{-|>}[d]&&
\qu^+_{\grey{1}}
        \ar@{-|>}[d]\\
\done^-_{\grey{0}}
        \ar@{.}@/^.1pc/[u]
        \ar@{-|>}[drrrrrrrr]&&
\done^-_{\grey{1}}
        \ar@{.}@/^.1pc/[u]
        \ar@{-|>}[drrrrrr]\\
&&&&&&&&\qu^+_{\grey{3}}
        \ar@{-|>}[dllll]
        \ar@{-|>}[dll]
        \ar@{-|>}[d]\\
&&&&\qu^-_{\grey{3, 0}}
        \ar@{.}@/^/[urrrr]
        \ar@{-|>}[d]
        \ar@{~}[rr]&&
\qu^-_{\grey{3,0}}
        \ar@{.}@/^/[urr]
        \ar@{-|>}[d]&&
\done^-_{\grey{3}}
        \ar@{-|>}[drr]\\
&&&&
\done^+_{\grey{3, 0}}
        \ar@{.}@/^.1pc/[u]&&
\done^+_{\grey{3, 0}}
        \ar@{.}@/^.1pc/[u]&&&&
\done^+ \ar@{.}@/_/[uuuuu]
}
\]
\caption{The meager form of $\mathsf{plet}_{\tunit, \tunit}$}
\label{fig:meager_strat}
\end{minipage}
\hfill
\begin{minipage}{.45\linewidth}
\[
\xymatrix@R=10pt@C=0pt{
&(\tunit
        \ar@{}[rrrr]|\to&&&&
\tunit) \ar@{}[rrrr]|\to&&&&
\tunit\\
&&&&&&&&&\qu^-
        \ar@{-|>}[dllll]\\
&&&&&\qu^+_{\grey{0}}
        \ar@{.}@/^/[urrrr]
        \ar@{-|>}[dlllll]
        \ar@{-|>}[dlll]\\
\qu^-_{\grey{0},\grey{i}}
        \ar@{.}@/^/[urrrrr]
        \ar@{-|>}[d]
&&\qu^-_{\grey{0},\grey{j}}
        \ar@{.}@/^/[urrr]
        \ar@{-|>}[dll]\\
\done^+_{\grey{0},\grey{i}}
        \ar@{.}@/^/[u]\\~
}
\]
\caption{Lossy meager form}
\label{fig:non_meager}
\end{minipage}
\end{figure}

This lets us define finiteness for parallel innocent strategies as in
Section~\ref{subsubsec:def_inn}:

\begin{defi}
Consider $A$ a concrete arena, and $\bsigma : A$ a parallel innocent
strategy.

We say that $\bsigma$ is \textbf{finite} if the set
$\ev{\rf{\bsigma}}_+ = \{m \in \ev{\rf{\bsigma}} \mid \pol(m) = +\}$
is finite.
\end{defi}

If $\bsigma$ is finite, its \textbf{size} is the cardinal of
$\ev{\rf{\bsigma}}_+$.

\subsubsection{Finite tests suffice} The key mechanism behind the
reduction to finite tests is to be able to restrict a parallel innocent
strategy following a finite subset of its meager form.

Say $x \in \conf{\bsigma}$ is \textbf{normal} iff for all $m^+
\imc_\bsigma m_1^-$ and $m^+ \imc_\bsigma m_2^-$ in $x$, $m_1 = m_2$. We
show that every normal $x\in \conf{\bsigma}$ has a unique representative
in $\rf{\bsigma}$.

\begin{lem}\label{lem:meager_ev}
Consider $A$ a concrete arena, $\bsigma : A$ a causal strategy, and $x
\in \conf{\bsigma}$ normal.

There is a unique $\rf{x} \in \conf{\rf{\bsigma}}$ s.t. $x \sym_\bsigma
\rf{x}$, and $\theta_x : x \sym_\bsigma \rf{x}$ is unique.
\end{lem}
\begin{proof}
\emph{Existence.} By induction on $x$. Consider $x \longcov{m} y$. If
$m$ is positive, by \emph{extension} there is an
extension of $\theta_x$ with $(m, n)$. As $n$ is positive
its negative dependencies are in $\rf{x}$ so their display
have copy index $\grey{0}$, so $n \in \ev{\rf{\bsigma}}$. If
$m$ is negative, by \emph{extension} on $A$, there is
\[
\pr_\bsigma\,\theta_x \cup \{(\pr_\bsigma(m), a)\} : \pr_\bsigma(x) \cup
\{\pr_\bsigma(m)\} \sym_A \pr_\bsigma(\rf{x}) \cup \{a\}\,,
\]
with $a$ characterised by $\ind(a)$, $\lbl(a)$, and $\apred(a)$. If $m$
is an answer, so is $a$ and by \emph{$\An$-narrow}, $\ind(a) =
\grey{0}$. If $m$ is a Question, we may not have $\ind(a) = \grey{0}$.
But then by \emph{$\Qu$-wide}, there is $a' \in \ev{A}$ s.t. $\apred(a')
= \apred(a)$, $\lbl(a') = \lbl(a)$ and $\ind(a') = \grey{0}$; but we must
prove that $a'$ is not already in $\pr_{\bsigma}(\rf{x})$. If it is,
there is $\theta_x(m') \in \rf{x}$ s.t. $\pr_\bsigma(\theta_x(m')) =
a'$. By courtesy $m'$ and $m$ are negative in $y$ with the same predecessor,
contradicting normality of $y$. So $y$ extends with $a'$ with copy index
$\grey{0}$. By \emph{transparent}, $\pr_{\bsigma}(\theta_x) \cup
\{(\pr_\bsigma(m), a')\} \in \tilde{A}$.
So, by $\sim$-receptivity of $\bsigma$, $\theta_x$ extends with $(m,
\rf{m})$ in $\bsigma$ s.t. $\rf{m} \in \ev{\rf{\bsigma}}$ as
required.

\emph{Uniqueness.} For $y_1, y_2 \in \conf{\rf{\bsigma}}$ s.t. $\theta_1
: x \sym_\bsigma y_1$ and $\theta_2 : x \sym_\bsigma y_2$, then $\theta
= \theta_2 \circ \theta_1^{-1} : y_1 \sym_\bsigma y_2$. But copy
indices of Opponent events in $\pr_{\bsigma} y_1$ and $\pr_\bsigma y_2$
are $\grey{0}$, so by \emph{$+$-transparent}, $\pr_\bsigma \theta \in
\ptilde{A}$. By Lemma 3.28 of \cite{cg2}, $y_1 = y_2$ and $\theta =
\id$, so $\theta_1 = \theta_2$.
\end{proof}

This does not depend on parallel innocence -- which comes in when
transporting \emph{events}:

\begin{lem}
Consider $A$ a concrete arena, $\bsigma : A$ parallel innocent, and $m
\in \ev{\bsigma}$.

Then, there exists a unique $\rf{m} \in \ev{\rf{\bsigma}}$ such that
$[m]_\bsigma \sym_\bsigma [\rf{m}]_\bsigma$.
\end{lem}
\begin{proof}
\emph{Pre-innocence} exactly states that the prime configuration
$[m]_\bsigma$ is normal. Hence, by Lemma \ref{lem:meager_ev}, there is a
unique $y \in \conf{\rf{\bsigma}}$ such that $[m]_\bsigma \sym_\bsigma
y$. But then, as symmetries are order-isomorphisms, $y$ is a prime
configuration $y = [\rf{m}]_\bsigma$ as required.
\end{proof}

Moreover, this assignment is preserved under symmetry:

\begin{lem}\label{lem:rf-stab-sym}
Consider $A$ a concrete arena, $\bsigma : A$ parallel innocent.

For any $\theta : x \sym_\bsigma y$ and $m \in x$, $\rf{m} =
\rf{\theta(m)}$.
\end{lem}
\begin{proof}
As $\theta$ is an order-iso, it restricts to $[m]_\bsigma \sym_\bsigma
[\theta(m)]_\bsigma$. Composition with $[m]_\bsigma \sym
[\rf{m}]_\bsigma$ and $[\theta(m)]_\bsigma \sym
[\rf{\theta(m)}]_\bsigma$ yields  $\varphi : [\rf{m}]_\bsigma
\sym_\bsigma [\rf{\theta(m)}]_\bsigma$, and $\pr_\bsigma(\varphi) \in
\ptilde{A}$ by \emph{$+$-transparent}. Hence, by Lemma 3.28 of
\cite{cg2}, $\varphi$ is an identity.
\end{proof}

From that, we may deduce the following:

\begin{cor}\label{cor:extr-finite}
For $A$ concrete, $\bsigma : A$ parallel innocent and
$\btau \cleq \rf{\bsigma}$ finite, 
\[
\ev{\bsigma \restrict \btau} = \{m \in \ev{\bsigma} \mid \rf{m} \in
\ev{\btau}\}
\]
with all other components inherited from $\bsigma$ yields a
finite innocent $\bsigma \restrict \btau : A$ s.t.
$\bsigma \restrict \btau \cleq \bsigma$.

Moreover, if $\bsigma$ is well-bracketed (resp. causally
well-bracketed), so is $\bsigma \restrict \btau$.
\end{cor}
\begin{proof}
The only non-trivial condition, \emph{extension} for $\tilde{\bsigma
\restrict \btau}$, follows from Lemma \ref{lem:rf-stab-sym}. 
\end{proof}

From that, we may finally prove that finite tests suffice.

\begin{cor}\label{cor:finite_tests_suffice}
Consider $A$ a concrete $-$-arena, and $\bsigma_1, \bsigma_2 : A$
parallel innocent.

If there is $\balpha : \oc A \vdash \tunit$ parallel innocent and
well-bracketed, such that
\[
\balpha \odot_\oc \bsigma_1 \eval\,,
\qquad
\qquad
\balpha \odot_\oc \bsigma_2 \div\,,
\]
there is $\balpha' \cleq \balpha$ parallel innocent,
well-bracketed, and finite, s.t. $\balpha' \odot_\oc \bsigma_1 \eval$ and
$\balpha' \odot_\oc \bsigma_2 \div$.
\end{cor}
\begin{proof}
Since $\balpha \odot_\oc \bsigma_1 \eval$, there is $x^\balpha \odot
x^{\bsigma_1} \in \conf{\balpha \odot \bsigma_1^\dagger}$ s.t.
$\pr_\balpha x^\balpha = x_A \parallel \{\qu^-, \done^+\}$. In
particular, $x^\balpha$ is finite. So the set $X = \{\rf{m} \mid m \in
x^\balpha\}$ is finite, so there is $\btau \cleq \rf{\balpha}$ a finite
meager strategy s.t. $X \subseteq \ev{\btau}$. By Corollary
\ref{cor:extr-finite}, $\balpha' = \balpha \restrict \btau : A$ is a
finite parallel innocent strategy s.t. $\balpha' \cleq \balpha$.
Moreover, by construction, $x^\balpha \in \conf{\balpha'}$ with the same
causal ordering as in $\balpha$, so that $\balpha' \odot_\oc \bsigma_1
\eval$ still. Finally, $\balpha' \odot_\oc \bsigma_2 \eval$ would
contradict $\balpha \odot_\oc \bsigma_2 \div$ since $\balpha' \cleq
\balpha$. As $\NAlt(\balpha') \subseteq \NAlt(\balpha)$, $\balpha'$ is
still well-bracketed. 
\end{proof}

In fact, meager forms of parallel innocent strategies can be expanded
back to the original strategy. This is not used in the technical
development, but we include it as Appendix \ref{app:exp_meager}.

\subsection{Factorization} \label{subsec:factor}
We focus on finite tests.  Unlike in the sequential argument of
Section~\ref{subsubsec:def_inn}, parallel innocent strategies have no ``first
Player move'' to reproduce first syntactically. Hence we organize our 
definability process differently. Its core is a \emph{factorization
result} (Corollary \ref{cor:factor_fst_arg}): namely, that every finite
test $\balpha : \oc (\with A_i) \vdash \tx$ may be obtained as 
\begin{eqnarray}
\balpha &\equiv& \fo(\balpha) \odot_\oc
\tuple{x_i\,\balpha_{k,1}\,\dots\,\balpha_{k,p_i} \mid i\in I,~k\in
K_i}\,,\label{eq:factor0}
\end{eqnarray}
with $\fo(\balpha)$ a strategy on a \emph{first-order type} and
$\balpha_{k,j}$ strictly smaller. This reduces finite definability to
that for finite \emph{first-order} strategies, dealt with in Section~\ref{subsubsec:fst_def}.

We first extract the components mentioned in
\eqref{eq:factor0}: the \emph{first-order substrategy} $\fo(\balpha)$,
and the \emph{argument substrategies} $\balpha_{k,j}$. 
\begin{figure}
\begin{minipage}{.46\linewidth}
\[
\scalebox{.9}{$
\xymatrix@C=-2pt@R=6pt{
\oc ((\tbool & \to & \tunit) &\with& ((\tunit & \to & \tbool) &\to&
\tbool))
&\vdash& \tbool\\
&&&&&&&&&& \red{\qu^-}
        \ar@[red]@{-|>}[dllllllll]
        \ar@[red]@{-|>}[dll]\\
&&\red{\qu}^{\red{+}}_{\grey{0}}
        \ar@{-|>}[dll]
        \ar@[red]@{-|>}[d]
&&&&&&\red{\qu}^{\red{+}}_{\grey{1}}
        \ar@{-|>}[dll]
        \ar@[red]@{-|>}[d]\\
\green{\qu}^{\green{-}}_{\grey{0,i}}
        \ar@{.}@/^/[urr]
        \ar@[green]@{-|>}[d]&&
\red{\done}^{\red{-}}_{\grey{0}}
        \ar@[red]@{.}@/^/[u]
        \ar@[red]@{-|>}[d]&&&&
\cyan{\qu}^{\cyan{-}}_{\grey{1,r}}
        \ar@{.}@/^/[urr]
        \ar@[cyan]@{-|>}[dll]
        \ar@[cyan]@{-|>}[drr]&&
\red{\ttrue}^{\red{-}}_{\grey{1}}
        \ar@[red]@{-|>}[dllllll]
        \ar@[red]@{.}@/^/[u]\\
\green{\ttrue}^{\green{+}}_{\grey{0,i}}
        \ar@[green]@{.}@/^/[u]&&
\red{\qu}^{\red{+}}_{\grey{2}}
        \ar@[red]@{-|>}[d]
        \ar@{-|>}[dll]&&
\cyan{\qu}^{\cyan{+}}_{\grey{1,r,0}}
        \ar@[cyan]@{.}@/^/[urr]
        \ar@[cyan]@{-|>}[d]&&&&
\cyan{\qu}^{\cyan{+}}_{\grey{2r+3}}
        \ar@[cyan]@{-|>}[d]\\
\purple{\qu}^{\purple{-}}_{\grey{2, k}}
        \ar@[purple]@{-|>}[d]
        \ar@{.}@/^/[urr]
&&\red{\done}^{\red{-}}_{\grey{2}}
        \ar@[red]@{.}@/^/[u]
        \ar@[red]@{-|>}[drrrrrrrr]&&
\cyan{\done}^{\cyan{-}}_{\grey{1,r,0}}
        \ar@[cyan]@{.}@/^/[u]
        \ar@[cyan]@{-|>}[drr]&&&&
\cyan{b}^{\cyan{-}}_{\grey{2r+3}}
        \ar@[cyan]@{.}@/^/[u]
        \ar@[cyan]@{-|>}[dll]\\
\purple{\tfalse}^{\purple{+}}_{\grey{2, k}}
        \ar@[purple]@{.}@/^/[u]
&&&&&&\cyan{b}^{\cyan{+}}_{\grey{1,r}}
        \ar@[cyan]@{.}@/^/[uuu]&&&&
\red{\ttrue}^{\red{+}}
        \ar@[red]@{.}@/_/[uuuuu]
}$}
\]
\end{minipage}
\hfill
\begin{minipage}{.46\linewidth}
\[
\scalebox{.9}{$
\xymatrix@C=-2pt@R=6pt{
\oc ((\tbool & \to & \tunit) &\with& ((\tunit & \to & \tbool) &\to&
\tbool))
&\vdash& \tbool\\
&&&&&&&&&& \red{\qu^-}
        \ar@[red]@{-|>}[dllllllll]
        \ar@[red]@{-|>}[dll]\\
&&\red{\qu}^{\red{+}}_{\grey{0}}
        \ar@{-|>}[dll]
        \ar@[red]@{-|>}[d]
&&&&&&\red{\qu}^{\red{+}}_{\grey{1}}
        \ar@{-|>}[dll]
        \ar@[red]@{-|>}[d]\\
\green{\qu}^{\green{-}}_{\grey{0,i}}
        \ar@{.}@/^/[urr]
        \ar@[green]@{-|>}[d]&&
\red{\done}^{\red{-}}_{\grey{0}}
        \ar@[red]@{.}@/^/[u]
        \ar@[red]@{-|>}[drrrrrrrr]&&&&
\cyan{\qu}^{\cyan{-}}_{\grey{1,r}}
        \ar@{.}@/^/[urr]
        \ar@[cyan]@{-|>}[dll]
        \ar@[cyan]@{-|>}[drr]&&
\red{\tfalse}^{\red{-}}_{\grey{1}}
        \ar@[red]@{-|>}[drr]
        \ar@[red]@{.}@/^/[u]\\
\green{\ttrue}^{\green{+}}_{\grey{0,i}}
        \ar@[green]@{.}@/^/[u]&&&&
\cyan{\qu}^{\cyan{+}}_{\grey{1,r,0}}
        \ar@[cyan]@{.}@/^/[urr]
        \ar@[cyan]@{-|>}[d]&&&&
\cyan{\qu}^{\cyan{+}}_{\grey{2r+3}}
        \ar@[cyan]@{-|>}[d]&&
\red{\tfalse}^{\red{+}}
        \ar@[red]@{.}@/_/[uuu]\\
&&
&&
\cyan{\done}^{\cyan{-}}_{\grey{1,r,0}}
        \ar@[cyan]@{.}@/^/[u]
        \ar@[cyan]@{-|>}[drr]&&&&
\cyan{b}^{\cyan{-}}_{\grey{2r+3}}
        \ar@[cyan]@{.}@/^/[u]
        \ar@[cyan]@{-|>}[dll]\\
&&&&&&\cyan{b}^{\cyan{+}}_{\grey{1,r}}
        \ar@[cyan]@{.}@/^/[uuu]&&&&
}$}
\]
\end{minipage}
\caption{A parallel innocent causally well-bracketed strategy}
\label{fig:ex_definability}
\end{figure}
We use as illustration the strategy with typical maximal augmentations
in Figure \ref{fig:ex_definability}. The \emph{first-order
sub-strategy}, in red, has events those depending on no Opponent
question besides the initial move: it is independent of Opponent's
exploration of the arguments, and is purely first-order.  The Player
questions in this first-order part play a special role; we call them
\emph{primary questions}. Intuitively, they correspond to occurrences of
variables not appearing in an argument to a variable call. In
Figure \ref{fig:ex_definability}, the primary questions are
$\red{\qu}^{\red{+}}_{\grey{0}}$, $\red{\qu}^{\red{+}}_{\grey{1}}$ and
$\red{\qu}^{\red{+}}_{\grey{2}}$.  Depending on their type, the primary
questions admit arguments that Opponent can access by playing
questions pointing to them. Parts of the strategy accessed in this way
are the \emph{argument sub-strategies} -- in Figure
\ref{fig:ex_definability} there are three,
respectively prompted by (\emph{i.e.} causally depending on)
$\green{\qu}^{\green{-}}_{\grey{0,i}}$,
$\purple{\qu}^{\purple{-}}_{\grey{2,k}}$ and
$\cyan{\qu}^{\cyan{-}}_{\grey{1,r}}$ and colored accordingly.
\begin{figure}
\[
\begin{array}{l}
f : \tbool \to \tunit, g : (\tunit \to \tbool) \to \tbool \vdash \\
\red{\mathbf{let}}
\left(
\begin{array}{l}
\red{x}\,\red{=}\,\red{f}\,\green{\ttrue}\\
\red{y}\,\red{=}\,\red{g}\,\color{cyan}(\lambda
z^{\tunit}.\,\mathbf{let}
\left(
\begin{array}{l}
u\,=\,z\\
v\,=\,g\,\bot
\end{array}
\right)\,\mathbf{in}\,v)
\end{array}
\right)
\,\red{\mathbf{in}}\,\red{(\mathbf{if}\,y\,\mathbf{then}\,(f\,}\purple{\tfalse}\red{;\,\ttrue)\,\mathbf{else}\,\tfalse)}
: \tbool
\end{array}
\]
\caption{Factorization and definability for Figure
\ref{fig:ex_definability}}
\label{fig:ex_def_term}
\end{figure}

The strategy of Figure \ref{fig:ex_definability} is exactly definable;
as illustration we show the term in Figure
\ref{fig:ex_def_term}, with subterms colored so as to match the four
components of the strategy\footnote{In the end, our definability process
will not quite give the term of Figure \ref{fig:ex_def_term} for the
strategy of Figure \ref{fig:ex_definability}, but a sequential version
as we do not know how to define first-order strategies in general -- see
Section~\ref{subsubsec:fst_def}.}.

The proof of factorization is organized as follows. In
Section~\ref{subsubsec:dec}, we extract the first-order part, and in
Section~\ref{subsubsec:flow} we reorganize it so as to be able to re-compose it
better. In Section~\ref{subsubsec:arg_substrat}, we extract the
argument substrategies. In Section~\ref{subsubsec:recomp}, we
conclude.

\subsubsection{Shallow substrategy}\label{subsubsec:dec} Consider
$A$ a type, and $\balpha : \intr{A}$ well-bracketed, parallel innocent,
causally well-bracketed, finite, and complete. We call such a strategy a
\textbf{test strategy}.

Necessarily $A$ has the form $A_1 \to \dots \to A_n \to \tx$ where $A_i
= A_{i,1} \to \dots \to A_{i, p_i} \to \tx_i$. Recall that $\tx, \tx_i$
range over ground types, \emph{i.e.} $\tunit$, $\tbool$ and $\tnat$. Up
to currying, we write
\[
\balpha : \oc (\with_{1\leq i \leq n} A_i) \vdash \tx\,,
\]
omitting semantic brackets. We often shorten the left hand side
part to $\oc (\with A_i)$, and reuse $A$ for the arena $\oc (\with A_i)
\vdash \tx$. Now, we start with the \emph{shallow substrategy}.

\begin{prop}
We define the \textbf{shallow substrategy} as having set of events
\[
\ev{\fst(\balpha)} = \{m \in \ev{\balpha} \mid \text{$[m]_\balpha$ has
at most one Opponent question}\}\,,
\]
and other components inherited from $\balpha$.
Then, $\fst(\balpha) : \oc (\with \tx_i) \vdash \tx$ is a test strategy.
\end{prop}
\begin{proof}
Write $\fst(A)$ for the arena $\oc (\with \tx_i) \vdash \tx$. First, for
each $m \in \ev{\fst(\balpha)}$, $\pr_\balpha(m) \in \ev{\fst(A)}$:
indeed, the least events in $\ev{A}$ but not in $\ev{\fst(A)}$
are Opponent questions. For \emph{extension}, as symmetries are
order-isos preserving polarities and Q/A labeling, they preserve
$\fst(\balpha)$. The conditions for a test strategy are direct by
restriction from $\balpha$.

We show $\fst(\balpha)$ complete. Consider $x \in
\confp{\fst(\balpha)}$. Since $\balpha$ is complete, there is $x
\subseteq y \in \conf{\balpha}$ complete. In particular, there is an
answer $a^+$ to the initial $\qu_0^-$. By Lemma
\ref{lem:key_caus_wb}, all moves of $x$ are below $a^+$ for
$\leq_\bsigma$. So setting $z = [a]_\bsigma$, $x \subseteq z$.
By Lemma \ref{lem:gcc_wb_aut}, all gccs leading to $a$ are
well-bracketed, so $[a]_\bsigma$ is complete. Finally, $[a]_\bsigma \in
\conf{\fst(\balpha)}$: indeed, if it has a non-initial
Opponent question $\qu^-$, it has an answer $b^+$ distinct from
$a^+$. But Player answers are maximal in gccs, contradicting
that any gcc can be extended to $a$.
\end{proof}

This captures the red part of Figure \ref{fig:ex_definability}.
All events of $\fst(\balpha)$ are in $\rf{\balpha}$: the negative
dependencies of $m \in \ev{\fst(\balpha)}$ are either answers or the
initial question, so by \emph{$\An$-narrow} their display has
copy index $\grey{0}$. Since $\balpha$ is finite, so is the set of
positive events of $\fst(\balpha)$. 

\subsubsection{The flow substrategy}\label{subsubsec:flow}
We must reconstruct $\balpha$ using the categorical structure of
$\CG$. But as distinct Player questions in the same $\tx_i$
may receive distinct argument substrategies, we need to relabel
$\fst(\balpha)$ to send those to distinct components. First we define:

\begin{defi}
A \textbf{primary question} of $\balpha$ is any $q^{\Qu,+} \in
\ev{\fst(\balpha)}$.
We write $\Q$ for the set of primary questions, and $\Q_i$ for the
primary questions displaying to $\tx_i$.
\end{defi}

By construction, $\Q = \uplus_{1\leq i \leq n} \Q_i$. As
$\ev{\fst(\balpha)}$ is finite, so is $\Q$, which allows us to set:

\begin{defi}\label{def:flow}
The \textbf{flow substrategy} $\flow(\balpha) : \bigotimes_{1\leq
i \leq n} \bigotimes_{q \in \Q_i} \tx_i \vdash \tx$ is $\fst(\balpha)$
with
\[
\begin{array}{rclcrcl}
\pr_{\flow(\balpha)}(m^{\Qu,-}) &=& (2,a) &\text{if}&
\pr_{\fst(\balpha)}(m^{\Qu,-})
&=& (2, a)\\
\pr_{\flow(\balpha)}(m^{\An,+}) &=& (2, a) &\text{if}&
\pr_{\fst(\balpha)}(m^{\An,+}) &=& (2, a)\\
\pr_{\flow(\balpha)}(m^{\Qu,+}) &=& (1, (i, (m, a))) &\text{if}&
\pr_{\fst(\balpha)}(m^{\Qu,+}) &=& (1, (\grey{j}, (i, a)))\\
\pr_{\flow(\balpha)}(m^{\An, -}) &=& (1, (i, (\just(m), a))) &\text{if}&
\pr_{\fst(\balpha)}(m^{\An, -}) &=& (1, (\grey{j}, (i, a)))\,,
\end{array}
\]
\emph{i.e.} sending each $q \in \Q$ and its answers to the copy of
$\tx_i$ specified by indices $i, q$.
\end{defi}

It is a test strategy.
We show in Figure \ref{fig:flow} the maximal augmentations of the flow
substrategy for Figure \ref{fig:ex_definability}, tagging each
component by the corresponding primary question.

\begin{figure}
\begin{minipage}{.45\linewidth}
\[
\xymatrix@C=-2pt@R=5pt{
\tunit^{\qu_{\grey{0}}^+} &\tensor& \tunit^{\qu_{\grey{2}}^+} &\tensor&
\tbool^{\qu_{\grey{1}}^+} & \vdash& \tbool\\
&&&&&&\qu^-
        \ar@{-|>}[dllllll]
        \ar@{-|>}[dll]\\
\qu^+   \ar@{-|>}[d]&&&&
\qu^+   \ar@{-|>}[d]\\
\done^- \ar@{-|>}[drr]
        \ar@{.}@/^/[u]&&&&
\ttrue^-\ar@{-|>}[dll]
        \ar@{.}@/^/[u]\\
&&\qu^+
        \ar@{-|>}[d]\\
&&\done^-
        \ar@{.}@/^/[u]
        \ar@{-|>}[drrrr]\\
&&&&&&\ttrue^+
        \ar@{.}@/_/[uuuuu]
}
\]
\end{minipage}
\hfill
\begin{minipage}{.45\linewidth}
\[
\xymatrix@C=-2pt@R=5pt{
\tunit^{\qu_{\grey{0}}^+} &\tensor& \tunit^{\qu_{\grey{2}}^+} &\tensor&
\tbool^{\qu_{\grey{1}}^+} & \vdash& \tbool\\
&&&&&&\qu^-
        \ar@{-|>}[dllllll]
        \ar@{-|>}[dll]\\
\qu^+   \ar@{-|>}[d]&&&&
\qu^+   \ar@{-|>}[d]\\
\done^- \ar@{-|>}[drrrrrr]
        \ar@{.}@/^/[u]&&&&
\tfalse^-\ar@{-|>}[drr]
        \ar@{.}@/^/[u]\\
&&&&&&\tfalse^+
        \ar@{.}@/_/[uuu]\\~\\~\\~
}
\]
\end{minipage}
\caption{Augmentations of the flow substrategy for Figure \ref{fig:ex_definability}}
\label{fig:flow}
\end{figure}

\subsubsection{The argument substrategies}\label{subsubsec:arg_substrat}
Next, we focus on the \emph{higher-order} structure, aiming to extract
the arguments to (the variable calls corresponding to) the primary
questions.

Fix a primary question $q \in \Q_i$. It displays to an initial event in
$A_i$, which is:
\[
\oc A_{i, 1} \lin \dots \lin~\oc A_{i, p_i} \lin \tx_i\,.
\]

Argument substrategies are accessed by Opponent questions pointing to
primary questions. Up to symmetry, there are $p_i$ Opponent questions
pointing to $q$, matching the $p_i$ arguments of $A_i$. From
now on, if $q \in \Q_i$ is a primary question and $q \imc_\bsigma
m^{\Qu,-}$ an Opponent question, we shall say that $m$ is \textbf{in
component $j$} if it displays to an initial move of $\oc A_{i,j}$. For
$q \in \Q_i$ and $1 \leq j \leq p_i$, we shall
extract the \emph{argument sub-strategy} $\alpha_{q, j}$ initiated by
Opponent questions pointing to $q$ in component $j$. As the
strategy provides the information for an argument of $A_i$ it must live
in $A_{i,j}$; but it can still access the context, so we aim
for:
\[
\balpha_{q, j} : \oc (\with A_i) \vdash\,\oc A_{i,j}\,.
\]

To do this, we assign to all events of $\balpha$ \emph{tags}, as
follows:

\begin{defi}\label{def:tag}
Consider $m \in \ev{\balpha}$. We write:
\[
\begin{array}{rclcl}
m &\inplus& \fst(\balpha)&\Leftrightarrow&
        \text{$[m]_\balpha$ comprises exactly one Opponent question,}\\
m &\inplus& \balpha_{q, j} &\Leftrightarrow&
        \text{there is $q \imc_\balpha n^{\Qu,-}$ in component $j$, such
that $n \leq_\balpha m$,} 
\end{array}
\]
where in the second clause $q \in \Q_i$ a primary question, and $1\leq
j \leq p_i$.
\end{defi}

Any $m \in \ev{\balpha}$ is tagged as either in $\fst(\balpha)$ or
in one argument sub-strategy. Crucially, each event receives
\emph{exactly one tag} -- this is where our structural constraints 
strike in:

\begin{lem}\label{lem:tag}
Every event $m \in \ev{\balpha}$ receives exactly one tag following
Definition \ref{def:tag}.
\end{lem}
\begin{proof}
First, each $m \in \ev{\balpha}$ receives \emph{at least} one tag.
Any $[m]_\balpha$ contains at least one Opponent question: the
initial move. If it contains exactly one Opponent question, $m
\inplus \fst(\balpha)$. Assume there are at least two. Take $n
\leq_\balpha m$ minimal s.t. it is a non-initial Opponent question.
Then its immediate predecessor is some $q \in \Q_i$; and so there
is $1 \leq j \leq p_i$ s.t. $m \inplus \balpha_{q, j}$.

We prove that $m$ receives \emph{at most} one tag.  Clearly if $m
\inplus \alpha_{q, j}$ for some $q \in \Q_i$ and $1\leq j \leq p_i$,
there are at least two Opponent questions in $[m]_\balpha$ so we cannot
have $m \inplus \fst(\balpha)$.  Assume $m \inplus \alpha_{q, j}$ and $m
\inplus \alpha_{q', j'}$ for $q \in \Q_i, q' \in \Q_{i'}$, $1\leq j \leq
p_i$ and $1\leq j' \leq p_{i'}$. We first show that $q = q'$; seeking a
contradiction assume they are distinct. But $q$ and $q'$ cannot be
comparable: if $q \leq_\balpha q'$, $[q']_\balpha$ has at least two
Opponent questions, contradicting $q' \in \ev{\fst(\balpha)}$.

Take $\rho \imc m, \rho' \imc m \in \gcc(\balpha)$, respectively
passing through $q$ and $q'$. Diagrammatically:
\[
\xymatrix@R=5pt@C=10pt{
&&&m_1  \ar@{-|>}[r]&
~       \ar@{.}[r]&
~       \ar@{-|>}[r]&
q       \ar@{-|>}[r]&
~       \ar@{.}[r]&
~       \ar@{-|>}[r]&
m_k     \ar@{-|>}[dr]\\
~\ar@{.}[r]&~\ar@{-|>}[r]&
m_0     \ar@{-|>}[ur]
        \ar@{-|>}[dr]&&&&&&&&
m_{k+1} \ar@{-|>}[r]&~\ar@{.}[r]&~\ar@{-|>}[r]&m\\
&&&n_1  \ar@{-|>}[r]&
~       \ar@{.}[r]&
~       \ar@{-|>}[r]&
q'      \ar@{-|>}[r]&
~       \ar@{.}[r]&
~       \ar@{-|>}[r]&
n_p
        \ar@{-|>}[ur]
}
\]
and since $q, q'$ are distinct, the diagram may be chosen with $X
= \{m_1, \dots, m_k\}$ and $Y = \{n_1, \dots, n_p\}$ disjoint. By
parallel innocence $m_1$ and $n_1$ are positive. By Lemma
\ref{lem:pol_caus} so must be $m_{k+1}$, so $m_k$ and $n_p$ are
negative. We have a \emph{globule} as in Definition
\ref{def:caus_wb}, so $X$ and $Y$ are complete, and in particular $q$ is
answered in $\rho$. Writing $q = m_i$, $m_{i+1}$ must answer
$q$. But since $m \inplus \alpha_{q, j}$, $q \imc_\balpha n^{Qu,-}
\leq_\balpha m $ where $n$ is a negative question in component $j$. Then
necessarily, $n = m_{i+1}$, or we get a contradiction with parallel
innocence. Thus $m_{i+1}$ is both a question and an answer,
contradiction. So, $q = q'$.  Finally, from Lemma \ref{lem:arguments}, $j = j'$. 
\end{proof}

This shows any $m \in \ev{\balpha}$ can always be attributed to
exactly \emph{one} of the sub-strategies we wish to extract.
Accordingly, the \emph{argument sub-strategies} will be defined with
events
\[
\ev{\balpha_{q ,j}} = \{m \in \ev{\balpha} \mid m \inplus \balpha_{\q,
j}\}\,,
\]
completed to ess by inheriting the components from $\balpha$, as will be
made explicit later.

The display map requires a careful reindexing of
events ending up on the right hand side, illustrated in
Figures \ref{fig:ev_substrat} and \ref{fig:sub_aug}.
For this we split $\ev{\balpha_{q, j}}$ in two
subsets: on the one hand, we have those events that depend
\emph{statically}, \emph{i.e.} with respect to $\leq_A$ (through
$\pr_\balpha$) on the primary question $q$ -- in Figure
\ref{fig:ev_substrat}, those are $\qu^-_{\grey{i,r}}, \qu^+_{\grey{1, r,
0}}, \done^-_{\grey{1, r, 0}}$ and $b^+_{\grey{1, r}}$. On the other
hand, the remaining events must follow from new calls to variables
in the context -- in Figure \ref{fig:ev_substrat}, those are
$\qu^+_{\grey{2r+3}}$ and $b^-_{\grey{2 r + 3}}$.
\begin{figure}
\begin{minipage}{.47\linewidth}
\[
\scalebox{.9}{$
\xymatrix@C=-2pt@R=5pt{
\oc ((\tbool & \to & \tunit) &\with& ((\tunit & \to & \tbool) &\to&
\tbool))
&\vdash& \tbool\\
&&&&&&&&&& \grey{\qu^-}
        \ar@[grey]@{-|>}[dllllllll]
        \ar@[grey]@{-|>}[dll]\\
&&\grey{\qu}^{\grey{+}}_{\grey{0}}
        \ar@[grey]@{-|>}[dll]
        \ar@[grey]@{-|>}[d]
&&&&&&\grey{\qu}^{\grey{+}}_{\grey{1}}
        \ar@[grey]@{-|>}[dll]
        \ar@[grey]@{-|>}[d]\\
\grey{\qu^-_{\grey{0,i}}}
        \ar@[grey]@{.}@/^/[urr]
        \ar@[grey]@{-|>}[d]&&
\grey{\done}^{\grey{-}}_{\grey{0}}
        \ar@[grey]@{.}@/^/[u]
        \ar@[grey]@{-|>}[d]&&&&
\qu^-_{\grey{1,r}}
        \ar@[grey]@{.}@/^/[urr]
        \ar@{-|>}[dll]
        \ar@{-|>}[drr]&&
\grey{\ttrue}^{\grey{-}}_{\grey{1}}
        \ar@[grey]@{-|>}[dllllll]
        \ar@[grey]@{.}@/^/[u]\\
\grey{\ttrue^+_{\grey{0,i}}}
        \ar@[grey]@{.}@/^/[u]&&
\grey{\qu}^{\grey{+}}_{\grey{2}}
        \ar@[grey]@{-|>}[d]
        \ar@[grey]@{-|>}[dll]&&
\qu^+_{\grey{1,r,0}}
        \ar@{.}@/^/[urr]
        \ar@{-|>}[d]&&&&
\qu^+_{\grey{2r+3}}
        \ar@{-|>}[d]\\
\grey{\qu^-_{\grey{2, k}}}
        \ar@[grey]@{-|>}[d]
        \ar@[grey]@{.}@/^/[urr]
&&\grey{\done}^{\grey{-}}_{\grey{2}}
        \ar@[grey]@{.}@/^/[u]
        \ar@[grey]@{-|>}[drrrrrrrr]&&
\done^-_{\grey{1,r,0}}
        \ar@{.}@/^/[u]
        \ar@{-|>}[drr]&&&&
b^-_{\grey{2r+3}}
        \ar@{.}@/^/[u]
        \ar@{-|>}[dll]\\
\grey{\tfalse^+_{\grey{2, k}}}
        \ar@[grey]@{.}@/^/[u]
&&&&&&b^+_{\grey{1,r}}
        \ar@{.}@/^/[uuu]&&&&
\grey{\ttrue}^{\grey{+}}
        \ar@[grey]@{.}@/_/[uuuuu]
}$}
\]
\caption{Events of $\ev{\balpha_{\qu_{\grey{1}}^+, 1}}$ of Figure
\ref{fig:ex_definability}}
\label{fig:ev_substrat}
\end{minipage}
\hfill
\begin{minipage}{.5\linewidth}
\[
\scalebox{.9}{$
\xymatrix@C=-5pt@R=7.5pt{
\oc ((\tbool & \to & \tunit) &\with& ((\tunit & \to & \tbool) &\to&
\tbool))
&\vdash& \oc (\tunit & \to & \tbool)\\
&&&&&&&&&&&&\qu^-_{\grey{r}}
        \ar@{-|>}[dllll]
        \ar@{-|>}[dll]\\
&&&&&&&&\qu^+_{\grey{2r + 3}}
        \ar@{-|>}[d]&&
\qu^+_{\grey{r,0}}
        \ar@{-|>}[d]
        \ar@{.}@/^/[urr]\\
&&&&&&&&b^-_{\grey{2r+3}}
        \ar@{.}@/_/[u]
        \ar@{-|>}[drrrr]&&
\done^-_{\grey{r,0}}
        \ar@{.}@/_/[u]
        \ar@{-|>}[drr]\\
&&&&&&&&&&&&b^+_{\grey{r}}
        \ar@{.}@/_/[uuu]\\~\\~\\~\\
}$}
\]
\caption{The corresponding aug. of $\balpha_{\qu_{\grey{1}}^+,1}$}
\label{fig:sub_aug}
\end{minipage}
\end{figure}
These two subsets are treated differently when defining the new
display map: the former are left unchanged, while the latter are
reindexed as in Figure \ref{fig:sub_aug}.

We introduce notations for the canonical embeddings of the set of moves
$\ev{A_i}$ and $\ev{A_{i,j}}$ into $\ev{A}$. More precisely, it will be
convenient, for each primary question $q \in \Q_i$, to write
\[
\inj_q(-) : \ev{A_i} \to \ev{A}
\]
the injection adding the sequence of tags addressing
$A_i$ within $A$, originating from the tagged disjoint unions
involved in all arena constructions -- in particular, it maps the
initial move of $A_i$ to $\pr_\balpha(q)$. Likewise,
$\inj_{q, j} : \ev{A_{i,j}} \to \ev{A}$ addresses the
$j$-th argument of $q$. Then:

\begin{defi}\label{def:pr_argument}
We define a display map for $\balpha_{q, j}$ by setting, for $m \in
\ev{\balpha_{q, j}}$:
\[
\begin{array}{rcllrcl}
\pr_{\balpha_{q, j}}(m) &=& \inj_\labr(a)
&\text{if}& \pr_\balpha(m) &=& \inj_{q, j}(a)\,,\\
\pr_{\balpha_{q, j}}(m) &=& \pr_\balpha(m)
&\text{otherwise,\!\!\!\!\!\!\!\!\!\!\!\!\!\!\!\!\!\!}
\end{array}
\]
where $\inj_\labl(a) = (1, a)$ and $\inj_\labr(a) = (2, a)$.
\end{defi}

Altogether, this lets us extract $\balpha_{q, j}$ as intended:

\begin{prop}
Consider $q \in \Q_i$ and $1 \leq j \leq p_i$.
The \textbf{argument substrategy} for $q, j$ is $(\ev{\balpha_{q, j}},
\leq_{q, j}, \conflict_{q, j}, \tilde{\balpha_{q, j}},
\pr_{\balpha_{q,j}})$, with components $\leq_{q, j}$ and $\conflict_{q,
j}$ the restrictions of $\balpha$, 
$\tilde{\balpha_{q, j}} = \{\theta \cap \ev{\balpha_{q, j}}^2 \mid \theta
\in \tilde{\balpha}\}$,
and $\pr_{\balpha_{q, j}}$ in Definition \ref{def:pr_argument}.

Then, $\balpha_{q, j} : \oc (\with A_i) \vdash\,\oc A_{i,j}$ is
well-bracketed, causally well-bracketed, parallel innocent.
\end{prop}
\begin{proof}
A routine verification. The key point is that as symmetries of $\balpha$
are order-isos displayed to symmetries of $A$, it follows that
they preserve the tag as in Definition \ref{def:tag}.
\end{proof}

Finally, we get rid of the $\oc$ on the right hand side, using
\emph{dereliction} $\der_A : \oc A \vdash A$.

\begin{prop}\label{prop:arg_substrat}
Consider $q \in \Q_i$ and $1 \leq j \leq p_i$.
Then, $\balpha_{q, j}^\bullet = \der_{A_{i,j}} \odot \balpha_{q, j} :
\oc (\with A_i) \vdash A_{i,j}$ is a test strategy with size strictly
lesser than $\balpha$.
\end{prop}
\begin{proof}
Using Proposition \ref{prop:comp_pcov}, it is easy that
$\balpha_{q, j}^\bullet$ is positively isomorphic to the strategy
obtained from $\balpha_{q, j}$ by restricting the initial Opponent
question to copy index $\grey{0}$. This also informs an injection of
$\rf{\balpha_{q, j}^\bullet}$ into $\rf{\balpha}$ not reaching the
primary question $q$, from which follows the announced size constraint.
\end{proof}

Note $\balpha_{q, j}$ is recovered from $\balpha_{q,
j}^\bullet$, via the ``Bang lemma'' \cite{ajm} (see Appendix
\ref{app:bang}):

\begin{restatable}{lem}{bang}\label{lem:bang}
For concrete arenas $A, B$ with $B$ pointed and $\bsigma
\in \CG(\oc A, \oc B)$,
\[
(\der_B \odot \bsigma)^\dagger \simstrat \bsigma\,.
\]
\end{restatable}

Summing up, from the original strategy
$\balpha : \oc (\with A_i) \vdash \tx$, we have now constructed:
\[
\begin{array}{rclcl}
\flow(\balpha) &:& \otimes_{1\leq i \leq n} \otimes_{q \in \Q_i}
\tx_i \vdash \tx\\
\balpha_{q, j}^\bullet &:& \oc (\with A_i) \vdash A_{i,j} 
&\qquad&
\text{for each $q \in \Q_i$ and $1\leq j \leq p_i$.}
\end{array}
\]

For each primary question $q \in \Q_i$, we use the cartesian closed
internal language to form
\[
x_1 : A_1, \dots, x_n : A_n \vdash x_i\,\balpha^\bullet_{q,
1}\,\dots\,\balpha^\bullet_{q, p_i} : \tx_i\,.
\]

Let us write $\balpha_q : \oc (\with A_i) \vdash \tx_i$ for the
resulting strategy.  Then, finally,
\begin{eqnarray}
\recomp(\balpha) ~~~ = ~~~ \flow(\balpha) \odot (\otimes_{1\leq i \leq
n}
\otimes_{q \in \Q_i} \balpha_{q}) \odot \bdelta_{\with A_i} &:& \oc
(\with A_i) \vdash
\tx\,,\label{eq:recomp} 
\end{eqnarray}
is our candidate to reconstruct $\balpha$. Here, for $B$ an arena and $n
\in \mathbb{N}$, we write $\bdelta_{B} : \oc B \vdash (\oc B)^{\tensor
n}$ for the obvious strategy (leaving $n$ implicit). In the sequel we
may only write $\bdelta$.

\subsubsection{Positions of $\recomp(\balpha)$}\label{subsubsec:recomp}
We expect that $\recomp(\balpha) \simstrat \balpha$, but we shall only
prove $\recomp(\balpha) \equiv \balpha$ --
this is simpler as positions compose relationally.

\begin{figure}
\begin{mathpar}
\scalebox{.95}{$
\inferrule
        { \x_A \in \coll A \\
          \x_B \in \coll B }
        { \x_A \tensor \x_B \in \coll (A\tensor B) }
$}
\and
\scalebox{.95}{$
\inferrule
        { \x_A \in \coll A \\
          \x_C \in \collo C }
        { \x_A \lin \x_C \in \collo (A\lin C) }
$}
\and
\scalebox{.95}{$
\inferrule
        { \x \in \collo A_i \\
          (i\in I) }
        { (i, \x) \in \collo (\with_{i\in I} A_i) }
$}
\and
\scalebox{.95}{$
\inferrule
        { (\x_C^i \in \collo C)_{i\in I} }
        { [\x_C^i \mid i \in I] \in \coll \oc C }
$}
\end{mathpar}
\caption{Syntax for positions for $A, B, C$ arenas with $C$ strict}
\label{fig:syn_pos}
\end{figure}
To help reason on positions, we adopt a syntax presented in Figure
\ref{fig:syn_pos}, following the bijections of  Lemma
\ref{lem:rel_conf}. We also write $\x_A \vdash \x_B \in \coll (A\vdash
B)$ for all $\x_A \in \coll A$ and $\x_B \in \coll B$. 

We now analyse the positions of the recomposition \eqref{eq:recomp}. We
start with:

\begin{lem}\label{lem:pos_delta}
For $B$ strict, the positions $\coll \bdelta_B$ of $\bdelta_B : \oc B
\vdash (\oc B)^{\tensor n}$ are exactly those 
\[
\left(\sum_{1 \leq i \leq n} \x_i \vdash \otimes_{1\leq i \leq n}
\x_i \right) 
\in
\int \left( \oc B \vdash (\oc B)^{\tensor n} \right)
\]
where $\x_i \in \coll \oc B$ for all $1 \leq i \leq n$.
\end{lem}

By a direct variation of Lemma \ref{lem:cc_pcov}.
Next we analyse the positions of $\balpha_q$ for $q \in \Q_i$:

\begin{lem}\label{lem:pos_alphaq}
For any $q \in \Q_i$, the non-empty positions $\collo \balpha_q$
are exactly those of the form
\[
\left([(i, \y_{q, 1} \lin \dots \lin \y_{q, p_i} \lin \v_q)] +
\sum_{1\leq j \leq
p_i}
\x_{q, j} \right) \vdash \v_q
\in
\int (\oc (\with A_i) \vdash \tx_i)
\]
where for each $1\leq j \leq p_i$, $(\x_{q, j} \vdash \y_{q, j})
\in \coll \balpha_{q, j}$, and for $\v_q \in \collo \tx_i$.
\end{lem}
\begin{proof}
From the laws of Seely categories, $\balpha_q$ is positively isomorphic
to the composition
\[
\xymatrix@C=70pt{
\oc (\with A_i) 
        \ar[r]^{\bdelta\,\,\,\,\,\,\,\,\,}&
\otimes_{p_i+1} \oc (\with A_i)
        \ar[r]^{(\otimes_{1\leq j \leq p_i}
(\balpha_{q,j}^\bullet)^\dagger) \tensor \intr{x_i}\,\,\,\,\,\,\,\,\,}
&
(\otimes_{1\leq j \leq p_i} \oc A_{i,j}) \tensor A_i
        \ar[r]^{\,\,\,\,\,\,\,\,\,\evm}&
\tx_i
}
\]
in $\CG$. The characterisation follows from Proposition
\ref{prop:vis_coll}, Lemma \ref{lem:pos_delta}, Lemma \ref{lem:bang},
and a direct verification analogous to Lemma \ref{lem:pos_delta} for
other copycat-like strategies involved. 
\end{proof}

From all those, we may characterise the positions of $\recomp(\balpha)$
as

\begin{cor}\label{cor:recomp}
Non-empty positions of $\recomp(\balpha)$ are exactly those of the form
\[
\sum_{1\leq i \leq n} \sum_{q \in \Q'_i} \left([(i, \y_{q,1} \lin \dots
\lin
\y_{q,p_i} \lin \v_q)] + \sum_{1\leq j \leq p_i} \x_{q, j}\right)
\vdash
\v
\qquad
\in
\qquad
\int (\oc (\with A_i) \vdash \tx)
\]
where for all $1\leq i \leq n$, $\Q'_i$ is a subset of $\Q_i$, all
$\v_q$ are non-empty, and:
\[
\left(\tensor_{1 \leq i \leq n} \tensor_{q \in \Q'_i} \v_q \vdash
\v\right) \in \collo \flow(\balpha)\,,
\qquad
\qquad
\left((\x_{q, j} \vdash \y_{q, j}) \in \coll \balpha_{q, j}\right)_{q
\in \Q'_i, 1\leq j \leq p_i}\,.
\]
\end{cor}
\begin{proof}
Direct from Lemmas \ref{lem:pos_delta}, \ref{lem:pos_alphaq}
and Proposition \ref{prop:vis_coll}.
\end{proof}

\subsubsection{Positions of $\balpha$}
We write $\confc{\balpha}$ for the complete, non-empty
configurations of $\balpha$. If $x \in \confc{\balpha}$, then
$\fst(x) = \{m \in x \mid m \inplus \fst(\balpha)\} \in
\confc{\fst(\balpha)}$
and we write $\Q^x = \Q \cap x$, and $\Q_i^x$ likewise. For each $q \in
\Q^x_i$ and $1 \leq j \leq p_i$, we also have
$x_{q, j} = \{m \in x \mid m \inplus \balpha_{q, j}\} \in
\conf{\balpha_{q,j}}$.

From Lemma \ref{lem:tag}, this informs 
$x = \fst(x) \uplus \left(\biguplus_{1\leq i \leq n} \biguplus_{q \in
\Q_i^x} \biguplus_{1\leq j \leq p_i} x_{q, j}\right)$.

We show that all complete non-empty
configurations of $\balpha$ arise in this way:

\begin{lem}\label{lem:char_pos_alpha}
This yields a bijection between $\confc{\balpha}$ and pairs
$(x, (x_{q, j})_{1\leq i \leq n, q \in \Q_i^{x}, 1\leq j
\leq p_i})$
where $x \in \confc{\fst(\balpha)}$, $x_{q, j} \in
\conf{\balpha_{q, j}}$ complete for all $1\leq i \leq n$, $q \in \Q_i^{x}$
and $1\leq j \leq p_i$. 

Moreover, writing $[x, (x_{q, j})_{i,q,j}] = x \uplus (\uplus_{i,q,j}
x_{q,j}) \in \conf{\balpha}$ this correspondence, we have
\[
\scalebox{.89}{$
\pr_{\balpha}\left([x, (x_{q, j})_{i,q,j}]\right)
=
\left(
\biguplus_{1\leq i \leq n} \biguplus_{q \in \Q^x_i}
\left(
 \left[ \inj_q(z_{q, 1} \lin \dots \lin z_{q, p_i} \lin
v_q) \right] \uplus 
\left[\biguplus_{1\leq j \leq p_i} \inj_\labl(y_{q, j}) \right]
\right)\right)\uplus \inj_\labr(v)
$}
\]
where we have, for all $1\leq i \leq n$, $q \in \Q_i^x$ and $1\leq j \leq p_i$:
\[
\pr_{\flow(\balpha)}(x) =
\inj_\labl \left(\tensor_{1 \leq i \leq n} \tensor_{q \in \Q^x_i}
v_q \right)
\uplus \inj_r(v)
\qquad
\qquad
\pr_{\balpha_{q, j}}(x_{q, j}) = 
\inj_\labl(y_{q, j}) \uplus \inj_\labr(z_{q, j})\,.
\]
\end{lem}
\begin{proof}
From $x \in \confc{\balpha}$, we get $(\fst(x), (x_{q, j})_{i,q,j})$
from the decomposition above. Reciprocally, from $(x, (x_{q,
j})_{i,q,j})$ we get a configuration in $\confc{\balpha}$ as
their union; it is down-closed by construction and consistent by
determinism of $\balpha$ (as any immediate negative conflict originates
from the arena, and hence appears in one of the components).  Finally,
the characterization of the display map follows from 
display maps of $\flow(\balpha)$ and $\balpha_{q, j}$.
\end{proof}

From this we may finally conclude the proof of factorization:

\begin{cor}\label{cor:fin_recomp}
The strategies $\balpha$ and $\recomp(\balpha)$ are positionally
equivalent.
\end{cor}
\begin{proof}
Taking symmetry classes from Lemma \ref{lem:char_pos_alpha},
we obtain the same characterization of non-empty complete positions of
$\balpha$ as in Corollary \ref{cor:recomp}.
\end{proof}

\subsubsection{Syntactic factorization} Finally, we must reformulate the
result relying on the cartesian closed structure only. The
\textbf{first-order substrategy} $\fo(\balpha) \in \CG_\oc(\with_{1\leq
i\leq n} \with_{q \in \Q_i} \tx_i, \tx)$ is obtained in the obvious way
from $\flow(\balpha)$ using the Seely category structure. Using
Corollary \ref{cor:fin_recomp}, Proposition \ref{prop:arg_substrat}, and
laws of a Seely category, we conclude: 

\begin{cor}\label{cor:factor_fst_arg}
Any test strategy $\balpha : \oc (\with A_i) \vdash \tx$ factors as a
composition of test strategies
\[
\balpha \equiv \fo(\balpha) \odot_\oc
\tuple{x_i\,\balpha_{q,1}^\bullet\,\dots\,\balpha_{q,p_i}^\bullet \mid
1\leq i \leq n,~q \in \Q_i}\,,
\]
where for all $1\leq i \leq n, q \in \Q_i$ and $1\leq j \leq p_i$,
$\balpha_{q, j}^\bullet$ has size strictly lesser than $\balpha$.
\end{cor}

\subsection{Finite Definability} \label{subsec:fin_def} Corollary
\ref{cor:factor_fst_arg} allows us to handle the higher-order structure,
it only remains to prove definability for first-order test strategies.

\begin{figure}
\[
\scalebox{.8}{$
\xymatrix@R=4pt@C=-3pt{
\oc (1 &\with& 1 & \with & 1 & \with & 1 & \with & 1) &\vdash& 1\\
&&&&&&&&&& \qu^-
        \ar@{-|>}[dllllllllll]
        \ar@{-|>}[dllllll]
        \ar@{-|>}[dll]\\
\qu^+_{\grey{0}}
        \ar@{-|>}[d]&&&&
\qu^+_{\grey{1}}
        \ar@{-|>}[d]&&&&
\qu^+_{\grey{2}}
        \ar@{-|>}[d]\\
\done^-_{\grey{0}}
        \ar@{.}@/^/[u]
        \ar@{-|>}[drr]
&&&&
\done^-_{\grey{1}}
        \ar@{.}@/^/[u]
        \ar@{-|>}[dll]
        \ar@{-|>}[drr]
&&&&
\done^-_{\grey{2}}
        \ar@{.}@/^/[u]
        \ar@{-|>}[dll]
\\
&&\qu^+_{\grey{3}}
        \ar@{-|>}[d]&&&&
\qu^+_{\grey{4}}
        \ar@{-|>}[d]\\
&&\done^-_{\grey{3}}
        \ar@{.}@/^/[u]
        \ar@{-|>}[drrrrrrrr]&&&&
\done^-_{\grey{4}}
        \ar@{.}@/^/[u]
        \ar@{-|>}[drrrr]\\
&&&&&&&&&&\done^+
        \ar@{.}@/_/[uuuuu]
}$}
\]
\caption{An undefinable first-order strategy}
\label{fig:ex_undef}
\end{figure}

\subsubsection{First-order definability}
\label{subsubsec:fst_def}
Not every first-order test strategy is exactly definable in $\PCFpar$.
For instance, that in Figure \ref{fig:ex_undef} is not series-parallel,
while it is fairly easy to prove that all $\PCFpar$-definable terms on
this type yield a series-parallel causal dependency. In general, 
we have not yet managed to properly understand which first-order
strategies are definable.  
Luckily, we do not need to. Indeed, given a test $\balpha$ it is 
sufficient to find $M$ such that $\intr{M}$ is \emph{positionally
equivalent} from $\balpha$. In $\PCFpar$, without interference, the
order of evaluation is unobservable; and positional equivalence is not
sensitive to it. So our definability process will simply sequentialize
$\balpha$, while preserving its positions. 

Consider $\Gamma = x_1 : \tx_1, \dots, x_n : \tx_n$, some ground type
$\tx$, and a test strategy:
\[
\balpha : \oc (\with \tx_i) \vdash \tx\,.
\]

If $\collo \balpha$ is empty, any diverging term $M$ will satisfy $\collo
\intr{M} = \emptyset$. Otherwise, there is some $x \in \confc{\balpha}$.
If $\balpha$ has no primary question, then $x =
\{q_0^{\Qu,-}, a^{\An,+}\}$, with $a$ answering $q_0$ -- write
$\pr_\balpha(a) = v$ some answer in $\tx$. But then by determinism of
$\balpha$, it cannot have any other move and $\balpha \simstrat
\intr{v}$. Otherwise, if $\balpha$ has a primary question, it has one $a
\in \Q$ which is \textbf{minimal}, \emph{i.e.} it only depends on the
initial move. But then $q$ appears in every $x \in \confc{\balpha}$:

\begin{lem}
For any minimal primary question $q$, for any $y \in \confc{\balpha}$,
we have $q \in y$.
\end{lem}
\begin{proof}
By determinism, $y \cup \{q_0, q\} \in \conf{\balpha}$. Since
$y$ is complete, there is $a^{\An,+} \in y$ such that $a$ answers $q_0$.
But then, by Lemma \ref{lem:key_caus_wb} we have $q \leq_\balpha a$. It
follows that $q \in y$ as required.
\end{proof}

Choose $q\in \Q_i$ minimal. As $q$ appears in all non-empty complete
configurations, it is safe to first make a call to $x_i$, then
branching on the possible return values. Since $\balpha$ is finite,
there is a finite set $V$ of values leading to an observable result.
Now, for each $v \in V$, we define $\balpha_{(q v)}$ the
\emph{residual} of $\balpha$ after $q$ yields value $v$; and then
proceed inductively. To define this residual, our first step is to
rename $\balpha$ to isolate this first call:

\begin{lem}\label{lem:sel_pq}
For $q \in \Q_i$ minimal, there is a test strategy
$\balpha_{(q)} : \oc ((\with_i \tx_i) \with \tx_i) \vdash \tx$
s.t.
\[
\begin{array}{rl}
\text{\emph{(1)}} & 
\text{for all $\x \vdash \w \in \collo \balpha_{(q)}$, then $\x = \x' +
[(n+1, \v)]$ such that $\x' + [(i, \v)] \vdash \w \in \collo \balpha$,}\\
\text{\emph{(2)}} &
\text{for all $\x \vdash \w \in \collo \balpha$, then $\x = \x' + [(i,
\v)]$ such that $\x' + [(n+1, \v)] \vdash \w \in \collo \balpha_{(q)}$.}
\end{array}
\]
\end{lem}
\begin{proof}
The strategy $\balpha_{(q)}$ has components as for $\balpha$ except
the display map, which sends $q$ and its answers to the new
component. The characterisation of positions is straightforward.
\end{proof}

We obtain the residual $\balpha_{(q v)}$ as
$
\balpha_{(q v)} = \balpha_{(q)} \odot_\oc \tuple{\der_{\with
\tx_i}, v} : \oc (\with \tx_i) \vdash \tx$
writing $v : \oc (\with \tx_i) \vdash
\tx_i$ the constant strategy. In order to characterize its
positions, we note:

\begin{lem}\label{lem:pos_subst_v}
For any $v \in V$, the positions of $\tuple{\der_{\with \tx_i},
v}^\dagger$ are exactly those of the form
\[
\x \vdash \x + p \cdot [(n+1, \v)]
\quad
\in 
\quad
\smallint \left(\oc (\with \tx_i) \vdash \oc ((\with \tx_i) \with
\tx_i)\right)\,,
\]
where $p \cdot [(n+1, \v)]$ denotes the $p$-fold sum, and
for any $\x \in \coll (\oc (\with \tx_i))$ and $p \in \mathbb{N}$.
\end{lem}
\begin{proof}
Straightforward verification.
\end{proof}

Using this lemma, we link the positions of $\balpha$ and $\balpha_{(q
v)}$.

\begin{lem}
We have the following properties: 
\[
\begin{array}{rl}
\text{\emph{(1)}} &
\text{for any $\x \vdash \w \in \collo \balpha$, there is $\x = \x' +
[(i, \v)]$ such that $\x' \vdash \w \in \collo \balpha_{(qv)}$,}\\
\text{\emph{(2)}} &
\text{for any $\x \vdash \w \in \collo \balpha_{(qv)}$, we have $\x +
[(i, \v)] \vdash \w \in \collo \balpha$.}
\end{array}
\]
\end{lem}
\begin{proof}
By definition, we have $\balpha_{(q v)} = \balpha_{(q)} \odot
\tuple{\der_{\with \tx_i}, v}^\dagger$, so by Proposition
\ref{prop:vis_coll},
\[
\coll \balpha_{(q v)} = 
\coll (\balpha_{(q)}) \odot 
\coll (\tuple{\der_{\with \tx_i}, v}^\dagger)\,.
\]

The lemma directly follows by Lemmas \ref{lem:sel_pq} and
\ref{lem:pos_subst_v}.
\end{proof}

We have $\balpha_{(qv)}$ a test strategy with size strictly smaller
than that of $\balpha$. By IH there is 
\[
x_1 : \tx_1, \dots, x_n : \tx_n \vdash N_{(q v)} : \tx_i\,,
\]
for each $v \in V$, such that $\intr{N_{(q v)}} \equiv \balpha_{(q
v)}$. Finally, we define $x_1 : \tx_1, \dots, x_n : \tx_n
\vdash M : \tx$ as
\[
\begin{array}{l}
\mathbf{case}\,x_i\,\mathbf{of}\\
\hspace{20pt}v_1 \mapsto N_{(q v_1)}\\
\hspace{20pt}v_2 \mapsto N_{(q v_2)}\\
\hspace{20pt}\dots\\
\hspace{20pt}v_p \mapsto N_{(q v_p)}
\end{array}
\qquad
\stackrel{\text{def}}{=}
\qquad
\begin{array}{l}
\mathbf{let}\,x\,=\,x_i\,\mathbf{in}\\
\hspace{38pt}\mathbf{if}\,x\,=_\tx\,v_1\,\mathbf{then}\,N_{(q v_1)}\\
\hspace{16pt}\mathbf{else}\,\mathbf{if}\,x\,=_\tx\,v_2\,\mathbf{then}\,N_{(q
v_2)}\\
\hspace{40pt}\dots\\
\hspace{16pt}\mathbf{else}\,\mathbf{if}\,x\,=_\tx\,v_p\,\mathbf{then}\,N_{(q
v_p)}\\
\hspace{91pt}\mathbf{else}\,\bot
\end{array}
\]
where $V = \{v_1, \dots, v_p\}$, using the syntactic sugar introduced
in Section~\ref{subsec:sugar}. Write
\[
x_1 : \tx_1, \dots, x_n : \tx_n, x : \tx_i \vdash M' : \tx
\]
for the iterated if statement, \emph{i.e.} $M$ is
$\mathbf{let}\,x\,=\,x_i\,\mathbf{in}\,M'$. It remains to
analyze the positions of $\intr{M}$ and $\intr{M'}$ to show that
$\intr{M} \equiv \balpha$ as required. We start with the positions of
$\intr{M'}$.

\begin{lem}\label{lem:pos_mp}
Consider $\x \in \coll{\,\oc (\with \tx_i)}$ and $\w \in
\collo{\tx}$.

Then, $\x \vdash \w \in \collo{\intr{N_{(qv)}}}$ iff there is $p \in
\mathbb{N}$ such that $\x + p \cdot [(n+1, \v)] \vdash \w \in
\collo{\intr{M'}}$.
\end{lem}
\begin{proof}
It is a direct verification, amounting to the correctness of our
definition
for equality test and the usual laws for conditionals, that for any $v
\in V$ we have $\intr{M'} \odot_\oc \tuple{\der, v} = \intr{N_{(q v)}}$.
The claim then follows by Proposition \ref{prop:vis_coll} and Lemma
\ref{lem:pos_subst_v}.
\end{proof}

It remains to take the interpretation of the $\mathbf{let}$
construction into account. Recall that
\[
\intr{M} 
\quad = \quad
\mathsf{let}_{\tx_i,\tx} \odot_\oc \tuple{\pi_i, \Lambda^\oc(\intr{M'})} 
\quad:\quad
\oc (\with \tx_i) \vdash \tx
\]
where $\mathsf{let}_{\tx_i, \tx} : \oc (\tx_i \with (\oc \tx_i \lin
\tx)) \vdash \tx$. We characterize the positions of
$\mathsf{let}_{\tx_i, \tx}$ as follows.

\begin{lem}\label{lem:pos_let}
The non-empty positions of $\mathsf{let}_{\tx_i, \tx}$ are exactly those
of the form
\[
[(1, \v)] + [(2, ((p \cdot [\v]) \lin \w))] \vdash \w
\quad
\in
\quad
\coll (\oc (\tx_i \with (\oc \tx_i \lin \tx)) \vdash \tx)
\]
for $\v \in \collo \tx_i$, $\w \in \collo \tx$, and $p \in \mathbb{N}$.
\end{lem}
\begin{proof}
A direct analysis of positions reached by complete configurations of
$\mathsf{let}_{\tx_i, \tx}$.
\end{proof}

We can now wrap up, showing that $\intr{M}$ has the same
non-empty positions as $\balpha$.

\begin{lem}
Consider $\x \in \coll{\,\oc (\with \tx_i)}$ and $\w \in \collo{\tx}$.

Then, $\x \vdash \w \in \collo{\intr{M}}$ iff $\x = \x' + [(i,\v)]$ with
$\x' \vdash \w \in \collo{\intr{N_{(qv)}}}$.
\end{lem}
\begin{proof}
By Lemmas \ref{lem:pos_mp} and \ref{lem:pos_let}.
\end{proof}

So we have $\intr{M} \equiv \balpha$ as desired. Summing up, we have
proved:

\begin{prop}\label{prop:def_test_fst}
For $\balpha : \oc (\with \tx_i) \vdash \tx$ any test strategy, there is
\[
x_1 : \tx_1, \dots, x_n : \tx_n \vdash M : \tx
\]
a term of $\PCF$ (not using parallel evaluation) such that $\intr{M}
\equiv \balpha$.
\end{prop}

\subsubsection{Finite definability} We may now conclude the proof of
finite definability.

\begin{cor}\label{cor:fin_def}
Let $\Gamma$ be a $\PCF$ context, $A$ a $\PCF$ type, and $\balpha :
\intr{\Gamma} \vdash \intr{A}$ a test strategy.

Then, there is $\Gamma \vdash M : A$ such that $\intr{M} \equiv
\balpha$.
\end{cor}
\begin{proof}
Up to currying, we write
$\balpha : \oc (\with_{1\leq i \leq n} A_i) \vdash \tx$,
writing $A_i = A_{i, 1} \to \dots \to A_{i, p_i} \to \tx_i$ for $1
\leq i \leq n$. We reason by induction on the size of $\balpha$.
By Corollary \ref{cor:factor_fst_arg}, $\balpha$ factors as
\[
\balpha \equiv \fo(\balpha) \odot_\oc
\tuple{x_i\,\balpha_{q,1}^\bullet\,\dots\,\balpha_{q,p_i}^\bullet \mid
1\leq i \leq
n,~q \in \Q_i}\,,
\]
with each $\Q_i$ finite, and
for $1\leq i \leq n$, $q \in \Q_i$ and $1\leq j \leq p_i$,
$\balpha_{q,j}^\bullet : \oc (\with A_i) \vdash A_{i,j}$
a test strategy of size strictly smaller than $\balpha$. By
induction hypothesis, there is
\[
x_1 : A_1, \dots, x_n : A_n \vdash N_{q, j} : A_{i,j}
\]
such that $\intr{N_{q, j}} \equiv \balpha_{q, j}$. Let us write $\Q_i
=
\{q_{i,1}, \dots, q_{i,{k_i}}\}$. By Proposition
\ref{prop:def_test_fst}
there is also
\[
x_{q_{1,1}} : \tx_1, \dots, x_{q_{1,k_1}} : \tx_1, \dots,
x_{q_{n,1}} : \tx_n, \dots, x_{q_{n,k_n}} : \tx_n 
\quad \vdash \quad
M_{\mathsf{fo}} : \tx
\]
such that $\intr{M_{\mathsf{fo}}} \equiv \fo(\balpha)$. Then, we define
the
term $x_1 : A_1, \dots, x_n : A_n \vdash M : \tx$ as
\[
x_1 : A_1, \dots, x_n : A_n \vdash
M_{\mathsf{fo}}[x_i\,N_{q_{i,l}, 1}\,\dots\,N_{q_{i,l}, p_i} /
x_{q_{i,l}}] : \tx\,.
\]

Then we may finally compute
\begin{eqnarray*}
\intr{M} &=& \intr{M_{\mathsf{fo}}} \odot_\oc
\tuple{\intr{x_i\,N_{q_{i,l}, 1}\,\dots\,N_{q_{i,l}, p_i}} \mid 1\leq
i \leq n, q_{i,l} \in \Q_i}\\
&\equiv& \intr{M_{\mathsf{fo}}} \odot_\oc
\tuple{x_i\,\intr{N_{q, 1}}\,\dots\,\intr{N_{q, n}} \mid 1\leq i \leq
n, q \in \Q_i}\\
&\equiv& 
\fo(\balpha) \odot_\oc
\tuple{x_i\,\balpha_{q,1}^\bullet\,\dots\,\balpha_{q,p_i}^\bullet \mid
1\leq i \leq
n,~q \in \Q_i}
\end{eqnarray*}
using the substitution lemma for cartesian closed
categories, compatibility of interpretation with the internal
language, the properties of $M_{\mathsf{fo}}$ and $N_{q,
j}$ and that $\equiv$ is a congruence.
\end{proof}

\subsection{Full Abstraction for \texorpdfstring{$\PCFpar$}{PCF//}} \label{subsec:fa_pcfpar}
We may now prove our final full abstraction result.

\begin{thm}
The model $\CGwbinn_\oc$ is intensionally fully abstract for $\PCFpar$.
\end{thm}
\begin{proof}
Consider $\vdash M, N : A$ such that $M \obs N$, and assume that
$\intr{M} \not \obs \intr{N}$, \emph{i.e.} there is a test $\balpha \in
\CGwbinn_\oc(\intr{A}, \intr{\tunit})$ such that,
\emph{w.l.o.g.},
$\balpha \odot_\oc \intr{M} \eval$ and $\balpha \odot_\oc \intr{N}
\div$.

By Corollary \ref{cor:finite_tests_suffice}, we assume 
$\balpha$ is additionally finite. We consider $\comp(\balpha)$ as
defined in Proposition \ref{prop:wb_pruning}. By construction it is
well-bracketed, parallel innocent, and finite. By Proposition
\ref{prop:wb_pruning} it is causally well-bracketed, so 
a \emph{test strategy}. Proposition \ref{prop:wb_pruning}
also states that $\comp(\balpha) \equiv \balpha$ which is a congruence
by Corollary \ref{cor:equiv_cong}, so 
\[
\comp(\balpha) \odot_\oc \intr{M} \eval\,,
\qquad
\qquad
\comp(\balpha) \odot_\oc \intr{N} \div\,.
\]

By Corollary \ref{cor:fin_def}, there is a term $x : A \vdash T :
\tunit$ such that $\intr{T} \equiv \comp(\balpha)$. Defining the context
$C[-] = (\lambda x^A.\,T)\,[-]$, it follows from the laws of cartesian
closed categories that
\[
\intr{C[M]} = \intr{(\lambda x^A.\,T)\,M} 
\simstrat \intr{T[M/x]}
\simstrat \intr{T} \odot_\oc \intr{M}
\equiv \comp(\balpha) \odot_\oc \intr{M} \eval\,,
\]
and likewise, $\intr{C[N]} \equiv \comp(\balpha) \odot_\oc
\intr{N}$. By Theorem \ref{th:adequacy_ipa}, $C[M] \eval$. By hypothesis
$M \obs N$, so $C[N] \eval$. By Theorem \ref{th:adequacy_ipa} again,
$\intr{C[N]} \eval$, hence $\comp(\balpha) \odot_\oc \intr{N} \eval$,
contradiction.
\end{proof}

We may finally answer our main question positively, with the following
theorem.

\begin{thm}
The model $\CGwb_\oc$ supports \emph{parallel innocence} and
\emph{sequentiality}, s.t.
\[
\begin{array}{rcl}
\CGwb_\oc          & \text{is fully abstract for} & \IPA\,,\\
\CGwbinn_\oc       & \text{is fully abstract for} & \PCFpar\,,\\
\CGwbseq_\oc       & \text{is fully abstract for} & \IA\,,\\
\CGwbseqinn_\oc    & \text{is fully abstract for} & \PCF\,,
\end{array}
\]
\end{thm}

Thus parallel innocence exactly bans interference, and sequentiality
exactly bans parallelism. Through this theorem, we have
successfully disentangled parallelism and interference.

\section{Conclusion}

It is puzzling that disentangling parallelism and interference requires
such an intricate machinery whereas the original \emph{semantic cube}
arose almost ``by accident'' from minor variations of the Hyland-Ong
model of $\PCF$. 

Our interpretation is that computational effects may be organized in
distinct categories. Some effects, such as interference and control,
bring more freedom as to how execution and its control flow explores a
piece of code. Once a sufficiently general mathematical description of
the control flow is given (such as the original Hyland-Ong setting for
sequential deterministic computation), this additional freedom may be
studied and characterized. In contrast, other effects such as
non-determinism and parallelism, affect the inherent structure of
execution itself: non-determinism quite explicitly so by generating
non-deterministic branching, and parallelism in a more DAG-like fashion
-- we refer to both under the umbrella name ``branching effects''. What
our paper demonstrates -- along with earlier papers on
non-deterministic innocence \cite{lics14,DBLP:conf/lics/TsukadaO15} --
is that if we aim to realize fully the ``unified semantic landscape'' of
Abramsky's programme, we must first deal with branching effects.
Non-branching effects should follow by characterizing the causal
patterns they allow\footnote{The line is not always so clear between
branching and non-branching effects: for instance, in a sequential
setting interference is non-branching, but in a parallel setting it
generates non-deterministic choice.}.

What other branching effects are around? One currently at the focus of
the semantics community is probabilistic choice. By itself, the
probabilistic branching structure is not much harder than
non-deterministic choice \cite{DBLP:journals/corr/TsukadaO14,lics18}.
However its interaction with non-deterministic and parallel branching is
a significant challenge, and one of the remaining scientific and
technical barriers for a truly unified game semantics landscape.

\paragraph{\textbf{Acknowledgments.}} First we would like to thank Glynn
Winskel. Though he has not participated in the developments in this
paper per se, we rely deeply on his pioneering work - and our subsequent
collaborations - on concurrent games. We owe a lot to his insight.

This work was supported by the ANR project DyVerSe
(ANR-19-CE48-0010-01); and by Labex MiLyon (ANR-10-LABX-0070) of
Universit\'e de Lyon, within the program ``Investissements d'Avenir''
(ANR-11-IDEX-0007), operated by the French ANR.

\bibliographystyle{alphaurl}
\bibliography{main}

\appendix

\newpage 

\section{Concurrent Games Toolbox}

\subsection{Event Structures and Maps} We start with recalling
\emph{maps of event structures}.

\begin{defi}\label{def:app_map_es}
Consider $E, F$ two es. A \textbf{map of event structures} from $E$ to
$F$ is a function $f : \ev{E} \to \ev{F}$ satisfying the following two
conditions:
\[
\begin{array}{rl}
\text{\emph{valid:}} & \text{for all $x \in \conf{E}$, we have $f\,x \in
\conf{F}$,}\\
\text{\emph{local injectivity:}} & \text{for all $e_1, e_2 \in x \in
\conf{E}$, if $f e_1 = f e_2$ then $e_1 = e_2$.}
\end{array}
\]

If $E$ and $F$ are esps, a \textbf{map of esps} is additionally required
to preserve polarities.
\end{defi}

\begin{lem}\label{lem:es_refl_caus}
Consider $f : E \to F$ a map of event structures, and $e_1, e_2 \in x \in \conf{E}$.

If $f(e_1) \leq_F f(e_2)$, then $e_1 \leq_E e_2$.
\end{lem}
\begin{proof}
Seeking a contradiction, assume we do not have $e_1 \leq_E e_2$. This
means that $e_1 \not \in [e_2]_E$. But $[e_2]_E \in \conf{E}$, so $f
[e_2]_E$ is down-closed. Moreover, $f e_2 \in f [e_2]_E$, so $f e_1 \in
f [e_2]_E$. So there is $e'_1 \in [e_2]_E$ such that $f e'_1 = f e_1$.
Finally, $[e_2]_E \subseteq x$, so $e_1, e'_1 \in x \in \conf{E}$.
Hence, by local injectivity, $e_1 = e'_1$. We deduce that $e_1 \in
[e_2]_E$ after all, contradiction.
\end{proof}

\subsection{Basic Properties of Strategies}
We gather some basic properties.

\begin{lem}\label{lem:app_caus_alt}
Consider $A$ an arena, $\bsigma : A$ a causal strategy, and $m, n \in
\ev{\bsigma}$.

If $m \imc_\bsigma n$, then $\pol(m) \neq \pol(n)$.
\end{lem}
\begin{proof}
If $\pol(m) = \pol(n)$, then by \emph{courteous} we have $\pr_\bsigma(m)
\imc_A \pr_\btau(n)$ as well. But since arenas are alternating, this is
absurd.
\end{proof}

\begin{lem}\label{lem:app_aux_caus}
Consider $A$ an arena, $\bsigma : A$ a causal strategy, and $m,  n^- \in
\ev{\bsigma}$ compatible.

Then, $m \imc_\bsigma n$ iff $\pr_\bsigma(m) \imc_A \pr_\bsigma(n)$.
\end{lem}
\begin{proof}
\emph{If.} Assume $\pr_\bsigma(m) \imc_A \pr_\bsigma(n)$. Since
$[n]_\bsigma \in \conf{\bsigma}$ and $\pr_\bsigma$ is a map of event
structures, $\pr_\bsigma [n]_\bsigma \in \conf{A}$, so it is
down-closed. Thus, there is $m' \in [n]_\bsigma$ such that
$\pr_\bsigma(m') = \pr_\bsigma(m)$. In particular, $n$ cannot be
minimal, so there is $m'' \imc_\bsigma n$. By
\emph{courteous}, since $n$ is negative we have $\pr_\bsigma(m'') \imc_A
\pr_\bsigma(n)$. But $A$ is \emph{forestial}, so $\pr_\bsigma(m'') =
\pr_\bsigma(m)$. Now, since $m, n$ are compatible they appear in a
configuration $x \in \conf{\bsigma}$, and in particular $m, m'' \in x$.
Thus, $m = m''$ by local injectivity.
\emph{Only if.} If $m \imc_\bsigma n^-$, then $\pr_\bsigma(m) \imc_A
\pr_\bsigma(n)$ by \emph{courteous}.
\end{proof}

\begin{lem}\label{lem:aux_imp_court}
Consider $A$ an arena, $\bsigma : A$ a causal strategy, and $m^- \in
\ev{\bsigma}$ s.t. $\pr_\bsigma m$ non-minimal.
Then, there is a unique $n \imc_\bsigma m$.
\end{lem}
\begin{proof}
\emph{Existence.} Write $a \imc \pr_\bsigma(m)$, which is unique since
$A$ is \emph{forestial}. Since $[m]_\bsigma \in \conf{\bsigma}$, we have
$\pr_\bsigma [m]_\bsigma \in \conf{A}$, therefore it is down-closed and
must contain $a$. Therefore, there is $n \leq_\bsigma m$ such that
$\pr_\bsigma(n) = a$. By Lemma \ref{lem:app_aux_caus}, we then have $n
\imc_\bsigma m$ as required.

\emph{Uniqueness.} If $n_1 \imc_\bsigma m$ and $n_2 \imc_\bsigma m$,
by \emph{courteous} $\pr_\bsigma n_1 \imc_A \pr_\bsigma m$
and $\pr_\bsigma n_2 \imc_A \pr_\bsigma m$. But then, $\pr_\bsigma n_1 =
\pr_\bsigma n_2$ as $A$ is \emph{forestial}. So $n_1 = n_2$ as
$\pr_\bsigma$ is \emph{locally injective}.
\end{proof}

\begin{lem}\label{lem:pol_caus}
Consider $A$ an arena, $\bsigma : A$ a causal strategy, and $m, n_1, n_2\in
\ev{\bsigma}$.

If $n_1 \imc_\bsigma m$ and $n_2 \imc_\bsigma m$ with $n_1, n_2$
distinct, then $\pol(m) = +$. 
\end{lem}
\begin{proof}
Seeking a contradiction, assume $\pol(m) = -$. Then, by
\emph{courteous}, we have $\pr_\bsigma(n_1) \imc_A \pr_\bsigma(m)$ and
$\pr_\bsigma(n_2) \imc_A \pr_\bsigma(m)$. As $A$ is \emph{forestial},
$\pr_\bsigma(n_1) = \pr_\bsigma(n_2)$. As $n_1, n_2 \in [m]_\bsigma \in
\conf{\bsigma}$, we have $n_1 = n_2$ by local injectivity,
contradiction.
\end{proof}

\begin{lem}\label{lem:app_aux_max_pos}
Consider $A$ an arena, $\bsigma : A$ a causal strategy, and $m^+ \in x \in
\conf{\bsigma}$.

Then, $m$ is maximal in $x$ iff $\pr_\bsigma m$ is maximal in
$\pr_\bsigma x$.
\end{lem}
\begin{proof}
\emph{If.} Assume $\pr_\bsigma m$ is maximal in $\pr_\bsigma x$.
Seeking a contradiction, assume that $m \imc_\bsigma n$. By Lemma
\ref{lem:app_caus_alt}, $\pol(n) = -$. Therefore, by
\emph{courteous}, $\pr_\bsigma(m) \imc_A \pr_\bsigma(n)$,
contradiction.

\emph{Only if.} Straightforward by Lemma \ref{lem:es_refl_caus}.
\end{proof}

\begin{lem}\label{lem:app_pres_neg_imm_confl}
Consider $A$ an arena, $\bsigma : A$ a causal strategy, and $m^-_1, m^-_2
\in \ev{\bsigma}$.
If $m_1$ and $m_2$ are in minimal conflict in $\bsigma$, then
$\pr_\bsigma(m_1)$ and $\pr_\bsigma(m_2)$ are in minimal conflict in $A$.
\end{lem}
\begin{proof}
We first prove that $\pr_\bsigma(m_1)$ and $\pr_\bsigma(m_2)$ are in
conflict. Seeking a contradiction, assume that it is not the case. Then,
as $m_1$ and $m_2$ are in \emph{minimal} conflict, we have 
$\pr_\bsigma([m_1]_\bsigma \cup [m_2]_\bsigma) \in \conf{A}$. Hence, by
\emph{receptive}, there is a unique $m'_2 \in \ev{\bsigma}$ such that
\[
[m_1]_\bsigma \cup [m_2)_\bsigma \vdash_\bsigma m'_2
\]
with $\pr_\bsigma(m'_2) = \pr_\bsigma(m_2)$ -- where $[m_2)_\bsigma =
\{n \in \ev{\bsigma} \mid n <_\bsigma m_2\}$.
But then, by Lemma \ref{lem:app_aux_caus}, $m_2$ is minimal in $\bsigma$
iff $m'_2$ is minimal in $\bsigma$, and so $m_2 =
m'_2$ by \emph{receptive}, contradicting that $m_1 \conflict_\bsigma
m_2$. Otherwise, consider $n \imc_\bsigma m_2$ and $n' \imc_\bsigma
m'_2$. Then, by \emph{courteous},
\[
\pr_\bsigma(n) \imc_A \pr_\bsigma(m_2)\,,
\qquad
\qquad
\pr_\bsigma(n') \imc_A \pr_\bsigma(m'_2)\,,
\]
but then $\pr_\bsigma(n) = \pr_\bsigma(n')$ since $A$ is
\emph{forestial}. But $n, n' \in [m_1]_\bsigma \cup [m_2)_\bsigma \in
\conf{\bsigma}$, so we must have $n = n'$ by locally injectivity. By
Lemma \ref{lem:aux_imp_court}, $n$ is the unique predecessor of $m_2$
and $m'_2$. So, $[n]_\bsigma \vdash_\bsigma m_2$ and $[n]_\bsigma
\vdash_\bsigma m'_2$ with the same image. So $m_2 = m'_2$,
contradiction. 

Finally, minimality of the conflict is obvious from that of $m_1$ and
$m_2$.
\end{proof}

\section{Construction of Alternating Strategies}
\label{app:alt}

In this second section of the appendix, we give more details on the
construction of $\NegAlt$.

First, a warning: quite a few superficial complications come from the
general construction $A \lin B$ for $B$ non-pointed, with 
morphisms from $A$ to $B$ being strategies on $A \lin B$. 
An alternative is to only consider $A \lin B$ for $B$ strict. Then we do
not have a symmetric monoidal closed category, only an exponential
ideal. This would be sufficient for the languages considered in this
paper, however, we opted to link with traditional categorical models.

\subsection{Basic Categorical Structure} We start by establishing the
categorical structure.

\subsubsection{Arrow arena} \label{app:arrow}
First, we give postponed proofs on the
construction of the arrow arena.

\conflictlin*
\begin{proof}
\emph{Existence.} We set $\conflict_{A\lin B}$ as the following
relation:
\[
\begin{array}{ccl}
(2, b_1) \conflict_{A \lin B} (2, b_2) &~\Leftrightarrow~&b_1 \conflict_{B} b_2\\
(1, (b, a_1)) \conflict_{A \lin B}  (1, (b, a_2))
&\Leftrightarrow & a_1 \conflict_A  a_2\\
(2, b_1) \conflict_{A \lin B} (1, (b_2, a))
&\Leftrightarrow&
b_1 \conflict_B b_2\\
(1, (b_1, a)) \conflict_{A \lin B} (2, b_2)
&\Leftrightarrow& b_1
\conflict_B b_2\\
(1, (b_1, a_1)) \conflict_{A \lin B}  (1, (b_2, a_2))
&\Leftrightarrow& (b_1 \conflict_B b_2) \vee (\min(a_1) = \min(a_2))
\vee (a_1 \conflict_A a_2)\,,
\end{array}
\]
where $b_1, b_2$ are assumed distinct, all other pairs left consistent,
and with $\min(a)$ the unique minimal antecedent of $a$. It is
routine that this conflict makes $A \lin B$ a $-$-arena.

Now, we check the additional condition. Consider $x \subseteq \ev{A\lin
B}$ down-closed, written as
\[
x = (\parallel_{b \in I} x_{A,b}) \parallel x_B
\]
where $I$ is a subset of the minimal events of $B$. Then, we show that
$x \in \conf{A \lin B}$ iff 
\[
\chi_{A, B}\,x = (\cup_{b \in I} x_{A,b}) \parallel x_B \in
\conf{A \vdash B}
\]
and $\chi_{A, B}$ is injective on $x$. \emph{Only if} is a direct
verification. For \emph{if}, if $\chi_{A, B}\,x \in \conf{A\lin B}$ then
the only possible conflict in $x$ is of the form, with $b_1 \conflict_B
b_2$ or $\min(a_1) = \min(a_2)$:
\[
(1, (b_1, a_1)) \conflict_{A\lin B} (1, (b_2, a_2))
\]

In the former case, by down-closure, $(2, b_1), (2, b_2) \in x$,
contradicting $\chi_{A, B}\,x \in \conf{A \vdash B}$. In the
latter case, by down-closure, $(1, (b_1, \min(a_1))), (1,
(b_2, \min(a_2))) \in x$ -- but they have the same image through
$\chi_{A, B}$, contradicting that it should be injective on $x$.
Uniqueness follows as with fixed events, an
event structure is determined by configurations.
\end{proof}

\subsubsection{Composition}\label{app:alt_comp}
Consider fixed $A, B$ and $C$ three $-$-arenas, and $\sigma : A \lin B$,
$\tau : B \lin C$ alternating prestrategies. Recall from the main
text:

\begin{defi}
A \textbf{pre-interaction} on $A, B, C$ is $u \in \ev{(A\lin
B) \lin C}^*$ satisfying:
\[
\begin{array}{rl}
\text{\emph{valid}:}& \forall 1\leq i \leq n,~\{s_1, \dots, s_i\} \in
\conf{(A \lin B) \lin C}
\end{array}
\]
\end{defi}

Remember that in $A \lin B$, events are either $(2, b)$ for $b
\in \ev{B}$, or $(1, (b, a))$ for $b \in \min(B)$ and $a \in \ev{A}$. By
convention, in this section, we write $[\labr, b]$ for $(2, b)$ and
$[\labl, a]_b$ for $(1, (b, a))$. In that case, $[\labr, b]$ is the unique
immediate predecessor of $[\labl, a]_b$, \emph{i.e.} its
\emph{justifier}.
Similarly, in $(A \lin B) \lin C$, events can be $(2, c)$ written $[\labr,
c]$; $(1, (c, (2, b)))$ written $[\labm, b]_c$; and $(1, (c, (1, (b,
a))))$ written $[\labl, a]_{b, c}$; we respectively say that they are
\textbf{in $C$}, \textbf{in $B$}, or \textbf{in $A$}.

\begin{figure}
\[
\scalebox{.75}{$
\begin{array}{rcl}
\varepsilon \restrict A, B &=& \varepsilon\\
u\,[\labr, c] \restrict A, B &=& u \restrict A, B\\
u\,[\labm, b]_c \restrict A, B &=& (u \restrict A, B)\,[\labr, b]\\
u\,[\labl, a]_{b,c} \restrict A, B &=& (u \restrict A, B)\,[\labl, a]_b
\end{array}$}
\,
\scalebox{.75}{$
\begin{array}{rcl}
\varepsilon \restrict B, C &=& \varepsilon\\
u\,[\labr, c] \restrict B, C &=& (u \restrict B, C)\,[\labr, c]\\
u\,[\labm, b]_c \restrict B, C &=& (u \restrict B, C)\,[\labl, b]_c\\
u\,[\labl, a]_{b,c} \restrict B, C &=& u \restrict B, C
\end{array}$}
\,
\scalebox{.75}{$
\begin{array}{rcl}
\varepsilon \restrict A, C &=& \varepsilon\\
u\,[\labr, c] \restrict A, C &=& (u \restrict A, C)\,[\labr, c]\\
u\,[\labm, b]_c \restrict A, C &=& u \restrict A, C\\
u\,[\labr, a]_{b,c} \restrict A, C &=& (u\restrict A, C)\,[\labl, a]_c
\end{array}$}
\]
\caption{Restrictions of pre-interactions}
\label{fig:restr_pre-int}
\end{figure}

Using this, we define in Figure \ref{fig:restr_pre-int} three
\textbf{restrictions} of a pre-interaction $u$, namely $u \restrict A, B
\in \ev{A \lin B}^*$, $u \restrict B, C \in \ev{B \lin C}^*$, and $u
\restrict A, C \in \ev{A \lin C}^*$. Now we set:

\begin{defi}
An \textbf{interaction} $u \in \tau \inter \sigma$ between $\sigma$ and
$\tau$ is a pre-interaction $u$ s.t.
\[
u \restrict A, B \in \sigma\,,
\qquad
u \restrict B, C \in \tau\,,
\qquad
u \restrict A, C \in \Alt(A\lin C)\,.
\]

The \textbf{composition} of $\sigma$ and $\tau$ is $\tau \odot \sigma =
\{u \restrict A, C \mid u \in \tau \inter \sigma\}$.
\end{defi}

The first step is to ensure that $\tau \odot \sigma$
is a prestrategy, and that if $\sigma$ and $\tau$ are strategies, then
so is $\tau \odot \sigma$. We start with the conditions of Definition
\ref{def:alt_strat}, postponing uniformity.
\emph{Non-empty} and \emph{prefix-closed} follow from those on
$\sigma$ and $\tau$. For \emph{deterministic}, we need more tools.
The main tool to study the interaction of alternating strategies is the
\emph{state analysis} of plays and interactions. Recall from Section~\ref{subsubsec:comp_seq_caus} that $s$ alternating is \textbf{in state
$O$} if it has even length, and \textbf{in state $P$} otherwise. Then,
we have the following property:

\begin{lem}\label{lem:state-alt}
If $u \in \tau \inter \sigma$, then we are in one of the following three
cases:
\[
\begin{array}{rl}
\text{\emph{(1)}} &
\text{$u \restrict A, B$, $u \restrict B, C$ and $u \restrict A, C$ are
respectively in state $O$, $O$, $O$.}\\
\text{\emph{(2)}} &
\text{$u \restrict A, B$, $u \restrict B, C$ and $u \restrict A, C$ are
respectively in state $O$, $P$, $P$.} \\
\text{\emph{(3)}} &
\text{$u \restrict A, B$, $u \restrict B, C$ and $u \restrict A, C$ are
respectively in state $P$, $O$, $P$.}
\end{array}
\]
\end{lem}
\begin{proof}
Standard argument, direct by induction on $u$.
\end{proof}

Next, we can prove the key property of the composition of alternating
strategies.

\begin{lem}\label{lem:alt_uniq_wit}
Consider $s \in \tau \odot \sigma$ of even length.

Then, there is a \emph{unique} \textbf{witness} $u \in \tau \inter
\sigma$ such that $u \restrict A, C = s$.
\end{lem}
\begin{proof}
\emph{Existence} is obvious by definition.

\emph{Uniqueness.} Seeking a contradiction, consider $u_1, u_2 \in \tau
\inter \sigma$ distinct such that $u_1 \restrict A, C = u_2 \restrict A,
C = s$. First, since $s$ has even length, $u_i \restrict A, C$ is in
state $O$, so $u_1$ and $u_2$ must be in state \emph{(1)} of Lemma
\ref{lem:state-alt}. It follows that their immediate prefix cannot be in
state \emph{(1)}, from which it follows last move is visible
(\emph{i.e.} in $A$ or $C$). So, $u_1$ and $u_2$ cannot be comparable
for the prefix order. Therefore, there is $u'$ maximal such that $u'
\prefix u_1$ and $u'\prefix u_2$, say we have $u' m_1 \prefix u_1$ and
$u' m_2 \prefix u_2$ for $m_1, m_2$ distinct.  By Lemma
\ref{lem:state-alt}, $u'$ is in one of the states \emph{(1)}, \emph{(2)}
or \emph{(3)}. If it is in state \emph{(1)}, then the next moves $m_1$
and $m_2$ are in $u \restrict A, C = s$, so they cannot be distinct. Say
$u'$ is in state \emph{(3)} -- the case \emph{(2)} is similar but
simpler. Necessarily, $m_1$ and $m_2$ are in $A$ or $B$, so that 
\[
u' m_1 \restrict A, B = (u' \restrict A, B) n_1 
\qquad
\qquad
u' m_2 \restrict A, B = (u' \restrict A, B) n_2
\]
are two plays of $\sigma$ with even length -- so that $n_1 = n_2$, by
\emph{determinism}. From the definition of restriction, the only case
where $m_1 \neq m_2$ with $n_1 = n_2$ is if $m_1 = [\labl, a]_{b,c}$,
$m_2 = [\labl, a]_{b, c'}$ with $c \neq c'$, so $n_1 = n_2 = [\labl,
a]_b$. Then, as $u'$ satisfies condition \emph{valid}, this
entails that their immediate dependencies $[\labm, b]_c, [\labm, b]_{c'}
\in \ev{u'}$ as well -- impossible since $[\labm, b]_c \conflict [\labm,
b]_{c'}$.
\end{proof}

It is then straightforward to prove that $\tau \odot \sigma$ satisfies
\emph{deterministic}.

\begin{prop}
We have that $\tau \odot \sigma : A \lin C$ is a prestrategy.

Moreover, if $\sigma$ and $\tau$ are strategies, then so is $\tau \odot
\sigma$.
\end{prop}
\begin{proof}
To obtain a prestrategy, it remains that $\tau \odot \sigma$
satisfies \emph{deterministic}. Consider $sn_1^+, sn_2^+ \in
\tau \odot \sigma$. Consider $u_1 m_1, u_2 m_2 \in \tau \inter \sigma$
their unique witness as given by Lemma \ref{lem:alt_uniq_wit}. We
reason as in Lemma \ref{lem:alt_uniq_wit}: if there is a diverging point
between $u_1 m_1$ and $u_2 m_2$, by Lemma \ref{lem:state-alt} the
divergence can be attributed to either $\sigma$ or $\tau$, contradicting
determinism.

Now, assume that $\sigma$ and $\tau$ satisfy condition \emph{receptive}.
Consider $s \in \tau \odot \sigma$ such that $s m^- \in \Alt(A\lin C)$.
More precisely, assume that $m = [(\labl, a)]_c$ as the case in $C$ is
simpler. Consider now $u \in \tau \inter \sigma$ such that $u \restrict
A, C = s$; necessarily $u$ is in state \emph{(1)} of Lemma
\ref{lem:state-alt}. Now, as $m$ is negative its immediate predecessor
of $m$ in $A\lin C$ is some $[(\labl, a')]_c$ in $s$. Since $s = u
\restrict A, C$, it corresponds to some $[(\labl, a')]_{b, c}$ in $u$,
for some $b \in \min(B)$. But then, it is a direct verification that $(u
\restrict A, B) [(\labl, a)]_b \in \Alt(A\lin B)$, so $(u \restrict A, B)
[(\labl, a)]_b \in \sigma$ by \emph{receptive}. Therefore, $u [(\labl,
a)]_{b,c} \in \tau \inter \sigma$ witnessing that $s m \in \tau \odot
\sigma$ as required.  
\end{proof}

Note that this argument and the prior state analysis of Lemma
\ref{lem:state-alt} are also found in the proof of composition of
sequentiality for causal strategies in Section~\ref{subsubsec:comp_seq_caus}.   

\subsubsection{Associativity} We now sketch associativity,
which follows standard lines, see \emph{e.g.} \cite{harmer2004innocent}.
Let us fix $\sigma : A \lin B, \tau : B \lin C$ and $\delta : C \lin D$
three alternating prestrategies. 

The first step is to define a notion of interaction between three
strategies. First:

\begin{defi}
A \textbf{$3$-pre-interaction} on $A, B, C, D$ is $w \in
\ev{((A \lin B) \lin C) \lin D}^*$ s.t.
\[
\begin{array}{rl}
\text{\emph{valid}:}& \forall 1\leq i \leq n,~\{w_1, \dots, w_i\} \in
\conf{(A \lin B) \lin C}\,,
\end{array}
\]
where $w = w_1 \dots w_n$.
\end{defi}

For a $3$-pre-interaction $w$, we define $w \restrict A, B \in \ev{A\lin
B}^*$, $w \restrict B, C \in \ev{B\lin C}^*$, $w \restrict C, D \in
\ev{C \lin D}^*$ and $w \restrict A, D \in \ev{A \lin D}^*$ with the
obvious adaptation of Figure \ref{fig:restr_pre-int}. An
\textbf{interaction} of $\sigma, \tau$ and $\delta$ is a
$3$-pre-interaction $w$ s.t. 
\[
w \restrict A, B \in \sigma\,,
\qquad
w \restrict B, C \in \tau\,,
\qquad
w \restrict C, D \in \delta\,,
\qquad
w \restrict A, D \in \Alt(A \lin D)\,,
\]
written $w \in \delta \inter \tau \inter \sigma$. We have four
additional restrictions
$w \restrict A, B, C \in \tau \inter \sigma$, $w \restrict B, C, D \in
\delta \inter \tau$, $w \restrict A, C, D \in \delta \inter (\tau \odot
\sigma)$ and $w \restrict A, B, D \in (\delta \odot \tau) \inter
\sigma$, defined in the obvious way. 

Then, the key argument of associativity is the so-called ``zipping
lemma'':

\begin{lem}[Zipping]\label{lem:app_alt_zip}
Consider $u \in \delta \inter (\tau \odot \sigma)$ and $v \in \tau
\inter \sigma$ such that $u \restrict A, C = v \restrict A,
C$.

Then, there is a unique $w \in \delta \inter \tau \inter \sigma$ such
that $w \restrict A, C, D = u$ and $w \restrict A, B, C = v$.
\end{lem}
\begin{proof}
By induction on $u$ -- by Lemma \ref{lem:state-alt} the moves in $B$
from $v$ can be interleaved with those in $v$ in a unique way; likewise
there is a unique way to set their dependency.
\end{proof}

We also have the mirror image, zipping $u \in (\delta
\odot \tau) \inter \sigma$ with $v \in \delta \inter \tau$. Altogether,

\begin{prop}\label{prop:app_alt_assoc}
We have $(\delta \odot \tau) \odot \sigma = \delta \odot (\tau \odot
\sigma)$.
\end{prop}
\begin{proof}
Consider $s \in \delta \odot (\tau \odot \sigma)$. It has a (unique)
witness $u \in \delta \inter (\tau \odot \sigma)$. Then, $u \restrict A,
C \in \tau \odot \sigma$, thus there is again $v \in \tau
\inter \sigma$ s.t. $u \restrict A, C = v \restrict A, C$. By Lemma
\ref{lem:app_alt_zip}, there is $w \in \delta \inter \tau \inter \sigma$
s.t. $w \restrict A, C, D = u$ and $w \restrict A, D, C = v$. But
then we may restrict $w$ to $w \restrict B, C, D \in \delta \inter \tau$, so
$w \restrict B, D \in \delta \odot \tau$. Moreover $w \restrict A,
B \in \sigma$ so $w \restrict A, B, D \in (\delta \odot \tau)
\inter \sigma$, from which $w \restrict A, D \in (\delta
\odot \tau) \odot \sigma$.
The other direction is symmetric.
\end{proof}

Note that associativity holds for prestrategies and does not depend on
\emph{receptive}.

\subsubsection{Identities} \label{app:alt_id}
Fix some $-$-arena $A$. For $s \in \Alt(A
\lin A)$, we define its restrictions
\[
\begin{array}{rcl}
\varepsilon \restrict \labl &=& \varepsilon\\
s\,[(\labl, a)]_{a'} \restrict \labl &=& (s \restrict \labl)\,a\\
s\,[(\labr, a)] \restrict \labr &=& s \restrict \labr
\end{array}
\qquad
\qquad
\begin{array}{rcl}
\varepsilon \restrict \labr &=& \varepsilon\\
s\,[(\labl, a)]_{a'} \restrict \labr &=& s \restrict \labr \\
s\,[(\labr, a)] \restrict \labr &=& (s \restrict \labr)\,a
\end{array}
\]
using which we may define the identity for alternating strategies:

\begin{defi}
The \textbf{copycat} $\cc_A : A \lin A$ comprises
all $s \in \Alt(A \lin A)$ s.t.
\[
\begin{array}{rl}
\text{\emph{(1)}}&
\text{for all $s' \prefix s$ of even length, $s' \restrict \labl = s'
\restrict \labr$,}\\
\text{\emph{(2)}}&
\text{for all $[(\labl, a)]_{a'} \in \ev{s}$, with $a \in \min(A)$, $a = a'$.}
\end{array}
\]
\end{defi}

Condition \emph{(2)} means that when playing a minimal event on the left
hand side, copycat justifies it with the same move on the right.
Such a condition is also required in Hyland-Ong games (though sometimes
mistakenly omitted). Copycat strategies provide identities: 

\begin{prop}
The $-$-arenas and alternating strategies form a category $\NegAlt$.
\end{prop}
\begin{proof}
It remains that for any $\sigma : A \lin B$,
$\sigma = \cc_B \odot \sigma \odot \cc_A$, which is elementary.
\end{proof}

\subsection{Monoidal Closed Structure} We now describe the monoidal structure.

\subsubsection{Tensor product}\label{app:alt_tensor}
On $-$-arenas, we have defined $A\tensor B$ simply as $A\parallel B$.

For strategies, the critical step is a suitable notion of
restriction. More precisely, for $A_1, A_2, B_1, B_2$ $-$-arenas
and $s \in \Alt(A_1 \tensor A_2 \lin B_1 \tensor B_2)$, we give a
partial definition
\[
s \restrict A_1, B_1 \in \ev{A_1 \lin B_1}^*
\qquad
\qquad
s \restrict A_2, B_2 \in \ev{A_2 \lin B_2}^*
\]
in 
\begin{figure}
\[
\scalebox{.8}{$
\begin{array}{rcl}
\varepsilon \restrict A_1, B_1 &=& \varepsilon\\
s\,[(\labl, (1, a))]_{(1, b)} \restrict A_1, B_1 
	&=& (s \restrict A_1, B_1)\,[(\labl, a)]_b\\
s\,[(\labl, (2, a))]_{(i, b)} \restrict A_1, B_1
	&=& s \restrict A_1, B_1\\
s\,[(\labr, (1, b)]] \restrict A_1, B_1
	&=& (s \restrict A_1, B_1)\,[(\labr, b)]\\
s\,[(\labr, (2, b)]] \restrict A_1, B_1 
	&=& s \restrict A_1, B_1
\end{array}$}
\quad
\scalebox{.8}{$
\begin{array}{rcl}
\varepsilon \restrict A_2, B_2 &=& \varepsilon\\
s\,[(\labl, (1, a))]_{(i, b)} \restrict A_1, B_1 
        &=& s \restrict A_2, B_2\\
s\,[(\labl, (2, a))]_{(2, b)} \restrict A_2, B_2
        &=& (s \restrict A_2, B_2)\,[(\labl, a)]_b\\
s\,[(\labr, (1, b)]] \restrict A_2, B_2
        &=& s \restrict A_2, B_2\\
s\,[(\labr, (2, b)]] \restrict A_2, B_2
        &=& (s \restrict A_2, B_2)\,[(\labr, b)]
\end{array}$}
\]
\caption{Partial restrictions for the tensor}
\label{fig:restr_tensor}
\end{figure}
Figure \ref{fig:restr_tensor} -- partial, because \emph{e.g.} $[(\labl,
(1, a))]_{(2, b)} \restrict A_1, B_1$ is left undefined. We then set:

\begin{defi}
Consider $\sigma_1 : A_1 \lin B_1$ and $\sigma_2 : A_2 \lin B_2$
alternating strategies. Then:
\[
\sigma_1 \tensor \sigma_2 = \{s \in \Alt(A_1 \tensor A_2 \lin B_1
\tensor B_2) \mid \forall i \in \{1, 2\},~s \restrict A_i, B_i \in
\sigma_i\}\,,
\]
implying in particular that for each $s \in \sigma_1 \tensor \sigma_2$
and $i \in \{1, 2\}$, $s \restrict A_i, B_i$ is defined.
\end{defi}

By definition, $\sigma_1 \tensor \sigma_2$ satisfies \emph{non-empty}
and \emph{prefix-closed}. As for composition, \emph{determinism}
involves performing a state analysis expressing that at each point,
only one of $\sigma_1$ or $\sigma_2$ has control. We skip the details.
See \emph{e.g.} \cite{harmer2004innocent} for an analogous proof, 
also reflected in the proof in Section~\ref{subsubsec:seq_seely} that
sequential causal strategies are stable under tensor.  

\begin{prop}
Consider $\sigma_1 : A_1 \lin B_1$ and $\sigma_2 : A_2 \lin B_2$
alternating strategies. 

Then, $\sigma_1 \tensor \sigma_2 : A_1 \tensor A_2 \lin B_1 \tensor B_2$
is an alternating strategy. 
\end{prop}

Fix $\sigma_1 : A_1 \lin B_1, \sigma_2 : A_2
\lin B_2, \tau_1 : B_1 \lin C_1$ and $\tau_2 : B_2 \lin C_2$. For $w \in
(\tau_1 \tensor \tau_2) \inter (\sigma_1 \tensor \sigma_2)$, we first
define partially $w \restrict A_1, B_1, C_1$ and $w \restrict A_2, B_2,
C_2$ analogously to Figure \ref{fig:restr_tensor} -- it is direct to
prove that $w \restrict A_1, B_1, C_1 \in \tau_1 \inter \sigma_1$ and $w
\restrict A_2, B_2, C_2 \in \tau_2 \inter \sigma_2$.

Functoriality is analogous to associativity in that it relies on a zipping lemma:

\begin{lem}\label{lem:zip_tensor}
Consider $u_1 \in \tau_1 \inter \sigma_1, u_2 \in \tau_2 \inter
\sigma_2$, and $s \in \Alt(A_1 \tensor A_2 \lin C_1 \tensor C_2)$ such
that $s \restrict A_1, C_1 = u \restrict A_1, C_1$ and $s \restrict A_2,
C_2 = u \restrict A_2, C_2$.

Then, there is a unique $w \in (\tau_1 \tensor \tau_2) \inter (\sigma_1
\tensor \sigma_2)$ such that
\[
s = w \restrict A_1 \tensor A_2, C_1 \tensor C_2\,,
\qquad
s\restrict A_1, B_1, C_1 = u_1\,,
\qquad
s \restrict A_2, B_2, C_2 = u_2\,.
\]
\end{lem}

\begin{prop}
The construction $\tensor$ extends to a bifunctor
\[
\tensor : \NegAlt \times \NegAlt \to \NegAlt\,.
\]
\end{prop}
\begin{proof}
Preservation of identities is direct. Functoriality
follows by Lemma \ref{lem:zip_tensor}
\end{proof}

Next, we complete the symmetric monoidal structure by providing the
structural natural isomorphisms.
We first introduce a few tools useful in giving a clean
definition of such structural isomorphisms, which are variants of the
copycat strategy.
We shall make use of certain morphisms, called \emph{renamings}, to act
on strategies.

\begin{defi}
A \textbf{renaming} from arena $A$ to $B$ is a function $f : \ev{A} \to
\ev{B}$ satisfying:
\[
\begin{array}{rl}
\text{\emph{validity:}} & \forall x \in \conf{A},~f x \in \conf{B}\\
\text{\emph{local injectivity:}} & \forall a_1, a_2 \in x \in
\conf{A},~f a_1 = f a_2 \implies a_1 = a_2\\
\text{\emph{polarity-preserving:}} & \forall a \in \ev{A},~\pol_B(f a) =
\pol_A(a)\\
\text{\emph{symmetry-preserving:}} & \forall \theta \in
\tilde{A} \text{(resp. $\ptilde{A}, \ntilde{A}$)},~f\,\theta \in
\tilde{B},\\
&\text{(resp. $\ptilde{B}, \ntilde{B}$)}\\
\text{\emph{strong-receptivity:}} & 
\text{for all $\theta \in \tilde{A}$, for all $f \theta \subseteq^-
\varphi \in \tilde{B}$,}\\
&\text{$\exists! \theta \subseteq^- \theta' \in
\tilde{A}$ $f \theta'= \varphi$}\\
\text{\emph{courtesy:}} & \forall a_1 \imc_A a_2,~(\pol_A(a_1) = + \vee
\pol_A(a_2) = -) \implies f a_1 \imc_B f a_2\,.
\end{array}
\]

We write $f : A \to B$ to mean that $f$ is a renaming from $A$ to $B$.
\end{defi}

This construction is imported from \cite{cg2}. Then we have:

\begin{defi}
Consider $\sigma : A \lin B$ and renamings $g : B \to B'$,
$f : A^\perp \to {A'}^\perp$. We set
\[
g \cdot \sigma \cdot f = \{g \cdot s \cdot f \mid s \in \sigma \} : A'
\lin B'
\]
where $g \cdot s \cdot f$ acts on $s$ event-wise, sending $[(\labl,
a)]_b$ to $[(\labl, f(a))]_{g(b)}$ and $[(\labr, b)]$ to $[(\labr, g(b))]$.
\end{defi}

It is direct that $g \cdot \sigma \cdot f$ is a strategy. The structural
isomorphisms that we aim to define are obtained by \emph{lifting}
renamings. Indeed if $f : A \to B$ is a renaming, we may define
\[
\overrightarrow{f} = f \cdot \cc_A : A \lin B\,.
\]
a renaming from copycat. Likewise, from $f : B^\perp \to A^\perp$ we set
$\overleftarrow{f} = \cc_B \cdot f : A \lin B$.

The main property satisfied by these constructions is the following
\emph{lifting lemma}.

\begin{lem}\label{lem:lifting}
Consider $\sigma : A \lin B$ a strategy, and $f : A^\perp \to
{A'}^\perp$, $g : B \to B'$ two renamings.
Then, we have $\overrightarrow{g} \odot \sigma \odot \overleftarrow{f} =
g \cdot \sigma \cdot f$.
\end{lem}
\begin{proof}
A direct adaptation of the neutrality of copycat under composition.
\end{proof}

Before we put these to use to construct the symmetric monoidal
structure, we deduce a few properties of lifted strategies. In the
statement below, for any renaming $f : A \to B$ which is additionally an
isomorphism, then we write $f^\perp : B^\perp \to A^\perp$ for its
inverse with polarities reversed -- it is immediate that it still
satisfies the conditions of a renaming.

\begin{prop}\label{prop:misc_lift}
Lifting is a functor from the category of renamings to $\NegAlt$.

Moreover, if $f : A \to B$ is an iso,
$\overrightarrow{f}$ is an iso; and we have $\overrightarrow{f}^{-1} =
\overrightarrow{f^{-1}}$ and $\overrightarrow{f} =
\overleftarrow{f^\perp}$.
\end{prop}
\begin{proof}
By Lemma \ref{lem:lifting} and direct verifications.
\end{proof}

With this, we may now define the structural isomorphisms
for the symmetric monoidal structure. We notice that for all arenas $A,
B, C$, there are invertible renamings:
\[
\begin{array}{rcrcl}
\rho_A &:& A \tensor 1 &\to& A\\
\lambda_A &:& 1 \tensor 1 &\to& A\\
\alpha_{A, B, C} &:& (A\tensor B) \tensor C &\to& A \tensor (B \tensor
C)\\
s_{A, B} &:& A \tensor B &\to& B \tensor A
\end{array}
\]
where $1$ is the empty arena, satisfying the coherence laws of a
symmetric monoidal category. The required structural isomorphisms are
simply obtained by lifting those as $\overrightarrow{\rho_A},
\overrightarrow{\lambda_A}, \overrightarrow{\alpha_{A, B, C}}$ and
$\overrightarrow{s_{A, B}}$. By Proposition \ref{prop:misc_lift}, the
required coherence diagrams are still satisfied.

\begin{prop}
The category $\NegAlt$ is a symmetric monoidal category.
\end{prop}
\begin{proof}
It remains to prove naturality, easy from Proposition
\ref{prop:misc_lift} and direct verifications. 
\end{proof}

\subsubsection{Cartesian products} $\NegAlt$ has cartesian products,
given by $\with$ on objects.

To perform the \textbf{pairing} of $\sigma_1 : A \lin B_1$ and $\sigma_2
: A \lin B_2$, we first build prestrategies
\[
\inj_1(\sigma_1) : A \lin B_1 \with B_2\,,
\qquad
\qquad
\inj_2(\sigma_2) : A \lin B_1 \with B_2
\]
defined by applying event-wise $\inj_i([(\labr, b)]) =
[(\labr, (i, b))]$, $\inj_i([(\labl, a)]_{b}) = [(\labl, a)]_{(i, b)}$.
We obtain $\tuple{\sigma_1, \sigma_2} = \inj_1(\sigma_1) \uplus
\inj_2(\sigma_2)$. It is direct that this yields a bijection:
\[
\tuple{-, -} : \NegAlt(A, B_1) \times \NegAlt(A, B_2) \to \NegAlt(A, B_1
\with B_2)
\]

The \textbf{projections} are $(\pi_1, \pi_2) = \tuple{-,
-}^{-1}(\id_{B_1 \with B_2})$, and we can verify
\[
\pi_1 \odot \tuple{\sigma_1, \sigma_2} = \sigma_1\,,
\qquad
\qquad
\pi_2 \odot \tuple{\sigma_1, \sigma_2} = \sigma_2
\]
which is enough to complete the cartesian structure of $\NegAlt$.

\subsubsection{Monoidal closure} First, for $-$-arenas $A, B$ and $C$,
there is a clear isomorphism
\[
(A \tensor B) \lin C \quad \iso \quad A \lin (B \lin C)
\]
which, applied event-wise, yields $\Lambda(-) : \NegAlt(A \tensor B, C)
\simeq \NegAlt(A, B \lin C)$. Then,
\[
\evm_{A, B} = \Lambda^{-1}(\id_{A\lin B}) : (A \lin B) \tensor A \lin B
\]
is the \textbf{evaluation} strategy. It is a verification akin to the
neutrality of copycat that for all $\sigma : A \lin (B
\lin C)$, we have $\evm_{A, C} \odot (\sigma \tensor B) =
\Lambda^{-1}(\sigma)$. 

\begin{prop}
$\NegAlt$ is a cartesian symmetric monoidal closed category.
\end{prop}

\subsection{Symmetry}\label{app:alt_sym}
We now develop the structure pertaining to
symmetry.

\subsubsection{Basic structure}\label{app:basic_alt_sym}
The extension of the construction above with
symmetry and uniformity unfolds essentially as in AJM games \cite{ajm}.
We describe the main steps.

First of all, we ensure all structural morphisms are uniform.
This is the purpose of:

\begin{lem}
For $A$ and $B$ $-$-arenas and $f : A \to B$ a renaming, 
$\overrightarrow{f} \simstrat \overrightarrow{f}$.
\end{lem}
\begin{proof}
Straightforward by \emph{symmetry-preserving} and
\emph{strong-receptivity} of renamings. 
\end{proof}

Next, we show operations on strategies are compatible with $\simstrat$.
The delicate case is composition. Fix $A, B$ and $C$ $-$-arenas, and
write $I = (A \lin B) \lin C$. We shall give to events of $(A \lin B) \lin
C$ a \emph{polarity} corresponding to their role in an interaction: a
move $m$ is \emph{negative} if it is in $A$ or $C$ and is negative for
$A \lin C$, and has \emph{polarity $p$} otherwise.

The main tool is the following lifting of Definition \ref{def:simstrat}
to interactions:

\begin{defi}
Consider $\bsigma, \bsigma' : A \lin B$ and $\btau, \btau' : B \lin C$.
We write $\btau \inter \bsigma \simstrat \btau' \inter \bsigma'$ iff
\[
\begin{array}{rl}
\text{$\to$-\emph{simulation:}} & \forall u m^p \in \tau \inter
\sigma,~v \in \tau' \inter \sigma',~u \sym_I v \implies \exists n^p,~v n
\in \tau'\inter \sigma'~\&~u m \sym_I v n\\
\text{$\ot$-\emph{simulation:}} & \forall u \in \tau \inter \sigma,~v
n^p \in \tau' \inter \sigma',~u \sym_I v \implies \exists m,~u m \in
\tau \inter \sigma~\&~u m \sym_I v n\\
\text{$\to$-\emph{receptive:}} & \forall u m^- \in \tau \inter \sigma,~v \in
\tau' \inter \sigma',~u m^-\,\sym_I\,v n^- \implies v n^- \in \tau'
\inter \sigma'\\ 
\text{$\ot$-\emph{receptive:}} & \forall u \in \tau \inter \sigma,~v n^- \in
\tau'\inter \sigma',~u m^-\,\sym_I\,v n^- \implies u m^- \in \tau \inter
\sigma
\end{array}
\]
\end{defi}

\begin{lem}\label{lem:simstrat_inter}
Consider 
$\bsigma, \bsigma' : A \lin B$ and $\btau, \btau' : B \lin C$ such that
$\bsigma \simstrat \bsigma'$ and $\btau \simstrat \btau'$.

Then, $\btau \inter \bsigma \simstrat \btau' \inter \bsigma'$.
\end{lem}
\begin{proof}
\emph{$\to$-simulation.} Consider $u m^p \in \tau \inter \sigma$ and $v
\in \tau' \inter \sigma'$, s.t. $u \sym_v v$. Necessarily, $m$ is
positive for $\sigma$ or $\tau$, \emph{w.l.o.g.} say $\sigma$.
Again, we distinguish cases whether $m$ is in $A$ or
$B$. We consider $B$, which is the most interesting case -- so that
$m = [(\labm, b)]_c$ for some $c \in \min(C)$.

We project $u [(\labm, b)]_c \restrict A, B = (u \restrict A, B) [(\labr,
b)]^+ \in \sigma$, $v \restrict A, B \in \sigma'$ with $u \restrict A, B
\sym_{A\lin B} v \restrict A, B$. By \emph{$\to$-simulation} for
$\sigma$, there is $[(\labr, b')]$ s.t. $(v \restrict A, B) [(\labr,
b')] \in \sigma'$ and 
\begin{eqnarray}
(u \restrict A, B) [(\labr, b)] \sym_{A\lin B} (v \restrict A, B) [(\labr,
b')]\,.\label{eq:app}
\end{eqnarray}

However, we also have $(u \restrict B, C) [(\labl, b)]_c^- \in \tau$, $v
\restrict B, C \in \tau'$. We also have $u \restrict B, C \sym_{B \lin
C} v \restrict B, C$. Necessarily, $(\labr, c) \in \ev{u \restrict B, C}$
and there is a symmetric $(\labr, c') \in \ev{v \restrict B, C}$. Then,
from \eqref{eq:app} and the definition of symmetry on $A \lin B$ and $B
\lin C$,  
\[
(u \restrict B, C) [(\labl, b)]_c \sym_{B \lin C} (v \restrict B, C)
[(\labl, b')]_{c'}
\]
so by receptivity, $(v \restrict B, C) [(\labl, b')]_{c'} \in \tau'$. It
follows that $v [(\labm, b')]_{c'} \in \tau' \inter \sigma'$ with 
$u [(\labm, b)]_c \sym_{I} v [(\labm, b')]_{c'}$ as required. The
condition $\ot$-simulation is symmetric.

Finally, \emph{$\to$-receptive} and \emph{$\ot$-receptive} are similar
but simpler.
\end{proof}

Finally, we may deduce compatibility of $\simstrat$ with composition:

\begin{cor}
Consider 
$\bsigma, \bsigma' : A \lin B$ and $\btau, \btau' : B \lin C$ such that
$\bsigma \simstrat \bsigma'$ and $\btau \simstrat \btau'$.

Then, $\btau \odot \bsigma \simstrat \btau' \odot \bsigma'$.
\end{cor}
\begin{proof}
Direct from Lemmas \ref{lem:alt_uniq_wit} and \ref{lem:simstrat_inter}. 
\end{proof}

Lemma \ref{lem:alt_uniq_wit} is crucial: this fails if we do
not have a unique witness. This is the main reason why this approach to
uniformity does not extend to $\NNegAlt$.
Other operations on strategies are easily seen to be compatible with
$\simstrat$. Therefore, considering $\NegAlt$ as having as only
morphisms the uniform strategies (\emph{i.e.} self-equivalent
for $\simstrat$), it is equipped with an additional equivalence relation
$\simstrat$ with respect to which all operations are compatible.

\subsection{Seely category} We now provide the last ingredients to the
Seely category.

First, we need a functor $\oc : \NegAlt \to \NegAlt$. The
construction is similar to the tensor and defined by a suitable
restriction,
\begin{figure}
\[
\begin{array}{rclcl}
\varepsilon \restrict \grey{i} &=& \varepsilon\\
s\,[(\labl, (\grey{i}, a))]_{(\grey{i}, b)} \restrict \grey{i}
	&=& (s \restrict \grey{i})\,[(\labl, a)]_b\\
s\,[(\labl, (\grey{j}, a))]_{(\grey{k}, b)} \restrict \grey{i} 
	&=& s \restrict \grey{i} && (\grey{i} \neq \grey{j})\\
s\,[(\labr, (\grey{i}, b))] \restrict \grey{i}
	&=& (s \restrict \grey{i})\,[(\labr, b)]\\
s\,[(\labr, (\grey{j}, b))] \restrict \grey{i}
	&=& s \restrict \grey{i}&&(i\neq j)
\end{array}
\]
\caption{Partial restrictions for $\oc$}
\label{fig:rest_oc}
\end{figure}
given in Figure \ref{fig:rest_oc}. Armed with this, we set:

\begin{defi}
Consider $\sigma : A \lin B$ an alternating strategy. Then, we set:
\[
\oc \sigma = \{s \in \Alt(\oc A \lin \oc B \mid \forall \grey{i} \in
\mathbb{N},~s \restrict \grey{i} \in \sigma\}\,,
\]
implying in particular that for each $s \in \oc \sigma$ and $\grey{i}
\in \mathbb{N}$, $s \restrict \grey{i}$ is defined.
\end{defi}

This is really an infinitary tensor of $\sigma$. That this
yields $\oc \sigma : \oc A \lin \oc B$ an alternating
strategy, along with functoriality, are as for the tensor. We
skip the details. 

\begin{prop}
The construction $\oc$ extends to a functor
$\oc : \NegAlt \to \NegAlt$.
\end{prop}

The Seely structure also includes structural
morphisms, all defined through lifting. We define renamings
\begin{figure}
\begin{minipage}{.45\linewidth}
\[
\begin{array}{rcrcl}
\eta_A &:& A^\perp &\to &(\oc A)^\perp\\
&& a &\mapsto& (\grey{i}, a)\\\\\\
\mu_A &:& (\oc \oc A)^\perp &\to& (\oc A)^\perp\\
&&(\grey{i}, (\grey{j}, a))&\mapsto& (\grey{\tuple{i,j}}, a)\\\\
\end{array}
\]
\caption{Comonad renamings}
\label{fig:ren_com}
\end{minipage}
\hfill
\begin{minipage}{.5\linewidth}
\[
\begin{array}{rcrcl}
\mathsf{see}^\to_{A,B} && \oc A \tensor \oc B &\to& \oc (A \with B)\\
&& (1, (\grey{i}, a)) &\mapsto& (\grey{2i}, (1, a))\\
&& (2, (\grey{i}, b)) &\mapsto& (\grey{2i+1}, (2, a))\\\\
\mathsf{see}^\ot_{A, B} && \oc (A \with B) &\to & \oc A \tensor \oc B\\
&&(\grey{i}, (1, a)) &\mapsto& (1, (\grey{i}, a))\\
&&(\grey{i}, (2, b)) &\mapsto& (2, (\grey{i}, b))
\end{array}
\]
\caption{Seely renamings}
\label{fig:ren_seely}
\end{minipage}
\hfill
\end{figure}
in Figures \ref{fig:ren_com} and \ref{fig:ren_seely} (with $\tuple{-,
-} : \mathbb{N} \times \mathbb{N} \to \mathbb{N}$ any bijection), 
set
\[
\begin{array}{rcrcl}
\der_A = \overleftarrow{\eta_A} &:& \oc A &\lin& A\\
\dig_A = \overleftarrow{\mu_A} &:& \oc A & \lin & \oc \oc A\\
\mon_{A, B} = \overrightarrow{\mathsf{see}_{A, B}^\to} &:& \oc A \tensor
\oc B &\lin & \oc (A\with B)\\
\mon_{A, B}^{-1} = \overrightarrow{\mathsf{see}_{A, B}^\ot} &:& \oc (A
\with B) &\lin& \oc A \tensor \oc B\,,
\end{array}
\]
and all coherence and naturality properties follow from Lemma
\ref{lem:lifting} with direct verifications. A
number of those only hold up to $\simstrat$: for instance, $\mon_{A, B}$
and $\mon^{-1}_{A,B}$ are inverses up to $\simstrat$, but not up to
equality. Finally, $\mon^0 : \ees \lin \oc \ees$ is the empty
strategy. Altogether:

\begin{prop}
The category $\NegAlt$ is a Seely category.
\end{prop}

\section{Construction of Non-Alternating Strategies}
\label{app:nalt}
Next, we construct the category $\NNegAlt$ of $-$-arenas and
non-alternating strategies. 

\subsection{Basic Categorical Structure} We first study composition and
copycat. For composition, the development is essentially a simpler
version of $\NegAlt$ as there is no determinism requirement. On the
other hand, copycat requires slightly more care.

\subsubsection{Composition} Consider fixed $A, B$ and $C$ three
$-$-arenas, and $\sigma : A \lin B$, $\tau : B \lin C$ non-alternating
prestrategies. Re-using pre-interactions and restrictions from before, we set:

\begin{defi}
An \textbf{interaction} $u \in \tau
\inter \sigma$, is a pre-interaction $u$ such that
\[
u \restrict A, B \in \sigma\,,
\qquad
u \restrict B, C \in \tau\,,
\]
and we automatically have $u \restrict A, C \in \NAlt(A\lin C)$.

The \textbf{composition} of $\sigma$ and $\tau$ is $\tau \odot \sigma =
\{u \restrict A, C \mid u \in \tau \inter \sigma\}$.
\end{defi}

It is clear this yields a non-alternating prestrategy, which only
requires \emph{non-empty} and \emph{prefix-closed}.
In composing non-alternating (pre)strategies, there is no analogue of
Lemma \ref{lem:state-alt}: all players can move anytime. The unique
witness property is lost: $s \in \tau \odot \sigma$ if there is
a witness $u \in \tau \inter \sigma$ such that $s = u \restrict A, C$,
but $u$ is in general not unique\footnote{Causal strategies in
$\CG$, with their explicit branching information, recover a unique witness
property.}.

That composition preserves strategies will follow from Propositions
\ref{prop:app_nalt_assoc} and \ref{prop:app_nalt_strat_char}. 

\subsubsection{Associativity} Fix $\sigma : A\lin B$, $\tau : B \lin C$
and $\delta : C \lin D$ non-alternating prestrategies. As above,
an \textbf{interaction} of $\sigma, \tau$ and $\delta$ is a
$3$-pre-interaction $w \in \ev{((A\lin B) \lin C) \lin D}^*$ such that
\[
w \restrict A, B \in \sigma\,,
\qquad
w \restrict B, C \in \tau\,,
\qquad
w \restrict C, D \in \delta\,,
\]
from which it follows automatically that $w \restrict A, D \in \NAlt(A\lin D)$.

Again, associativity relies on a ``zipping lemma''. 

\begin{lem}[Zipping]\label{lem:app_nalt_zip}
Consider $u \in \delta \inter (\tau \odot \sigma)$ and $v \in \tau
\inter \sigma$ such that $u \restrict A, C = v \restrict A,
C$.

Then, there is $w \in \delta \inter \tau
\inter \sigma$ such that $w \restrict A, C, D = u$ and $w \restrict A,
B, C = v$. 
\end{lem}

Again, $w$ is not unique. The interactions $u$ and $v$ impose causal
constraints on $\ev{w}$, and $w$ may be chosen as any ternary
interaction respecting those. As in the alternating case there is a
mirror lemma, and associativity follows as in Proposition
\ref{prop:app_alt_assoc}. 

\begin{prop}\label{prop:app_nalt_assoc}
We have $(\delta \odot \tau) \odot \sigma = \delta \odot (\tau \odot
\sigma)$.
\end{prop}

Associativity works for prestrategies, \emph{i.e.} it does not
rely on \emph{receptive} and \emph{courteous}.

\subsubsection{Copycat} Non-alternating strategies are intended to
model \emph{asynchronous} concurrency -- accordingly, the
non-alternating copycat, still written $\cc_A : A \lin A$, is an
asynchronous forwarder. 
We first describe the configurations on $A \vdash A$
reached by copycat: those are $x \parallel y$ for $x, y \in
\conf{A}$, such that $y \scott x$, where $\scott$ is the
``Scott order'' \cite{DBLP:conf/fossacs/Winskel13} defined as
\[
y \scott x \quad \Leftrightarrow \quad y \supseteq^- x \cap y
\subseteq^+ x
\]
with polarities taken on $A$. Whenever $\cc_A$ receives a negative
event, it forwards it to the other side, but asynchronously: $y$ may
contain negative events not yet forwarded to $x$, and $x$ may contain
positive events (\emph{i.e.} negative for $A \vdash A$) not yet
forwarded to $y$. 

But copycat plays on $A \lin A$, not on $A \vdash A$. Remember from
Section~\ref{subsubsec:gen_arr} that there is
\[
\chi_{A, B} : (A \lin B) \to (A\vdash B)
\]
satisfying the axioms of a map of event structures. We may set:

\begin{defi}\label{def:asy_cc}
For $A$ any $-$-arena, we set $\cc_A$ to comprise all $s \in
\NAlt(A \lin A)$ s.t.:
\[
\begin{array}{rl}
\text{\emph{balanced:}} & 
	\text{for all $1\leq i \leq n$, writing $\chi_{A, A}\,\ev{s_1
\dots s_n} = x \parallel y$, we have $y \scott x$,}\\
\text{\emph{well-linked:}} &
	\text{for all $[(\labl, a)]_{a'}$, we have $a = a'$,}
\end{array}
\]
where $s = s_1\dots s_n$.
\end{defi}

It is direct that copycat is \emph{receptive} and
\emph{courteous}. In fact, it turns out that \emph{receptive} and
\emph{courteous} are \emph{exactly} the required conditions for copycat
to be neutral for composition:

\begin{prop}\label{prop:app_nalt_strat_char}
Consider $A$ and $B$ two $-$-arenas, and $\sigma : A \lin B$ a
prestrategy.

Then, $\sigma$ is a strategy if and only if $\cc_B \odot \sigma \odot
\cc_A = \sigma$.
\end{prop}
\begin{proof}[Sketch]
\emph{If.} It is direct that all non-alternating prestrategies $\cc_B
\odot \sigma \odot \cc_A$ are \emph{receptive} and \emph{courteous}. 
\emph{Only if.} Considering $w \in \cc_B \inter \sigma \inter \cc_A$,
using that $\sigma$ is \emph{receptive} and \emph{courteous}, $w$ can be
transformed by permuting into $w'$ with the same outer restriction, but
where all moves are immediately forwarded by copycat. It follows that the
outer restriction is in $\sigma$.
\end{proof}

This was noticed in \cite{gm}, and also holds for
causal strategies \cite{lics11}. We deduce:

\begin{cor}
The $-$-arenas and non-alternating strategies form a category
$\NNegAlt$.
\end{cor}
\begin{proof}
Follows from Propositions
\ref{prop:app_nalt_assoc} and \ref{prop:app_nalt_strat_char}.
\end{proof}

\subsubsection{Further structure} As for $\NegAlt$, we have -- with the
same constructions:

\begin{prop}
$\NNegAlt$ is a cartesian symmetric monoidal closed category.
\end{prop}

We could easily adapt the developments of Section~\ref{subsec:seq_ia} to
show that there is a subcategory of \emph{sequential} non-alternating
strategies, which maps functorially to $\NegAlt$ (without uniformity).
However, our attempts to endow $\NNegAlt$ with symmetry failed, because
of the lack of \emph{unique witness} in compositions. It seems it could
be done by using a different approach to symmetry, namely 
Melli\`es' group-theoretic formulation of uniformity \cite{ag1}; this is
however left for future work.
We conclude this section with:

\naltsmcc*
\begin{proof}
It remains to prove that copycat strategies are well-bracketed and that
well-bracketing is preserved by operations on strategies, which is
a routine verification.
\end{proof}

\section{Thin Concurrent Games}

In this section, we give some proof for thin concurrent games. Notably, 
we detail the proofs for the characterisations of configurations of
interaction and composition used in this paper: they do not appear in
\cite{cg2}, as they were elaborated more recently.

\subsection{Representable functions} We investigate the functions
between domains of configurations representable as maps of event
structures.

\subsubsection{Representable functions between configurations}
We will be
interested in those functions between configurations 
arising from maps of event structures:

\begin{defi}
For $A, B$ es, a function $f : \conf{A} \to \conf{B}$
is \textbf{representable} if there is a (necessarily unique) map of
event structure $\hat{f} : A \to B$ s.t. for all $x \in \conf{A}$,
$\hat{f}(x) = f(x)$.
\end{defi}
\begin{proof}
For uniqueness, if $f_1, f_2 : A \to B$ have the same image for
configurations, we have
\[
f_1([a)_A) \longcov{f_1(a)} f_1([a]_A)\,,
\qquad
f_2([a)_A) \longcov{f_2(a)} f_2([a]_A)\,,
\]
where $[a)_A = \{a' \in \ev{A} \mid a' <_A a\}$. 

Since $f_1([a)_A) = f_2([a)_A)$ and $f_1([a]_A) = f_2([a]_A)$ this
entails $f_1(a) = f_2(a)$. 
\end{proof}

We shall use the following characterisation of representable functions:

\begin{lem}\label{lem:mapification}
For $A, B$ es, a function $f : \conf{A} \to \conf{B}$ is representable
iff it:
\[
\begin{array}{rl}
\text{\emph{preserves the empty set:}} &
	f(\emptyset) = \emptyset\\
\text{\emph{preserves covering:}} & 
	\text{for $x, y \in \conf{A}$, if $x \cov y$, then $f(x) \cov f(y)$}\\
\text{\emph{preserves unions:}} &
	\text{for $x, y \in \conf{A}$, if $x \cup y \in \conf{A}$ then
$f(x\cup y) = f(x) \cup f(y)$.}
\end{array}
\]
\end{lem}
\begin{proof}
\emph{If.} For any $a \in \ev{A}$, write $[a)_A = [a]_A \setminus
\{a\}$. It is always a configuration of $A$, and $[a)_A \cov [a]_A$.
Because $f$ preserves covering, there is $b \in \ev{B}$ such that 
\[
f([a)_A) \longcov{b} f([a]_A)\,,
\]
\emph{i.e.} $f([a]_A) = f([a)_A) \cup \{b\}$, write $b = \hat{f}(a)$.

This defines a function on events. Next, we show that its direct
image of configurations agrees with $f$. We first claim that 
for all $x, y \in \conf{A}$ such that $x \cov y$ and
$y = x \cup \{a\}$, we have $f(y) = f(x) \cup \{\hat{f}(a)\}$ as well.
To prove this, consider the diagram
\[
\xymatrix@C=-3pt@R=0pt{
[a)_A	&\cov& x_1 & \cov &\dots & \cov & x_{n-1} &\cov & x\\
\rotatebox{270}{$\hspace{-10pt}\cov$}&&
\rotatebox{270}{$\hspace{-10pt}\cov$}&&&&
\rotatebox{270}{$\hspace{-10pt}\cov$}&&
\rotatebox{270}{$\hspace{-10pt}\cov$}\\
[a]_A	&\cov& y_1 & \cov & \dots& \cov & y_{n-1} & \cov & y
}
\]
where, for each $1\leq i < n$, $y_i = x_i \cup \{a\}$. Because $f$
preserves covering and unions, $f$ preserves covering squares, hence by
straightforward induction, for all $1\leq i < n$,
\[
f(x_i) \longcov{\hat{f}(a)} f(y_i)\,,
\]
so $f(y) = f(x) \cup \{\hat{f}(a)\}$ as required. 

Now, we can prove that for all $x \in \conf{A}$, we
have $f(x) = \hat{f}(x)$. Indeed, there is
\[
\emptyset \longcov{a_1} \dots \longcov{a_n} x
\]
a covering chain, which as $f$ preserves the empty set and by the
property just above, entails 
\[
\emptyset \longcov{\hat{f}(a_1)} \dots \longcov{\hat{f}(a_n)} f(x)\,,
\]
as well, showing that $f(x) = \hat{f}(x)$ as needed.
From this point it is trivial that $\hat{f}$ is a map of es:
preservation of configurations follows from $f : \conf{A}
\to \conf{B}$, and local injectivity from the preservation of covering.
\emph{Only if.} Immediate verification.
\end{proof}

\subsubsection{Representable functions with symmetry}
We extend this construction with symmetry.

\begin{defi}\label{def:rep_map_ess}
Let $A, B$ be ess. A function $f : \conf{A} \to \conf{B}$ is
\textbf{representable} if there is a (necessarily unique) map of ess
$\hat{f} : A \to B$ such that for all $x \in \conf{A}$, $\hat{f}(x) =
f(x)$.
\end{defi}

As above, we give a characterization of representable functions between
configurations.

\begin{lem}\label{lem:map_ess_rep}
Let $A, B$ be ess. Then, a function $f : \conf{A} \to \conf{B}$ is
\textbf{representable} iff it:
\[
\begin{array}{rl}
\text{\emph{preserves the empty set:}} &
        f(\emptyset) = \emptyset\\
\text{\emph{preserves covering:}} & 
        \text{for $x, y \in \conf{A}$, if $x \cov y$, then $f(x) \cov
f(y)$}\\
\text{\emph{preserves unions:}} &
        \text{for $x, y \in \conf{A}$, if $x \cup y \in \conf{A}$ then
$f(x\cup y) = f(x) \cup f(y)$.}\\
\text{\emph{preserves symmetry:}}&
\text{there is a (necessarily unique) monotone $\widetilde{f} :
\tilde{A} \to \tilde{B}$}\\ &\text{such that $\dom \circ \widetilde{f} =
f \circ \dom$ and $\cod \circ \widetilde{f} = f \circ \cod$.} 
\end{array}
\]
\end{lem}
\begin{proof}
\emph{If.}
If $f$ satisfies the first three axioms, it is representable
without symmetry. By Lemma \ref{lem:mapification} there is a map of
event structures $\hat{f} : A \to B$ s.t. for all $x\in \conf{A}$,
$\hat{f}(x) = f(x)$.

Now, assuming that $f$ preserves symmetry, we show that actually we must
have
\[
\widetilde{f}(\theta) = \{(\hat{f}(a), \hat{f}(a')) \mid (a, a') \in
\theta\}\,,
\]
\emph{i.e.} $\hat{f}(\theta)$, for all $\theta \in \tilde{A}$. Indeed,
consider a
covering chain for $\theta$, \emph{i.e.} a sequence in $\tilde{A}$:
\[
\emptyset = \theta_0 \longcov{(a_1, a'_1)} \theta_1 \longcov{(a_2, a'_2)} \dots
\longcov{(a_n, a'_n)} \theta_n = \theta\,.
\]

We show by induction on $0\leq i \leq n$ that $\widetilde{f}(\theta_i) =
\hat{f}(\theta_i)$. First, $\dom \circ \widetilde{f}(\emptyset) = f(\emptyset) =
\emptyset$, so $\widetilde{f}(\emptyset) = \emptyset =
\hat{f}(\emptyset)$. For $0 \leq i \leq n$, by IH,
$\widetilde{f}(\theta_i) = \hat{f}(\theta_i)$. We then have:
\begin{eqnarray*}
\dom(\widetilde{f}(\theta_{i+1})) &=& f(\dom(\theta_{i+1}))\\
&=& f(\dom(\theta_i) \cup \{a_{i+1}\})\\
&=& f(\dom(\theta_i)) \cup \{\hat{f}(a_{i+1})\}\,,
\end{eqnarray*}
and the symmetric reasoning shows 
$\cod(\widetilde{f}(\theta_{i+1})) =
\cod(\widetilde{f}(\theta_{i})) \cup \{\hat{f}(a'_{i+1})\}$
as well. But finally, we also have $\widetilde{f}(\theta_i) \subseteq
\widetilde{f}(\theta_{i+1})$ since $\widetilde{f}$ is monotone. So we
must have 
\begin{eqnarray*}
\widetilde{f}(\theta_{i+1}) &=& \widetilde{f}(\theta_i) \cup
\{(\hat{f}(a_{i+1}), \hat{f}(a'_{i+1}))\}\\
&=& \hat{f}(\theta_i) \cup \{(\hat{f}(a_{i+1}), \hat{f}(a'_{i+1}))\}\\
&=& \hat{f}(\theta_{i+1})\,,
\end{eqnarray*}
as required. Now this exactly means that for all $\theta \in \tilde{A}$
we have $\hat{f}(\theta) \in \tilde{B}$, and thus $\hat{f}$ is a map of
event structures with symmetry.
\emph{Only if.} Obvious.
\end{proof}

\subsection{Interaction and Composition} 
Fix $\bsigma : A \vdash B$ and
$\btau : B \vdash C$ (pre)strategies.

\subsubsection{Interaction} \label{app:interaction}
First, we characterise interaction in terms
of its configurations.

We start by recalling that by Lemma 3.12 of \cite{cg2}, there is a
pullback
\[
\xymatrix@R=20pt@C=0pt{
&\btau \inter \bsigma
        \ar[dl]_{\Pi_\bsigma}
        \ar[dr]^{\Pi_\btau}
        \pb{270}\\
\bsigma \parallel C
        \ar[dr]_{\pr_{\bsigma} \parallel C}&&
A \parallel \btau
        \ar[dl]^{A \parallel \pr_{\btau}}\\
&A\parallel B \parallel C
}
\]
in the category of event structures with symmetry, the \emph{interaction
pullback}. First, we have:

\begin{lem}\label{lem:proj_caus_comp}
For any $x\in \conf{\btau \inter \bsigma}$, writing
\[
\Pi_\bsigma\,x = x^\bsigma \parallel x^\btau_C \in \conf{\bsigma
\parallel C}
\qquad
\qquad
\Pi_\btau\,x = x^\bsigma_A \parallel x^\btau \in \conf{A\parallel
\btau}\,,
\]
then we have $x^\bsigma \in \conf{\bsigma}$ and $x^\btau \in
\conf{\btau}$ causally compatible.
\end{lem}
\begin{proof}
First, $x^\bsigma$ and $x^\btau$ are matching. For causal
compatibility, we consider 
\[
\varphi_{x^\bsigma, x^\btau}
\quad
:
\quad
x^\bsigma \parallel x^\tau_C 
\quad
\stackrel{\pr_\bsigma \parallel x^\btau_C}{\simeq} 
\quad
x^\bsigma_A \parallel x_B \parallel x^\btau_C 
\quad
\stackrel{x^\bsigma_A \parallel \pr_\btau^{-1}}{\simeq}
\quad
x^\bsigma_A \parallel x^\btau
\]
from Definition \ref{def:caus_comp}, with $(m, n)
\vartriangleleft (m', n')$ iff $m <_{\bsigma \parallel C} m'$ or 
$n <_{A \parallel \btau} n'$ which we must prove acyclic. 
Now, the function which to $c \in x$ associates $(\Pi_\bsigma(c),
\Pi_\btau(c))$ is a bijection
\[
\psi : x \simeq \varphi_{x^\bsigma, x^\btau}\,,
\]
and we now claim that for all $c, c' \in x$, if $\psi(c) \vartriangleleft
\psi(c')$, then $c <_{\btau\inter \bsigma} c'$. Indeed, say $c = (m, n)$ and $c' = (m',
n')$, with \emph{w.l.o.g.} $m <_{\bsigma \parallel C} m'$. Of course
then, $m = \Pi_\bsigma(c)$ and $m' = \Pi_{\btau}(c')$. Now, we use that
$\Pi_\bsigma$ is a map of event structures, and those always locally
reflect causality (Lemma \ref{lem:es_refl_caus}). Therefore,
we must have $c <_{\btau \inter \bsigma} c'$ as required.

Therefore a cycle for $\vartriangleleft$ in
$\varphi_{x^\bsigma, x^\btau}$ would induce a cycle for $<_{\btau\inter
\bsigma}$ in $x$, contradiction.
\end{proof}

\charinter*
\begin{proof}
\emph{Existence.} 
We must provide the two order-isomorphisms announced. Take $x^\bsigma
\in \conf{\bsigma}$ and $x^\btau \in \conf{\btau}$ causally compatible.
By Definition \ref{def:caus_comp}, the bijection
\[
\varphi_{x^\bsigma, x^\btau}
\quad
:
\quad
x^\bsigma \parallel x^\tau_C 
\quad
\stackrel{\pr_\bsigma \parallel x^\btau_C}{\simeq} 
\quad
x^\bsigma_A \parallel x_B \parallel x^\btau_C 
\quad
\stackrel{x^\bsigma_A \parallel \pr_\btau^{-1}}{\simeq}
\quad
x^\bsigma_A \parallel x^\btau
\]
is secured, \emph{i.e.} the relation 
$(m, n) \vartriangleleft (m', n') 
\Leftrightarrow
m \leq_{\bsigma \parallel C} m'
\vee
n \leq_{A \parallel \btau} n'
$
defined on the graph of $\varphi_{x^\bsigma, x^\btau}$ is acyclic, so
its reflexive transitive closure $\leq_{x^\bsigma, x^\btau}$ is a
partial order. This turns $\varphi_{x^\bsigma, x^\btau}$ into an event
structure, and in fact an ess with identity symmetries. Moreover, there
are obvious maps of ess
\[
\pi_\bsigma : \varphi_{x^\bsigma, x^\btau} \to \bsigma \parallel C
\qquad
\qquad
\pi_\btau : \varphi_{x^\bsigma, x^\btau} \to A \parallel \btau
\]
commuting with display maps, so by the universal property, 
$\tuple{\pi_\bsigma, \pi_\btau} : \varphi_{x^\bsigma, x^\btau} \to
\btau \inter \bsigma$ and
\[
x^\btau \inter x^\bsigma = \tuple{\pi_\bsigma, \pi_\btau}
(\ev{\varphi_{x^\bsigma, x^\btau}}) \in \conf{\btau \inter \bsigma}
\]
concludes the definition of the action of $(-\inter -)$ on causally
compatible pairs. Reciprocally, if $x \in \conf{\btau \inter \bsigma}$
then its projections yield $\Pi_\bsigma(x) = x^\bsigma \parallel
x^\btau_C$ and $\Pi_\btau(x) = x^\btau_A \parallel x^\btau$, and 
by Lemma \ref{lem:proj_caus_comp}, $x^\bsigma$ and $x^\btau$ are
causally compatible.
Finally, for $x^\bsigma$ and $x^\btau$ causally compatible,
\[
\Pi_\bsigma(x^\tau \inter x^\bsigma) = x^\bsigma \parallel x^\btau_C
\qquad
\qquad
\Pi_\btau(x^\btau \inter x^\bsigma) = x^\bsigma_A \parallel x^\btau
\]
by construction and if $x \in \conf{\btau \inter \bsigma}$, $x =
x^\btau \inter x^\bsigma$ by universal property of the pullback.
The projections are monotone, and the monotonicity of $(-\inter -)$
follows from the universal property. 
For symmetries, causal compatibility of $\theta^\bsigma \in
\tilde{\bsigma}$ and $\theta^\btau \in \tilde{\btau}$ amounts
to
\[
\xymatrix@R=5pt@C=5pt{
x^\bsigma \parallel x^\btau_C
	\ar@{}[dd]_{\theta^\bsigma \parallel \theta^\btau_C~}
	\ar@{}[rr]^{\Pi_\bsigma^{-1}} & \simeq & 
x^\btau \inter x^\bsigma 
	\ar@{}[rr]^{\Pi_\btau}&
\simeq& x^\bsigma_A \parallel x^\btau
	\ar@{}[dd]^{~\theta^\bsigma_A \parallel \theta^\btau}\\
\rotatebox{90}{$\iso$}&&&&\rotatebox{90}{$\iso$}\\
y^\bsigma \parallel y^\btau_C
	\ar@{}[rr]_{\Pi_\bsigma}&
\simeq&
y^\btau \inter y^\bsigma
	\ar@{}[rr]_{\Pi_\btau}&
\simeq&
y^\bsigma_A \parallel y^\btau
}
\]
commuting, inducing $\theta^\inter : x^\btau \inter
x^\bsigma \simeq y^\btau \inter y^\bsigma$. But symmetries on $\btau
\inter \bsigma$ are precisely those bijections $x^\btau
\inter x^\bsigma \simeq y^\btau \inter y^\bsigma$ projecting to
$\tilde{S\parallel C}$ and $\tilde{A \parallel T}$ as above, so
$\theta^\inter \in \tilde{\btau \inter \bsigma}$. Reciprocally,
$\theta^\inter \in \tilde{\btau \inter \bsigma}$ induces $\theta^\bsigma
\in \tilde{\bsigma}$ and $\theta^\btau \in \tilde{\btau}$ by
projections. That this yields an order-iso compatible
with $\dom$ and $\cod$ is direct.

\emph{Uniqueness.} Two event structures with symmetry
satisfying the hypotheses obviously have isomorphic domains of
configurations (and isomorphic domains of symmetries, in a compatible
manner). But such order-isomorphisms between
domains of configurations and symmetries are automatically representable
in the sense of Definition \ref{def:rep_map_ess}. Therefore, by Lemma
\ref{lem:map_ess_rep} the isomorphisms are generated by isomorphisms of
event structures with symmetry, as required. Preservation of display
maps holds for configurations by hypothesis; and by uniqueness in Lemma
\ref{lem:map_ess_rep}, two maps of ess with the same action on
configurations must be equal. Therefore, the isomorphism commutes with
display maps and is an isomorphism of pre-interactions as required.
\end{proof}

\subsubsection{Composition}\label{app:char_conf_comp}
We aim to prove Propositions \ref{prop:char_conf_comp} and
\ref{prop:comp_pcov}, exploiting Proposition
\ref{prop:main_interaction2}, along with a few extra lemmas. First, we
notice a connection between \emph{minimality} of causally compatible
pairs and that maximal events of interactions are visible. 

\begin{lem}
For $x^\bsigma \in \conf{\bsigma}$ and $x^\btau \in \conf{\btau}$
causally compatible, they are \emph{minimal} causally compatible
iff the maximal events of $x^\btau \inter x^\bsigma$ are
\textbf{visible}, \emph{e.g.} occur in $A$ or $C$.
\end{lem}
\begin{proof}
\emph{If.} Consider $y^\bsigma \in \conf{\bsigma}$ and $y^\btau \in
\conf{\btau}$ causally compatible such that $y^\bsigma \subseteq
x^\bsigma$, $y^\btau \subseteq x^\btau$, while $x^\bsigma_A =
y^\bsigma_A$ and $x^\btau_C = y^\btau_C$. Assume, seeking a
contradiction, that we have $y_B \subset x_B$ a strict inclusion.
Necessarily, we have $y^\btau \inter y^\bsigma \subset x^\btau \inter
x^\bsigma$. Take $m \in (y^\btau \inter
y^\bsigma) \setminus (x^\btau \inter x^\bsigma)$, \emph{w.l.o.g.} we can
assume that $m$ is maximal for $\leq_{\btau \inter \bsigma}$ in $x^\btau
\inter x^\bsigma$. By hypothesis, $m$ occurs in $A$ or $C$. But this
immediately contradicts the hypothesis that $x^\bsigma_A =
y^\bsigma_A$ and $x^\btau_C = y^\btau_C$.

\emph{Only if.} Assume $x^\bsigma, x^\btau$ are minimal causally
compatible, and take $m \in x^\btau \inter x^\bsigma$ maximal. Seeking a
contradiction, assume that $m$ occurs in $B$. So, projecting
\[
\Pi_\bsigma\,(x^\btau \inter x^\bsigma) = x^\bsigma \parallel
x^\btau_C\,,
\qquad
\qquad
\Pi_\btau\,(x^\btau \inter x^\bsigma) = x^\bsigma_A \parallel x^\btau\,,
\]
we have $\Pi_\bsigma m = (1, s)$ with $s \in x^\bsigma$ and $\Pi_\btau m
= (2, t)$ with $t \in x^\btau$. As maps of event structures locally
reflect causality (Lemma \ref{lem:es_refl_caus}), $s$ is maximal in
$x^\bsigma$ and $t$ is maximal in $x^\btau$. Hence, $y^\bsigma =
x^\bsigma \setminus \{s\} \in \conf{\bsigma}$ and $y^\btau = x^\btau
\setminus \{t\} \in \conf{\btau}$. By construction they are causally
compatible with the same projections to $A$ and $C$, contradicting
minimality of $x^\bsigma$ and $x^\btau$.
\end{proof}

We write $\confv{\btau \inter \bsigma}$ for the configurations whose
maximal events are visible and likewise for symmetries. Then, the lemma
above means that we can refine Proposition \ref{prop:main_interaction2}
to: 

\begin{prop}\label{prop:char_comp_half1}
The order-isomorphisms of Proposition \ref{prop:main_interaction2}
restrict to
\[
\begin{array}{rcrcl}
(- \inter -) \!\!
&:&
\!\!\{(x^\btau, x^\bsigma) \in \conf{\btau} \times \conf{\bsigma} \mid
\text{$x^\bsigma$ and $x^\btau$ \emph{minimal} caus. comp.}\}
\!\!
&\simeq&
\!\! \confv{\btau \inter \bsigma}\\
(- \inter -) \!\!
&:&
\!\!\{(\theta^\btau, \theta^\bsigma) \in \tilde{\btau} \times
\tilde{\bsigma} \mid \text{$\theta^\bsigma$ and $\theta^\btau$
\emph{minimal} caus. comp.}\} \!\!
&\simeq&
\!\! \tildev{\btau \inter \bsigma}
\end{array}
\]
\end{prop}

Now, it remains to link configurations of $\btau \odot \bsigma$ with
configurations of $\btau \inter \bsigma$ with visible maximal events,
and likewise for symmetries. First, we fix a few notations.
Any configuration $x \in \conf{\btau \inter \bsigma}$ yields a
configuration of the composition, its \textbf{hiding}, defined as
${{\proj{x}}} = x \cap \ev{\btau \odot \bsigma} \in \conf{\btau \odot
\bsigma}$. Reciprocally, if $x \in \conf{\btau \odot \bsigma}$ is a configuration
of the composition, its \textbf{witness} is defined as
\[
[x]_{\btau \inter \bsigma} = \{m \in \ev{\btau \inter \bsigma} \mid
\exists n \in x,~m\leq_{\btau \inter \bsigma} n\} \in \conf{\btau \inter
\bsigma}\,.
\]

The next point we make, is that interactions with visible maximal events
are exactly those arising as witnesses of configurations of the
composition. More precisely, we have:

\begin{lem}\label{lem:char_comp_half2}
There are order-isomorphisms compatible with $\dom, \cod$ and display
maps:
\[
\begin{array}{rcrclcl}
\proj{(-)} &:& \confv{\btau \inter \bsigma} &\simeq& \conf{\btau \odot
\bsigma} &:& [-]_{\btau \inter \bsigma}\\
\proj{(-)} &:& \tildev{\btau \inter \bsigma} &\simeq& \tilde{\btau \odot
\bsigma} &:& [-]_{\btau \inter \bsigma}
\end{array}
\]
\end{lem}
\begin{proof}
For configurations, to $x \in \confv{\btau \inter \bsigma}$ we associate
its hiding $\proj{x} \in \conf{\btau \odot \bsigma}$. Reciprocally, to
$x \in \conf{\btau \odot \bsigma}$, we associate its witness $[x]_{\btau
\inter \bsigma} \in \confv{\btau \inter \bsigma}$. Those operations
preserve inclusion, and it is an elementary verification that they are
inverses. 

For symmetries, any $\theta : x \sym_{\btau \inter \bsigma} y$ must
preserve visible events, so it induces by \emph{hiding}
\[
\proj{\theta} : \proj{x} \sym_{\btau \odot \bsigma} \proj{y}
\]
a symmetry on the composition; and hiding preserves inclusion.
Reciprocally, if $\theta : x \sym_{\btau \odot \bsigma} y$ then by
definition there is $\theta \subseteq \theta' : x' \sym_{\btau \inter
\bsigma} y'$. Necessarily, $[x]_{\btau \inter \bsigma} \subseteq x'$ and
$[y]_{\btau \inter \bsigma} \subseteq y'$. Since $\theta'$ is an
order-iso, by \emph{restriction} we may assume \emph{w.l.o.g.}
$\theta' : [x]_{\btau \inter \bsigma} \sym_{\btau \inter \bsigma}
[y]_{\btau \inter \bsigma}$.

But the witness $\theta'$ is \emph{unique}: by Lemma
3.33 of \cite{cg2}, if $\theta'' : [x]_{\btau \inter \bsigma}
\sym_{\btau \inter \bsigma} [y]_{\btau \inter \bsigma}$ is such that
$\proj{\theta''} = \proj{\theta'} = \theta$, then $\theta' = \theta''$.
So to $\theta : x \sym_{\btau \odot \bsigma} y$ we associate
this unique $\theta' : [x]_{\btau \inter \bsigma} \sym_{\btau \inter
\bsigma} [y]_{\btau \inter \bsigma}$. Monotonicity and that
these operations are inverses also follow immediately from uniqueness of
the witness symmetry, \emph{i.e.} Lemma 3.33 of \cite{cg2}.
\end{proof}

Composing the order-isomorphisms of Proposition
\ref{prop:char_comp_half1} and Lemma \ref{lem:char_comp_half2}, we get:

\charcomp*
\begin{proof}
\emph{Existence.} Direct by Proposition
\ref{prop:char_comp_half1} and Lemma \ref{lem:char_comp_half2}.

\emph{Uniqueness.} As in the proof of Proposition
\ref{prop:main_interaction2}.
\end{proof}

This means that any configuration $x\in \conf{\btau \odot \bsigma}$ may
be written uniquely as $x^\btau \odot x^\bsigma$ for $x^\bsigma \in
\conf{\bsigma}$ and $x^\btau \in \conf{\btau}$ minimal causally
compatible -- note that from the construction of the order-isomorphism
in the proposition above, we then have $x^\btau \inter x^\bsigma =
[x^\btau \odot x^\bsigma]_{\btau \inter \bsigma}$. 

\subsubsection{$+$-covered case}\label{app:pluscov}
Finally, it remains to prove Proposition \ref{prop:comp_pcov}. 

In essence, Proposition \ref{prop:comp_pcov} is a specialization of
Proposition \ref{prop:char_conf_comp} to strategies. Consider therefore
from now on that $\bsigma : A \vdash B$ and $\btau : B \vdash C$ are
\emph{strategies}. First of all, we show that for $+$-covered
configurations we can omit the minimality assumption. 

\begin{lem}\label{lem:pcov1}
Consider $x^\bsigma \in \confp{\bsigma}$ and $x^\btau \in \confp{\btau}$
causally compatible $+$-covered.

Then, $x^\bsigma$ and $x^\btau$ are \emph{minimal} causally compatible.
\end{lem}
\begin{proof}
Seeking a contradiction, assume $y^\bsigma$ and $y^\btau$ are causally
compatible such that $y^\btau \inter y^\bsigma \subset x^\btau \inter x
^\bsigma$ with $x^\bsigma_A = y^\bsigma_A$ and $x^\btau_C = y^\btau_C$.
Without loss of generality, consider $m \in (x^\btau \inter x^\bsigma)
\setminus (y^\btau \inter y^\bsigma)$ maximal for $\leq_{\btau\inter
\bsigma}$. By hypothesis, $m$ occurs in $B$. Therefore, projecting
\[
\Pi_\bsigma(m) = (1, s)
\qquad
\qquad
\Pi_\btau(m) = (2, t)\,,
\]
maximality of $m$ in $x^\btau \inter x^\bsigma$ entails via Lemma
\ref{lem:es_refl_caus} that $s$ is maximal in $x^\bsigma$ and $t$ maximal in
$x^\btau$. But necessarily, $s$ and $t$ has dual polarities:
\emph{w.l.o.g.} say that $\pol_\bsigma(s) = +$ and $\pol_\btau(t) = -$.
So, $t$ is negative maximal in $x^\btau$, contradicting
that $x^\btau \in \confp{\btau}$ is $+$-covered.
\end{proof}

So in a synchronization between $+$-covered configurations,
the maximal events are visible as if they are synchronized, they will be
both maximal and negative for one of the two players. Resulting
configurations of the composition are automatically $+$-covered:

\begin{lem}\label{lem:pcov2}
For $x^\bsigma \in \confp{\bsigma}, x^\btau \in \confp{\btau}$ 
causally compatible, $x^\btau \odot x^\bsigma \in \confp{\btau
\odot \bsigma}$.
\end{lem}
\begin{proof}
Consider $m \in x^\btau \odot x^\bsigma$ maximal. This means that $m$ is
also maximal in $x^\btau \inter x^\bsigma = [x^\btau \inter
x^\bsigma]_{\btau \inter \bsigma}$. Necessarily, $m$ occurs in $A$ or
$C$, \emph{w.l.o.g.} assume it occurs in $C$. Then, $\Pi_\btau(m)$ has
the form $(2, t)$ where using Lemma \ref{lem:es_refl_caus}, necessarily $t$
is maximal (for $\leq_\btau$) in $x^\btau$. But then, since $x^\btau$ is
$+$-covered, $t$ is positive -- hence, $m$ is positive as well.
\end{proof}

So causally compatible $+$-covered $x^\bsigma
\in \confp{\bsigma}$ and $x^\btau \in \confp{\btau}$ are 
minimal, and their composition yields $x^\btau \odot
x^\bsigma \in \confp{\btau \odot \bsigma}$ $+$-covered. We 
prove the converse:

\begin{lem}\label{lem:pcov3_appendix}
Consider $x^\bsigma \in \conf{\bsigma}$ and $x^\btau \in \conf{\btau}$
minimal causally compatible.

If $x^\btau \odot x^\bsigma \in \confp{\btau \odot \bsigma}$ is
$+$-covered, so are $x^\bsigma \in \confp{\bsigma}$ and $x^\btau \in
\confp{\btau}$.
\end{lem}
\begin{proof}
Consider $s \in x^\bsigma$ maximal. Necessarily, there is a unique $m
\in x^\btau \inter x^\bsigma$ such that $\Pi_\bsigma(m) = (1, s)$.
Assume first that $m$ is maximal in $x^\btau \inter x^\bsigma$. As
$x^\btau \inter x^\bsigma = [x^\btau \odot x^\bsigma]_{\btau \inter
\bsigma}$, if $m$ is maximal it must be visible and maximal in $x^\btau
\odot x^\bsigma$. Therefore, it is positive by hypothesis, and $s$ is
positive.

Otherwise, assume $m$ is not maximal, so there is some $m
\imc_{\btau \inter \bsigma} n$. By Lemma \ref{lem:caus_int}, 
\[
\Pi_\bsigma(m) \imc_{\bsigma \parallel C} \Pi_\bsigma(n)
\qquad
\text{\emph{or}}
\qquad
\Pi_\tau(m) \imc_{A \parallel \btau} \Pi_\btau(n)\,.
\]

If this is the former, then there is $s \imc_\bsigma s'$ with $s' \in
x^\bsigma$, absurd by maximality of $s$. If this is the latter, then two
cases arise. If $m$ occurs in $A$, then $\Pi_\btau(m) = (1, a)$ and
$\Pi_\btau(n) = (1, a')$ with $a \imc_A a'$. Likewise, $\Pi_\bsigma(m) =
(1, s)$ and $\Pi_\bsigma(n) = (1, s')$. But then by Lemma
\ref{lem:es_refl_caus}, we must have $s <_\bsigma s'$ contradicting again the
maximality of $s$. Finally, if $m$ occurs in $B$, then $\Pi_\btau(m) =
(2, t)$ and $\Pi_\btau(n) = (2, t')$ with $t \imc_\btau t'$. We split
cases one last time, depending on the polarity of $t$ in $\btau$. If $t$
is negative, then $s$ is positive in $\bsigma$ as required. Otherwise,
by courtesy of $\btau$ we must have $\pr_\btau(t) \imc_{B\vdash C}
\pr_\btau(t')$. In particular, $t'$ must also occur in $B$ and we must
have $\Pi_\bsigma(n) = (1, s')$ for $s' \in x^\bsigma$, with moreover
$\pr_\bsigma(s) \imc_{A\vdash B} \pr_\bsigma(s')$. Therefore, again by
Lemma \ref{lem:es_refl_caus}, we must have $s <_\bsigma s'$ contradicting
the maximality of $s$. 

The symmetric reasoning shows that any $t \in x^\btau$ maximal
in $x^\btau$ is positive.
\end{proof}

We are almost in position to prove Proposition \ref{prop:comp_pcov} --
the only missing piece is \emph{uniqueness}:

\lempcovtrois*
\begin{proof}
We extend $\psi$ to all configurations and all symmetries,
and conclude via Lemma \ref{lem:map_ess_rep}.

Let $x \in \conf{\bsigma}$. Consider $x^+ \in \conf{\bsigma}$ minimal
such that $x^+ \subseteq^- x$. Necessarily, $x^+ \in \confp{\bsigma}$,
so we may take $\psi(x^+) \in \confp{\btau}$. Now, since $\psi$ is
compatible with display maps, 
\[
\pr_\btau(\psi(x^+)) \subseteq^- \pr_\bsigma(x)\,,
\]
therefore by receptivity and courtesy (see Lemma 3.13 from \cite{cg1}),
there is a unique $y \in \conf{\btau}$ such that $\psi(x^+) \subseteq^-
y$ and $\pr_\btau(y) = \pr_\bsigma(x)$; we set $\psi(x) = y$.

We must show this extended $\psi$ preserves inclusion; we show it
preserves covering, and distinguish the positive and negative cases.
First, consider configurations in $\conf{\bsigma}$: 
\[
x \longcov{s^+} y\,,
\]
which means $x^+ \subseteq^- z \longcov{s^+} y^+$. Now, by
hypothesis $\psi(x^+) \subseteq \psi(y^+)$, so $\psi(x^+) \subseteq
\psi(y)$. Moreover, this inclusion must contain exactly one positive
event, write it $t^+ \in \psi(y) \setminus \psi(x^+)$ -- 
necessarily, $\pr_\btau(t) = \pr_\bsigma(s)$. Now, notice
$\pr_\btau(t)$ is maximal in $\pr_\btau(\psi(y))$: indeed, 
\[
\pr_\btau(\psi(x)) = \pr_\bsigma(x) \longcov{\pr_\bsigma(s)}
\pr_\bsigma(y) = \pr_\btau(\psi(y))
\]
by hypothesis. So, by courtesy, $t$ is maximal in $\psi(y)$ as well.
So, we have $z = \psi(y) \setminus \{t\} \in \conf{\btau}$. Moreover, 
$\pr_\btau(z) = \pr_{\btau}(\psi(y)) \setminus \{\pr_\btau(t)\} =
\pr_\bsigma(x)$. So, $\psi(x^+) \subseteq^- z$ with $\pr_\btau(z) =
\pr_\bsigma(x)$, therefore $z = \psi(x)$; and $\psi(x) \cov \psi(y)$ as
required. Considering now a negative extension, \emph{i.e.}
\[
x \longcov{s^-} y\,,
\]
it follows that $x^+ = y^+$ and $\psi(x) \cov \psi(y)$ is immediate by
receptivity. Altogether we have
\[
\psi : \conf{\bsigma} \to \conf{\btau}
\]
compatible with display maps and preserving inclusion. Likewise we
construct $\psi^{-1} : \conf{\btau} \to \conf{\bsigma}$ preserving
inclusion from its action on $+$-covered configurations. That they are
inverses is immediate from $\psi$ being a bijection between
$+$-covered configurations, and receptivity.

Now, we consider the action of $\psi$ on symmetries. If $\theta \in
\tilde{\bsigma}$, consider $\theta^+ \in \conf{\bsigma}$ minimal s.t.
$\theta^+ \subseteq^- \theta$ -- recall that as a symmetry, $\theta$ is
an order-iso preserving polarities, so this is well-defined. Now,
$\theta^+ \in \tildep{\bsigma}$, so that we may take $\psi(\theta^+) \in
\tildep{\btau}$ as for configurations. Again, since $\psi$ is compatible
with display maps, we have 
\[
\pr_{\btau}(\psi(\theta^+)) \subseteq^- \pr_\bsigma(\theta)\,.
\]

By receptivity and courtesy of $\btau$, there are unique extensions of
$\dom(\psi(\theta^+))$ and $\cod(\psi(\theta^+))$ to
$\psi(\dom(\theta))$ and $\psi(\cod(\theta))$, projecting via
$\pr_\btau$ to $\dom(\pr_\bsigma(\theta))$ and
$\cod(\pr_\bsigma(\theta))$ respectively. We get  
$\psi(\theta^+) \subseteq^- \Omega \in \tilde{\btau}$
by iterating $\sim$-receptivity for $\btau$, 
such that $\pr_\btau(\Omega) = \pr_\bsigma(\theta)$ -- which
characterises $\Omega$ as the composition
\[
\psi(\dom(\theta)) \stackrel{\pr_\btau}{\simeq} \dom(\pr_\bsigma(\theta))
\stackrel{\pr_\bsigma(\theta)}{\sym_A} \cod(\pr_\bsigma(\theta))
\stackrel{\pr_\btau}{\simeq} \psi(\cod(\theta))\,,
\]
so a \emph{unique} extension of $\psi(\theta^+)$ projecting to 
$\pr_\bsigma(\theta)$  -- we fix $\psi(\theta) = \Omega$. 
Monotonicity is immediate from compatibility with $\dom$ and $\cod$ and
that $\psi$ preserves covering $\cov$ on configurations.
Likewise we extend $\psi^{-1}$ to all symmetries; that $\psi$ and
$\psi^{-1}$ are inverses follows as they preserve covering
and are inverses on configurations. 
\end{proof}

We may now conclude our final characterization of composition:

\charcompbis*
\begin{proof}
\emph{Existence.} Simply a restriction of the isomorphisms of
Proposition \ref{prop:char_conf_comp}. By Lemma \ref{lem:pcov1},
causally compatible $x^\bsigma \in \confp{\bsigma}$ and $x^\btau \in
\confp{\btau}$ are automatically minimal, and by Lemma \ref{lem:pcov2},
$x^\btau \odot x^\bsigma \in \confp{\btau \odot \bsigma}$ is
$+$-covered. Reciprocally, if $x^\btau \odot x^\bsigma \in \confp{\btau
\odot \bsigma}$ is $+$-covered, then by Lemma \ref{lem:pcov3_appendix}, so are
$x^\bsigma \in \confp{\bsigma}$ and $x^\btau \in \confp{\btau}$. Since
the isomorphism $(-\odot -)$ is compatible with $\dom$ and $\cod$
and symmetries are order-isomorphisms, there observations apply to
symmetries. 
\emph{Uniqueness.} By Lemma \ref{lem:pcov3_main}.
\end{proof}

\subsection{Charactering immediate causality}\label{app:charcaus}

\lemcharcaus*
\begin{proof}
If $m \imc_{\btau \inter \bsigma} m'$, then $x = [m]_{\btau
\inter \bsigma} \setminus \{m, m'\} \in \conf{\btau \inter \bsigma}$.
Then we have
\[
x 
\quad
=
\quad
 x^\btau \inter x^\bsigma 
\quad
\longcov{m} 
\quad
y^\btau \inter y^\bsigma
\quad
\longcov{m'}
\quad
z^\btau \inter z^\bsigma 
\quad
=
\quad
[m]_{\btau \inter \bsigma}
\]
in $\conf{\btau \inter \bsigma}$, inlining the order-isomorphism of
Proposition \ref{prop:main_interaction}. Let us focus first on 
\[
x^\btau \inter x^\bsigma 
\quad
\longcov{m} 
\quad
y^\btau \inter y^\bsigma\,.
\]

Since $(-\inter -)$ is an order-iso, this yields a covering in the
partial order of causally compatible pairs, ordered by pairwise
inclusion. By compatibility with display maps these inclusions add
exactly one event in $A$, in $B$, or in $C$. This yields three cases:
\emph{(a)} $m$ occurs in $A$, $x^\tau = y^\tau$, and $x^\bsigma \cov
y^\bsigma$ adds one event $s \in \ev{\bsigma}$; \emph{(b)} $m$ occurs in
$C$, $x^\bsigma = y^\bsigma$ and $x^\btau \cov y^\btau$ adds one event
$t \in \ev{\btau}$; or \emph{(c)} $m$ occurs in $B$, $x^\bsigma \cov
y^\bsigma$ adds one event $s$ and $x^\btau \cov y^\btau$ adds one event
$t$ with $\pr_\bsigma(s) = (2, b)$ and $\pr_\btau(t) = (1, b)$.

Now, back to considering $m$ and $m'$. If both occur in $C$, then we
have
\[
x^\btau \inter x^\bsigma 
\quad
\longcov{m}
\quad
y^\btau \inter x^\bsigma
\quad
\longcov{m'}
\quad
z^\btau \inter x^\bsigma\,,
\]
with $x^\btau \longcov{t} y^\btau \longcov{t'} z^\btau$. If $t<_\btau
t'$, then $t \imc_\btau t'$ and we are done. Otherwise, we also have
\[
x^\btau \, \longcov{t'} \, u^\btau \, \longcov{t} \, z^\btau\,,
\]
and as $t'$ occurs in $C$ we also have $x^\bsigma$ and $u^\btau$
causally compatible. Therefore
\[
x^\btau \inter x^\bsigma 
\quad
\longcov{n'}
\quad
u^\btau \inter x^\bsigma
\quad
\longcov{n}
\quad
z^\btau \inter x^\bsigma
\]
for some $n, n' \in \ev{\btau \inter \bsigma}$ since $-\inter -$ is an
order-isomorphism, so $\{n, n'\} = \{m, m'\}$. But since $u^\btau \inter
x^\bsigma \neq y^\btau \inter x^\bsigma$ we must have $n = m$ and $n' =
m'$, contradicting $m <_{\btau \inter \bsigma} m'$.

The case where $m$ and $m'$ both occur in $A$ is symmetric. If $m$
occurs in $A$ and $m'$ in $C$,
\[
x^\btau \inter x^\bsigma
\quad
\longcov{m}
\quad
x^\btau \inter z^\bsigma
\quad
\longcov{m'}
\quad
z^\btau \inter z^\bsigma
\]
with $x^\bsigma \longcov{s} z^\bsigma$ and $x^\btau \longcov{t}
z^\btau$. But then $x^\bsigma$ and $z^\btau$ are also causally
compatible, and
\[
x^\btau \inter x^\bsigma
\quad
\cov
\quad
z^\btau \inter x^\bsigma
\quad
\cov
\quad
z^\btau \inter z^\bsigma
\]
which as above contradicts $m<_{\btau \inter \bsigma} m'$. The case
where $m$ occurs in $C$ and $m'$ occurs in $A$ is symmetric. Now, assume
$m$ occurs in $A$ and $m'$ occurs in $B$. Then we have
\[
x^\btau \inter x^\bsigma 
\quad
\longcov{m}
\quad
x^\btau \inter y^\bsigma
\quad
\longcov{m'}
\quad
z^\btau \inter z^\bsigma
\]
where $x^\bsigma \longcov{s} y^\bsigma$, $x^\btau \longcov{t'} z^\btau$,
and $y^\bsigma \longcov{s'} z^\bsigma$. If $s <_\bsigma s'$, then $s
\imc_\bsigma s'$ and we are done. Otherwise, we also have $x^\bsigma
\cov u^\bsigma \cov z^\bsigma$, $u^\bsigma$ and $z^\btau$ are also
causally compatible, and
\[
x^\btau \inter x^\bsigma
\quad
\cov
\quad
z^\btau \inter u^\bsigma
\quad
\cov
\quad
z^\btau \inter z^\bsigma\,,
\]
contradiction. All cases with one event occurring in $B$ and the other
in $A$ or $C$ are symmetric. Finally, assume both $m$ and $m'$ occur in
$B$. In that case, we have
\[
x^\bsigma 
\,
\longcov{s}
\,
y^\bsigma
\,
\longcov{s'}
\,
z^\bsigma\,,
\qquad
\text{and}
\qquad
x^\btau
\,
\longcov{t}
\,
y^\tau
\,
\longcov{t'}
\,
z^\btau\,.
\]

If we have $s <_\bsigma s'$ then $s \imc_\bsigma s'$ and we are done,
and likewise for $t \imc_\btau t'$. Otherwise, 
$x^\bsigma \cov u^\bsigma \cov z^\bsigma$ and $x^\btau \cov u^\btau \cov
z^\btau$, and
$
x^\btau \inter x^\bsigma
\cov
u^\btau \inter u^\bsigma
\cov
z^\btau \inter z^\bsigma$,
contradiction.
\end{proof}

\begin{lem}\label{lem:imc_aux}
Consider $m, m' \in \ev{\btau \inter \bsigma}$ such that $m
\imc_{\btau\inter \bsigma} m'$.

Then, if $m_\bsigma <_\bsigma m'_\bsigma$ we have $m_\bsigma
\imc_\bsigma m'_\bsigma$, and likewise for $\btau$.
\end{lem}
\begin{proof}
If $m \imc_{\btau\inter \bsigma} m'$, then there is $x^\btau \inter
x^\bsigma, y^\btau \inter y^\bsigma$ and $z^\btau \inter z^\bsigma$ in
$\conf{\btau \inter \bsigma}$ such that
\[
x^\btau \inter x^\bsigma
\quad
\longcov{m}
\quad
y^\btau \inter y^\bsigma
\quad
\longcov{m'}
\quad
z^\btau \inter z^\bsigma\,,
\]
but if $m_\bsigma$ and $m'_\bsigma$ are defined then
$x^\bsigma \longcov{m_\bsigma} y^\bsigma \longcov{m'_\bsigma}
z^\bsigma$. If $m_\bsigma <_\bsigma m'_\bsigma$,  $m_\bsigma
\imc_\bsigma m'_\bsigma$.
\end{proof}

\subsection{The Bang Lemma}\label{app:bang}

Now, we prove the \emph{bang lemma} from AJM games
\cite{ajm}. Fix $A$ and $B$ two concrete arenas with $B$ pointed, and
$\bsigma : \oc A \vdash \oc B$ a causal strategy.

By \emph{receptive}, for each $\grey{i} \in \mathbb{N}$, there is a
unique $q_{\grey{i}} \in \min(\bsigma)$ such that
$\pr_{\bsigma}(q_{\grey{i}}) = (2, (\grey{i}, b))$ the unique minimal
move of $B$ of copy index $\grey{i}$. Let us write 
\[
\ev{\bsigma_{\grey{i}}} = \{m \in \ev{\bsigma} \mid q_{\grey{i}}
\leq_\bsigma m\}
\]
for the set of events of $\bsigma$ above $q_{\grey{i}}$. Since $\bsigma$
is pointed, it follows that for $\grey{i}, \grey{j} \in \mathbb{N}$
distinct, $\ev{\bsigma_{\grey{i}}}$ and $\ev{\bsigma_{\grey{j}}}$ are
disjoint. Likewise, since arenas are concrete, moves in immediate
conflict have the same predecessor -- it follows that if $m \in
\ev{\bsigma_{\grey{i}}}$ and $n \in \ev{\bsigma_{\grey{j}}}$, the
negative dependencies of $m$ and $n$ are compatible, hence $m$ and $n$
are compatible by \emph{determinism}. Therefore, we have
$\bsigma \iso \parallel_{\grey{i} \in \mathbb{N}} \bsigma_{\grey{i}}$
where $\bsigma_{\grey{i}}$ has a structure of ess directly imported from
$\bsigma$. Furthermore, $\bsigma_{\grey{i}} : \oc A \vdash B$
with the display map $\pr_{\bsigma_{\grey{i}}}$ defined in the obvious way.
The key argument is:

\begin{lem}\label{lem:bang_main}
For any $\grey{i}, \grey{j} \in \mathbb{N}$, we have $\bsigma_{\grey{i}}
\simstrat \bsigma_{\grey{j}}$.
\end{lem}
\begin{proof}
We exploit Lemma \ref{lem:mapification} and build a (necessarily
representable) iso between the domains of
configurations, compatible with symmetry. 
Consider $x_{\grey{i}} \in \conf{\bsigma_{\grey{i}}}$, with 
\[
\pr_{\bsigma}(x_{\grey{i}}) = x_A \parallel \{\grey{i}\} \times x_B\,.
\]

From Definition \ref{def:concrete-arena2}, $x_A
\parallel \{\grey{i}\} \times x_B \sym_{\oc A \vdash \oc B}^- x_A
\parallel \{\grey{j}\} \times x_B$ with the obvious symmetry
$\theta_{\grey{i}, \grey{j}}^-$. We must transport $x_{\grey{i}}$ along this
negative symmetry $\theta_{\grey{i}, \grey{j}}$. By Lemma B.4 from
\cite{cg2}, there are unique $x_{\grey{j}} \in \conf{\bsigma}$ and
$\psi : x_{\grey{i}} \sym_\bsigma x_{\grey{j}}$ s.t.
$\pr_\bsigma \psi = \theta^+ \circ \theta_{\grey{i}, \grey{j}}^-$
for some 
\[
\theta^+ : x_A \parallel \{\grey{j}\} \times x_B \sym_{\oc A \parallel
\oc B}^+ y_A \parallel \{\grey{j}\} \times y_B
\]
where we know that $\grey{j}$ is unchanged, because of condition
\emph{$+$-transparent} of Definition \ref{def:concrete-arena2}.
Therefore, $x_\grey{j} \in \conf{\bsigma_{\grey{j}}}$ as
required. Monotonicity of this operation and the fact that it is a
bijection between configurations follow from the \emph{uniqueness}
clause for Lemma B.4 of \cite{cg2}; compatibility with symmetry follows
from composition with the symmetry $\psi$.

This induces an isomorphism of ess $\varphi : \bsigma_{\grey{i}}
\simstrat \bsigma_{\grey{j}}$ which we must still check is a positive
isomorphism. But for $x_{\grey{i}} \in \conf{\bsigma_{\grey{i}}}$, the
symmetry $\theta^+$ above entails 
\[
\pr_{\bsigma_{\grey{i}}}(x_{\grey{i}}) = x_A \parallel x_B
\quad
\stackrel{\varphi^+}{\sym_{\oc A \vdash B}^+} 
\quad
y_A \parallel y_B = \pr_{\bsigma_{\grey{j}}}(x_{\grey{j}})
\]
ensuring that the triangle commutes up to positive symmetry as required.
\end{proof}

Next we lift a positive isomorphism on one copy index to the whole
strategy: 

\begin{lem}\label{lem:bang_aux}
Consider $A$ and $B$ concrete $-$-arenas with $B$ pointed, and $\bsigma,
\btau : \oc A \vdash \oc B$.

If $\bsigma_{\grey{0}} \simstrat \btau_{\grey{0}}$, then $\bsigma
\simstrat \btau$.
\end{lem}
\begin{proof}
By Lemma \ref{lem:bang_main}, for $\grey{i} \in \mathbb{N}$, 
$\bsigma_{\grey{i}} \simstrat \bsigma_{\grey{0}} \simstrat
\btau_{\grey{0}} \simstrat \btau_{\grey{i}}$, we conclude by 
parallel composition.
\end{proof}

We may finally deduce the bang lemma:

\bang*
\begin{proof}
By Proposition \ref{prop:comp_pcov} and a reasoning analogous to
Proposition \ref{prop:cc_neutral}, we have
$\der_B \odot \bsigma \simstrat \bsigma_{\grey{0}}$.
Therefore, by Lemma \ref{lem:bang_aux}, for any two $\bsigma, \btau :
\oc A \vdash \oc B$, if $\der_B \odot \bsigma \simstrat \der_B \odot
\btau$, then $\bsigma \simstrat \btau$. The lemma follows directly from
that and the Seely category laws.
\end{proof}

\subsection{Expansion of Meager Strategies}
\label{app:exp_meager}

Consider $\bsigma : A$ parallel innocent. We first observe that any move
$m \in \ev{\bsigma}$ is determined by $\rf{m}$, along with the copy
index of its negative dependencies. An
\textbf{exponential slice} for $y \in \conf{\rf{\bsigma}}$ normal is an
assignment of copy indices for all negative questions of $y$ -- or more
precisely,
\[
\alpha : y^{\Qu,-} \to \mathbb{N}\,,
\]
with $y^{\Qu,-}$ the negative questions of $y$. To any $x \in
\conf{\bsigma}$ we have associated $\rf{x} \in \conf{\rf{\bsigma}}$. To
complete $\rf{x}$, we also associate to $x$ an exponential slice for
$\rf{x}$:
\[
\begin{array}{rcrcl}
\slice(x) &:& (\rf{x})^{\Qu,-} &\to& \mathbb{N}\\
&& n^- &\mapsto& \ind(\pr_\bsigma(\theta_x^{-1}(n)))\,.
\end{array}
\]

We now show how to reconstruct events of $\bsigma$ from $\rf{\bsigma}$
and an exponential slice.

\begin{lem}\label{lem:meager_ev2}
Consider $A$ a concrete arena, $\bsigma : A$ a causal strategy,
and $y \in \conf{\rf{\bsigma}}$.

For any $\alpha : y^{\Qu,-} \to \mathbb{N}$, there is a unique $x \in
\conf{\bsigma}$ such that $y = \rf{x}$, $\alpha = \slice(x)$, and $x
\sym_\bsigma y$. Moreover, the symmetry $\varphi_y : y \sym_\bsigma x$
is unique.
\end{lem}
\begin{proof}
Same proof as for Lemma \ref{lem:meager_ev}, setting up indices
following $\alpha$ instead of $\grey{0}$.
\end{proof}

Just as Lemma \ref{lem:meager_ev}, for $\bsigma : A$ parallel innocent,
Lemma \ref{lem:meager_ev2} can be used to assemble an event $m \in
\ev{\rf{\bsigma}}$ and an exponential slice $\alpha : [m]^{\Qu,-} \to
\mathbb{N}$ into $n \in \ev{\bsigma}$ s.t. $\rf{n} = m$ and
$\slice([n]_\bsigma) = \alpha$. So together, Lemmas \ref{lem:meager_ev}
and \ref{lem:meager_ev2} establish a bijection between $\ev{\bsigma}$
and pairs $(m, \alpha)$ of $m \in \ev{\rf{\bsigma}}$ and an exponential
slice $\alpha : [m]_\bsigma^{\Qu,-} \to \mathbb{N}$. From this it seems
clear how to reconstruct $\bsigma$ from $\rf{\bsigma}$: first, we
reconstruct a partial order.

\begin{prop}\label{prop:iso-rf-exp}
For $A$ a concrete arena and $\bsigma : A^+$, we define a partial order
$\exp(\bsigma)$:
\[
\begin{array}{rcl}
\ev{\exp(\bsigma)} &=& \{(m, \alpha) \mid m \in
\ev{\bsigma}~\wedge~\alpha : [m]_\bsigma^{\Qu,-} \to \mathbb{N}\}\,\\
(m_1, \alpha_1) \leq_{\exp(\bsigma)} (m_2, \alpha_2) &\Leftrightarrow&
m_1 \leq_\bsigma m_2 ~\wedge~ \forall n^- \leq_\bsigma m_1,~\alpha_1(n)
= \alpha_2(n)\,.
\end{array}
\]

Then, for any $\bsigma : A$, $\bsigma$ and $\exp(\rf{\bsigma})$ are
isomorphic partial orders.
\end{prop}
\begin{proof}
Direct consequence of Lemmas \ref{lem:meager_ev} and
\ref{lem:meager_ev2}.
\end{proof}

It remains to
complete $\exp(\bsigma)$ into a causal strategy.
The display map is determined by a choice of copy index for positive
questions. Hence, for any $m^{\Qu,+} \in \ev{\bsigma}$, assume fixed
some
\[
f_m : \mathbb{N}^{[m]_\bsigma^{\Qu,-}} \to \mathbb{N}
\]
specifying, for each Player question, its copy index depending on
indices of Opponent questions in its causal dependency. For simplicity,
we assume this choice is globally injective, \emph{i.e.} $f_m, f_n$
have disjoint codomains for distinct $m, n \in \ev{\bsigma}$. 
The resulting strategy will not depend on the choice of the family
$(f_m)_{m \in \ev{\bsigma}^{\Qu,+}}$ up to positive isomorphism.

\begin{prop}
Consider $A$ a concrete arena, $\bsigma : A^+$ parallel innocent
causally well-bracketed.
We define a display map $\pr_{\exp(\bsigma)}$ for $\exp(\bsigma)$ by
induction, with image
\[
\begin{array}{rclcl}
\pr_{\exp(\bsigma)}(m^{\An}, \alpha) &=& a\,, &&\\
\pr_{\exp(\bsigma)}(m^{\Qu,-}, \alpha) &=& q &\text{s.t.}& \ind(q) =
\alpha(m)\,,\\
\pr_{\exp(\bsigma)}(m^{\Qu,+}, \alpha) &=& q &\text{s.t.}& \ind(q) =
f_m(\alpha)\,,
\end{array}
\]
where $a, q$ is the unique event of $A$ with label $\lbl(m)$,
predecessor $\pr_{\exp(\bsigma)}(\just(m, \alpha))$ with $\just(m,
\alpha)$ defined as $\just(m)$ with slice the restriction of $\alpha$,
label $\lbl(m)$; satisfying the additional constraint given. Further
components are:
\[
\begin{array}{rcl}
(m, \alpha) \conflict_{\exp(\bsigma)} (n, \beta)
&\Leftrightarrow& \text{$m \conflict_\bsigma n$, and $\alpha$ and
$\beta$ coincide on their common domain}\\
\theta \in \tilde{\exp(\bsigma)} &\Leftrightarrow& \text{$\theta : x \iso
y$ order-iso s.t. $\pi_1 = \pi_1 \circ \theta$.}
\end{array}
\]

Then, $\exp(\bsigma) : A$ is parallel innocent causally well-bracketed.
\end{prop}

The only difficulty is in handling conflict.  A positive iso $\varphi :
\bsigma_1 \simstrat \bsigma_2$ directly lifts to $\exp(\bsigma_1)
\simstrat \exp(\bsigma_2)$, applying $\varphi$ to moves while keeping
copy indices unchanged. Finally:

\begin{cor}
For $A$ a concrete arena, the operations $\rf{-}$ and $\exp(-)$ yield a
bijection between (positive isomorphism classes of) causally
well-bracketed strategies on $A$ and $A^+$. 
\end{cor}
\begin{proof}
For $\bsigma : A^+$, $\rf{\exp(\bsigma)} \simstrat \bsigma$ by
construction. For $\bsigma : A$, $\exp(\rf{\bsigma}) \simstrat \bsigma$
is obtained from the iso of Proposition \ref{prop:iso-rf-exp}
and a verification that this preserves further structure.
\end{proof}

\end{document}